\documentclass[
  sigconf,
  screen, ]{acmart}
\pdfoutput=1

\copyrightyear{2024}
\acmYear{2024}
\setcopyright{rightsretained}
\acmConference[LICS '24]{39th Annual ACM/IEEE Symposium on Logic in Computer Science}{July 8--11, 2024}{Tallinn, Estonia}
\acmBooktitle{39th Annual ACM/IEEE Symposium on Logic in Computer Science (LICS '24), July 8--11, 2024, Tallinn, Estonia}
\acmDOI{10.1145/3661814.3662135}
\acmISBN{979-8-4007-0660-8/24/07}

\usepackage{tikz,tikz-cd,relsize,hyperref}
\usepackage[nameinlink,capitalise]{cleveref}
\usepackage{enumitem}
\usepackage{mathtools,subcaption}
\usepackage[utf8]{inputenc} 

\usetikzlibrary{positioning,arrows}
        
\theoremstyle{definition}
\newtheorem{theorem}{Theorem}[section]
\newtheorem{corollary}[theorem]{Corollary}
\newtheorem{lemma}[theorem]{Lemma}

\newtheorem{construction}[theorem]{Construction}
\newtheorem{intuition}[theorem]{Intuition}
\newtheorem{definition}[theorem]{Definition}
\newtheorem{notation}[theorem]{Notation}
\newtheorem{example}[theorem]{Example}

\newtheorem{note}[theorem]{Note}

\newtheorem{fact}[theorem]{Fact}

\newtheorem{proposition}[theorem]{Proposition}

\crefname{theorem}{Theorem}{Theorems}
\crefname{construction}{Construction}{Constructions}
\crefname{corollary}{Corollary}{Corollaries}
\crefname{lemma}{Lemma}{Lemmas}
\crefname{hope}{Hope}{Hopes}
\crefname{belief}{Belief}{Beliefs}
\crefname{construction}{Construction}{Constructions}
\crefname{intuition}{Intuition}{Intuitions}
\crefname{definition}{Definition}{Definitions}
\crefname{notation}{Notation}{Notations}
\crefname{example}{Example}{Examples}
\crefname{warning}{Warning}{Warnings}
\crefname{claim}{Claim}{Claims}
\crefname{procedure}{Procedure}{Procedures}
\crefname{note}{Note}{Notes}
\crefname{question}{Question}{Questions}
\crefname{fact}{Fact}{Facts}
\crefname{assumption}{Assumption}{Assumptions}
\crefname{metatheorem}{Metatheorem}{Metatheorems}
\crefname{proposition}{Proposition}{Propositions}

\DeclareFontFamily{U}{min}{}
\DeclareFontShape{U}{min}{m}{n}{<-> udmj30}{}
\newcommand\yo{\!\text{\usefont{U}{min}{m}{n}\symbol{'207}}\!}

\newcommand\ofty{\,{:}\,}

\newcommand\colim{\operatorname{colim}}

\newcommand\img{\operatorname{im}}

\newcommand\psh{\operatorname{PSh}}
\newcommand\Aut{\operatorname{Aut}}

\newcommand{\dom}{\operatorname{dom}}
\newcommand{\cod}{\operatorname{cod}}

\newcommand\calA{\mathcal A}
\newcommand\calB{\mathcal B}

\newcommand\calF{\mathcal F}
\newcommand\calG{\mathcal G}
\newcommand\calH{\mathcal H}

\newcommand\calK{\mathcal K}
\newcommand\calL{\mathcal L}
\newcommand\calM{\mathcal M}
\newcommand\calN{\mathcal N}
\newcommand\calO{\mathcal O}
\newcommand\calP{\mathcal P}

\newcommand\calR{\mathcal R}
\newcommand\calS{\mathcal S}
\newcommand\calU{\mathcal U}

\newcommand\llbr[1]{\left\llbracket #1\right\rrbracket}

\newcommand\R{\mathbb R}
\newcommand\N{\mathbb N}
\newcommand\Q{\mathbb Q}
\newcommand\Z{\mathbb Z}
\newcommand\Ber{\operatorname{Ber}}

\newcommand\angled[1]{\langle #1\rangle}
\newcommand\frakA{\mathfrak A}
\newcommand\frakB{\mathfrak B}
\newcommand\frakC{\mathfrak C}
\newcommand\frakD{\mathfrak D}

\newcommand\frakF{\mathfrak F}

\newcommand\commsquare[8]{
\begin{tikzcd}[ampersand replacement=\&]
{#1} \arrow[d, "{#4}"'] \arrow[r, "{#2}"] \& {#3} \arrow[d, "{#5}"] \\
{#6} \arrow[r, "{#7}"']                \& {#8}
\end{tikzcd}
}

 \usetikzlibrary{patterns,patterns.meta}

\makeatletter
\newcommand{\labelword}[2]{\phantomsection
  #1\def\@currentlabel{\unexpanded{#1}}\label{#2}}
\makeatother

\ifdefined\noapp
\newcommand\appref[1]{the appendix}
\newcommand\annoyingappref[1]{Theorem~C.33}
\else
\newcommand\appref[1]{\cref{#1}}
\newcommand\annoyingappref[1]{\cref{#1}}
\fi

\newcommand\arxivversion{}
\ifdefined\arxivversion
\newcommand\dlc[1]{}
\else
\newcommand\dlc[1]{\input{#1}}
\fi

\newcommand\goodparagraph[1]{\hspace{0em}\\[-0.5em]\textbf{\textit{#1.~}}}

\newcommand\fremlin[1]{\citet[#1]{fremlin2000measure}}
\newcommand\fremlinft[3]{\citet[\href{https://www1.essex.ac.uk/maths/people/fremlin/chap#2.pdf\#page=#3}{#1}]{fremlin2000measure}}
\newcommand\fremlinf[3]{\cite[\href{https://www1.essex.ac.uk/maths/people/fremlin/chap#2.pdf\#page=#3}{#1}]{fremlin2000measure}}

\newcommand\lilac[2]{\cite[\href{https://johnm.li/lilac.pdf\#page=#2}{#1}]{li2023lilac}}
\newcommand\lilact[2]{\citet[\href{https://johnm.li/lilac.pdf\#page=#2}{#1}]{li2023lilac}}
\newcommand\nref[2]{\hyperref[#1]{#2}}

\newcommand\flavortext{}

\newcommand\powerset{\mathscr P}
\newcommand\trueprop{\mathsf{True}}
\newcommand\falseprop{\mathsf{False}}
\newcommand\vtrue{\top}
\newcommand\vfalse{\bot}

\newcommand\hexnum[1]{\texttt{0x#1}}

\newcommand\funcat[2]{[#1;#2]}
\newcommand\catelts{{\nref{def:catelts}{\mathrm{El}}}}
\newcommand\Shjat{\operatorname{Sh}_{{\rm atomic}}}
\newcommand\FinInj{{\nref{def:FinInj}{\mathbf{Inj}_{<\omega}}}}
\newcommand\FinInjop{{\nref{def:FinInj}{\mathbf{Inj}_{<\omega}\op}}}

\newcommand\tgsets[1]{{#1{\nref{app:def:gsets}{\,\Set}}}}
\newcommand\op{^{\rm op}}
\newcommand\monunit{\mathrm{I}}
\newcommand\scmfst{\mathop{\nref{app:def:semicartesian-monoidal-category}{\mathrm{fst}}}}
\newcommand\scmsnd{\mathop{\nref{app:def:semicartesian-monoidal-category}{\mathrm{snd}}}}
\newcommand\Set{\mathbf{Set}}
\newcommand\Core{\mathrm{Core}}

\newcommand\measalg{\operatorname{\nref{app:cons:measalg}{alg}}}
\newcommand\lmafp{\mathbin{\hat\otimes}}
\newcommand\atom[1]{\operatorname{atom}(#1)}
\newcommand\aseq{=_{\rm a.s.}}

\newcommand\shotimes{\mathop{\otimes}}
\newcommand\neset[2]{\mathrm{D}_{#1,#2}}
\newcommand\negligibles{\mathrm{negligibles}}
\newcommand\texforgetmu{\mathrm{U}}
\newcommand\forgetmu{{\nref{def:forgetmu}{\texforgetmu}}}
\newcommand\dforgetmu{{\nref{def:dforgetmu}{\mathrm{U}_{\rm d}}}}
\newcommand\pbmeas[2]{{\nref{app:cons:pbmeas}{#1}^{-1}}#2}
\newcommand\texhcube{{{{\mathbb I}^\omega}}}
\newcommand\hcube{{\nref{def:hcube}{\texhcube}}}

\newcommand\texagoodname{absolutely continuous set}
\newcommand\texAgoodname{Absolutely continuous set}
\newcommand\agoodname{\texagoodname}
\newcommand\aagoodname{an \texagoodname}
\newcommand\Agoodname{\texAgoodname}
\newcommand\Agoodnamenolink{\texAgoodname}
\newcommand\agoodnamenolink{\texagoodname}
\newcommand\dagoodname{discrete \texagoodname{}}
\newcommand\agoodnamecat{\mathbf{Set}^\ll}
\newcommand\sheafcatadj{enhanced measurable}
\newcommand\Sheafcatadj{Enhanced measurable}
\newcommand\texgoodsheaf{\sheafcatadj{} sheaf}
\newcommand\texGoodsheaf{\Sheafcatadj{} sheaf}
\newcommand\goodsheaf{\texgoodsheaf}
\newcommand\Goodsheaf{\texGoodsheaf}
\newcommand\texgoodsheaves{\sheafcatadj{} sheaves}
\newcommand\texGoodsheaves{\Sheafcatadj{} sheaves}
\newcommand\goodsheaves{\texgoodsheaves}
\newcommand\Goodsheaves{\texGoodsheaves}
\newcommand\Goodsheavesnolink{\Sheafcatadj{} sheaves}
\newcommand\goodsheavesnolink{\sheafcatadj{} sheaves}
\newcommand\goodsheafcat{\mathbf{EMS}}

\newcommand\dgoodsheaves{discrete \sheafcatadj{} sheaves}
\newcommand\Dgoodsheaves{Discrete \sheafcatadj{} sheaves}

\newcommand\texStdProb{\mathbf{Prob}_{\rm std}}
\newcommand\StdProb{{\nref{def:stdprob}{\texStdProb}}}

\newcommand\StdProbAlg{\nref{def:StdProbAlg}{\mathbf{ProbAlg}_{\rm std}}}
\newcommand\StdProbAlgop{\nref{def:StdProbAlg}{\mathbf{ProbAlg}_{\rm std}\op}}
\newcommand\texStdMble{\mathbf{EMS}_{\rm std}}
\newcommand\StdMble{{\nref{def:mblecat}{\texStdMble}}}
\newcommand\StdMbleop{{\nref{def:mblecat}{\texStdMble\op}}}
\newcommand\StdMbleAlg{{\nref{def:StdMbleAlg}{\mathbf{MbleAlg}_{\rm std}}}}
\newcommand\StdMbleAlgop{{\nref{def:StdMbleAlg}{\mathbf{MbleAlg}_{\rm std}\op}}}
\newcommand\CountProb{{\nref{def:CountProb}{\mathbf{Prob}_{\le\omega}^+}}}
\newcommand\CSur{{\nref{def:CSur}{\mathbf{Surj}_{\le\omega}}}}
\newcommand\CSurop{{\nref{def:CSur}{\mathbf{Surj}_{\le\omega}\op}}}

\newcommand\Store{{\nref{def:store}{\mathrm{S}}}}
\newcommand\Prop{{\nref{prop:ext-ren-inv-prop}{\mathrm{Prop}}}}
\newcommand\Loc{{\nref{prop:warp-in-the-atomic-sheaves}{\mathrm{Loc}}}}

\newcommand\Sch{{\nref{def:Sch}{\mathbf{Sch}}}}

\newcommand\inl{\mathrm{inl}}
\newcommand\inr{\mathrm{inr}}
\newcommand\enc{\mathrm{enc}}
\newcommand\dec{\mathrm{dec}}

\newcommand\nomgrp{{\nref{def:nomgrp}{S_\omega}}}

\newcommand\Nom{{\nref{def:Nom}{\mathbf{Nom}}}}
\newcommand\nomstore{\mathbin{\xrightharpoonup{\rm fin}}}
\newcommand\nomLoc{{\nref{prop:warp-in-the-gsets}{\overline{\mathrm{Loc}}}}}
\newcommand\nomProp{{\nref{def:nomProp}{\overline{\mathrm{Prop}}}}}
\newcommand\nomStore{{\nref{def:nomStore}{\overline{\mathrm{S}}}}}

\newcommand\sepcon{\mathbin{*}}
\newcommand\sepimp{\mathbin{-\mkern-2mu*}}

\newcommand\dauti{\nref{def:dauti}{\Aut[0,1]}}
\newcommand\nomDpspcsprelim{{\nref{def:nomDpspcsprelim}{\overline{\mathbb{P}}}}}
\newcommand\nomDpspcs{{\nref{def:nomDpspcs}{\overline{\mathbb{P}}}}}
\newcommand\nomDRVprelim{{\nref{def:nomDRVprelim}{\overline{\mathrm{RV}}}}}
\newcommand\nomDRV{{\nref{prop:dwarp-in-the-gsets}{\overline{\mathrm{RV}}}}}

\newcommand\discsubscript{_{\rm d}}
\newcommand\dgsets{{\nref{def:dgsets}{\agoodnamecat\discsubscript}}}

\newcommand\Dpspcsprelim{{\nref{def:Dpspcsprelim}{\mathbb{P}}}}
\newcommand\Dpspcs{{\nref{def:Dpspcs}{\mathbb{P}}}}
\newcommand\DRVprelim{{\nref{def:DRVprelim}{\mathrm{RV}}}}
\newcommand\DRV{{\nref{prop:dwarp-in-the-atomic-sheaves}{\mathrm{RV}}}}

\newcommand\dsheafcat{{\nref{def:dsheafcat}{\goodsheafcat\discsubscript}}}

\newcommand\RV{\mathrm{RV}}
\newcommand\RVshf{{\nref{lem:rv-mble-sheaf}{\RV}}}

\newcommand\pcmjoin{{\nref{app:def:pspcs-tensor-join}{\mathrm{join}}}}
\newcommand\pcmunit{{\nref{app:def:pspcs-join-unit}{\mathrm{emp}}}}

\newcommand\texpspcs{\mathbb P}
\newcommand\pspcs{{\nref{def:mble-pspcs}{\texpspcs}}}
\newcommand\gProp{{\mathrm{Prop}}}

\newcommand\sheafcat{{\nref{def:sheafcat}{\goodsheafcat}}}

\newcommand\gpcmjoin{{\nref{app:def:gpspcs-join}{\overline{\mathrm{join}}}}}
\newcommand\gpcmunit{{\nref{app:def:emp-gpspcs}{\overline{\mathrm{emp}}}}}
\newcommand\gkrmorder{\mathbin{\nref{app:def:pspc-ord}{\overline{\sqsubseteq}}}}
\newcommand\gkrmorderop{\mathbin{\nref{app:def:pspc-ord}{\overline{\sqsupseteq}}}}
\newcommand\texauti{\mathbf{G}^{\ll}}
\newcommand\auti{{\nref{def:auti}{\texauti}}}
\newcommand\preauti{{{\Aut_\StdMble\hcube}}}
\newcommand\Fix{\operatorname{Fix}}
\newcommand\nomFix{\operatorname{\nref{def:nomFix}{Fix}}}
\newcommand\Stab{\operatorname{Stab}}
\newcommand\texgpspcs{\overline{\mathbb P}}
\newcommand\gpspcs{{\nref{def:gpspcs}{\texgpspcs}}}
\newcommand\texgRV{\overline{\operatorname{RV}}}
\newcommand\gRV{{\nref{def:gRV}{\texgRV}}}

\newcommand\cinfty{{{\rm c}_\infty}}
\newcommand\subgrpcat{{\bf Fix}}

\newcommand\cgsets{{\nref{def:cgsets}{\agoodnamecat}}}

\makeatletter
\newcommand*\pdot{\mathpalette\pdot@{.8}}
\newcommand*\pdot@[2]{\mathbin{\vcenter{\hbox{\scalebox{#2}{$\m@th#1\bullet$}}}}}
\makeatother
\newcommand\krmorder{\mathrel{\nref{app:def:mble-pspcs-ordering}{\sqsubseteq}}}
\newcommand\Dicom{\mathbin{\nref{def:Dicom}{\overline\pdot}}}
\newcommand\Dorder{\mathbin{\nref{def:Dorder}{\overline\sqsubseteq}}}

\author{John M. Li}
\affiliation{\institution{Northeastern University}
  \city{Boston, MA}
  \country{U.S.A.}}
\email{li.john@northeastern.edu}

\author{Jon Aytac}
\affiliation{\institution{Sandia National Laboratories}
  \city{Livermore, CA}
  \country{U.S.A.}}
\email{jmaytac@sandia.gov}

\author{Philip Johnson-Freyd}
\affiliation{\institution{Sandia National Laboratories}
  \city{Livermore, CA}
  \country{U.S.A.}}
\email{pajohn@sandia.gov}

\author{Amal Ahmed}
\affiliation{\institution{Northeastern University}
  \city{Boston, MA}
  \country{U.S.A.}}
\email{amal@ccs.neu.edu}

\author{Steven Holtzen}
\affiliation{\institution{Northeastern University}
  \city{Boston, MA}
  \country{U.S.A.}}
\email{s.holtzen@northeastern.edu}

\title{A Nominal Approach to Probabilistic Separation Logic}

\begin{abstract}

  Currently, there is a gap between the tools used
  by probability theorists and those used in formal reasoning about
  probabilistic programs. On the one hand, a probability theorist decomposes
  probabilistic state along the simple and natural product of probability
  spaces. On the other hand, recently developed probabilistic separation logics
  decompose state via relatively unfamiliar
  measure-theoretic constructions for computing unions of
  sigma-algebras and probability measures. We bridge the gap between these two
  perspectives by showing that these two methods of decomposition are equivalent
  up to a suitable equivalence of categories. Our main result is a probabilistic analog of the
  classic equivalence between the category of nominal sets and the Schanuel
  topos. Through this equivalence, we validate design decisions in prior work
  on probabilistic separation logic and create new connections to
  nominal-set-like models of probability.

\end{abstract}

\begin{document}
  \maketitle

\section{Introduction}

Separation logic~\cite{reynolds2002separation}, now a standard tool for reasoning about
programs with shared mutable state, grew out of Reynolds's 
Syntactic Control of Interference~\cite{reynolds1978syntactic} --- a
substructural system for controlling the interaction of imperative program fragments.
The basic ingredients
for today's interpretations of separation logic connectives,
present in the original model of Syntactic Control of Interference~\cite{o1993model},
can be seen as living in a category of functors known as the Schanuel topos,
with noninterference defined in terms of the
coproduct of finite sets.
Over the years, this model has been reformulated to suit the needs of
formal reasoning about imperative programs: modern models of separation logic
live not in the Schanuel topos, but in categories more like $\Set$,
and separation is interpreted not by coproduct,
but by algebraic structures such as
partial commutative monoids (PCMs)~\cite{galmiche2005semantics,biering2004logic}.
In particular, the now-standard model of separation logic
in which separating conjunction splits
stores into disjoint pieces is defined in terms of
the partial function $\uplus$
sending a pair of disjoint stores to their union,
giving rise to a PCM of stores.
This shift in perspective
is justified by a classic equivalence of categories:
\begin{fact} \label{fact:folklore}
The Schanuel topos $\Sch$ is equivalent to the category $\Nom$
of nominal sets,
and the original coproduct-based model of separation
in $\Sch$
corresponds to the standard union-based model in $\Nom$
across this equivalence.\footnote{For a good reference documenting this equivalence, see
\citet[\S6.3]{pitts2013nominal}.}
\end{fact}

Today, there is a pressing need for syntactic control of \emph{probabilistic}
interference --- that is, for establishing the \emph{probabilistic independence}
of program fragments. In response to this need, recent work has
developed a number of probabilistic separation
logics~\cite{barthe2019probabilistic,bao2021bunched,bao2022separation,li2023lilac},
whose semantic models are given by PCMs made of probability-theoretic objects.
Lilac~\cite{li2023lilac} is a separation logic whose PCM-based model
is particularly well-behaved: its notion of separation 
coincides with probabilistic independence~\cite[Lemma 2.5]{li2023lilac},
and yields a frame rule identical to the standard one for store-based separation logics.

However, Lilac's PCM model does not match a probability theorist's intuition.
One expects separation to be interpreted via a standard product
of probability spaces~\citep{kallenberg1997foundations}, but Lilac interprets separation
using \emph{independent combination}: a partial binary operation
on probability spaces constructed out of low-level set-theoretic
operations on $\sigma$-algebras.
Moreover, Lilac's model fixes up front an unconventional sample
space --- the space $[0,1]^\omega$
of infinite streams of real numbers in the interval, known as the Hilbert
cube --- and the soundness of Lilac's proof rules
depends on various properties specific to it.
These contrasts between Lilac's model and textbook probability
raise a question:
\emph{how do we know Lilac provides a good notion of separation for
probabilistic separation logic?}

We answer this question by showing Lilac's seemingly non-standard 
independent combination is in fact \emph{equivalent} to a probability theorist's 
product-based intuition of state decomposition.
Our result is a probabilistic analog of \cref{fact:folklore}:
just as the coproduct model of separation
corresponds to the now-standard model based on $\uplus$
across an equivalence between the Schanuel topos and $\Nom$,
the probability theorist's intuitive product-based model
of independence corresponds to Lilac's independent-combination-based model
across an equivalence between
a category of \emph{\goodsheaves}
and a category of \emph{\agoodname{s}} (\cref{thm:krm-corresp}).
Our contributions are as follows:
\begin{itemize}[leftmargin=*]
  \item 
    We introduce \emph{\agoodname{s}}:
    just as nominal sets are sets equipped with an action
    by permutations of names,
    \agoodname{s}
    are sets equipped with a continuous action 
    by measurable permutations
    of the Hilbert cube.
  \item
    We prove analogs of
    the equivalence $\Sch\simeq\Nom$
    for both discrete and continuous probability~(\cref{thm:discrete,thm:cont}).
    In particular, we show that
    the category $\cgsets$ of \agoodname{s}
    is equivalent to a topos $\sheafcat$ of \emph{\goodsheaves{}}:
    a probabilistic analog of the Schanuel topos.
  \item
    We show that $\cgsets$ provides a natural background category
    for a fragment of Lilac. 
    \cref{thm:krm-corresp} then shows that, by transporting across
    the equivalence $\cgsets\simeq\sheafcat$,
    Lilac's model corresponds
    to a 
    model
    in $\sheafcat$
    where separation
    arises naturally from product of probability spaces via Day convolution~\cite{dongol2016convolution,o1995syntactic,biering2004logic}.
\end{itemize}

\section{The Nominal Situation} \label{sec:nom-situation}

Our main result
is a probabilistic analog
of \cref{fact:folklore} (\cref{thm:krm-corresp}).
To set the stage, we first make \cref{fact:folklore}
a precise mathematical statement (\cref{prop:sch-nom-krm}).
We devote this section to describing the necessary pieces in 
this comfortable setting;
the material in this section is standard, but we will deviate
occasionally from the usual presentation in order to
focus on the aspects that are most relevant
to our eventual probabilistic counterpart.

At its core, \cref{fact:folklore} states that two
distinct approaches to modelling store-separation
are equivalent.
To illustrate this fact we will study
a tiny separation logic
consisting of propositions $P,Q$ about integer-valued stores:
\begin{align}
P,Q ::= x\mapsto i \mid \trueprop \mid P\sepcon Q.
\tag{\textsc{TinySep}}
\label{eq:tinysep}
\end{align}
\ref{eq:tinysep} propositions are well-formed according to a judgment $\Gamma\vdash P$
defined as usual:
a context $\Gamma$ is a set of logical variables $x$,
and $\Gamma \vdash P$
if $\Gamma$ contains the variables used in $P$.
\cref{fact:folklore} asserts the equivalence of
two different models for \ref{eq:tinysep}:

\newcommand\modelsh[2]{\llbr{#1}_1^{#2}}
\newcommand\modelgset[1]{\llbr{#1}_2}
\goodparagraph{Model~1: separation as coproduct}
  In this model, a store consists of two components: (1) a \emph{shape} $L$
  given as a finite set of available locations (i.e., memory addresses), and (2)
  a \emph{valuation} $s : L\rightharpoonup \Z$, a partial function assigning
  values to a subset of the shape. An example is shown in Figure~\ref{fig:coproduct};
  the store $s$ has shape $\{\hexnum0,\hexnum1, \hexnum2\}$, and the 
  valuation maps $\hexnum0 \mapsto 8$ and so on.
  Under this model, the meaning of a proposition
  depends on the shape $L$:
  the interpretation of a proposition $\Gamma \vdash P$
  has form $\modelsh{\Gamma\vdash P}L: (\Gamma\to L)\to \calP(L\rightharpoonup\Z)$,
  associating each substitution $\gamma : \Gamma\to L$ to
  the set $\modelsh{\Gamma\vdash P}L(\gamma)$ of $L$-shaped valuations
  satisfying $P$.

  Under this interpretation, we define
  $s\in\modelsh{\trueprop}L(\gamma)$ always
  and $s\in\modelsh{x\mapsto i}L(\gamma)$
  if and only if $s(\gamma(x)) = i$.
  Separating conjunction is defined via
  the coproduct of store shapes:
  $P_1\sepcon P_2$
  holds of an $L$-shaped valuation $s$
  if and only if there are valuations
  $s_1$ of shape $L_1$ and $s_2$ of shape $L_2$,
  and an injective function $i : L_1 + L_2\hookrightarrow L$
  embedding the coproduct $L_1+L_2$ into $L$ such that
  $s_1$ satisfies $P_1$ and $s_2$ satisfies $P_2$
  and $s_1,s_2$ embed into $s$ along $i$.
  This situation 
  is visualized in Figure~\ref{fig:coproduct}. For example,
  \begin{align*}
  s \in \modelsh{(x\mapsto 8)\sepcon(y\mapsto 3)}{
    \{\hexnum0,\hexnum1,\hexnum2\}}(\{x\mapsto \hexnum0,
    y\mapsto \hexnum1\})
  \end{align*}
  is witnessed by setting
  $s_1$ to the $\{\hexnum0\}$-shaped valuation
  $\{\hexnum0 \mapsto 8\}$
  and $s_2$ to the $\{\hexnum0\}$-shaped valuation
  $\{ \hexnum0 \mapsto 3\}$
  and $i : \{\hexnum0\}+\{\hexnum1\}\hookrightarrow
  \{\hexnum0,\hexnum1\}$
  to the injection defined by 
  $i(\inl(\hexnum0)) = \hexnum0$
  and $i(\inr(\hexnum0)) = \hexnum1$, where 
  $\inl : L_1 \to L_1 + L_2$ and $\inr : L_2 \to L_1 + L_2$ are the coproduct injections.

  \begin{figure}
    \centering
  \begin{subfigure}[b]{0.4\linewidth}
    \centering
  \begin{tikzpicture}
    \matrix [column sep=1mm]
    {
                            & \node{$s$};      &[0.4cm] &  \node {$s_1$}; \\
      \node{\hexnum0};      & \node[draw] (s0) {8}; &      &  \node[draw] (s10) {8}; & \node{\hexnum0};\\
      \node{\hexnum1};      & \node[draw] (s1) {3}; &      &  \\
      \node{\hexnum2};      & \node[draw] (s2) {4}; &      &  \node[draw] (s21) {3}; & \node{\hexnum0};\\
                            &                  &      &  \node {$s_2$}; \\
    };
    \draw[->] (s10) -- (s0) node[fill=white,inner sep=2pt,midway] {$i$};
    \draw[->] (s21) -- (s1) node[fill=white,inner sep=2pt,midway] {$i$};
  \end{tikzpicture}
  \caption{Model~1: coproduct.}
  \label{fig:coproduct}
  \end{subfigure}
  \hfill
  \begin{subfigure}[b]{0.5\linewidth}
    \centering
  \begin{tikzpicture}
    \matrix [column sep=1mm]
    {
                           & \node {$s_1$};    &[0.1cm]               &[0.1cm] \node {$s_2$}; &[0.1cm]              & [0.1cm] \node {$s$};\\
      \node{\hexnum0};     & \node[draw] {8};  &                      & \node[draw] {};       &                     &  \node[draw] {8};      \\
      \node{\hexnum1};     & \node[draw] {};   & \node {$\biguplus$}; & \node[draw] {3};      & \node{$\subseteq$}; &  \node[draw] {3};     \\
      \node{\hexnum2};     & \node[draw] {};   &                      & \node[draw] {};      &                     &  \node[draw] {4};     \\
      \node{$\vdots$};     & \node{$\vdots$};  &                      & \node{$\vdots$};      &                     &  \node{$\vdots$}; \\
    };
  \end{tikzpicture}
  \caption{Model~2: union.}
  \label{fig:disjoint-union}
  \end{subfigure}
  \caption{Visualizing separation in Model~1 and Model~2.}
  \label{fig:schan-nom}
  \end{figure}
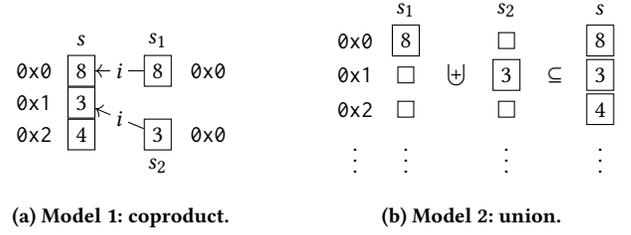

\goodparagraph{Model~2: separation as union}
In this model, one fixes upfront
a ``universal store shape''
into which all store shapes can be embedded.
Any countably-infinite set will do;
we choose the natural numbers $\N$.
A store is a partial
function $s : \N\nomstore\Z$
defined on finitely many values of its domain,
and a proposition $\Gamma\vdash P$
denotes a function
$\modelgset{\Gamma\vdash P} : (\Gamma\to\N)\to\calP(\N\nomstore\Z)$.
The interpretations of $\trueprop$ and $x\mapsto i$
are as in the shape-indexed model:
$s\in \modelgset\trueprop(\gamma)$
always and $s\in\modelgset{x\mapsto i}(\gamma)$
if and only if $s(\gamma(x)) = i$.
Separating conjunction is defined via union
of stores: a store $s$ is in $\modelgset{P_1\sepcon P_2}(\gamma)$ 
if and only if there exist disjoint stores $s_1$ and $s_2$ 
with $s_1 \uplus s_2 \subseteq s$ such that
$s_1$ is in $\modelgset{P_1}(\gamma)$
and $s_2$ is in $\modelgset{P_2}(\gamma)$.
Figure~\ref{fig:disjoint-union} visualizes an example:
$s\in\modelgset{(x\mapsto 8)\sepcon(y\mapsto 3)}\{x\mapsto \hexnum0,y\mapsto \hexnum1\}$
holds because $s_1$ and $s_2$
have a union contained in $s$
and $s_1$ satisfies $x \mapsto 8$
and $s_2$ satisfies $y \mapsto 3$.

\goodparagraph{Relating the two models}
Model~1 and Model~2 are equivalent by
\cref{fact:folklore}.
The equivalence is based on the following idea.
Every store shape $L$ can be 
encoded as a finite subset of $\N$
via a suitable pair of functions
$\enc_L: L\to\N$
and 
$\dec_L : \N\rightharpoonup L$.
Choosing an arbitrary such pair $(\enc_L,\dec_L)$ for every $L$
allows translating Model~1 into Model~2
in a bijective way:
an $L$-shaped store $s : L\rightharpoonup \Z$
corresponds to a finite partial function $s\circ \dec_L : \N\nomstore\Z$,
and a Model-1-substitution $\gamma : \Gamma\to L$
corresponds to a Model-2-substitution $\enc_L \circ \gamma : \Gamma\to \N$.
Via these translations, it holds that
 $s\in \modelsh{\Gamma\vdash P}L(\gamma)$ if and only if 
$s\circ\dec_L\in \modelgset{\Gamma\vdash P}(\enc_L\circ \gamma)$
for all $L$-shaped valuations $s$, propositions $\Gamma\vdash P$,
and substitutions $\gamma : \Gamma\to L$,
so both models induce the same notion of store-satisfaction.

This equivalence should seem plausible enough
given how tiny \ref{eq:tinysep} is.
What is remarkable about \cref{fact:folklore} is that this equivalence
continues to hold
when the interpretations $\modelsh-{(-)}$ and $\modelgset-$
are extended
to include all the usual features of separation logic,
including separating implication $\sepimp$,
the intuitionistic connectives $\land,\lor,\to,\falseprop$,
quantification
at both first-order and higher type, quantification over propositions,
predicates defined by structural recursion, and so on.
In short, the semantic domains of Model~1 are equivalent in expressive power
to those of Model~2.

Rather than laboriously verifying one by one
that the standard interpretations of
each of these features coincide,
\cref{fact:folklore} establishes a general result.
The key is to place Model~1 and Model~2
into the context of suitable categories
that bring out their essential structure.
Model~1
naturally lives in a category $\Sch$
called the \emph{Schanuel topos}:
the interpretation $\modelsh{\Gamma\vdash P}{(-)}$
of a proposition $P$ defines a $\Sch$-morphism
from a $\Sch$-object representing $\Gamma$-substitutions
to a $\Sch$-object representing store-predicates.
Model~2 naturally lives
in the category $\Nom$
of \emph{nominal sets}~\cite{pitts2013nominal}:
the interpretation
$\modelgset{\Gamma\vdash P}$
defines a $\Nom$-morphism from a nominal set
of $\Gamma$-substitutions to a nominal set of store-predicates.
Having placed Model~1 and Model~2
into suitable
background categories,
\cref{fact:folklore} follows from
a classic theorem:
the categories $\Sch$ and $\Nom$
are known to be equivalent~\cite[\S 6.3]{pitts2013nominal},
and inspecting the proof of this equivalence
shows that the functor $\Sch\to\Nom$ witnessing it
sends Model~1 to Model~2 via the construction
involving $(\enc_L,\dec_L)$.

The rest of this section is devoted to filling
in the details of this category-theoretic setup.
First we will describe how Model~1 lives in $\Sch$
and Model~2 lives in $\Nom$.
Then we will highlight the essential properties
of this setup that make the equivalence $\Sch\simeq\Nom$
possible,
and how Model~1 and Model~2
are instances of the same structure
across this equivalence;
\cref{thm:krm-corresp}
relies crucially on
identifying analogous properties in the probabilistic setting.

\subsection{Model~1 in the Schanuel topos} \label{sec:schanuel-topos}
\newcommand\Extension{\nref{def:Extension}{Extension}}
\newcommand\Renaming{\nref{def:Renaming}{Renaming}}
\newcommand\Restriction{\nref{def:Restriction}{Restriction}}

In this section we describe how
Model~1
of \cref{sec:nom-situation}
naturally lives in the Schanuel topos $\Sch$.
The benefit of this
is that it makes the invariants
maintained by 
$\modelsh-{(-)}$
explicit: the category $\Sch$
is such that all
constructions that make categorical sense --- i.e.,
are well-defined as objects and morphisms of $\Sch$ ---
are forced to preserve all invariants.
The invariants in this case
are the following principles
one intuitively expects to hold
when reasoning about stores:
\begin{itemize}[leftmargin=*]
  \item \emph{\labelword{Extension}{def:Extension}}:
    propositions should continue to hold
    when new locations are introduced (such as
    when declaring a local variable or allocating a reference).
    More precisely, if
    $s\in\modelsh{\Gamma\vdash P}L(\gamma)$
    for some $L$-shaped valuation $s$
    and substitution $\gamma : \Gamma\to L$,
    and if $L$ is a subset of some extended set of locations $L'$,
    then it should hold that
    $s\in\modelsh{\Gamma\vdash P}{L'}(\gamma)$,
    where we have implicitly coerced $s$
    into an $L'$-shaped valuation
    and $\gamma$ into an $L'$-shaped substitution $\Gamma\to L'$
    along the inclusion $L\subseteq L'$.
    \footnote{This makes the separation logic affine
    rather than linear; we will restrict our attention to affine separation logics in this paper,
    as Lilac is affine and our main goal is to obtain models for it.}
  \item \emph{\labelword{Renaming}{def:Renaming}}: propositions should be 
    stable under renaming of locations.
    More precisely, if $s\in\modelsh{\Gamma\vdash P}L(\gamma)$
    for some $L$-shaped valuation $s$
    and substitution $\gamma : \Gamma\to L$,
    and if $f$ is a bijective function $L\to L'$,
    then it should hold that
    $s\circ f^{-1}\in\modelsh{\Gamma\vdash P}{L'}(f\circ \gamma)$.
  \item \emph{\labelword{Restriction}{def:Restriction}}: the truth of a proposition
    should not depend on any unused locations.
    For example, suppose a proposition $P$ 
    holds of the $\{\ell_1,\ell_2\}$-shaped valuation $\{\ell_1\mapsto 1\}$,
    which does not use the location $\ell_2$.
    Then $P$ should also hold of
    $\{\ell_1\mapsto 1\}$ considered as an $\{\ell_1\}$-shaped valuation.
\end{itemize}
As basic principles of
store-based reasoning, it is crucial
that these invariants are preserved by the
basic separation logic connectives:
if $P$ and $Q$ satisfy \Extension{}, \Renaming{}, and \Restriction{},
then their separating conjunction $P\sepcon Q$,
separating implication $P\sepimp Q$,
conjunction $P\land Q$, and implication $P\to Q$
should as well.

A general strategy for preserving
invariants like this is to work with $\Set$-valued functors out of a
category $C$ capturing them.
Such functors are very well-behaved:
in particular, many subcategories of the
functor category $\funcat{C\op}{\Set}$,
called \emph{categories of sheaves on $C$},
are automatically cartesian closed, and
can be used to quickly obtain
invariant-preserving interpretations of logical connectives.
Placing Model~1 into the Schanuel topos $\Sch$
is an instance of this idea.
The Schanuel topos is a particular subcategory of $\funcat{C\op}\Set$,
where $C$ is chosen so that functors $C\op\to\Set$
capture \Extension{} and \Renaming{},
consisting only of functors that are \emph{atomic sheaves} in order to capture \Restriction{}.
We build up to this model in steps.

\subsubsection{The base category $C$} \label{sec:nominal-fixing-base-category}
Essentially, \Extension{} says propositions should be stable under
subset-inclusions $L\subseteq L'$
and \Renaming{} says they should be stable under
bijections.
These two invariants
can be packaged into a \emph{category of store shapes}:
\newcommand\StoreSh{{\nref{def:storesh}{\mathbf{Shp}}}}
\begin{definition} \label{def:storesh}
  Let $\StoreSh$ be the category whose objects
  are finite sets $L$ and whose morphisms from $L$
  to $M$ are functions $M\to L$
  definable by composing subset-inclusions and bijections.
\end{definition}
Note the direction $M\to L$ is the reverse of what one
might expect; this is because we will consider contravariant functors
on $\StoreSh$.
Intuitively, there is a morphism $M \to L$ if $L$ is a ``smaller'' shape
than $M$.
Since every composite of subset-inclusions and bijections
is an injective function, and every injective function
is bijective onto its image,
the category $\StoreSh$ has a simple abstract description:
\begin{proposition}
  The category $\StoreSh$
  is equal to $\FinInjop$,
  \labelword{where}{def:FinInj} $\FinInj$ is the category of
  injective functions between finite sets.
\end{proposition}
With $\StoreSh$ in hand,
functors $\StoreSh\op\to\Set$ (equivalently,
functors $\FinInj\to\Set$) model
\Extension{}- and \Renaming{}-invariant concepts.
In particular, there is a functor
modelling stores:
\begin{definition}[Store functor] \label{def:store}
  The \emph{store functor} $\Store : \StoreSh\op\to\Set$ is a 
  functor that sends a finite set
  $L$ to the set of 
  all $L$-shaped valuations
  and a 
  $\FinInj$-morphism $i : L\hookrightarrow M$ 
  to a function coercing $\Store(M)$ into $\Store(L)$
  defined by
  $\Store(i)(L,s) = (M,s')$, where
  $s'$ is the valuation $M\rightharpoonup\Z$
  defined by $s'(m) = s(l)$ iff $m = i(l)$ for some $l$ in $L$.
\end{definition}

The action of $\Store$ on $\StoreSh$-morphisms
captures the operations that we expect to be invariant under:
if $i$ is a subset inclusion $L\subseteq L'$,
then $\Store(i)$ coerces $L$-shaped stores
into $L'$-shaped stores as in the description of \Extension{},
and if $f$ is a bijective function $L\to L'$,
then $\Store(f)$ sends an $L$-shaped valuation $s$
to an $L'$-shaped valuation $s\circ f^{-1}$
as in the description of \Renaming{}.

\subsubsection{Using sheaves to capture Restriction}

Recall the example used to illustrate \Restriction:
if a proposition holds of the $\{\ell_1,\ell_2\}$-shaped
valuation $\{\ell_1\mapsto 1\}$,
then it should also hold of $\{\ell_1\mapsto 1\}$
considered as an $\{\ell_1\}$-shaped valuation.
We say that $\{\ell_1\mapsto 1\}\in\Store\{\ell_1,\ell_2\}$
\emph{restricts to $\{\ell_1\mapsto 1\}\in\Store\{\ell_1\}$
along $i$}, where $i$ is the 
subset-inclusion $\{\ell_1\}\subseteq\{\ell_1,\ell_2\}$.
This is an instance of a more general property satisfied
by the functor $\Store$:

\begin{proposition} \label{prop:store-atomic-sheaf-cond}
  Let $i : L\hookrightarrow M$ be an injective function between finite sets $L$ and $M$,
  and $s\in\Store(M)$ an $M$-shaped valuation.
  If $\dom(s)\subseteq\img(i)$,
  then there exists a unique $L$-shaped valuation $s'\in\Store(L)$,
  the \emph{restriction of $s$ along $i$},
  such that $\Store(i)(s') = s$.
\end{proposition}
\cref{prop:store-atomic-sheaf-cond} can be expressed more abstractly:
\begin{definition}
  Let $F$ be a functor $\StoreSh\op\to\Set$
  and $i : L\hookrightarrow M$ a $\FinInj$-morphism.
  An element $y$ of $F(M)$
  is \emph{restrictable along $i$}
  if for all $\FinInj$-objects
  $N$ and $\FinInj$-morphisms $j,k : M\to N$
  with $j\circ i = k\circ i$
  it holds that $F(j)(y) = F(k)(y)$.
\end{definition}

\begin{definition} \label{def:atomic-sheaf-condition}
  A functor $F : \StoreSh\op\to\Set$ 
  \emph{has a restriction operation}
  if for all $\FinInj$-morphisms $i : L\hookrightarrow M$
  and elements $y$ of $F(M)$ that are restrictable along $i$,
  there exists a unique $x\in F(L)$,
  called the \emph{restriction of $y$ along $i$},
  such that $y = F(i)(x)$.
\end{definition}
With these definitions in hand,
one can show
\cref{prop:store-atomic-sheaf-cond}
is equivalent to $\Store$ having a restriction operation.
Functors with a restriction operation
have a special name: they are
called \emph{atomic sheaves on $\StoreSh$}~\cite[Lemma III.4.2]{maclane2012sheaves}.
\labelword{The}{def:Sch} Schanuel topos $\Sch$ is the full subcategory of
$\funcat{\StoreSh\op}\Set$ consisting of atomic sheaves.

In these new terms,
\cref{prop:store-atomic-sheaf-cond} says
$\Store$ is an atomic sheaf on $\StoreSh$,
and so an object of $\Sch$.
Just as $\Store$ captures the concept of stores as shape-indexed valuations,
there are other atomic sheaves
for each of the other concepts used to define 
Model~1:

\begin{proposition} \label{prop:warp-in-the-atomic-sheaves} \label{prop:ext-ren-inv-prop}
  The following are objects of $\Sch$: \begin{itemize}[leftmargin=*]
    \item The constant functor $\Prop$
        sending every object of $\StoreSh$
        to the set $\{\vtrue,\vfalse\}$
        and every morphism of $\StoreSh$ to the identity function.
    \item The functor $\Loc$ of locations,
      defined by $\Loc(L) = L$
      on objects of $\StoreSh$
      and $\Loc(i : L\hookrightarrow L')(l : L) = i(l)$
      on $\FinInj$-morphisms $i : L\hookrightarrow L'$.
    \item The functor $\Loc^\Gamma$
      of $\Gamma$-substitutions,
      which maps objects $L$ to the set of all substitutions $L \rightarrow \Gamma$,
      and action on $\FinInj$-morphisms inherited pointwise from $\Loc$.
  \end{itemize}
\end{proposition}
With these sheaves in hand,
one can show Model~1 lives in $\Sch$:
\begin{proposition} \label{prop:model-one-lives-in-sch}
  If $\Gamma\vdash P$
  then the $L$-indexed family of functions
  \[
    \left(
    \modelsh{\Gamma\vdash P}L :
     (\Gamma\to L)
     \to
     \calP(L\rightharpoonup\Z)
    \right)_{L\in\StoreSh}
  \]
  is natural in $L$, so
  defines a morphism $\Loc^\Gamma\to \Prop^\Store$
  in $\Sch$,
  where $\Prop^\Store$ is the exponential guaranteed to exist
  because $\Sch$ is cartesian closed by virtue of being a category of sheaves.
  Moreover, every morphism of this type
  satisfies \Extension{}, \Renaming{}, and \Restriction{}.
\end{proposition}

\subsection{Model~2 in nominal sets} \label{sec:nominal-sets}

We now turn to the other side of the 
equivalence given by \cref{fact:folklore}:
the category of nominal sets $\Nom$,
and how it naturally houses
Model~2
of \cref{sec:nom-situation},
in which separation is defined via union of
finite partial functions on $\N$.

Just as $\Sch$ is a category capturing the invariants
implicitly maintained by Model~1,
$\Nom$ is a category
capturing the invariants implicitly maintained
by Model~2.
In this case, the invariants are:
\newcommand\Permutation{{\nref{def:Permutation}{Permutation}}}
\newcommand\Finiteness{{\nref{def:Finiteness}{Finiteness}}}
\begin{itemize}[leftmargin=*]
  \item \emph{\labelword{Permutation}{def:Permutation}}:
    propositions should be stable under permuting locations.
    If $s\in\modelgset{\Gamma\vdash P}(\gamma)$
    for some store $s : \N\nomstore\Z$ and substitution $\gamma : \Gamma\to \N$,
    and $\pi : \N\to\N$ is a permutation of finitely-many natural numbers,
    then it should hold that $s\circ \pi \in \modelgset{\Gamma\vdash P}(\pi^{-1}\circ \gamma)$.
  \item \emph{\labelword{Finiteness}{def:Finiteness}}: more subtly,
    stores and substitutions
    can only mention finitely-many locations $n\in\N$;
    this models the fact that physical stores are necessarily finite,
    and ensures that one always has the ability to allocate fresh locations.
\end{itemize}
To capture \Permutation{}, the objects of $\Nom$
are sets equipped with an action by a group of permutations to be invariant under.
  \labelword{Specifically}{def:nomgrp}, let $\nomgrp$ be the
  group of permutations of finitely-many natural numbers:
  elements of $\nomgrp$ are bijective functions $\pi : \N\to\N$
  such that there exists some $n\in \N$
  with $\pi(m) = m$ for all $m\ge n$.
  An \emph{$\nomgrp$-set}
  is a set $X$ equipped with
  a right action by $\nomgrp$:
  a function $(\cdot) : X\times\nomgrp\to X$
  satisfying $x\cdot 1 = x$ and $x\cdot (\pi\sigma) = (x\cdot \pi)\cdot \sigma$
  for all $x\in X$ and $\pi,\sigma\in\nomgrp$.
There is an $\nomgrp$-set
$\nomStore$ of stores,
whose group action says what it means to
permute the locations in a store:\footnote{In general, we will overline
 objects of $\Nom$ to distinguish them
 from their $\Sch$-counterparts.
}
\begin{definition} \label{def:nomStore}
  Let $\nomStore$
  be the $\nomgrp$-set
  of stores $s : \N\nomstore\Z$
  with action
  $s\cdot \pi = s\circ \pi$.
\end{definition}
A morphism of $\nomgrp$-sets
$(X,\cdot_X)\to(Y,\cdot_Y)$
is an \emph{equivariant function}:
a function $f : X\to Y$
satisfying $f(x\cdot_X\pi) = f(x)\cdot_Y \pi$
for all $x\in X$ and $y\in Y$ and $\pi\in\nomgrp$.
This captures invariance under \Permutation{}:
$\nomgrp$-morphisms $\nomStore\to\nomProp$,
\labelword{where}{def:nomProp} $\nomProp$ is the $\nomgrp$-set 
$\{\vtrue,\vfalse\}$
with trivial action $p\cdot \pi = p$,
are the permutation-invariant predicates on stores.

To capture \Finiteness{}, the category $\Nom$
is a full subcategory of the category of $\nomgrp$-sets
consisting of those $\nomgrp$-sets $(X,\cdot_X)$
in which every $x\in X$ only uses finitely many locations.
The concept of ``using'' a location is made precise by looking at
stabilizer subgroups: if
$x\cdot\pi = x$
(i.e., $\pi$ is in the stabilizer of $x$)
then $x$ can only ``use'' those locations fixed by $\pi$.
An element $x$ uses finitely many locations if its stabilizer is \emph{open}
for a suitable topology:
\begin{definition}[Topology on $\nomgrp$] \label{def:open-stab}
  A subset $U$ of $\nomgrp$ is \emph{open}
  if for every $\pi$ in $U$
  there exists a finite subset $A$ of $\N$
  such that $\pi\in\nomFix A\subseteq U$,
  \labelword{where}{def:nomFix} $\nomFix A$
  is the subgroup of $\nomgrp$-permutations $\pi$
  that fix every element of $A$; i.e., $\pi(a) = a$ for all $a$ in $A$.
\end{definition}

Intuitively, a stabilizer subgroup $\Stab x$ is open
if every $\pi$ stabilizing $x$ does so for some ``finite reason'' $A$:
there is some subset $A$ fixed by $\pi$ such that any other permutation $\pi'$
fixing $A$ also stabilizes $x$.
Nominal sets are $\nomgrp$-sets with open stabilizers~\citep[\S 6.2]{pitts2003nominal}:

\begin{definition} \label{def:Nom}
  A \emph{nominal set}
  is a $\nomgrp$-set $(X,\cdot)$
  such that for every $x$ in $X$
  the stabilizer subgroup $\Stab x$ is \nref{def:open-stab}{open}.
  $\Nom$ is the category of nominal sets and equivariant functions.
\end{definition}

For example, $\nomStore$ is a nominal set:
if $s$ is a store with $s\circ\pi = s$,
then $\pi$ fixes the finite set $\dom(s)$,
and moreover every permutation fixing $\dom(s)$
fixes $s$, so $\dom(s)\subseteq\Stab x$ and $\Stab s$ is open.
There are nominal sets
capturing each
of the other concepts used in Model~2:
\begin{proposition} \label{prop:warp-in-the-gsets}
  The following are objects of $\Nom$:\begin{itemize}[leftmargin=*]
    \item The $\nomgrp$-set $\nomProp$ of propositions
    \item The $\nomgrp$-set $\nomLoc$ of locations $\N$
      with action $x\cdot \pi = \pi^{-1}(x)$.
    \item The $\nomgrp$-set $\nomLoc^\Gamma$ of $\Gamma$-substitutions
      $\gamma : \Gamma\to \N$
      with action defined by $\gamma\cdot \pi = \pi^{-1}\circ \gamma$.
  \end{itemize}
\end{proposition}
With these in hand, one can show Model~2 lives in $\Nom$:
\begin{proposition} \label{prop:model-two-lives-in-nom}
  If $\Gamma\vdash P$ then 
  the function
  $\modelgset{\Gamma\vdash P}$
  is a morphism $\nomLoc^\Gamma\to\nomProp^\nomStore$ in $\Nom$,
  and every morphism of this type
  satisfies \Permutation{} and \Finiteness{}.
\end{proposition}

\subsection{The equivalence} \label{sec:sch-nom} \label{sec:eqv-models-schnom}

This section sketches the classic equivalence $\Sch\simeq\Nom$
and how Models 1 and 2 correspond across it.
We will not be concerned so much with the details of this
particular equivalence, but rather with highlighting
the key properties of $\Sch$
and $\Nom$ that make it possible --- \cref{thm:krm-corresp}
relies on identifying analogous properties
in the probabilistic setting.

In \cref{sec:nom-situation} we sketched
the correspondence between Model~1 and Model~2,
based on the idea that every store shape $L$
can be encoded as a finite set of natural numbers
via a pair of functions $(\enc_L,\dec_L)$.
This idea also forms the basis for the
equivalence $\Sch\simeq\Nom$.
In the language of \cref{sec:schanuel-topos},
every $\enc_L$
encodes the object $L$ of $\StoreSh$
as a subset $\img(\enc_L)$ of $\N$.
This encoding extends to morphisms of $\StoreSh$:
every $\StoreSh$-morphism $M\to L$,
equivalently an injective function $f : L\hookrightarrow M$,
can be encoded as a permutation $\pi\in\nomgrp$
that sends $\img(\enc_L)$ to $\img(\enc_M)$.
More precisely,
\newcommand\Homogeneity{{\nref{lem:nom-homogeneity}{Homogeneity}}}
\savebox0{\cite[L1.14]{pitts2013nominal}}
\begin{proposition}[Homogeneity~\usebox0] \label{lem:nom-homogeneity}
  Let $L,M$ be finite sets and $\enc_L$ and $\enc_M$
  injective functions $L\hookrightarrow\N$ and $M\hookrightarrow\N$.
  For any injective function $i : L\hookrightarrow M$,
  there exists $\pi\in\nomgrp$
  such that $\pi\circ \enc_L = \enc_M \circ i$,
  making the following square commute:
    \[\begin{tikzcd}
    \N \arrow[r, "\pi", dashed]                 & \N                    \\
    L \arrow[u, "\enc_L", hook] \arrow[r, "i"', hook] & M \arrow[u, "\enc_M"', hook]
    \end{tikzcd}\]
\end{proposition}
Furthermore,
the relationships between encoded store shapes
$\img(\enc_L)$ are faithfully captured by
relationships between subgroups of $\nomgrp$:

\newcommand\Correspondence{{\nref{def:Correspondence}{Correspondence}}}
\begin{proposition}[Correspondence] \label{def:Correspondence}
    For any two store shapes $L$ and $M$,
    it holds that
    $\nomFix(\img(\enc_L))\subseteq\nomFix(\img(\enc_M))$
    if and only if $\img(\enc_L)\supseteq \img(\enc_M)$.
\end{proposition}

\newcommand\schPred{{\nref{def:schPred}{\mathrm{Pred}}}}
\newcommand\schSubst[1]{{\nref{def:schSubst}{\llbr{#1}}}}
\newcommand\nomPred{{\nref{def:nomPred}{\overline{\mathrm{Pred}}}}}
\newcommand\nomSubst[1]{{\nref{def:nomSubst}{\overline{\llbr{#1}}}}}
\Homogeneity{} and \Correspondence{} together
give the equivalence $\Sch\simeq\Nom$.
For details, see \citet[Theorem III.9.2]{maclane2012sheaves}.
With this equivalence in hand, we are finally in a position to make \cref{fact:folklore} precise.
\labelword{Abbreviating}{def:schSubst} $\Loc^\Gamma$ as $\schSubst\Gamma$
\labelword{and}{def:schPred} the exponential $\Prop^\Store$ as $\schPred$,
\cref{prop:model-one-lives-in-sch}
shows that the Hom-set $\Sch(\schSubst\Gamma,\schPred)$
serves as a semantic domain for Model-1 interpretations of
\ref{eq:tinysep} propositions in context $\Gamma$.
Analogously, \cref{prop:model-two-lives-in-nom}
shows that $\Nom(\nomSubst\Gamma,\nomPred)$
serves as a semantic domain for Model~2,
\labelword{where}{def:nomSubst} $\nomSubst\Gamma$ is $\nomLoc^\Gamma$
\labelword{and}{def:nomPred} $\nomPred$ the exponential $\nomProp^\nomStore$.
The next proposition establishes that these semantic
domains correspond across $\Sch\simeq\Nom$:
\begin{proposition} \label{prop:semantic-domains-corresp}
  Across the equivalence $\Sch\simeq\Nom$, the sheaf
  $\Store$ corresponds to the nominal set $\nomStore$,
  $\Prop$ to $\nomProp$,
  $\Loc$ to $\nomLoc$,
  $\schSubst\Gamma$ to $\nomSubst\Gamma$,
  $\schPred$ to $\nomPred$,
  and 
$\Sch(\schSubst\Gamma,\schPred)$
to
$\Nom(\nomSubst\Gamma,\nomPred)$.
\end{proposition}

\newcommand\schDay{\mathbin{\nref{def:schDay}{\otimes}}}
\newcommand\schinc{\nref{def:schinc}{i}}
\newcommand\schjoin{\mathbin{\nref{def:schjoin}{\pdot}}}
\newcommand\schord{\mathrel{\nref{def:schord}{\sqsubseteq}}}
\newcommand\internalsepcon{\mathbin{\nref{def:internalsepcon}{\circledast}}}
It remains to show that Model~1 intepretations
$\modelsh{\Gamma\vdash P}{(-)}$
correspond to their Model~2 counterparts
$\modelgset{\Gamma\vdash P}$.
This is straightforward when $P$ is
$\trueprop$
or $x\mapsto i$; the interesting case
is the separating conjunction $P_1*P_2$.
One could show
$\modelsh{P_1\sepcon P_2}{(-)}$
corresponds to
$\modelgset{P_1\sepcon P_2}$
by unwinding definitions
and showing, via a careful calculation,
that they correspond across the functor $\Sch\to\Nom$ witnessing
the equivalence.
But \cref{fact:folklore} is far more general.
The idea is to use
the \emph{internal language} of $\Sch$: as a category of sheaves,
any construction in higher-order logic can be interpreted
in $\Sch$~\cite[VI.7.1]{maclane2012sheaves}.
In this internal language,
types denote sheaves and functions denote natural transformations,
and Model~1's separating conjunction
can be defined as
$\modelsh{\Gamma\vdash P_1\sepcon P_2}{(-)}
=\modelsh{\Gamma\vdash P_1}{(-)}\internalsepcon\modelsh{\Gamma\vdash P_2}{(-)}$,
\labelword{where}{def:internalsepcon} $\internalsepcon$
is a special $\Sch$-morphism denoting separating conjunction in the 
internal language of $\Sch$. The meaning of $\internalsepcon$ 
can be described by conveniently using the
higher-order logic of $\Sch$:
\begin{align} \label{eqn:sepcon-defn}
 \begin{aligned}
 &(\internalsepcon) :
   \schPred^{\schSubst\Gamma}
   \times
   \schPred^{\schSubst\Gamma}
   \to
   \schPred^{\schSubst\Gamma}
   \\
 &(f_1 \internalsepcon f_2)(\gamma : {\schSubst\Gamma})(s : \Store) =
  \left(~\begin{aligned}
     &\exists~s_1~s_2:\Store.~
     s_1\pdot s_2 \text{ defined} \land{}\\
     &s_1\pdot s_2 \sqsubseteq s\land{}
     f_1\,\gamma\,s_1 \land f_2\,\gamma\,s_2
    \end{aligned}\right)
  \end{aligned}
\end{align}
This definition is made of the following key ingredients:
\newcommand\schSepStores{{\nref{def:schSepStores}{\mathrm{S}^2_\perp}}}
\begin{itemize}[leftmargin=*]
\item A symbol $\sqsubseteq$,
  which in the internal language
  looks like an ordering relation on stores,
  and externally denotes 
  a suitable natural transformation
   $\Store\times\Store\to\Prop$.
\item A symbol $\pdot$,
  which internally looks like a partial function combining stores,
  and externally denotes
  a natural transformation $\schSepStores\to \Store$,
  where $\schSepStores$ is a subobject
  $i : \schSepStores\hookrightarrow\Store\times\Store$
  of the sheaf $\Store\times\Store$ of pairs of stores
  carving out the domain on which $\pdot$ is defined.
\end{itemize}
  \labelword{The}{def:schord} ordering $\sqsubseteq$
  is the natural transformation
  $(\schord) : \Store\times\Store\to\Prop$
  defined by
    $(s_1 \schord_L s_2) = \top$ if and only if $s_1$ is a subvaluation of $s_2$.
  \labelword{The}{def:schjoin} combining operation $\pdot$
  is a natural transformation
  $\schjoin : \schSepStores\to\Store$.
  \labelword{Its}{def:schDay} domain $\schSepStores$ 
  is a sheaf defined in terms of the coproduct of finite sets.
  \labelword{Each}{def:schSepStores} $\schSepStores(L)$ is a set consisting of pairs of $L$-shaped valuations
  that ``factor through'' a coproduct $L_1+L_2$ along some
  $i : L_1+L_2\hookrightarrow L$:
  \[ 
     \schSepStores(L) =
     \begin{aligned}
       &\{\,(\Store(i\circ\inl)(s_1), \Store(i\circ\inr)(s_2)) \\
       &\!\mid L_1,L_2\in\StoreSh,s_1\in \Store(L_1), s_2\in\Store(L_2),
         i : L_1+L_2\hookrightarrow L \,\}
     \end{aligned}
  \]
  The morphism $\schjoin$
  sends each pair $(\Store(i\circ\inl)(s_1), \Store(i\circ\inr)(s_2))$ of separated stores
  to the combined store $\Store(i)[s_1,s_2]$,
  where the valuation $[s_1,s_2]$ of type
  $L_1+L_2\to\Z$ is the unique one defined by $[s_1,s_2]\circ\inl = s_1$
  and $[s_1,s_2]\circ\inr = s_2$.
  \labelword{Each}{def:schinc}
  $\schSepStores(L)$
  is a subset of $(\Store\times\Store)(L)$,
  and collecting the canonical subset-inclusions
  into an $L$-indexed family gives
  a monic natural transformation
  $\schinc : \schSepStores\hookrightarrow \Store\times\Store$.

\newcommand\nomDay{\mathbin{\nref{def:nomDay}{\sepcon}}}
\newcommand\nomSepStores{\mathbin{\nref{def:nomSepStores}{\overline{\mathrm{S}}^2_\perp}}}
\newcommand\nominc{\nref{def:nominc}{\overline{i}}}
\newcommand\nomjoin{\mathbin{\nref{def:nomjoin}{\overline{\pdot}}}}
\newcommand\nomord{\mathrel{\nref{def:nomord}{\overline{\sqsubseteq}}}}
\newcommand\schemp{{\nref{def:schemp}{\mathrm{emp}}}}
\newcommand\nomemp{{\nref{def:nomemp}{\overline{\mathrm{emp}}}}}
  In the internal language, $\schjoin$ looks like a partial function
  that is associative and commutative and monotone with respect to $\schord$,
  \labelword{with}{def:schemp} unit the natural transformation $\schemp : 1\to \Store$
  sending every store shape $L$ to the empty valuation on $L$.
  Together, the tuple $(\schord,\schSepStores,\schinc,\schjoin,\schemp)$
  packages up the ingredients needed to model separation logic in $\Sch$
  into a \emph{resource monoid} internal to $\Sch$:
  \begin{definition}
    A \emph{resource monoid}~\cite{galmiche2005semantics} is a poset $(R,\sqsubseteq)$ with a least element $\bot$
    and a monotone partial function $(\cdot) : R\times R\rightharpoonup R$
    such that $(R,\cdot,\bot)$ forms a partial commutative monoid.\footnote{In this paper we are concerned with affine models of separation logic,
    and so consider an affine variant of the resource monoids defined in
    \citet{galmiche2005semantics}. Our definition is closest in spirit to the
    affine PDMs sketched there.}
  \end{definition}
  We can similarly construct a resource monoid in $\Nom$.
  \labelword{There}{def:nomord}
  is an equivariant function
  $(\schord) : \nomStore\times\nomStore\to\nomProp$
  sending a pair $(s_1,s_2)$ of finite partial functions on $\N$
  to $\vtrue$ if and only if $s_1\subseteq s_2$,
  \labelword{with}{def:nomemp} least element $\nomemp$ the empty finite partial function.
\labelword{There}{def:nomSepStores} is a nominal set
  $\nomSepStores$ of separated stores:
  the set 
  \[\{(s_1,s_2)\mid s_1,s_2\in\nomStore, \dom(s_1)\cap\dom(s_2) = \emptyset\}\]
  of pairs of stores with disjoint domain, and pointwise action.
  \labelword{Both}{def:nominc} the
  canonical inclusion $\nominc: \nomSepStores\hookrightarrow\nomStore\times\nomStore$
  and 
  \labelword{the}{def:nomjoin}
  function
  $(\nomjoin) : \nomSepStores\to\nomStore$
  sending a pair $(s_1,s_2)$ of disjoint stores
  to their union $s_1\uplus s_2$
  are equivariant, hence morphisms in $\Nom$.
  Finally, $\nomjoin$ is monotone in $\nominc$
  and has unit $\nomemp$
  so $(\nomord,\nomSepStores,\nominc,\nomjoin,\nomemp)$ forms a resource monoid internal to $\Nom$,
and reinterpreting
\cref{eqn:sepcon-defn}
inside $\Nom$
with $(\nomord,\nomSepStores,\nominc,\nomjoin,\nomemp)$
in place of
$(\schord,\schSepStores,\schinc,\schjoin,\schemp)$
yields Model~2's separating conjunction.
The following proposition, connecting the two resource monoids,
makes \cref{fact:folklore} precise:
\begin{proposition} \label{prop:sch-nom-krm}
The resource monoid $(\schord,\schSepStores,\schinc,\schjoin,\schemp)$
corresponds to $(\nomord,\nomSepStores,\nominc,\nomjoin,\nomemp)$
across the equivalence $\Sch\simeq\Nom$.
\end{proposition}
We are at last ready to appreciate the full power of this fact.
First, it shows $\modelsh{P_1\sepcon P_2}{(-)}$
corresponds to $\modelgset{P_1\sepcon P_2}$:
both arise from the same internal-language definition,
up to the replacement of types and
function symbols
following \cref{prop:semantic-domains-corresp,prop:sch-nom-krm}.
Next, since the separating implication $\sepimp$
and all intuitionistic connectives can be defined
similarly using the internal language,
they must correspond as well; this extends the equivalence
of Models 1 and 2 from \ref{eq:tinysep}
to all standard separation logic connectives.
More generally, \cref{fact:folklore} says
that any construction in higher-order logic
that only uses the types and functions of
\cref{prop:semantic-domains-corresp,prop:sch-nom-krm}
corresponds across the equivalence $\Sch\simeq\Nom$.

\section{The discrete case} \label{sec:discrete}

\cref{thm:krm-corresp}
imports quite a bit of measure theory
in order to support continuous probability.
To describe the key ideas, we
temporarily set the measure theory aside
by first presenting in detail a version of
\cref{thm:krm-corresp}
adapted to discrete probability.

The structure of this section is completely analogous
to \cref{sec:nom-situation}.
We first present two different
probabilistic separation logics:
one
where separation is defined via the product
of sample spaces,
and a second based on \citet{li2023lilac}
where separation is defined via \emph{independent combination}.
Then, we will show how separation-as-product
naturally lives in a category $\dsheafcat$
of \emph{\dgoodsheaves{}} analogous to the Schanuel topos,
and how separation-as-independent-combination
naturally lives in a category $\dgsets$ of \emph{\dagoodname{s}}.
Finally, we show these two categories equivalent,
and that the two notions of separation correspond
across this equivalence, giving an analog of \cref{fact:folklore}
suitable for discrete probability.

In \cref{sec:nom-situation} we considered a tiny separation logic
\ref{eq:tinysep} for integer-valued stores.
Analogously, we consider here a logic
for integer-valued random variables:
\begin{align}
P,Q ::= X\sim \mu \mid \trueprop \mid P\sepcon Q.
\tag{\textsc{TinyProbSep}}
\label{eq:tinyprobsep}
\end{align}
The proposition $X\sim\mu$ asserts that the logical variable $X$
stands for an integer-valued random variable
with probability mass function $\mu : \Z\to[0,1]$.
As in \ref{eq:tinysep},
a proposition is well-formed in context $\Gamma$,
written $\Gamma\vdash P$,
if $\Gamma$ contains the variables used by $P$.
We shall establish the equivalence of two different models for \ref{eq:tinyprobsep}.
In both cases, the basic idea is that a proposition denotes a predicate
on \emph{probability spaces} and logical variables denote \emph{random
variables}, just as a proposition in ordinary separation logic
denotes a predicate on stores with logical variables denoting heap locations.
The difference is in how these objects are represented:

\goodparagraph{Model~1: separation as product}
In this model, a \emph{probability space}
consists of two components:
(1) a nonempty countable set $\Omega$
called the \emph{sample space},
and (2) a \emph{probability space $\calP{}$ on $\Omega$}
consisting of a pair $(\calF,\mu)$
with $\calF$ a $\sigma$-algebra on $\Omega$
and $\mu : \calF\to[0,1]$ a probability measure.
A \emph{random variable on $\Omega$}
is a function $\Omega\to\Z$.
\labelword{We}{def:Dpspcsprelim} \labelword{will}{def:DRVprelim} write $\Dpspcsprelim(\Omega)$
and $\DRVprelim(\Omega)$
for the set of probability spaces and random variables on $\Omega$ respectively.

The meaning of a proposition depends on the underlying sample space:
$\Gamma\vdash P$
denotes a map
$\modelsh{\Gamma\vdash P}\Omega
: (\Gamma\to \DRVprelim(\Omega))\to \powerset(\Dpspcsprelim(\Omega))$
associating each \emph{random substitution} $G : \Gamma\to\DRVprelim(\Omega)$
to the set
$\modelsh{\Gamma\vdash P}\Omega(G)$
of probability spaces on $\Omega$ satisfying $P$.

Under this interpretation, we have $\calP\in\modelsh\trueprop\Omega(G)$
for all probability spaces $\calP$ on $\Omega$,
and $(\calF,\mu)\in\modelsh{X\sim\nu}\Omega(G)$
if and only if $G(X)$ is $\calF$-measurable
and has distribution $\nu$; i.e., for all $i\in \Z$ it holds that
$G(X)^{-1}(i)\in\calF$
and $\mu(G(X)^{-1}(i)) = \nu(i)$.
Separating conjunction is defined in terms of products
of sample spaces. To make this precise, we need the following definitions:

\begin{definition}[Pullback probability space]
  Let $X$ be a nonempty countable set, $(Y,\calG,\nu)$
  a countable probability space,
  and $f : X\twoheadrightarrow Y$
  a surjective function.
  The \emph{pullback of $(\calG,\nu)$ along $f$},
  written $f^{-1}(\calG,\nu)$,
  is the probability space $(\calF,\mu)$ on $X$ 
  defined by
  $\calF = \{f^{-1}(G) \mid G \in\calG\}$
  and $\mu(f^{-1}(G)) = \nu(G)$.
  Note $\mu$ is well-defined because
  $f$ surjective, so $f^{-1}$ injective.
\end{definition}

\newcommand\subpspcs{\mathrel{\nref{def:subpspcs}{\sqsubseteq}}}
\begin{definition}[Subspace] \label{def:subpspcs}
  Given two probability spaces $(\calF,\mu)$
  and $(\calG,\nu)$ on $\Omega$,
  say $(\calF,\mu)$ is a \emph{subspace} of $(\calG,\nu)$,
  written $(\calF,\mu)\subpspcs(\calG,\nu)$,
  if $\calF\subseteq\calG$ and $\nu|_\calF = \mu$.\end{definition}

With these definitions, the separating conjunction $P_1\sepcon P_2$
holds of a probability space $\calP$
on $\Omega$
if and only if there exist probability spaces $\calP_1$ on $\Omega_1$
and $\calP_2$ on $\Omega_2$
and a surjective function
$p : \Omega\twoheadrightarrow\Omega_1\times\Omega_2$
such that $\calP_1$ satisfies $P_1$
and $\calP_2$ satisfies $P_2$
and $p^{-1}(\calP_1\otimes\calP_2)$
is a subspace of $\calP$,
where $\calP_1\otimes\calP_2$
is the product probability space
on $\Omega_1\times\Omega_2$
whose measure is the product measure
induced by the measures of $\calP_1$ on $\Omega_1$
and $\calP_2$ on $\Omega_2$ in the usual way.

For example, let $\Omega$ be the sample space $\{0,1\}^3$
of points $(x,y,z)\in\R^3$ with $x,y,z$ all either $0$ or $1$.
Let $G$ be the random substitution of type $\{X,Y\}\to\DRVprelim(\Omega)$
where $G(X)$ is the random variable $(x,y,z)\mapsto x$
and $G(Y)$ is the random variable $(x,y,z)\mapsto y$.
Let $(\calF,\mu)$ be the uniform probability space on $\Omega$,
assigning each tuple $(x,y,z)$ probability $1/8$.
It holds that
\[
  (\calF,\mu)\in\modelsh{(X\sim\Ber(1/2))\sepcon(Y\sim\Ber(1/2))}\Omega(G),
\]
witnessed by setting $p$
to the projection $\Omega\twoheadrightarrow\{0,1\}\times\{0,1\}$
defined by $p(x,y,z) = (x,y)$.

\newcommand\Isigalg{{\nref{def:Isigalg}{\Sigma_{[0,1]}}}}
\goodparagraph{Model~2: separation as independent combination}
In this model, one fixes upfront a single measurable space
to serve as a ``universal sample space''
into which all discrete sample spaces can be embedded.
\labelword{Any}{def:Isigalg} standard Borel space will do;
we choose the interval $[0,1]$.
The idea is that, just as every finite store shape $L$ can be 
encoded as a finite subset of $\N$ via an injective function $\enc_L : L\hookrightarrow \N$,
every nonempty countable sample space $\Omega$
can be encoded as a countable partition of the interval
via a random variable $\dec_\Omega : [0,1]\to\Omega$
with each $\dec_\Omega^{-1}(\omega)$ nonnegligible, 
visualized as:
  
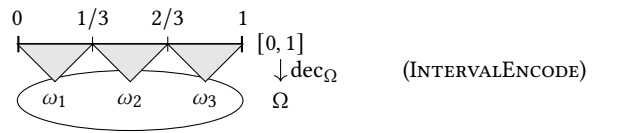
\begin{equation}
    \begin{aligned}
    \tikz{
          \draw [|-|,thick] (0,0) -- (3,0);
          \node[] at (0, 0.35) {$0$};
          \node[] at (3, 0.35) {$1$};
          \draw [] (1,0.1) -- (1,-0.1) node[above, yshift=5] {$1/3$};
          \draw [] (2,0.1) -- (2,-0.1) node[above, yshift=5] {$2/3$};

          \node [] at (0.5,-0.75) {$\omega_1$};
          \node [] at (1.5,-0.75) {$\omega_2$};
          \node [] at (2.5,-0.75) {$\omega_3$};

          \draw (1.5,-0.75) ellipse (1.5cm and 0.4cm);
          \node[] at(3.5, 0) {$[0,1]$};
          \draw[->] (3.5, -0.2) -- (3.5, -0.5) node[right, xshift=0.5, yshift=4] {$\dec_\Omega$};
          \node[] at(3.5, -0.75) {$\Omega$};

          \filldraw[draw=black, fill=gray!20] (0,0) -- (1,0) -- (0.5, -0.5) -- (0,0) -- cycle;
          \filldraw[draw=black, fill=gray!20] (1,0) -- (2,0) -- (1.5, -0.5) -- (1,0) -- cycle;
          \filldraw[draw=black, fill=gray!20] (2,0) -- (3,0) -- (2.5, -0.5) -- (2,0) -- cycle;
    }
    \end{aligned}
    \tag{\textsc{IntervalEncode}}
    \label{fig:encoding-omega}
  \end{equation}

This illustration gives one possible encoding
of the three-point space $\{\omega_1,\omega_2,\omega_3\}$
as the partition $\{[0,1/3), [1/3,2/3], (2/3,1]\}$,
generated by the random variable $\dec_\Omega : [0,1]\to \{\omega_1,\omega_2,\omega_3\}$
taking value $\omega_1$ on $[0,1/3)$,
$\omega_2$ on $[1/3,2/3]$,
and $\omega_3$ on $(2/3,1]$.

With this fixed sample space in hand,
the random variables $\Omega\to\Z$
of Model~1
can be encoded as 
Lebesgue-measurable functions $[0,1]\to \Z$, quotiented
by almost-everywhere equality.
\labelword{We}{def:nomDRVprelim}
will write $\nomDRVprelim$ for the set of integer-valued random variables.

In order for our encoding
of sample spaces as partitions generated by random variables
to respect almost-everywhere equality of random variables,
we must consider such partitions up to negligibility:
for example,
the partitions $\{[0,1/3), [1/3,2/3], (2/3,1]\}$
and $\{[0,1/3], (1/3,2/3), [2/3,1]\}$
should be considered equivalent,
as they arise from almost-everywhere-equal random variables.
This idea motivates the following definition.

\newcommand\parteq{\mathrel{\nref{def:parteq}{\aseq}}}
\begin{definition}
  A \emph{countable measurable partition of $[0,1]$}
  is a countable partition $\{A_i\}_{i\in I}$
  with each $A_i$ a Lebesgue-measurable and nonnegligible subset of $[0,1]$,
  quotiented by almost-everywhere equality:
  \labelword{two}{def:parteq} partitions are almost-everywhere equal,
  written $\{A_i\}_{i\in I} \parteq \{B_j\}_{j\in J}$,
  if for all $i$ in $I$ there exists a unique $j$ in $J$
  such that the symmetric difference $A_i\triangle B_j$ is Lebesgue-negligible.
\end{definition}

Just as any $L$-shaped valuation can be encoded as a finite partial function on $\N$,
any discrete probability space can be encoded as a countable measurable partition equipped with a measure:

\begin{definition}
  A \emph{countable measured partition of $[0,1]$}
  is a pair $(\{A_i\}_{i\in I}, \mu)$
  with $\{A_i\}_{i\in I}$ a countable measurable partition of $[0,1]$
  and $\mu : \{A_i\}_{i\in I}\to [0,1]$
  a function satisfying $\sum_i \mu(A_i) = 1$.
  Two such partitions
  are equal if their measurable partitions are equal
  and their measures agree.
  \labelword{Let}{def:nomDpspcsprelim} $\nomDpspcsprelim$ be the set of countable measured partitions of $[0,1]$.
\end{definition}

Now that we have a way of encoding discrete probability spaces as countable measured partitions of $[0,1]$,
we can define a model of \ref{eq:tinyprobsep} purely in terms of countable measured partitions.
A proposition $\Gamma\vdash P$ denotes a map
$\modelgset{\Gamma\vdash P} : (\Gamma\to\nomDRVprelim)\to \powerset(\nomDpspcsprelim)$
assigning each random substitution $G : \Gamma\to\nomDRVprelim$
the set of countable measured partitions satisfying $P$.
The interpretations of $\trueprop$ and $X\sim\mu$ are as in
Model~1:
$\calP\in\modelgset\trueprop(G)$ for any countable measured partition $\calP$,
and $(\{A_i\}_{i\in I},\mu)\in\modelgset{X\sim\nu}(G)$
if and only if
for all $k\in \Z$ there exists $i\in I$ with $G(X)^{-1}(k)\triangle A_i$ negligible
and $\mu(G(X)^{-1}(k)) = \nu(k)$.
Separating conjunction is defined via \emph{independent combination},
following \citet{li2023lilac}:
\begin{definition}[Discrete independent combination] \label{def:Dicom}
  A countable measured partition $(\calA,\mu)$
  is an \emph{independent combination}
  of $(\calA_1,\mu_1)$ and $(\calA_2,\mu_2)$
  if (1) the partition $\calA$ is generated by
  the intersections $A_1\cap A_2$ for $A_1$ in $\calA_1$ 
  and $A_2$ in $\calA_2$
  and (2) $\mu(A_1\cap A_2) = \mu_1(A_1)\mu_2(A_2)$
  for all $A_1$ in $\calA_1$ and $A_2$ in $\calA_2$.
  Independent combinations are unique if they exist~\lilac{Lemma 2.3}{11},
  defining a partial function $\Dicom$
  with $(\calA_1,\mu_1)\Dicom(\calA_2,\mu_2) = (\calA,\mu)$
  if and only if $(\calA,\mu)$ is the independent combination
  of $(\calA_1,\mu_1)$
  and $(\calA_2,\mu_2)$.
\end{definition}

\begin{definition}[Ordering on partitions] \label{def:Dorder}
  For two countable measured partitions $(\calA,\mu)$
  and $(\calB,\nu)$,
  let $\Dorder$ be the ordering relation
  defined by $(\calA,\mu)\Dorder(\calB,\nu)$
  if and only if the partition $\calA$ is coarser than $\calB$
  and $\nu$ restricts to $\mu$.
\end{definition}
These definitions give an interpretation of separating conjunction:
a countable measured partition $(\calA,\mu)$ on $[0,1]$
satisfies $P_1\sepcon P_2$ with random substitution $G$,
written $(\calA,\mu)\in\modelgset{P_1\sepcon P_2}(G)$,
if and only if
there exist $(\calA_1,\mu_1)$
and $(\calA_2,\mu_2)$ independently combinable
with $(\calA_1,\mu_1)\Dicom(\calA_2,\mu_2)\Dorder(\calA,\mu)$
such that $(\calA_1,\mu_1)$ is in $\modelgset{P_1}(G)$
and $(\calA_2,\mu_2)$ is in $\modelgset{P_2}(G)$.

\goodparagraph{Relating the two models}
We will show that Model~1 and Model~2 are equivalent.
As shown in \ref{fig:encoding-omega},
every nonempty countable sample space $\Omega$ can be encoded as a
countable measured partition on $[0,1]$ via a suitable
random variable $\dec_\Omega : [0,1]\to\Omega$.
Choosing a $\dec_\Omega$
for every $\Omega$
allows translating Model~1 into Model~2:
a Model~1 random variable $X \in \DRV(\Omega)$
corresponds to a Model~2 random variable $X\circ\dec_\Omega\in\nomDRVprelim$,
and probability spaces can be translated similarly.
To extend this into an equivalence analogous to \cref{fact:folklore},
we repeat the recipe of \cref{sec:nom-situation}: we will
place Models 1 and 2 into suitable categories, show the categories equivalent, and
show that the models correspond across this equivalence.

\subsection{\Dgoodsheaves{}}

In \cref{sec:schanuel-topos} we saw how the Schanuel topos $\Sch$
captured the invariants maintained by Model~1 of \ref{eq:tinysep}.
In this section we describe analogously how
a category of \emph{\dgoodsheaves},
written $\dsheafcat$,
captures the invariants maintained by Model~1 of \ref{eq:tinyprobsep}.
Whereas the invariants of \cref{sec:schanuel-topos}
were about extensions and restrictions of the store shape $L$,
the invariants in our probabilistic setting
are about extensions and restrictions of the sample space $\Omega$,
as observed by \citet{simpson2017probability,simpsonsheafslides}:
\newcommand\probExtension{\nref{def:probExtension}{Extension}}
\newcommand\probRestriction{\nref{def:probRestriction}{Restriction}}
\begin{itemize}[leftmargin=*]
  \item \emph{\labelword{Extension}{def:probExtension}}:
    propositions that hold in sample space $\Omega$
    should continue to hold when $\Omega$
    is extended to a larger sample space $\Omega'$
    via a surjective function $p : \Omega'\twoheadrightarrow\Omega$.
    More precisely, if $(\calF,\mu)\in\modelsh{\Gamma\vdash P}\Omega(G)$
    for some probability space $(\calF,\mu)$ on $\Omega$
    and random substitution $G : \Gamma\to\DRV(\Omega)$,
    then it should hold that $p^{-1}(\calF,\mu) \in \modelsh{\Gamma\vdash P}{\Omega'}(G\cdot p)$,
    where $G\cdot p$ is the random substitution
    $(G\cdot p)(X) = G(X)\circ p$.\footnote{Note that this rolls the two invariants \Extension{} and \Renaming{}
    of \cref{sec:schanuel-topos} into one: it captures invariance under
    permutations of the underlying sample space in the case where $p$
    is a bijection.}

  \item \emph{\labelword{Restriction}{def:probRestriction}}:
    the truth of a proposition should not depend on any unused samples.
    For example, let $\Omega$ be an arbitrary sample space.
    Suppose $G : \Gamma\to\RV(\Omega)$
    sends every $X$ in $\Gamma$
    to the constant random variable $0$, so
    $G(X)(\omega) = 0$ for all $\omega$,
    and let $(\calF,\mu)$ be the minimal probability space on $\Omega$
    where $\calF$ is the minimal $\sigma$-algebra $\{\emptyset,\Omega\}$
    and $\mu$ the minimal probability measure with $\mu(\emptyset) = 0$
    and $\mu(\Omega) = 1$.
    Both $G$ and $(\calF,\mu)$
    don't use any of the samples in $\Omega$: every random variable $G(X)$ 
    is a deterministic value, and $\mu$ only assigns probabilities to
    the deterministic events $\emptyset$ and $\Omega$.
    \probRestriction{} says that if a proposition $P$ holds in this situation,
    then it should hold of the one-point probability space
    on the one-point set with substitution
    sending every $X$ in $\Gamma$ to the constant random variable $0$.
\end{itemize}
To capture these invariants, we replay the construction of the Schanuel topos:
whereas the Schanuel topos is a category of atomic sheaves on the category $\StoreSh$
of store shapes, $\dsheafcat$ is a category of atomic sheaves
on a category of discrete sample spaces.

\labelword{First}{def:CSur}, we fix a base category capturing \probExtension{}:
the category $\CSur$ of nonempty countable sets
and surjective functions.
The idea is that an object $\Omega$ of $\CSur$
is a countable sample space,
and a morphism $p : \Omega'\twoheadrightarrow\Omega$
extends a sample space $\Omega$
to a larger space $\Omega'$
in which every sample $\omega$ in $\Omega$
is expanded to a set of samples $p^{-1}(\omega)\subseteq\Omega'$;
surjectivity of $p$ ensures that every $p^{-1}(\omega)$
is nonempty, so $p$ never deletes any samples in $\Omega$.

Functors $\CSurop\to\Set$ model sample-space-dependent concepts.
In particular, there is a functor modelling probability spaces:
  \labelword{For}{def:Dpspcs} a $\CSur$-morphism $p : \Omega'\twoheadrightarrow\Omega$,
  setting $\Dpspcs(p)$ to the function $\Dpspcsprelim(\Omega)\to\Dpspcsprelim(\Omega')$
  that sends a probability space $\calP$
  on $\Omega$ to its pullback $p^{-1}\calP$
  makes $\Dpspcsprelim$ a functor $\CSurop\to\Set$.

Next, we capture \probRestriction{} by cutting the functor category $\funcat{\CSurop}\Set$
down to a full subcategory of atomic sheaves.
The notion of atomic sheaf is given by the notion of \emph{atomic topology},
which exists for a given category if and only if the following \emph{Ore property}
holds~\cite[p.115]{maclane2012sheaves}:
\begin{definition} \label{prop:csur-ore}
  A category $C$ has the \emph{right Ore property}
  if for all $X\xrightarrow f Z\xleftarrow g Y$ there exists
  $X\xleftarrow h W\xrightarrow k Y$ such that $fh = gk$.
\end{definition}
That $\CSur$ satisfies this condition can be straightforwardly verified:
any cospan can be completed to a commutative square
by taking a pullback in $\Set$.
Thus the notion of atomic sheaf makes sense for $\CSur$,
a functor is an atomic sheaf
if and only if it has a restriction operation in the sense of
\cref{def:atomic-sheaf-condition},
and the following definition makes sense:
\begin{definition} \label{def:dsheafcat}
  Let $\dsheafcat$ be the full subcategory of
  the category $\funcat{\CSurop}\Set$
  consisting of those functors that are atomic sheaves.
\end{definition}

Just as $\Dpspcs(\Omega)$ models the concept of probability spaces
on $\Omega$, there are other atomic sheaves corresponding to each
of the other concepts used to define Model~1:

\newcommand\DProp{\nref{def:DProp}{\mathrm{Prop}}}
\begin{proposition} \label{prop:dwarp-in-the-atomic-sheaves}
  The following are objects of $\dsheafcat$: \begin{itemize}[leftmargin=*]
    \item \labelword{The}{def:DProp} constant functor $\DProp$ of propositions sending every object
       of $\CSur$ to the set $\{\vtrue,\vfalse\}$ and every morphism
       of $\CSur$ to the identity function.
    \item The functor $\DRV$ of random variables,
      with action on morphisms
      defined by $\DRV(p : \Omega'\twoheadrightarrow\Omega)(X : \DRVprelim(\Omega))
      = (X\circ p : \DRVprelim(\Omega'))$.
    \item The functor $\DRV^\Gamma$
      of $\Gamma$-substitutions
      with $\DRV^\Gamma(\Omega) = \Gamma\to\DRV(\Omega)$
      and action on morphisms
      defined by lifting $\DRV$ pointwise.
  \end{itemize}
\end{proposition}
\noindent With these sheaves in hand,
one can show Model~1 lives in $\dsheafcat$:
\begin{proposition}
  If $\Gamma\vdash P$ then
  the $\Omega$-indexed family of functions
  \[
    \left(
      \modelsh{\Gamma\vdash P}\Omega
      : (\Gamma\to\DRV(\Omega))\to
      \calP(\Dpspcs(\Omega))
    \right)_{\Omega\in\CSur}
    \]
    is natural in $\Omega$,
    so defines a morphism $\DRV^\Gamma\to\Prop^\Dpspcs$
    in $\dsheafcat$, and every morphism
    of this type satisfies \probExtension{}
    and \probRestriction{}.
\end{proposition}

\subsection{Discrete \agoodnamenolink{s}}

We now turn to Model~2 of \ref{eq:tinyprobsep} described in \cref{sec:discrete}.
Just as Model~2 of \ref{eq:tinysep} naturally lives in the category $\Nom$
of nominal sets,
Model~2 of \ref{eq:tinyprobsep} naturally lives in
a category $\dgsets$ of \emph{\dagoodname{s}}.
$\Nom$
captures two invariants held by Model~2 of \ref{eq:tinysep}:
\Permutation{} and \Finiteness{}.
Model~2 of \ref{eq:tinyprobsep} maintains analogous invariants:
\newcommand\prPerm{\nref{def:prPerm}{Permutation}}
\newcommand\Sparsity{\nref{def:Sparsity}{Sparsity}}
\begin{itemize}[leftmargin=*]
  \item \emph{\labelword{Permutation}{def:prPerm}}:
    propositions should be stable under permuting the sample space $[0,1]$.
    More precisely, if $(\calA,\mu)\in\modelgset{\Gamma\vdash P}(G)$
    for some countable measured partition $(\calA,\mu)$
    and random substitution $G : \Gamma\to\nomDRVprelim$,
    and if $\pi : [0,1]\to[0,1]$ is a measurable bijection,
    then it should hold that $(\calA,\mu)\cdot\pi\in\modelgset{\Gamma\vdash P}(G\cdot\pi)$,
    where $(\calA,\mu)\cdot\pi$ and $G\cdot\pi$
    are the results of the permutation $\pi$ acting on $(\calA,\mu)$
    and $G$.

  \item \emph{\labelword{Sparsity}{def:Sparsity}}:
    more subtly, the countable measured partitions
    $(\calA,\mu)$ represent \emph{countable} probability spaces only.
    This ensures $(\calA,\mu)$ always leaves ``enough room'' in $[0,1]$
    for ``fresh randomness'': for any other discrete probability space,
    there exists an encoding of it as a countable measured partition $(\calB,\nu)$
    such that the discrete independent combination
    $(\calA,\mu)\Dicom(\calB,\nu)$ is defined.
\end{itemize}
To capture \prPerm{}, the objects of $\dgsets$
are sets equipped with an action by a group of measurable automorphisms.
  \labelword{Specifically}{def:dauti}, let $\dauti$ be the group of measurable maps
  $\pi : [0,1]\to[0,1]$ that are bijective mod almost-everywhere equality.
  The category of $\dauti$-sets is the category whose objects
  are sets $X$ equipped with a right action by $\dauti$
  and whose morphisms are equivariant functions.
  Just as there is a $\nomgrp$-set $\nomStore$ of stores,
  there is a $\dauti$-set $\nomDpspcs$ of countable measured partitions on $[0,1]$:
\begin{definition} \label{def:nomDpspcs}
  Let $\nomDpspcsprelim$ be the set of countable measured partitions on $[0,1]$
  with action $(\{A_i\}_{i\in I},\mu)\cdot\pi = (\{\pi^{-1}(A_i)\}_{i\in I},\mu\circ\pi)$.
\end{definition}

\Sparsity{} is captured by
topologizing $\dauti$
via countable measurable partitions,
so a stabilizer $\Stab x$ is open if
every $\pi$ stabilizing $x$ does so for a ``countable reason'':
there is a partition $\calA$ fixed by $\pi$ such that any other permutation
fixing $\calA$ also stabilizes $x$. 

\newcommand\setaseq{\mathrel{=_{\rm a.e.}}}
\newcommand\dautiFix{\operatorname{\nref{def:dautiFix}{Fix}}}
\begin{definition}[Topology on $\dauti$] \label{def:dauti-open}
A subset $U$ of $\dauti$
is \emph{open} if for every $\pi\in U$,
  there exists a countable measurable partition $\calA$ of $[0,1]$
  such that $\pi\in\dautiFix\calA\subseteq U$,
  \labelword{where}{def:dautiFix}
  $\dautiFix \calA$ is the subgroup of $\dauti$
  consisting of those permutations $\pi$ that fix every element of $\calA$;
  i.e., $\pi(A) \setaseq A$ for all $A\in\calA$.
\end{definition}

\begin{definition} \label{def:dgsets}
  A \emph{\dagoodname{}} is a $\dauti$-set 
whose elements have open stabilizers.
Let $\dgsets$ be the category of \dagoodname{s} and equivariant functions.
\end{definition}

In addition to $\Dpspcs$,
there are objects of $\dgsets$ corresponding to each of the other concepts
used to define Model~2 of \ref{eq:tinyprobsep}:
\newcommand\nomDProp{\nref{prop:dwarp-in-the-gsets}{\overline{\mathrm{Prop}}}}
\begin{proposition} \label{prop:dwarp-in-the-gsets}
  The following are objects of $\dgsets$: \begin{itemize}[leftmargin=*]
    \item The $\dauti$-set $\nomDProp = \{\vtrue,\vfalse\}$ with the trivial action.
    \item The $\dauti$-set $\nomDRVprelim$ of random variables
      with action $X\cdot\pi = X\circ \pi$.
    \item The $\dauti$-set $\nomDRVprelim^\Gamma$
      of random $\Gamma$-substitutions $\Gamma\to\nomDRVprelim$
      with action defined by lifting $\nomDRV$
      pointwise.
  \end{itemize}
\end{proposition}
\noindent With these in hand, one can show Model~2 lives
in $\dgsets$:
\begin{proposition} 
  If $\Gamma\vdash P$ then
  the function $\modelgset{\Gamma\vdash P}$
  is a morphism $\nomDRV^\Gamma\to\nomDProp^\nomDpspcs$ in $\dgsets$,
  and every morphism of this type
  satisfies \Permutation{} and \Finiteness{}.
\end{proposition}

\subsection{The equivalence of categories} \label{sec:discrete-eqv}

Having placed Models 1 and 2 of \ref{eq:tinyprobsep} described in \cref{sec:discrete}
into the categories $\dsheafcat$ and $\dgsets$ respectively,
we describe in this section how $\dsheafcat$ are $\dgsets$ equivalent,
giving an analog of $\Sch\simeq\Nom$ for discrete probability.
The key step is to establish probabilistic analogs
of \Homogeneity{} and \Correspondence{}:

\begin{lemma}[Homogeneity] \label{def:DHomogeneity}
  Let $\Omega,\Omega'$ be nonempty countable sets
  and let $\dec_\Omega$ and $\dec_{\Omega'}$
  be measurable functions $[0,1]\to\Omega$ and $[0,1]\to\Omega'$
  with $\dec_\Omega^{-1}(\omega)$ and $\dec_{\Omega'}^{-1}(\omega')$
  nonnegligible for all $\omega\in\Omega$ and $\omega'\in\Omega'$.
  For any surjective function $f : \Omega'\twoheadrightarrow\Omega$,
  there exists $\pi\in\dauti$
  making the following square commute:
  \[\begin{tikzcd}
  {[0,1]} \arrow[d, "\dec_{\Omega'}"'] \arrow[r, "\pi", dashed] & {[0,1]} \arrow[d, "\dec_\Omega"] \\
  \Omega' \arrow[r, "p"', two heads]                            & \Omega                          
  \end{tikzcd}\]
\end{lemma}

This lemma is particularly important, so we give some intuition about its proof.
Consider the following visualization of a surjection $p$ from 
$\Omega' = \{\omega_1', \omega_2', \omega_3'\}$ onto $\Omega = \{\omega_1, \omega_2\}$
and
two decoding functions $\dec_{\Omega'}$ and $\dec_\Omega$ visualized as
in \ref{fig:encoding-omega}:

\begin{center}
  \pgfdeclarelayer{bg}    \pgfdeclarelayer{mid}   \pgfsetlayers{bg,mid,main}  \begin{tikzpicture}
  \draw [|-|,thick] (0,0) -- (3,0);
  \draw [] (1,0.1) -- (1,-0.1) node[above, yshift=5] {$1/3$};
  \draw [] (2,0.1) -- (2,-0.1) node[above, yshift=5] {$2/3$};

  \node [] (omegap1) at (0.5,-0.75) {$\omega'_1$};
  \node [] (omegap2) at (1.5,-0.75) {$\omega'_2$};
  \node [] (omegap3) at (2.5,-0.75) {$\omega'_3$};

  \draw (1.5,-0.75) ellipse (1.5cm and 0.3cm);
  \node[] at(-0.4, -0.75) {$\Omega'$};

  \filldraw[draw=black, fill=gray!20] (0,0) -- (1,0) -- (0.5, -0.5) -- (0,0) -- cycle;
  \filldraw[draw=black, fill=gray!20] (1,0) -- (2,0) -- (1.5, -0.5) -- (1,0) -- cycle;
  \filldraw[draw=black, fill=gray!20] (2,0) -- (3,0) -- (2.5, -0.5) -- (2,0) -- cycle;

  \draw [|-|,thick] (4,0) -- (6,0);
  \draw [] (5,0.1) -- (5,-0.1)  node[above, yshift=5] {$1/2$};
  \node [] (omega1) at (4.5,-0.75) {$\omega_1$};
  \node [] (omega2) at (5.5,-0.75) {$\omega_2$};
  \filldraw[draw=black, fill=gray!20] (4,0) -- (5,0) -- (4.5, -0.5) -- (4,0) -- cycle;
  \filldraw[draw=black, fill=gray!20] (5,0) -- (6,0) -- (5.5, -0.5) -- (5,0) -- cycle;

  \draw (5,-0.75) ellipse (1cm and 0.3cm);
  \node[] at(6.3, -0.75) {$\Omega$};

  \draw [->, ] (omegap1) to [bend right = 20] (omega1);

\begin{pgfonlayer}{bg}
    \draw [->] (omegap2) to [bend right = 20] (omega2);
  \end{pgfonlayer}
  \begin{pgfonlayer}{mid}
    \filldraw[white] (3,-1.15) circle (5pt);
  \end{pgfonlayer}

  \draw [->, ] (omegap3) to [bend right = 20] (omega1);

  \node[] at(3.5, -1.4) {$p$};

\end{tikzpicture}
\end{center}

\cref{def:DHomogeneity} asserts that there exists $\pi$
with $p \circ \dec_{\Omega'} =  \dec_\Omega \circ \pi$. Indeed 
we can explicitly construct such a $\pi$ for this example: let $\pi$ be
an automorphism
that sends the interval $[0, 1/3]$ to $[0, 1/4]$,
the interval $[2/3, 1]$ to $[1/4, 1/2]$, and finally $[1/3, 2/3]$ to $[1/2, 1]$.
This construction generalizes nicely to any situation where
the preimages $\dec_{\Omega'}^{-1}(\omega_i')$
and $\dec_{\Omega}^{-1}(\omega_i)$ are all intervals.
In the fully general case,
these preimages can be 
arbitrary Lebesgue-measurable sets, but
every such set is measurably isomorphic 
to an interval~\fremlinf{344J}{34}{38}, so the general case reduces to the one sketched above.

\begin{lemma}[Correspondence] \label{def:DCorrespondence}
  For any countable measurable partition $\{A_i\}_{i\in I}$ of $[0,1]$,
  let $\dautiFix\{A_i\}_{i\in I}$ be the subgroup of $\dauti$
  consisting of those automorphisms $\pi\in\dauti$ fixing $\{A_i\}_{i\in I}$, so
  that $\pi(A_i) \triangle A_i$ negligible for all $i\in I$.
  For any two partitions $\calA$ and $\calB$, it holds that
  $\dautiFix\calA\subseteq\dautiFix\calB$
  iff $\calA$ is finer than $\calB$.
\end{lemma}
\begin{proof}
  If $\calA$ is finer than $\calB$ then certainly
  every $\pi$ fixing $\calA$ fixes $\calB$.
  For the converse,
  suppose for contradiction that $\calA$ is not finer than $\calB$,
  so there is some $A\in\calA$ and $B_1,B_2\in\calB$
  with $A\cap B_1$ and $A\cap B_2$ both nonnegligible.
  Pick an arbitrary $\pi$ swapping $A\cap B_1$
  with $A\cap B_2$;
  $\pi$ fixes $\calA$ but not $\calB$,
  contradicting $\dautiFix\calA\subseteq\dautiFix\calB$.
\end{proof}

The equivalence follows from these lemmas:
\begin{theorem} \label{thm:discrete}
  $\dsheafcat\simeq\dgsets$.
\end{theorem}

\newcommand\dshDay{\mathbin{\nref{def:dshDay}{\otimes}}}
\newcommand\dshinc{\nref{def:dshinc}{i}}
\newcommand\dshjoin{\mathbin{\nref{def:dshjoin}{\pdot}}}
\newcommand\dshemp{\mathbin{\nref{def:dshemp}{\mathrm{emp}}}}
\newcommand\dshjoindom{{\nref{def:dshjoindom}{\mathbb P^2_\perp}}}
\newcommand\dshord{\mathrel{\nref{def:dshord}{\sqsubseteq}}}
\newcommand\dshPred{{\nref{def:dshPred}{\mathrm{Pred}}}}
\newcommand\dshSubst[1]{{\nref{def:dshSubst}{\llbr{#1}}}}
For details, see the appendix: \cref{thm:discrete} follows 
from a specialization of \annoyingappref{app:thm:nominal-situation} 
using \cref{def:DHomogeneity,def:DCorrespondence} to satisfy the preconditions.
This equivalence of categories
extends to an equivalence of Models 1 and 2 of \ref{eq:tinyprobsep}.
The argument is as in \cref{sec:eqv-models-schnom}: we
package Models 1 and 2 into resource monoids in $\dsheafcat$ and $\dgsets$
respectively, and then show they correspond across the equivalence $\dsheafcat\simeq\dgsets$.

To construct the resource monoid packaging Model~1 into $\dsheafcat$,
we make use of a general recipe for constructing models of separation logic
via the \emph{Day convolution}~\cite{o1995syntactic,dongol2016convolution,biering2004logic}.
  The Day convolution is a general construction
  lifting a monoidal structure on a base category $C$
  to a monoidal structure on $\funcat{C\op}\Set$, see \citet{day2006closed}.
  The resource monoid $(\schord,\schSepStores,\schinc,\schjoin)$
  in $\Sch$ described in \cref{sec:sch-nom} can be constructed using the Day convolution:
  the base category $\StoreSh$
  has a monoidal product given by coproduct of finite sets,
  and Day convolution lifts this to a monoidal product $\otimes$
  on $\funcat{\StoreSh\op}\Set$;
  applying the Day convolution to the sheaf $\Store$
  gives the functor $\Store\otimes\Store$, which
  one can show is naturally isomorphic to $\schSepStores$;
  the operations $\schord,\schinc,\schjoin$ can then be defined
  straightforwardly.

To apply this recipe for discrete probability,
we replace $\StoreSh$
with $\CSur$ 
and coproduct $+$ of finite sets
with product $\times$ of sample spaces.
This makes $(\CSur,\times,1)$ a monoidal category,
where the unit $1$ is the one-point sample space.
\labelword{Via}{def:dshDay} the Day convolution, $\times$ lifts to a monoidal product
$\dshDay$ on $\funcat{\CSurop}\Set$.
Just as $\Store\otimes\Store$ is isomorphic to the functor $\schSepStores$ modelling separated stores,
the Day convolution $\Dpspcs\dshDay\Dpspcs$ is isomorphic to
a sheaf of
probability spaces that can be rendered independent with a suitable joint measure:
\begin{proposition}
  \labelword{The}{def:dshjoindom} functor $\Dpspcs\dshDay\Dpspcs$ is isomorphic to
  an atomic sheaf $\dshjoindom$ sending each $\Omega\in\CSur$ to
  the set
  \[ 
     \begin{aligned}
       &\{\,((\pi_1\circ p)^{-1}(\calP_1), (\pi_2\circ p)^{-1}(\calP_2)) \\
       &\!\mid \Omega_1,\Omega_2\in\CSur,\calP_1\in \Dpspcs(\Omega_1), \calP_2\in\Dpspcs(\Omega_2),
         p : \Omega\twoheadrightarrow\Omega_1\times\Omega_2\,\}
     \end{aligned}
  \]
  of pairs of probability spaces on $\Omega$
  that ``factor through'' a product $\Omega_1\times\Omega_2$
  along some projection $p : \Omega\twoheadrightarrow\Omega_1\times\Omega_2$.
\end{proposition}

The resource monoid operations can be defined as follows.
  \labelword{First}{def:dshord}, the subspace ordering $\subpspcs$
  forms a natural transformation $(\dshord) : \Dpspcs\times\Dpspcs\to\DProp$.
  \labelword{Next}{def:dshjoin}, there
  is natural transformation
  $(\dshjoin) : \dshjoindom\to\Dpspcs$
  sending a pair
  $((\pi_1\circ p)^{-1}(\calP_1), (\pi_2\circ p)^{-1}(\calP_2))$
  of probability spaces that factor through
  some $p : \Omega\twoheadrightarrow\Omega_1\times\Omega_2$
  to
  $p^{-1}(\calP_1\otimes\calP_2)$,
  where $\calP_1\otimes\calP_2$
  is the usual product probability space on $\Omega_1\times\Omega_2$.
  \labelword{Each}{def:dshinc}
  $\dshjoindom(\Omega)$
  is a subset of $(\Dpspcs\times\Dpspcs)(\Omega)$,
  and collecting the canonical subset-inclusions
  into an $\Omega$-indexed family
  forms a natural transformation
  $\dshinc: \dshjoindom\hookrightarrow\Dpspcs\times\Dpspcs$.
  Finally, $\dshjoin$ is associative and
  commutative and monotone with respect to $\dshord$,
  \labelword{and}{def:dshemp} has unit the natural transformation $\dshemp : 1\to \Dpspcs$
  sending a sample space $\Omega$ to the trivial probability space $(\Omega,\{\emptyset,\Omega\},\mu)$
  with $\mu(\Omega) = 1$.

\begin{proposition} \label{prop:dsh-rm}
  $(\dshord,\dshjoindom,\dshinc,\dshjoin,\dshemp)$
  is a resource monoid in $\dsheafcat$.
\end{proposition}

\newcommand\dgDay{\mathbin{\nref{def:dgDay}{\sepcon}}}
\newcommand\dginc{\nref{def:dginc}{\overline{i}}}
\newcommand\dgjoin{\mathbin{\nref{def:dgjoin}{\overline{\pdot}}}}
\newcommand\dgemp{\mathbin{\nref{def:dgemp}{\overline{\mathrm{emp}}}}}
\newcommand\dgord{\mathrel{\nref{def:dgord}{\overline{\sqsubseteq}}}}
\newcommand\dgjoindom{{\nref{def:dgjoindom}{\overline{\mathbb P}^2_\perp}}}
\newcommand\dgPred{{\nref{def:dgPred}{\overline{\mathrm{Pred}}}}}
\newcommand\dgSubst[1]{{\nref{def:dgSubst}{\overline{\llbr{#1}}}}}
Model~2 can be packaged into a resource monoid
in $\dgsets$ analogously.
\labelword{Let}{def:dgjoindom} $\dgjoindom$
be the \dagoodname{}
  $\{(\calP_1,\calP_2)\mid \calP_1,\calP_2\in\nomDpspcs, \calP_1\Dicom\calP_2\text{ defined}\}$
  of pairs of independently combinable countable measured partitions of $[0,1]$
  with pointwise group action.
  \labelword{This}{def:dginc}
  is a subset of $\nomDpspcs\times\nomDpspcs$;
  both the canonical inclusion map $\dginc$
  \labelword{and}{def:dgjoin}
  the function $\dgjoin : \dgjoindom\to\nomDpspcs$
  sending a pair $(\calP_1,\calP_2)$ of independently combinable
  probability spaces on $[0,1]$
  to their independent combination $\calP_1\Dicom\calP_2$
  are equivariant.
  \labelword{Finally}{def:dgord}, the
  ordering relation $\Dorder$ on probability spaces on $[0,1]$
  is equivariant, so defines a morphism
  $(\dgord) : \nomDpspcs\times\nomDpspcs\to\nomDProp$,
  \labelword{and}{def:dgemp} this ordering relation has as least element $\dgemp$
  the measured partition containing a single component with probability $1$.

The following theorem establishes that 
$(\dgord,\dgjoindom,\dginc,\dgjoin)$ is a resource monoid
together with an analog of \cref{fact:folklore} for discrete probability:

\begin{theorem} \label{thm:discrete-krm}
  The resource monoid $(\dshord,\dshjoindom,\dshinc,\dshjoin,\dshemp)$
  corresponds to $(\dgord,\dgjoindom,\dginc,\dgjoin,\dgemp)$
  across the equivalence $\dsheafcat\simeq\dgsets$.
\end{theorem}

\section{The continuous case} \label{sec:cont}

In this section we generalize \cref{sec:discrete}
from discrete to continuous probability:
$\dsheafcat$ becomes
a category $\sheafcat$ of \emph{\goodsheaves{}},
and $\dgsets$ becomes a category $\cgsets$
of \emph{\agoodname{s}}.
Due to the amount of measure theory
required, we stick to stating the key definitions and lemmas;
the full details can be found in \appref{app:sec:goodsheaves,app:sec:agoodnames}.

\subsection{\Goodsheavesnolink{}}

The first step in generalizing $\dsheafcat$
to $\sheafcat$ is to replace the base category $\CSur$
of discrete sample spaces with a base category of continuous
sample spaces.

The starting point for this generalization is the following observation.
  \labelword{Let}{def:CountProb} $\CountProb$ be the category
  whose objects are countable probability spaces
  $(\Omega,\mu : \Omega\to[0,1])$
  with $\mu(\omega) > 0$ for all $\omega\in\Omega$,
  and whose morphisms
  $(\Omega, \mu)\to(\Omega',\mu')$
  are measure-preserving maps $f : \Omega\to\Omega'$;
  i.e., $\sum_{f(\omega) = \omega'}\mu(\omega) = \mu'(\omega')$
  for all $\omega'\in\Omega'$.
  \labelword{There}{def:dforgetmu} is a 
  functor $\dforgetmu : \CountProb\to\CSur$
  that forgets the measures $\mu$:
  measure-preserving maps $f$
  between probability spaces with strictly positive measure
  are surjective because
  $f^{-1}(y)$ must be nonempty for all $y\in\cod(f)$.
  The category $\CSur$ is the image of $\CountProb$ under
  $\dforgetmu$:
  every nonempty countable set $\Omega$ can be equipped with a strictly positive
  probability measure,
  and for every surjective function $p : \Omega'\twoheadrightarrow\Omega$,
  there exist strictly positive probability measures $\mu'$ on $\Omega'$
  and $\mu$ on $\Omega$ making $p$ a measure-preserving map $(\Omega',\mu')\to(\Omega,\mu)$.
Thus $\CSur$ can be thought of as a category of probability spaces
where one has forgotten all measures.

To generalize this situation from discrete to continuous probability,
we replace the category $\CountProb$ of countable probability spaces
with a category of continuous probability spaces:

\begin{definition}[The category $\StdProb$] \label{def:stdprob}
Let $\StdProb$ be the category
of \emph{standard probability spaces}~\citep{rohlin1949fundamental}
and measure-preserving maps
quotiented by almost-everywhere equality.
\end{definition}

Then, we replace the functor
$\dforgetmu : \CountProb\to\CSur$
with a functor $\forgetmu$ that forgets
continuous probability measures.
The idea behind this forgetting process is as follows.
Given a probability space $(X,\calF,\mu)$,
one can forget everything about the measure $\mu$
except for which subsets are negligible,
leaving behind an \emph{enhanced measurable space} $(X,\calF,\calN)$,
where $\calN$ is the $\sigma$-ideal of $\mu$-negligible sets~\cite[Definition 4.4]{pavlov2022gelfand}.
    Given a measure-preserving map $[f] : (X,\calF,\mu)\to (Y,\calG,\nu)$
    quotiented by almost-everywhere equality
    where $\mu$ has negligibles $\calN$ and $\nu$ has negligibles $\calM$,
    one can forget everything about $[f]$ measure-preserving except
    for the fact that
    $\nu(G) = 0$ iff $\mu(f^{-1}(G)) = 0$,
    leaving behind an equivalence class
    $[f]$
    with $f^{-1}(G)\in\calN$ iff $G\in\calM$
    for all $G\in\calG$.
This motivates the following definitions.

\begin{definition}[The category $\StdMble$] \label{def:mblecat}
A \emph{standard enhanced measurable space}
is tuple $(X,\calF,\calN)$ for which there
exists a measure $\mu$ making $(X,\calF,\mu)$ a standard probability space
with negligibles $\calN$.
Given enhanced measurable spaces $(X,\calF,\calN)$
and $(Y,\calG,\calM)$,
a measurable map $f : (X,\calF)\to(Y,\calG)$
is \emph{negligible-preserving and reflecting}
if $f^{-1}(G)\in\calN$ iff $G\in\calM$ for all $G\in\calG$;
two such maps $f,f'$ are \emph{almost-everywhere equal}
if $f^{-1}(G)\triangle f'^{-1}(G')\in\calN$ for all $G,G'\in\calG$ with $G\triangle G'\in\calM$.
Let $\StdMble$ be the category of
standard enhanced measurable spaces
and negligible-preserving-and-reflecting maps quotiented by almost-everywhere equality.
\end{definition}

\begin{proposition}
\labelword{Let}{def:forgetmu} $\forgetmu : \StdProb\to\StdMble$
be the functor that sends probability spaces $(X,\calF,\mu)$
with negligibles $\calN$ to enhanced measurable spaces $(X,\calF,\calN)$.
This functor is surjective on objects,
and any morphism of standard enhanced measurable spaces
arises from a measure-preserving map equipping
those spaces with standard probability measures.
\end{proposition}

Then, just as $\dsheafcat$ is the category of atomic sheaves on $\CSur$,
$\sheafcat$ is the category of atomic sheaves on $\StdMble$:
\begin{proposition} \label{lem:mble-right-ore}
  $\StdMble$ has the \nref{prop:csur-ore}{right Ore property}.
\end{proposition}

\begin{definition} \label{def:goodsheaf} \label{def:sheafcat}
  Let $\sheafcat$ be the full subcategory of $\funcat{\StdMbleop}\Set$
  consisting of atomic sheaves.
  Objects of $\sheafcat$ will be called
  \emph{\goodsheaves{}}.
\end{definition}

Inside $\sheafcat$, there are continuous analogs
of the \dgoodsheaves{} $\DRV$ of random variables
and $\Dpspcs$ of discrete probability spaces.
The continuous analog of $\DRV$ models
$A$-valued random variables for $A$ Polish,
following \citet{simpsonsheafslides,simpson2017probability}:

\begin{definition} \label{lem:rv-mble-sheaf}
  For any measurable space $(A,\calG)$
  arising from a Polish space,
  the \emph{sheaf of random variables} is:
  \begin{align*}
    &\RVshf_A(\Omega,\calF,\calN) = \{ \text{measurable maps }(\Omega,\calF)\to (A,\calG) \}\,/\,=_{\calN\text{-\rm a.e.}}
    \\
    &\RVshf_A(p : \Omega' \to \Omega)([X] : \RVshf_A(\Omega)) : \RVshf_A(\Omega') = [X\circ p]
  \end{align*}
\end{definition}

For proof that $\RVshf_A$ is indeed a sheaf, 
see \appref{app:lem:rv-mble-sheaf}.
Next, to generalize $\Dpspcs$ from discrete to continuous probability,
we make use of the following observation:
every discrete probability space $(\Omega,\calF,\mu)$
arises via pullback from a surjection $X : \Omega\twoheadrightarrow A$
in which the set $A$ is equipped
with a probability mass function $\nu : A\to [0,1]$,
by setting $\calF := \{X^{-1}(a) \mid a \in A\}$
and $\mu(X^{-1}(a)) = \nu(a)$.
Thus, discrete probability spaces $(\Omega,\calF,\mu)$
can be represented by $\CSur$ morphisms
$\Omega\to \dforgetmu(A,\mu)$ for $(A,\mu)\in\CountProb$,
where $\dforgetmu$ is the functor
$\CSur\to\CountProb$ that forgets measures.
This motivates the following generalization to the continuous setting.

\begin{definition} \label{def:mble-pspcs}
  The \emph{sheaf of probability spaces} is
  \[\pspcs:=\colim_{A\,:\,\Core(\StdProb)} \yo(\forgetmu A),\]
  where $\yo$ is the Yoneda embedding, $\Core(\StdProb)$
  is the subcategory of $\StdProb$-isomorphisms,
  and the colimit is taken in presheaves.
  (See \appref{app:def:mble-pspcs} for proof that $\pspcs$ is indeed a sheaf.)
  Concretely, the presheaf $\pspcs$ sends $(\Omega,\calF,\calN):\StdMble$
  to the set of pairs $((A,\calG,\mu),X)$
  where $(A,\calG,\mu):\StdProb$
  and $X$ is a $\StdMble$-map
  from $(\Omega,\calF,\calN)$
  to $\forgetmu(A,\calG,\mu)$,
  quotiented by
      $((A,\calG,\mu),X)\sim((A',\calG',\mu'),X')$ iff
        there is a $\StdProb$-iso $i : (A,\calG,\mu)\to (A',\calG',\mu')$
        with $X' = U(i)\, X$.
  The action on morphisms is given by precomposition.
\end{definition}

\newcommand\StdMbleotimes{{\nref{def:stdmble-tensor}{\otimes}}}
Using $\RVshf$ and $\pspcs$,
we generalize the resource monoid of \cref{thm:discrete-krm}
to a resource monoid of continuous probability spaces.
The monoidal category $(\CSur,\times,1)$
of discrete sample spaces
becomes a monoidal category $(\StdMble,\otimes,1)$
of continuous sample spaces,
with monoidal product $\otimes$
inherited from the usual tensor
product $\otimes_{\StdProb}$ of standard probability spaces:

\newcommand\shjoindom{{\nref{def:shjoindom}{\mathbb P^2_\perp}}}
\newcommand\shjoin{\mathbin{\nref{def:shjoin}{\pdot}}}
\newcommand\shinc{{\nref{def:shinc}{i}}}
\newcommand\shord{\mathrel{\nref{def:shord}{\sqsubseteq}}}
\newcommand\CProp{{\nref{def:CProp}{\mathrm{Prop}}}}
\begin{definition} \label{def:stdmble-tensor}
  Given two standard enhanced measurable spaces $X,Y$,
  their \emph{tensor product}
  $X\StdMbleotimes Y$
  is defined to be
  $\forgetmu (X'\otimes_{\StdProb} Y')$,
  where
  $X'$ and $Y'$ are arbitrary standard probability spaces
  with $\forgetmu (X') = X$ and $\forgetmu (Y') = Y$.
\end{definition}
This is well-defined --- the choice of $X',Y'$ does not matter ---
and extends to a bifunctor on $\StdMble$
making $(\StdMble,\StdMbleotimes,1)$ a symmetric monoidal category
with unit the one-point space $1$.
For details, see \appref{app:sec:stdmblespc}.
Lifting $\StdMbleotimes$ to $\funcat\StdMbleop\Set$
via the Day convolution
yields a resource monoid in $\sheafcat$:

\begin{lemma}
  \labelword{The}{def:shinc} Day convolution $\pspcs\otimes\pspcs$
  is a sheaf, and there is a monic map of sheaves $\shinc : \pspcs\otimes\pspcs
  \hookrightarrow \pspcs\times\pspcs$.
\end{lemma}

\newcommand\shemp{{\nref{def:shemp}{\mathrm{emp}}}}
\begin{lemma} \label{def:shord}
  \labelword{There}{def:shemp} is a map of sheaves $\shord : \pspcs\times\pspcs\to\CProp$,
  \labelword{where}{def:CProp} $\CProp$ is the constant sheaf at $\{\vtrue,\vfalse\}$,
  and a map of sheaves $\shemp : 1\to\pspcs$,
  making $(\pspcs,\shemp)$ a poset in $\sheafcat$ with least element $\shemp$.
\end{lemma}

\begin{lemma} \label{def:shjoin}
  There is a map of sheaves $\shjoin : \pspcs\otimes\pspcs\to\pspcs$,
  monotone with respect to $\shord$,
  such that $(\pspcs,\shjoin,\shemp)$ is a partial commutative monoid in $\sheafcat$.
\end{lemma}

\begin{theorem} \label{def:sheafresourcemonoid}
  $(\shord,\shjoindom,\shinc,\shjoin,\shemp)$
  is a resource monoid in $\sheafcat$.
\end{theorem}

For details, see \appref{app:sec:goodsheaves}.
While the colimit presentation of $\pspcs$ makes it easier to check for sheafhood
and to construct the above resource monoid,
it is difficult to work with in the concrete calculations to follow.
To address this, we show $\pspcs$ equivalent to a sheaf of continuous probability spaces
that arise via pullback along $\StdMble$-maps.
To do this,
we must take care to define pullback in a way
that respects the negligible ideals contained in $\StdMble$-objects.

\newcommand\epbmeas[2]{{\nref{def:epb}{#1^*}}#2}

\begin{definition} \label{def:epb}
  For $(X,\calF,\calN)\in\StdMble$
  and $(Y,\calG,\mu)\in\StdProb$
  and $f : (X,\calF,\calN)\to\forgetmu(Y,\calG,\mu)$,
  the \emph{enhanced pullback of $(Y,\calG,\mu)$ along $f$},
  written $\epbmeas f(\calG,\mu)$,
  is the pair $(\epbmeas f \calG,\epbmeas f \mu)$ defined by
  \begin{align*}
    &\epbmeas f\calG = \{f^{-1}(G)\triangle N \mid G\in\calG,N'\in\calN\} \\
    &\epbmeas f\mu(f^{-1}(G)\triangle N) = \mu(G)\text{ for all $G\in\calG,N\in\calN$}
  \end{align*}
  Enhanced pullback makes $(X,\epbmeas f \calG,\epbmeas f \mu)$
  a probability space with negligibles $\calN$
  and $f$ a measure-preserving map $(X,\epbmeas f\calG,\epbmeas f\mu)\to(Y,\calG,\mu)$.
\end{definition}

\begin{definition} \label{def:pspc-on-ems}
  A \emph{probability space on} $(X,\calF,\calN)\in\StdMble$
  is a pair $(\calG,\mu)$
  with $\calN\subseteq \calG\subseteq\calF$
  and $\mu$ a probability measure with negligibles $\calN$.
  \labelword{Call}{def:standardizable} such a pair \emph{standardizable} if $(X,\calG,\mu)$
  arises via \nref{def:epb}{enhanced pullback}
  along a map $f : (X,\calF,\calN)\to \forgetmu(Y,\calG,\mu)$
  for some $(Y,\calG,\mu)\in\StdProb$.
\end{definition}

With these definitions in hand,
the colimit $\pspcs$
is equivalent to 
a sheaf of standardizable probability spaces,
with action on morphisms given by enhanced pullback:

\newcommand\epspcs{{\nref{def:mble-pspcs-computed}{\hat\texpspcs}}}
\begin{lemma} \label{def:mble-pspcs-computed}
  $\pspcs$ is equivalent to the following sheaf:
  \begin{align*}
    &\epspcs(\Omega)
    = \{(\calG,\mu) \mid \text{$(\calG,\mu)$ \nref{def:pspc-on-ems}{standardizable on}
    $\Omega$}\}
    \\
    &\epspcs(f : \Omega'\to\Omega)(\calG,\mu)
    = \epbmeas f{(\calG,\mu)}
  \end{align*}
\end{lemma}

Moreover, the Day convolution $\pspcs\otimes\pspcs$
corresponds to a sheaf of independently combinable probability spaces:

\begin{lemma} \label{def:shjoindom}
  $\pspcs\otimes\pspcs$ is equivalent to
  the following sheaf $\shjoindom$:
  \begin{align*}
    &\shjoindom(\Omega) = \left\{ ((\calG,\mu),(\calH,\nu)) 
    \left|~ \begin{aligned}
      &\text{$(\calG,\mu)$ and $(\calH,\nu)$ standardizable}\\
      &\text{and independently combinable}
    \end{aligned}\right.\right\}
    \\
    &\shjoindom(f : \Omega'\to\Omega)((\calG,\mu),(\calH,\nu))
    = (\epbmeas f {(\calG,\mu)}
    , \epbmeas f {(\calH,\nu)}
    )
  \end{align*}
\end{lemma}
Via these equivalences,
the resource monoid in \cref{def:sheafresourcemonoid}
parallels its discrete analog (\cref{prop:dsh-rm}).
Across $\pspcs\cong\epspcs$, the ordering $\shord$ corresponds to the generalization of \cref{def:Dorder}
from countable measured partitions to standardizable probability spaces.
Across $\pspcs\otimes\pspcs\cong\shjoindom$, the monic map $\shinc$ corresponds
to the canonical inclusion $\shjoindom\hookrightarrow \pspcs\times\pspcs$,
and the combining operation $\shjoin$ corresponds to
the map $\shjoindom\to\pspcs$ that
sends independently-combinable pairs
of standardizable probability spaces to their independent combination.
For details, see \appref{app:sec:prob-concepts-as-goodnames}.

\subsection{\Agoodnamenolink{s}} \label{sec:pnom}

Finding a continuous analog to $\dgsets$ boils down
to showing continuous analogs of \cref{def:DHomogeneity,def:DCorrespondence}.
In the discrete setting, these lemmas hold
because every discrete probability space
can be encoded as a measured partition that leaves
enough room in the sample space $[0,1]$ for fresh randomness.
To create a continuous analog,
we fix an enormous sample space following \citet{li2023lilac}:

\newcommand\hcubeproj[1]{{\nref{def:projn}{\mathrm{proj}_{1..#1}}}}
\newcommand\embedmap[1]{{#1\nref{def:embedmap}{\otimes1_\hcube}}}

\begin{definition} \label{def:hcube}
  The Hilbert cube $\hcube$ is the standard enhanced measurable space $([0,1]^\omega,\calF,\calN)$
  of infinite sequences in the interval $[0,1]$.
  The $\sigma$-algebra $\calF$
  and negligibles $\calN$ are those of the usual Lebesgue measure
  on $[0,1]^\omega$.
\end{definition}

Then, to ensure that there is always enough room left over in $\hcube$ for 
fresh randomness, we encode all probability spaces 
using only finitely many dimensions at a time:

\begin{definition} \label{def:ff-pspc}
  A \nref{def:standardizable}{standardizable} probability space $(\calG,\mu)$ on $\hcube$
  has \emph{finite footprint}
  if it arises by \nref{def:epb}{enhanced pullback} along a
  map $\hcube\to X$ that factors through 
  $\hcubeproj n$ for some $n$,
  \labelword{where}{def:projn}
  $\hcubeproj n$ is the canonical projection $\hcube\to[0,1]^n$.
\end{definition}

Analogously, the group $\dauti$ of $\dgsets$
becomes a group of finite-dimensional permutations of the Hilbert cube:

\begin{definition} \label{def:autoff}
\label{def:auti}
  A $\StdMble$-automorphism $\pi : \hcube\to\hcube$
  has \emph{finite width}
  if it is of the form $\pi'\times 1_{\hcube}$
  for some $\StdMble$-automorphism
  $\pi' : [0,1]^n\to[0,1]^n$.
  Let $\auti$ be the subgroup of $\Aut_\StdMble\hcube$
  consisting only of those automorphisms with \nref{def:autoff}{finite width}.
\end{definition}

Then, the topology on $\dauti$ generated by countable measurable partitions
becomes a topology on $\auti$ generated by standardizable sub-$\sigma$-algebras with finite footprint:

\newcommand\autiFix{\operatorname{\nref{def:autiFix}{Fix}}}
\begin{definition}[Topology on $\auti$] \label{def:auti-open} \label{def:autiFix}
  A subgroup $U$ of $\auti$ is
  \emph{open} if for every $\pi$ in $U$
  there exists $(\calF,\mu)$ with finite footprint
  such that $\pi\in\autiFix\calF\subseteq U$,
  where $\autiFix \calF$ is the subgroup of
  those $\pi$ in $\auti$ with $\pi(F) \setaseq F$ for all $F\in\calF$.
\end{definition}

\begin{definition} \label{def:pnom} \label{def:cgsets}
  $\cgsets$ is the category of
  $\auti$-sets with \nref{def:auti-open}{open} stabilizers
  and equivariant functions between them;
  objects of $\cgsets$ will be called \emph{\agoodname{s}}.
\end{definition}

There are \agoodname{s}
analogous to the sheaves $\RVshf_A$ of random variables
and $\pspcs$ of standardizable probability spaces:

\begin{definition} \label{def:gRV}
  For $A$ a Polish space,
  a random variable $X : \hcube\to A$
  has \emph{finite footprint}
  if it factors through $\hcubeproj n$ for some $n$.
  Let $\gRV_A$ be the set of random variables with finite footprint.
  This forms \aagoodname{}, with action $X\cdot\pi = X\circ \pi$.
\end{definition}

\begin{definition} \label{def:gpspcs}
  Let $\gpspcs$ be the set of standardizable probability spaces on $\hcube$
  with finite footprint. This forms \aagoodname{},
  with action $(\calF,\mu)\cdot\pi = \epbmeas{\pi}{(\calF,\mu)}$.
\end{definition}

\newcommand\gjoindom{{\nref{def:gjoindom}{\overline{\mathbb P}^2_\perp}}}
\newcommand\gjoin{\mathbin{\nref{def:gjoin}{\overline\pdot}}}
\newcommand\gemp{{\nref{def:gemp}{\overline{\mathrm{emp}}}}}
\newcommand\ginc{{\nref{def:ginc}{\overline i}}}
\newcommand\gord{\mathrel{\nref{def:gord}{\overline\sqsubseteq}}}
\newcommand\GProp{{\nref{def:GProp}{\overline{\mathrm{Prop}}}}}
These yield a resource monoid in $\cgsets$.
\begin{theorem} \label{thm:cgsets-monoid}
$(\gord,\gjoindom,\ginc,\gjoin,\gemp)$
is a $\cgsets$ resource monoid, where \begin{itemize}[leftmargin=*]
  \item 
  \labelword{$\gord : \gpspcs\times\gpspcs\to\GProp$}{def:gord}
  is the map that sends $((\calG,\mu),(\calH,\nu))$
  to $\top$ iff $\calG\subseteq\calH$ and $\nu|_\calG = \mu$,
  \labelword{where}{def:GProp} $\GProp$ is the two-element set with trivial action.

  \item 
  \labelword{$\gjoindom$}{def:gjoindom} is the set
  of pairs $((\calG,\mu),(\calH,\nu))\in\gpspcs\times\gpspcs$
  for which $(\calG,\mu)$ and $(\calH,\nu)$
  are independently combinable.

  \item 
  \labelword{$\ginc$}{def:ginc} is the inclusion 
  $\gjoindom\hookrightarrow\gpspcs\times\gpspcs$.

  \item 
  \labelword{$\gjoin : \gjoindom\to\gpspcs$}{def:gjoin}
  is the map that sends independently-combinable pairs
  to their independent combination.

  \item 
  \labelword{$\gemp : 1\to\gpspcs$}{def:gemp} is the constant map
  at the probability space $\epbmeas{f}1$ on $\hcube$
  arising from enhanced pullback along
  the unique $\StdMble$-map $\hcube\to\forgetmu(1)$
  into the one-point probability space $1$.
\end{itemize}
\end{theorem}

\subsection{The equivalence}

By choosing $\hcube$ as underlying sample space
and topologizing $\Aut\hcube$
to permit only objects that use finitely-many dimensions
of $\hcube$ at a time,
we obtain
continuous analogs of \Homogeneity{} and \Correspondence{}.
This relies crucially on
both the finiteness of footprints and the inclusion of negligible ideals
in the base category 
$\StdMble$.
Negligible ideals
allow passing to \emph{measure algebra}~\fremlinf{321A}{32}{1}:

\begin{definition}
  A \emph{measure algebra} is a tuple $(\frakA,\overline\mu)$
  consisting of a complete Boolean algebra $\frakA$
  and a function $\overline\mu : \frakA\to[0,1]$
  such that (1) $\overline\mu(A) > 0$ for $A\ne\bot$
  and (2)
  $\overline\mu$ is countably additive in the sense that
  $\overline\mu(\bigvee_i A_i) = \sum_i\overline\mu(A_i)$
  for all countable families $\{A_i\}_{i\in I}$ with $A_i\wedge A_j = \bot$
  for all $i\ne j$.
  A \emph{measure algebra homomorphism}
  from $(\frakA,\overline\mu)$ to $(\frakB,\overline\nu)$
  is a complete Boolean algebra homomorphism $f : \frakA\to\frakB$,
  measure-preserving in the sense that $\overline\nu(f(A)) = \overline\mu(A)$
  for all $A\in\frakA$.
\end{definition}

Every $(X,\calF,\mu)$ in $\StdProb$ yields a measure algebra
$(\calF/\mu,\overline\mu)$,
where $\calF/\mu$ is the complete Boolean algebra
of events $F\in\calF$ mod $F\sim F'$ iff $\mu(F\triangle F') = 0$,
and $\overline\mu([F]) = \mu(F)$~\fremlinf{321H}{32}{3}.
Every measure-preserving map $f$ from $(X,\calF,\mu)$ to $(Y,\calG,\nu)$
defines a homomorphism $f^*$ from $(\calG/\nu,\overline\nu)$ to $(\calF/\mu,\overline\mu)$
sending $[G]\in\calG/\nu$ to $[f^{-1}(G)]\in\calF/\mu$~\fremlinf{324M}{32}{26}.
This gives a duality:
\begin{definition} \label{def:StdProbAlg}
  A \emph{standard probability algebra}
  is a measure algebra $(\frakA,\overline\mu)$
  arising from a standard probability space
  as described above.
  Let $\StdProbAlg$ be the category of standard probability algebras
  and measure algebra homomorphisms between them.
\end{definition}
\begin{lemma} \label{lem:dual1}
  $\StdProb\simeq \StdProbAlgop$. \end{lemma}

A similar duality holds also for $\StdMble$:

\begin{definition} \label{def:StdMbleAlg}
  A \emph{standard measurable algebra}
  is a complete Boolean algebra $\frakA$
  arising from a standard probability space;
  i.e. $\frakA$ is isomorphic to a Boolean algebra $\calF/\mu$
  for some $(X,\calF,\mu)\in\StdProb$.
  Let $\StdMbleAlg$ be the category
  of standard measurable algebras
  and \emph{injective} complete boolean algebra homomorphisms.
\end{definition}

\begin{lemma}  \label{lem:dual2}
  $\StdMble\simeq \StdMbleAlgop$. \end{lemma}

\cref{lem:dual1,lem:dual2}
allow importing the extensive technical development
of measure algebras 
from \citet{fremlin2000measure}.
In particular, the algebraic perspective reveals
that the finite-footprint property from \cref{sec:pnom}
is a means of producing \emph{relatively-atomless} subalgebras:

\begin{definition}[\fremlinft{331A}{33}{1}]
  Let $\frakA$ be a complete Boolean algebra
  and $\frakB\subseteq\frakA$ a subalgebra.
  An element $a\in\frakA$ is a \emph{$\frakB$-relative atom}
  of $\frakA$ if the principal ideal generated by $a$ in $\frakA$
  is $\{a\cap b \mid b\in\frakB\}$. The algebra $\frakA$ is \emph{$\frakB$-relatively atomless}
  if it has no $\frakB$-relative atoms.
\end{definition}

\begin{theorem} \label{thm:relatively-atomless-enough-room}
  Let $\frakA$ be the measurable algebra
  of $\hcube$.
  For any $(\calG,\mu)$ with finite footprint,
  $\frakA$ is $\calG/\mu$-relatively atomless.
\end{theorem}

Relative-atomlessness is
key to obtaining continuous analogs of
\Homogeneity{} and \Correspondence{},
which hold 
specifically for the
case where subalgebras are relatively atomless:

\begin{lemma}[Homogeneity] \label{lem:mble-homogeneity}
  For $\frakA$ a standard measurable algebra
  and subalgebras $\frakB,\frakC\subseteq\frakA$
  that render it relatively-atomless, and 
  a $\StdMbleAlg$-morphism
   $f : \frakB\hookrightarrow\frakC$,
  there exists a complete Boolean algebra automorphism $\pi : \frakA\to\frakA$
  with $\pi(b) = f(b)$ for all $b\in\frakB$.
\end{lemma}

\begin{lemma}[Correspondence] \label{lem:mble-distinguishability}
  Let $\frakA$ be a standard measurable algebra.
  For any subalgebra $\frakC\subseteq\frakA$, 
  let $\Fix\frakC$ be the group
  of $\frakA$-automorphisms fixing every $c$ in $\frakC$.
  If $\frakA$ is $\frakC$-relatively atomless
  then $\Fix\frakC\subseteq\Fix\frakD$ iff
  $\frakD\subseteq\frakC$.
\end{lemma}

These yield a continuous analog of \cref{thm:discrete}:
\begin{theorem} \label{thm:cont}
  $\sheafcat\simeq\cgsets$.
\end{theorem}

Finally, a careful calculation across this equivalence
shows that the resource monoids in 
\cref{def:sheafresourcemonoid,thm:cgsets-monoid}
indeed correspond, yielding
an analog of \cref{fact:folklore} for continuous probability:

\begin{theorem} \label{thm:krm-corresp}
Across $\sheafcat\simeq\cgsets$,
the sheaf $\pspcs$ corresponds to $\gpspcs$,
the sheaf $\RVshf_A$ corresponds to $\gRV_A$,
and the resource monoid
$(\shord,\pspcs\otimes\pspcs,\shinc,\shjoin,\shemp)$
in $\sheafcat$
corresponds to
$(\gord,\gjoindom,\ginc,\gjoin,\gemp)$
in $\cgsets$.
\end{theorem}

\section{Discussion \& related work} \label{sec:discussion}

\goodparagraph{Atomic sheaves for probability} \citet{tao_2015}
defines probabilistic notions as those invariant under extension of
the sample space. Along these lines,
\citet{simpson2017probability} constructs a topos of
atomic sheaves on a category of probability spaces and measure-preserving
maps; in it, he presents a sheaf of random variables
and an extension of the Giry monad~\cite{giry2006categorical} to sheaves,
and shows how concepts such as independence
and expectation can be internalized~\cite{simpsonsheafslides,simpsonSyntheticProbabilityTheory}.

Simpson's topos is similar to our $\sheafcat$, but
our base category $\StdMble$ omits measures
and its maps are quotiented by almost-everywhere equality;
we instead model measures explicitly via the sheaf $\pspcs$.
As we have focused on separation logic,
we have not investigated whether the Giry monad
extends to $\sheafcat$ and
the probabilistic concepts that can be expressed internally;
this would make interesting future work.
\citet{simpson2017probability}
mentions a resemblance to nominal sets,
but does not extensively develop the notion to the best of our knowledge. 

Simpson's topos also serves as a model of Atomic Sheaf Logic~\cite{equivalenceandconditionalindependenceinatomicsheaflogic},
a recently-developed logic axiomatizing the interaction between conditional independence and a notion of
\emph{atomic equivalence}, which in the probabilistic setting denotes equidistribution of random variables,
with potential applications to developing proof-relevant probabilistic separation logics;
it would be interesting to explore whether our topos admits analogous constructions.

\goodparagraph{Categorical probability}
There are numerous categorical formulations of probability.
\citet{fritz2020synthetic} develops probability theory purely synthetically by
axiomatizing equational properties known to hold for Markov kernels. 
\citet{jackson2006sheaf}, building on \citet{breitsprecher2006concept}, 
gives an alternative sheaf-theoretic model of
probability by taking sheaves on a single measurable space rather than a
category of measurable spaces; we speculate that there could be a relationship
between this model and ours similar to the relationship between 
petit and gros topoi of sheaves on topological spaces~\citep{maclane2012sheaves}.

\goodparagraph{Quasi-Borel spaces}
The category QBS of quasi-Borel spaces~\citep{heunen2017convenient}
is a richly developed model of higher-order probability.
Whereas QBS has been used extensively to
model higher-order probabilistic
languages~\citep{vakar2019domain,scibior2017denotational,aguirre2021higher,sato2019formal},
our goal in constructing $\sheafcat$ and $\cgsets$ has been focused on 
refining models of probabilistic
separation logic.
Structurally, QBS and $\sheafcat$ are quite different:
QBS is a well-pointed quasi-topos while $\sheafcat$ is a non-well-pointed topos. However,
as remarked in \citet[Prop. 34]{heunen2017convenient}, 
QBS is related to particular presheaves
on the category of standard measurable spaces.
This suggests connections to
$\sheafcat$, since it is a category of sheaves on $\StdMble$,
but there is a gap between these two settings:
$\StdMble$-morphisms are quotiented by
almost-everywhere equality whereas maps of standard measurable spaces are not.
We leave elucidating the relationship
between our setting and QBS to future work.

\goodparagraph{General representation theorems}
The equivalence $\Sch\simeq\Nom$ 
can be obtained via
a \emph{Fra{\"i}ss{\'e} limit}~\cite[\S 7.1]{hodges1993model},
a recipe for making universal objects (e.g., $\N$)
capable of representing a class of models (e.g., finite sets).
More generally, there is a long line of results giving
groupoid-based representations of
categories~\cite{butz1996representing,blass1989freyd,dubuc2003localic,joyal1984extension,blass1983Boolean,makkai1982full},
with a history going back to Grothendieck~\cite{grothendieck2002rev,dubuc2000galois}.
\citet{caramello2016topological} is particularly relevant, as it
gives conditions closely resembling \cref{def:DHomogeneity,def:DCorrespondence}
under which categories of atomic sheaves
are equivalent to categories of continuous $\Aut(u)$-sets for suitable objects $u$.
We are currently investigating whether \cref{thm:discrete,thm:cont}
can be obtained via these general results,
with an eye towards generalizing beyond probability
to the quantum setting.

\goodparagraph{Probabilistic separation logic}
PSL~\citep{barthe2019probabilistic} is the
first separation logic whose separating conjunction models
independence, by splitting random substitutions; 
it has since been extended to support
conditional independence~\citep{bao2021bunched} 
and negative dependence~\citep{bao2022separation},
and to the quantum setting~\citep{zhou2021quantum}.
In contrast to PSL and its extensions,
Lilac~\citep{li2023lilac} has
an alternative model of separation,
via independent combination.
Lilac's model is complicated:
independent combination is an intricate measure-theoretic operation,
an intricate proof is required to show it forms a monoid,
and many side conditions on this monoid are
needed for soundness of Lilac's proof rules.

\cref{thm:krm-corresp} simplifies and clarifies Lilac's model.
It shows that independent combination arises naturally
from the well-known tensor product of standard probability spaces;
that independent combination
forms a monoid then follows from the fact that tensor product is monoidal.
The resource monoid in \cref{thm:krm-corresp}
replaces the side conditions on Lilac's monoid
with the single notion of \nref{def:standardizable}{standardizability} ---
a condition well-motivated by the intuition that probability spaces
should arise via pullback along $\StdMble$-maps.

\cref{thm:krm-corresp} also
improves on the model in \citet{li2023lilac}
in multiple ways.
Quotienting by negligiblity
yields a model invariant under
almost-everywhere equality, whereas
the model in \citet{li2023lilac} must manually track $\sigma$-ideals 
of negligible sets.
Interpreting propositions as equivariant maps
implies our model is
invariant under finite-width permutations of $\hcube$.
Finally, using the internal language of $\sheafcat$,
one can interpret quantification over propositions,
allowing to generalize
Lilac to a higher-order logic;
in the future, we would like to explore whether this higher-order generalization
can be used to specify properties of higher-order programs.

An aspect of Lilac not captured by our model
is its \emph{conditioning modality},
interpreted by disintegration~\citep{chang1997conditioning}.
This is difficult to capture in our model because $\StdMble$-objects
come with a fixed collection of negligible sets, whereas
disintegration can change which sets are negligible.

\goodparagraph{Probability and name generation}
Recent work has identified connections between probability theory
and name generation:
\citet{staton2016semantics} provides a semantics for a
probabilistic language
that treats random variables as dynamically-allocated read-only
names, and
\citet{sabok2021probabilistic} show that QBS can be used
to characterize observational equivalence of stateful imperative programs
by interpreting dynamic allocation
as probabilistic sampling.
The resemblance between our probabilistic \cref{thm:krm-corresp} and
the store-based \cref{fact:folklore} provides further evidence
along these lines.

\goodparagraph{Nominal sets}
Many constructions exist in
$\Nom$ beyond its ability to capture permutation-invariance:
\emph{freshness quantification}~\cite{menni2003quantifiers}
captures the informal convention of picking fresh names~\cite{barendregt1984lambda},
a \emph{name abstraction}~\cite[\S4]{pitts2013nominal} type former
gives a uniform treatment of binding,
and \emph{nominal restriction sets}~\cite[\S9.1]{pitts2013nominal}
models languages with locally generated names~\cite{odersky1994functional,pitts1993observable}.
It would be interesting to explore whether analogous constructions
can be carried out in $\cgsets$, to obtain analogous treatments
of the informal convention of picking fresh sample spaces~\fremlinf{\S27}{27}{1}
and to provide models of probabilistic languages with locally generated random variables.

\section{Conclusion}
We unify two different approaches to separating probabilistic state: the usual
product of probability spaces and independent combination.
To do this, we show that separation-as-product lives in a category
$\sheafcat$ of \goodsheaves{},
that separation-as-independent-combination lives in a category
$\cgsets$ of \agoodname{s},
and that these two notions of separation correspond
across an equivalence $\sheafcat\simeq\cgsets$.
This validates the use of independent combination
in probabilistic separation logic~\citep{li2023lilac},
clarifies independent combination's relationship with traditional formulations
of independence, and suggests improvements to existing models.
Finally, as a probabilistic analog of $\Nom$,
the category $\cgsets$ creates new probabilistic interpretations
of nominal concepts, which we hope will
create more opportunities for using nominal techniques in probability.

\section{Acknowledgements}

We thank Minsung Cho, 
Anthony D'Arienzo,
Ryan Doenges, 
John Gouwar,
and Max New
for their careful feedback
and suggestions.
This work was supported by the National Science Foundation under Grant No. \#CCF-2220408.
Sandia National Laboratories is a multimission laboratory
managed and operated by National Technology \& Engineering
Solutions of Sandia, LLC, a wholly owned subsidiary
of Honeywell International Inc., for the U.S. Department of
Energy's National Nuclear Security Administration under
contract DE-NA0003525. SAND No. SAND2024-03245O C.

\bibliographystyle{ACM-Reference-Format}
\bibliography{main}


\begin{thebibliography}{55}


\ifx \showCODEN    \undefined \def \showCODEN     #1{\unskip}     \fi
\ifx \showDOI      \undefined \def \showDOI       #1{#1}\fi
\ifx \showISBNx    \undefined \def \showISBNx     #1{\unskip}     \fi
\ifx \showISBNxiii \undefined \def \showISBNxiii  #1{\unskip}     \fi
\ifx \showISSN     \undefined \def \showISSN      #1{\unskip}     \fi
\ifx \showLCCN     \undefined \def \showLCCN      #1{\unskip}     \fi
\ifx \shownote     \undefined \def \shownote      #1{#1}          \fi
\ifx \showarticletitle \undefined \def \showarticletitle #1{#1}   \fi
\ifx \showURL      \undefined \def \showURL       {\relax}        \fi
\providecommand\bibfield[2]{#2}
\providecommand\bibinfo[2]{#2}
\providecommand\natexlab[1]{#1}
\providecommand\showeprint[2][]{arXiv:#2}

\bibitem[Aguirre et~al\mbox{.}(2021)]%
        {aguirre2021higher}
\bibfield{author}{\bibinfo{person}{Alejandro Aguirre}, \bibinfo{person}{Gilles Barthe}, \bibinfo{person}{Marco Gaboardi}, \bibinfo{person}{Deepak Garg}, \bibinfo{person}{Shin-ya Katsumata}, {and} \bibinfo{person}{Tetsuya Sato}.} \bibinfo{year}{2021}\natexlab{}.
\newblock \showarticletitle{Higher-order probabilistic adversarial computations: categorical semantics and program logics}.
\newblock \bibinfo{journal}{\emph{Proceedings of the ACM on Programming Languages}} \bibinfo{volume}{5}, \bibinfo{number}{ICFP} (\bibinfo{year}{2021}), \bibinfo{pages}{1--30}.
\newblock


\bibitem[Axler(2020)]%
        {axler2020measure}
\bibfield{author}{\bibinfo{person}{Sheldon Axler}.} \bibinfo{year}{2020}\natexlab{}.
\newblock \bibinfo{booktitle}{\emph{Measure, integration \& real analysis}}.
\newblock \bibinfo{publisher}{Springer Nature}.
\newblock


\bibitem[Bao et~al\mbox{.}(2021)]%
        {bao2021bunched}
\bibfield{author}{\bibinfo{person}{Jialu Bao}, \bibinfo{person}{Simon Docherty}, \bibinfo{person}{Justin Hsu}, {and} \bibinfo{person}{Alexandra Silva}.} \bibinfo{year}{2021}\natexlab{}.
\newblock \showarticletitle{A bunched logic for conditional independence}. In \bibinfo{booktitle}{\emph{2021 36th Annual ACM/IEEE Symposium on Logic in Computer Science (LICS)}}. IEEE, \bibinfo{pages}{1--14}.
\newblock


\bibitem[Bao et~al\mbox{.}(2022)]%
        {bao2022separation}
\bibfield{author}{\bibinfo{person}{Jialu Bao}, \bibinfo{person}{Marco Gaboardi}, \bibinfo{person}{Justin Hsu}, {and} \bibinfo{person}{Joseph Tassarotti}.} \bibinfo{year}{2022}\natexlab{}.
\newblock \showarticletitle{A separation logic for negative dependence}.
\newblock \bibinfo{journal}{\emph{Proceedings of the ACM on Programming Languages}} \bibinfo{volume}{6}, \bibinfo{number}{POPL} (\bibinfo{year}{2022}), \bibinfo{pages}{1--29}.
\newblock


\bibitem[Barendregt et~al\mbox{.}(1984)]%
        {barendregt1984lambda}
\bibfield{author}{\bibinfo{person}{Hendrik~P Barendregt} {et~al\mbox{.}}} \bibinfo{year}{1984}\natexlab{}.
\newblock \bibinfo{booktitle}{\emph{The lambda calculus}}. Vol.~\bibinfo{volume}{3}.
\newblock \bibinfo{publisher}{North-Holland Amsterdam}.
\newblock


\bibitem[Barthe et~al\mbox{.}(2019)]%
        {barthe2019probabilistic}
\bibfield{author}{\bibinfo{person}{Gilles Barthe}, \bibinfo{person}{Justin Hsu}, {and} \bibinfo{person}{Kevin Liao}.} \bibinfo{year}{2019}\natexlab{}.
\newblock \showarticletitle{A probabilistic separation logic}.
\newblock \bibinfo{journal}{\emph{Proceedings of the ACM on Programming Languages}} \bibinfo{volume}{4}, \bibinfo{number}{POPL} (\bibinfo{year}{2019}), \bibinfo{pages}{1--30}.
\newblock


\bibitem[Biering(2004)]%
        {biering2004logic}
\bibfield{author}{\bibinfo{person}{Bodil Biering}.} \bibinfo{year}{2004}\natexlab{}.
\newblock \showarticletitle{On the Logic of Bunched Implications and its Relation to Separation Logic}.
\newblock  (\bibinfo{year}{2004}).
\newblock


\bibitem[Blass and Scedrov(1983)]%
        {blass1983Boolean}
\bibfield{author}{\bibinfo{person}{Andreas Blass} {and} \bibinfo{person}{Andrej Scedrov}.} \bibinfo{year}{1983}\natexlab{}.
\newblock \showarticletitle{Boolean classifying topoi}.
\newblock \bibinfo{journal}{\emph{Journal of Pure and Applied Algebra}} \bibinfo{volume}{28}, \bibinfo{number}{1} (\bibinfo{year}{1983}), \bibinfo{pages}{15--30}.
\newblock


\bibitem[Blass and {\v{S}}{\v{c}}edrov(1989)]%
        {blass1989freyd}
\bibfield{author}{\bibinfo{person}{Andreas Blass} {and} \bibinfo{person}{Andrej {\v{S}}{\v{c}}edrov}.} \bibinfo{year}{1989}\natexlab{}.
\newblock \bibinfo{booktitle}{\emph{Freyd's models for the independence of the axiom of choice}}. Vol.~\bibinfo{volume}{404}.
\newblock \bibinfo{publisher}{American Mathematical Soc.}
\newblock


\bibitem[Breitsprecher(2006)]%
        {breitsprecher2006concept}
\bibfield{author}{\bibinfo{person}{Siegfried Breitsprecher}.} \bibinfo{year}{2006}\natexlab{}.
\newblock \showarticletitle{On the concept of a measurable space I}. In \bibinfo{booktitle}{\emph{Applications of Sheaves: Proceedings of the Research Symposium on Applications of Sheaf Theory to Logic, Algebra, and Analysis, Durham, July 9--21, 1977}}. Springer, \bibinfo{pages}{157--168}.
\newblock


\bibitem[Butz and Moerdijk(1996)]%
        {butz1996representing}
\bibfield{author}{\bibinfo{person}{Carsten Butz} {and} \bibinfo{person}{Ieke Moerdijk}.} \bibinfo{year}{1996}\natexlab{}.
\newblock \showarticletitle{Representing topoi by topological groupoids}.
\newblock  (\bibinfo{year}{1996}).
\newblock


\bibitem[Caramello(2016)]%
        {caramello2016topological}
\bibfield{author}{\bibinfo{person}{Olivia Caramello}.} \bibinfo{year}{2016}\natexlab{}.
\newblock \showarticletitle{Topological galois theory}.
\newblock \bibinfo{journal}{\emph{Advances in Mathematics}}  \bibinfo{volume}{291} (\bibinfo{year}{2016}), \bibinfo{pages}{646--695}.
\newblock


\bibitem[Chang and Pollard(1997)]%
        {chang1997conditioning}
\bibfield{author}{\bibinfo{person}{Joseph~T Chang} {and} \bibinfo{person}{David Pollard}.} \bibinfo{year}{1997}\natexlab{}.
\newblock \showarticletitle{Conditioning as disintegration}.
\newblock \bibinfo{journal}{\emph{Statistica Neerlandica}} \bibinfo{volume}{51}, \bibinfo{number}{3} (\bibinfo{year}{1997}), \bibinfo{pages}{287--317}.
\newblock


\bibitem[Day(1970)]%
        {day2006closed}
\bibfield{author}{\bibinfo{person}{Brian Day}.} \bibinfo{year}{1970}\natexlab{}.
\newblock \showarticletitle{On closed categories of functors}. In \bibinfo{booktitle}{\emph{Reports of the Midwest Category Seminar IV}}. Springer, \bibinfo{pages}{1--38}.
\newblock


\bibitem[Dongol et~al\mbox{.}(2016)]%
        {dongol2016convolution}
\bibfield{author}{\bibinfo{person}{Brijesh Dongol}, \bibinfo{person}{Ian~J Hayes}, {and} \bibinfo{person}{Georg Struth}.} \bibinfo{year}{2016}\natexlab{}.
\newblock \showarticletitle{Convolution as a unifying concept: Applications in separation logic, interval calculi, and concurrency}.
\newblock \bibinfo{journal}{\emph{ACM Transactions on Computational Logic (TOCL)}} \bibinfo{volume}{17}, \bibinfo{number}{3} (\bibinfo{year}{2016}), \bibinfo{pages}{1--25}.
\newblock


\bibitem[Dubuc(2003)]%
        {dubuc2003localic}
\bibfield{author}{\bibinfo{person}{Eduardo~J Dubuc}.} \bibinfo{year}{2003}\natexlab{}.
\newblock \showarticletitle{Localic galois theory}.
\newblock \bibinfo{journal}{\emph{Advances in Mathematics}} \bibinfo{volume}{175}, \bibinfo{number}{1} (\bibinfo{year}{2003}), \bibinfo{pages}{144--167}.
\newblock


\bibitem[Dubuc and de~la Vega(2000)]%
        {dubuc2000galois}
\bibfield{author}{\bibinfo{person}{Eduardo~J Dubuc} {and} \bibinfo{person}{C~Sanchez de~la Vega}.} \bibinfo{year}{2000}\natexlab{}.
\newblock \showarticletitle{On the Galois theory of Grothendieck}.
\newblock \bibinfo{journal}{\emph{arXiv preprint math/0009145}} (\bibinfo{year}{2000}).
\newblock


\bibitem[Edalat(1999)]%
        {edalat1999semi}
\bibfield{author}{\bibinfo{person}{Abbas Edalat}.} \bibinfo{year}{1999}\natexlab{}.
\newblock \showarticletitle{Semi-pullbacks and bisimulation in categories of Markov processes}.
\newblock \bibinfo{journal}{\emph{Mathematical Structures in Computer Science}} \bibinfo{volume}{9}, \bibinfo{number}{5} (\bibinfo{year}{1999}), \bibinfo{pages}{523--543}.
\newblock


\bibitem[Fremlin(2000)]%
        {fremlin2000measure}
\bibfield{author}{\bibinfo{person}{David~Heaver Fremlin}.} \bibinfo{year}{2000}\natexlab{}.
\newblock \bibinfo{booktitle}{\emph{Measure theory}}. Vol.~\bibinfo{volume}{4}.
\newblock \bibinfo{publisher}{Torres Fremlin}.
\newblock


\bibitem[Fritz(2020)]%
        {fritz2020synthetic}
\bibfield{author}{\bibinfo{person}{Tobias Fritz}.} \bibinfo{year}{2020}\natexlab{}.
\newblock \showarticletitle{A synthetic approach to Markov kernels, conditional independence and theorems on sufficient statistics}.
\newblock \bibinfo{journal}{\emph{Advances in Mathematics}}  \bibinfo{volume}{370} (\bibinfo{year}{2020}), \bibinfo{pages}{107239}.
\newblock


\bibitem[Galmiche et~al\mbox{.}(2005)]%
        {galmiche2005semantics}
\bibfield{author}{\bibinfo{person}{Didier Galmiche}, \bibinfo{person}{Daniel M{\'e}ry}, {and} \bibinfo{person}{David Pym}.} \bibinfo{year}{2005}\natexlab{}.
\newblock \showarticletitle{The semantics of BI and resource tableaux}.
\newblock \bibinfo{journal}{\emph{Mathematical Structures in Computer Science}} \bibinfo{volume}{15}, \bibinfo{number}{6} (\bibinfo{year}{2005}), \bibinfo{pages}{1033--1088}.
\newblock


\bibitem[Giry(2006)]%
        {giry2006categorical}
\bibfield{author}{\bibinfo{person}{Michele Giry}.} \bibinfo{year}{2006}\natexlab{}.
\newblock \showarticletitle{A categorical approach to probability theory}. In \bibinfo{booktitle}{\emph{Categorical Aspects of Topology and Analysis: Proceedings of an International Conference Held at Carleton University, Ottawa, August 11--15, 1981}}. Springer, \bibinfo{pages}{68--85}.
\newblock


\bibitem[Grothendieck(1971)]%
        {grothendieck2002rev}
\bibfield{author}{\bibinfo{person}{Alexander Grothendieck}.} \bibinfo{year}{1971}\natexlab{}.
\newblock \bibinfo{booktitle}{\emph{Rev{\'e}tements {\'E}tales et Groupe Fondamental (SGA 1)}}. \bibinfo{series}{Lecture Notes in Mathematics}, Vol.~\bibinfo{volume}{224}.
\newblock \bibinfo{publisher}{Springer-Verlag}.
\newblock


\bibitem[Heunen et~al\mbox{.}(2017)]%
        {heunen2017convenient}
\bibfield{author}{\bibinfo{person}{Chris Heunen}, \bibinfo{person}{Ohad Kammar}, \bibinfo{person}{Sam Staton}, {and} \bibinfo{person}{Hongseok Yang}.} \bibinfo{year}{2017}\natexlab{}.
\newblock \showarticletitle{A convenient category for higher-order probability theory}. In \bibinfo{booktitle}{\emph{2017 32nd Annual ACM/IEEE Symposium on Logic in Computer Science (LICS)}}. IEEE, \bibinfo{pages}{1--12}.
\newblock


\bibitem[Hodges(1993)]%
        {hodges1993model}
\bibfield{author}{\bibinfo{person}{Wilfrid Hodges}.} \bibinfo{year}{1993}\natexlab{}.
\newblock \bibinfo{booktitle}{\emph{Model theory}}.
\newblock \bibinfo{publisher}{Cambridge university press}.
\newblock


\bibitem[Jackson(2006)]%
        {jackson2006sheaf}
\bibfield{author}{\bibinfo{person}{Matthew Jackson}.} \bibinfo{year}{2006}\natexlab{}.
\newblock \emph{\bibinfo{title}{A sheaf theoretic approach to measure theory}}.
\newblock \bibinfo{thesistype}{Ph.\,D. Dissertation}. \bibinfo{school}{University of Pittsburgh}.
\newblock


\bibitem[Joyal and Tierney(1984)]%
        {joyal1984extension}
\bibfield{author}{\bibinfo{person}{Andr{\'e} Joyal} {and} \bibinfo{person}{Myles Tierney}.} \bibinfo{year}{1984}\natexlab{}.
\newblock \bibinfo{booktitle}{\emph{An extension of the Galois theory of Grothendieck}}. Vol.~\bibinfo{volume}{309}.
\newblock \bibinfo{publisher}{American Mathematical Soc.}
\newblock


\bibitem[Kallenberg(1997)]%
        {kallenberg1997foundations}
\bibfield{author}{\bibinfo{person}{Olav Kallenberg}.} \bibinfo{year}{1997}\natexlab{}.
\newblock \bibinfo{booktitle}{\emph{Foundations of modern probability}}. Vol.~\bibinfo{volume}{2}.
\newblock \bibinfo{publisher}{Springer}.
\newblock


\bibitem[Li et~al\mbox{.}(2023)]%
        {li2023lilac}
\bibfield{author}{\bibinfo{person}{John~M Li}, \bibinfo{person}{Amal Ahmed}, {and} \bibinfo{person}{Steven Holtzen}.} \bibinfo{year}{2023}\natexlab{}.
\newblock \showarticletitle{Lilac: a Modal Separation Logic for Conditional Probability}.
\newblock \bibinfo{journal}{\emph{Proceedings of the ACM on Programming Languages}} \bibinfo{volume}{7}, \bibinfo{number}{PLDI} (\bibinfo{year}{2023}), \bibinfo{pages}{148--171}.
\newblock


\bibitem[MacLane and Moerdijk(2012)]%
        {maclane2012sheaves}
\bibfield{author}{\bibinfo{person}{Saunders MacLane} {and} \bibinfo{person}{Ieke Moerdijk}.} \bibinfo{year}{2012}\natexlab{}.
\newblock \bibinfo{booktitle}{\emph{Sheaves in geometry and logic: A first introduction to topos theory}}.
\newblock \bibinfo{publisher}{Springer Science \& Business Media}.
\newblock


\bibitem[Makkai(1982)]%
        {makkai1982full}
\bibfield{author}{\bibinfo{person}{M Makkai}.} \bibinfo{year}{1982}\natexlab{}.
\newblock \showarticletitle{Full continuous embeddings of toposes}.
\newblock \bibinfo{journal}{\emph{Trans. Amer. Math. Soc.}} \bibinfo{volume}{269}, \bibinfo{number}{1} (\bibinfo{year}{1982}), \bibinfo{pages}{167--196}.
\newblock


\bibitem[Menni(2003)]%
        {menni2003quantifiers}
\bibfield{author}{\bibinfo{person}{Matias Menni}.} \bibinfo{year}{2003}\natexlab{}.
\newblock \showarticletitle{About $\newq$-quantifiers}.
\newblock \bibinfo{journal}{\emph{Applied categorical structures}}  \bibinfo{volume}{11} (\bibinfo{year}{2003}), \bibinfo{pages}{421--445}.
\newblock


\bibitem[Odersky(1994)]%
        {odersky1994functional}
\bibfield{author}{\bibinfo{person}{Martin Odersky}.} \bibinfo{year}{1994}\natexlab{}.
\newblock \showarticletitle{A functional theory of local names}. In \bibinfo{booktitle}{\emph{Proceedings of the 21st ACM SIGPLAN-SIGACT symposium on Principles of programming languages}}. \bibinfo{pages}{48--59}.
\newblock


\bibitem[O'Hearn(1993)]%
        {o1993model}
\bibfield{author}{\bibinfo{person}{Peter~W. O'Hearn}.} \bibinfo{year}{1993}\natexlab{}.
\newblock \showarticletitle{A model for syntactic control of interference}.
\newblock \bibinfo{journal}{\emph{Mathematical structures in computer science}} \bibinfo{volume}{3}, \bibinfo{number}{4} (\bibinfo{year}{1993}), \bibinfo{pages}{435--465}.
\newblock


\bibitem[O'Hearn et~al\mbox{.}(1995)]%
        {o1995syntactic}
\bibfield{author}{\bibinfo{person}{Peter~W O'Hearn}, \bibinfo{person}{AJ Power}, \bibinfo{person}{M Takeyama}, {and} \bibinfo{person}{Robert~D Tennent}.} \bibinfo{year}{1995}\natexlab{}.
\newblock \showarticletitle{Syntactic control of interference revisited}.
\newblock \bibinfo{journal}{\emph{Electronic notes in Theoretical computer science}}  \bibinfo{volume}{1} (\bibinfo{year}{1995}), \bibinfo{pages}{447--486}.
\newblock


\bibitem[Pavlov(2022)]%
        {pavlov2022gelfand}
\bibfield{author}{\bibinfo{person}{Dmitri Pavlov}.} \bibinfo{year}{2022}\natexlab{}.
\newblock \showarticletitle{Gelfand-type duality for commutative von Neumann algebras}.
\newblock \bibinfo{journal}{\emph{Journal of Pure and Applied Algebra}} \bibinfo{volume}{226}, \bibinfo{number}{4} (\bibinfo{year}{2022}), \bibinfo{pages}{106884}.
\newblock


\bibitem[Pitts(2003)]%
        {pitts2003nominal}
\bibfield{author}{\bibinfo{person}{Andrew~M Pitts}.} \bibinfo{year}{2003}\natexlab{}.
\newblock \showarticletitle{Nominal logic, a first order theory of names and binding}.
\newblock \bibinfo{journal}{\emph{Information and computation}} \bibinfo{volume}{186}, \bibinfo{number}{2} (\bibinfo{year}{2003}), \bibinfo{pages}{165--193}.
\newblock


\bibitem[Pitts(2013)]%
        {pitts2013nominal}
\bibfield{author}{\bibinfo{person}{Andrew~M Pitts}.} \bibinfo{year}{2013}\natexlab{}.
\newblock \bibinfo{booktitle}{\emph{Nominal sets: Names and symmetry in computer science}}.
\newblock \bibinfo{publisher}{Cambridge University Press}.
\newblock


\bibitem[Pitts and Stark(1993)]%
        {pitts1993observable}
\bibfield{author}{\bibinfo{person}{Andrew~M Pitts} {and} \bibinfo{person}{Ian~DB Stark}.} \bibinfo{year}{1993}\natexlab{}.
\newblock \showarticletitle{Observable properties of higher order functions that dynamically create local names, or: What's new?}. In \bibinfo{booktitle}{\emph{International Symposium on Mathematical Foundations of Computer Science}}. Springer, \bibinfo{pages}{122--141}.
\newblock


\bibitem[Reynolds(1978)]%
        {reynolds1978syntactic}
\bibfield{author}{\bibinfo{person}{John~C Reynolds}.} \bibinfo{year}{1978}\natexlab{}.
\newblock \showarticletitle{Syntactic control of interference}. In \bibinfo{booktitle}{\emph{Proceedings of the 5th ACM SIGACT-SIGPLAN symposium on Principles of programming languages}}. \bibinfo{pages}{39--46}.
\newblock


\bibitem[Reynolds(2002)]%
        {reynolds2002separation}
\bibfield{author}{\bibinfo{person}{John~C Reynolds}.} \bibinfo{year}{2002}\natexlab{}.
\newblock \showarticletitle{Separation logic: A logic for shared mutable data structures}. In \bibinfo{booktitle}{\emph{Proceedings 17th Annual IEEE Symposium on Logic in Computer Science}}. IEEE, \bibinfo{pages}{55--74}.
\newblock


\bibitem[Rohlin(1949)]%
        {rohlin1949fundamental}
\bibfield{author}{\bibinfo{person}{VA Rohlin}.} \bibinfo{year}{1949}\natexlab{}.
\newblock \showarticletitle{On the fundamental ideas of measure theory}.
\newblock \bibinfo{journal}{\emph{Mat. Sb.(NS)}} \bibinfo{volume}{25}, \bibinfo{number}{67} (\bibinfo{year}{1949}), \bibinfo{pages}{107--150}.
\newblock


\bibitem[Sabok et~al\mbox{.}(2021)]%
        {sabok2021probabilistic}
\bibfield{author}{\bibinfo{person}{Marcin Sabok}, \bibinfo{person}{Sam Staton}, \bibinfo{person}{Dario Stein}, {and} \bibinfo{person}{Michael Wolman}.} \bibinfo{year}{2021}\natexlab{}.
\newblock \showarticletitle{Probabilistic programming semantics for name generation}.
\newblock \bibinfo{journal}{\emph{Proceedings of the ACM on Programming Languages}} \bibinfo{volume}{5}, \bibinfo{number}{POPL} (\bibinfo{year}{2021}), \bibinfo{pages}{1--29}.
\newblock


\bibitem[Sato et~al\mbox{.}(2019)]%
        {sato2019formal}
\bibfield{author}{\bibinfo{person}{Tetsuya Sato}, \bibinfo{person}{Alejandro Aguirre}, \bibinfo{person}{Gilles Barthe}, \bibinfo{person}{Marco Gaboardi}, \bibinfo{person}{Deepak Garg}, {and} \bibinfo{person}{Justin Hsu}.} \bibinfo{year}{2019}\natexlab{}.
\newblock \showarticletitle{Formal verification of higher-order probabilistic programs: reasoning about approximation, convergence, bayesian inference, and optimization}.
\newblock \bibinfo{journal}{\emph{Proceedings of the ACM on Programming Languages}} \bibinfo{volume}{3}, \bibinfo{number}{POPL} (\bibinfo{year}{2019}), \bibinfo{pages}{1--30}.
\newblock


\bibitem[{\'S}cibior et~al\mbox{.}(2017)]%
        {scibior2017denotational}
\bibfield{author}{\bibinfo{person}{Adam {\'S}cibior}, \bibinfo{person}{Ohad Kammar}, \bibinfo{person}{Matthijs V{\'a}k{\'a}r}, \bibinfo{person}{Sam Staton}, \bibinfo{person}{Hongseok Yang}, \bibinfo{person}{Yufei Cai}, \bibinfo{person}{Klaus Ostermann}, \bibinfo{person}{Sean~K Moss}, \bibinfo{person}{Chris Heunen}, {and} \bibinfo{person}{Zoubin Ghahramani}.} \bibinfo{year}{2017}\natexlab{}.
\newblock \showarticletitle{Denotational validation of higher-order Bayesian inference}.
\newblock \bibinfo{journal}{\emph{Proceedings of the ACM on Programming Languages}} \bibinfo{volume}{2}, \bibinfo{number}{POPL} (\bibinfo{year}{2017}), \bibinfo{pages}{1--29}.
\newblock


\bibitem[Simpson(2016)]%
        {simpsonsheafslides}
\bibfield{author}{\bibinfo{person}{Alex Simpson}.} \bibinfo{year}{2016}\natexlab{}.
\newblock \bibinfo{title}{Probability sheaves}.
\newblock
\newblock
\urldef\tempurl%
\url{https://synapse.math.univ-toulouse.fr/index.php/s/QWrxKeXn31mN3gz}
\showURL{%
\tempurl}
\newblock
\shownote{Accessed: 2023-10-02}.


\bibitem[Simpson(2017)]%
        {simpson2017probability}
\bibfield{author}{\bibinfo{person}{Alex Simpson}.} \bibinfo{year}{2017}\natexlab{}.
\newblock \showarticletitle{{Probability Sheaves and the Giry Monad}}. In \bibinfo{booktitle}{\emph{7th Conference on Algebra and Coalgebra in Computer Science (CALCO 2017)}} \emph{(\bibinfo{series}{Leibniz International Proceedings in Informatics (LIPIcs)}, Vol.~\bibinfo{volume}{72})}, \bibfield{editor}{\bibinfo{person}{Filippo Bonchi} {and} \bibinfo{person}{Barbara K{\"o}nig}} (Eds.). \bibinfo{publisher}{Schloss Dagstuhl--Leibniz-Zentrum fuer Informatik}, \bibinfo{address}{Dagstuhl, Germany}, \bibinfo{pages}{1:1--1:6}.
\newblock
\showISBNx{978-3-95977-033-0}
\showISSN{1868-8969}
\urldef\tempurl%
\url{https://doi.org/10.4230/LIPIcs.CALCO.2017.1}
\showDOI{\tempurl}


\bibitem[Simpson(2018)]%
        {simpsonSyntheticProbabilityTheory}
\bibfield{author}{\bibinfo{person}{Alex Simpson}.} \bibinfo{year}{2018}\natexlab{}.
\newblock \bibinfo{title}{Synthetic Probability Theory}.  (\bibinfo{year}{2018}).
\newblock
\urldef\tempurl%
\url{http://tobiasfritz.science/2019/cps_workshop/slides/simpson.pdf}
\showURL{%
\tempurl}
\newblock
\shownote{Mathematics and Theoretical Computing Seminar at the University of Ljubljana}.


\bibitem[Simpson(2024)]%
        {equivalenceandconditionalindependenceinatomicsheaflogic}
\bibfield{author}{\bibinfo{person}{Alex Simpson}.} \bibinfo{year}{2024}\natexlab{}.
\newblock \showarticletitle{Equivalence and Conditional Independence in Atomic Sheaf Logic}. In \bibinfo{booktitle}{\emph{2024 39th Annual ACM/IEEE Symposium on Logic in Computer Science (LICS)}}. IEEE.
\newblock
\urldef\tempurl%
\url{https://doi.org/10.1145/3661814.3662132}
\showDOI{\tempurl}


\bibitem[Staton(2007)]%
        {staton2007name}
\bibfield{author}{\bibinfo{person}{Sam Staton}.} \bibinfo{year}{2007}\natexlab{}.
\newblock \bibinfo{booktitle}{\emph{Name-passing process calculi: operational models and structural operational semantics}}.
\newblock \bibinfo{type}{{T}echnical {R}eport}. \bibinfo{institution}{University of Cambridge, Computer Laboratory}.
\newblock


\bibitem[Staton et~al\mbox{.}(2016)]%
        {staton2016semantics}
\bibfield{author}{\bibinfo{person}{Sam Staton}, \bibinfo{person}{Hongseok Yang}, \bibinfo{person}{Frank Wood}, \bibinfo{person}{Chris Heunen}, {and} \bibinfo{person}{Ohad Kammar}.} \bibinfo{year}{2016}\natexlab{}.
\newblock \showarticletitle{Semantics for probabilistic programming: higher-order functions, continuous distributions, and soft constraints}. In \bibinfo{booktitle}{\emph{Proceedings of the 31st Annual ACM/IEEE Symposium on Logic in Computer Science}}. \bibinfo{pages}{525--534}.
\newblock


\bibitem[Tao(2015)]%
        {tao_2015}
\bibfield{author}{\bibinfo{person}{Terence Tao}.} \bibinfo{year}{2015}\natexlab{}.
\newblock \bibinfo{title}{254A, notes 0: A review of probability theory}.
\newblock
\newblock
\urldef\tempurl%
\url{https://terrytao.wordpress.com/2010/01/01/254a-notes-0-a-review-of-probability-theory/}
\showURL{%
\tempurl}
\newblock
\shownote{Accessed: 2023-10-02}.


\bibitem[V{\'a}k{\'a}r et~al\mbox{.}(2019)]%
        {vakar2019domain}
\bibfield{author}{\bibinfo{person}{Matthijs V{\'a}k{\'a}r}, \bibinfo{person}{Ohad Kammar}, {and} \bibinfo{person}{Sam Staton}.} \bibinfo{year}{2019}\natexlab{}.
\newblock \showarticletitle{A domain theory for statistical probabilistic programming}.
\newblock \bibinfo{journal}{\emph{Proceedings of the ACM on Programming Languages}} \bibinfo{volume}{3}, \bibinfo{number}{POPL} (\bibinfo{year}{2019}), \bibinfo{pages}{1--29}.
\newblock


\bibitem[Wendt(1998)]%
        {wendt1998change}
\bibfield{author}{\bibinfo{person}{Michael~A Wendt}.} \bibinfo{year}{1998}\natexlab{}.
\newblock \showarticletitle{Change of base for measure spaces}.
\newblock \bibinfo{journal}{\emph{Journal of Pure and Applied Algebra}} \bibinfo{volume}{128}, \bibinfo{number}{2} (\bibinfo{year}{1998}), \bibinfo{pages}{185--212}.
\newblock


\bibitem[Zhou et~al\mbox{.}(2021)]%
        {zhou2021quantum}
\bibfield{author}{\bibinfo{person}{Li Zhou}, \bibinfo{person}{Gilles Barthe}, \bibinfo{person}{Justin Hsu}, \bibinfo{person}{Mingsheng Ying}, {and} \bibinfo{person}{Nengkun Yu}.} \bibinfo{year}{2021}\natexlab{}.
\newblock \showarticletitle{A quantum interpretation of bunched logic \& quantum separation logic}. In \bibinfo{booktitle}{\emph{2021 36th Annual ACM/IEEE Symposium on Logic in Computer Science (LICS)}}. IEEE, \bibinfo{pages}{1--14}.
\newblock
\urldef\tempurl%
\url{https://doi.org/10.1109/LICS52264.2021.9470673}
\showDOI{\tempurl}


\end{thebibliography}

\pagebreak
\appendix
\onecolumn

\ifdefined\noapp
\else
\renewcommand\StdProb{\nref{app:def:stdprob}{\texStdProb}}
\renewcommand\StdProbAlg{\nref{app:def:stdprobalg}{\mathbf{ProbAlg}_{\rm std}}}
\renewcommand\StdProbAlgop{\nref{app:def:stdprobalg}{\mathbf{ProbAlg}_{\rm std}\op}}
\renewcommand\StdMble{{\nref{app:def:mblecat}{\texStdMble}}}
\renewcommand\StdMbleop{{\nref{app:def:mblecat}{\texStdMble\op}}}
\renewcommand\StdMbleAlg{{\nref{app:def:stdmblealg}{\mathbf{MbleAlg}_{\rm std}}}}
\renewcommand\StdMbleAlgop{{\nref{app:def:stdmblealg}{\mathbf{MbleAlg}_{\rm std}\op}}}
\renewcommand\forgetmu{{\nref{app:def:forgetmu}{\texforgetmu}}}
\renewcommand\goodsheaf{{\nref{app:def:goodsheaf}{\texgoodsheaf}}}
\renewcommand\Goodsheaf{{\nref{app:def:goodsheaf}{\texGoodsheaf}}}
\renewcommand\goodsheaves{{\nref{app:def:goodsheaf}{\texgoodsheaves}}}
\renewcommand\Goodsheaves{{\nref{app:def:goodsheaf}{\texGoodsheaves}}}
\renewcommand\pspcs{{\nref{app:def:mble-pspcs}{\texpspcs}}}
\renewcommand\hcube{{\nref{app:def:hcube}{\texhcube}}}
\renewcommand\auti{{\nref{app:def:auti}{\texauti}}}
\renewcommand\agoodname{{\nref{app:def:pnom}{\texagoodname}}}
\renewcommand\Agoodname{{\nref{app:def:pnom}{\texAgoodname}}}
\renewcommand\gpspcs{{\nref{app:def:pspcs-pnom}{\texgpspcs}}}
\renewcommand\gRV{{\nref{app:def:rv-pnom}{\texgRV}}}

\section{Overview}

The goal of this appendix is to build up towards
\cref{app:sec:goodsheaves,app:sec:agoodnames},
which give sheaf- and permutation-based models of
probabilistic separation logic, prove their equivalence, and show the correspondence
between the tensor-product- and independent-combination-based
models of probabilistic separation.
To lead into this result, the appendix is structured as follows.
\begin{itemize}
    \item 
      \cref{app:sec:measure-theory}
      imports the measure theory
      and
      \cref{app:sec:sheaf-theory}
      the sheaf theory
      needed to present the probabilistic counterpart of
      \cref{sec:nom-situation}.
      Most of the results are for handling continuous probability,
      but some (in particular, \cref{app:lem:tensor-preserves-ussheaves}
      characterizing the Day convolution)
      are needed even in the discrete case.
      These sections are best not read linearly on first reading, but 
      used as reference when reading \cref{app:sec:goodsheaves,app:sec:agoodnames}.
    \item \cref{app:sec:goodsheaves} presents \goodsheavesnolink{}.
    \item \cref{app:sec:agoodnames} presents \agoodnamenolink{}
      and our probabilistic analogue of \cref{fact:folklore}.
\end{itemize}

On reading, we encourage you to follow concepts by clicking on 
hyperlinks that take you directly to definitions.

\section{General measure theory} \label{app:sec:measure-theory}
Throughout this section we rely heavily on \citet{fremlin2000measure}.

\subsection{Standard probability spaces and probability algebras}

\begin{definition}[standard Borel space] \label{app:def:sbs}
  A \emph{standard Borel space}
  is a measurable space $(X,\calF)$
  such that $X$ can be made into a complete separable metrizable
  space (in other words, a Polish space) whose Borel $\sigma$-algebra is $\calF$~\fremlinf{424A}{42}{30}.
\end{definition}

\begin{definition}[standard probability space] \label{app:def:sps}
  A \emph{standard probability space}
  is a probability space $(X,\calF,\mu)$
  where $(X,\calF)$ is the completion of a standard Borel space
  with respect to the negligibles of $\mu$.
  (For more on standard probability spaces,
  see \citet{rohlin1949fundamental},
  where they are called Lebesgue spaces;
  in particular, \citet[\S 2.7]{rohlin1949fundamental} justifies the definition given here.)
\end{definition}

\begin{definition}[the category $\StdProb$] \label{app:def:stdprob}
  Let $\StdProb$
  be the category whose objects are standard
  probability spaces
  and whose morphisms are measure-preserving maps
  quotiented by almost-sure equality:
  \begin{align*}
    &\StdProb((X, \calF,\mu), (Y,\calG,\nu))
    = \{ f : (X,\calF)\to(Y,\calG)\text{ measurable}
    \mid \mu(f^{-1}G) = \nu(G)\text{ for all $G$ in $\calG$}
    \}/\aseq\\
    &\text{where $f\aseq g$ iff $\{x\in X\mid f(x)\ne g(x)\}$ is $\mu$-negligible}
  \end{align*}
\end{definition}

\begin{definition}[measure algebra] \label{app:def:measalg}
  A \emph{measure algebra} 
  is a pair $(\frakA,\overline\mu)$
  where $\frakA$ is a complete boolean algebra
  (i.e., a boolean algebra with all small meets and joins)
  and $\overline\mu$ is a function $\frakA\to[0,1]$
  satisfying \begin{itemize}
    \item $\overline\mu(\bot) = 0$
    \item $\overline\mu(a) > 0$ for all $a\ne\bot$
    \item $\overline\mu \big(\bigvee_i a_i\big) = \sum_i \overline\mu(a_i)$
      for all sequences $(a_i)_{i\in \N}$
      with $a_i\wedge a_j = \bot$ for all $i\ne j$.
  \end{itemize}
  A \emph{probability algebra} 
  is a measure algebra $(\frakA,\overline\mu)$
  with $\overline\mu(\top) = 1$.
\end{definition}

See \fremlinft{321A}{32}{1} for more information on measure algebras.
Our definitions deviate slightly from the definitions given there:
our measure-algebras are
closed under all small meets and joins,
so correspond to the notion of Dedekind-complete
measure algebra~~\fremlinf{314A(a)}{31}{32}, whereas the definition in
\fremlinft{321A}{32}{1} only requires the measure algebras
to be Dedekind $\sigma$-complete (closed under countable meets and joins).
This change in terminology is motivated
by the fact that all probability spaces will give rise to
Dedekind-complete measure algebras:

\begin{construction} \label{app:cons:measalg}
  Every probability space $(X,\calF,\mu)$
  gives rise to a corresponding measure algebra $(\frakA,\overline\mu)$
  as follows. Let $\calN$ be the $\sigma$-ideal of $\mu$-negligible sets.
  Set $\frakA$ to the quotient boolean algebra $\calF/\calN$.
  Elements of $\frakA$ are equivalence classes $[F]$ for $F\in\calF$,
  modulo $[F] = [F']$ iff $F\triangle F'$ (the symmetric difference of $F$ and $F'$)
  is $\mu$-negligible.
  Define $\overline\mu$ by $\overline\mu[F] = \mu(F)$.
  Let $\measalg(X,\calF,\mu)$
  denote the measure algebra constructed from
  a probability space $(X,\calF,\mu)$ in this way.
\end{construction}

For more on this construction, see \fremlinft{321H}{32}{3}.

\begin{lemma} \label{app:lem:sigalg-mod-negl-cba}
  For every probability space $(X,\calF,\mu)$
  it holds that $\measalg(X,\calF,\mu)$
  is a Dedekind-complete measure algebra in the sense of
  Fremlin, so a measure algebra in our sense.
\end{lemma}
\begin{proof}
  $\measalg(X,\calF,\calN)$
  is a probability algebra~\fremlinf{322A(a)}{32}{5},
  so localizable~\fremlinf{322C}{32}{6},
  so Dedekind-complete~\fremlinf{322A(e)}{32}{5}.
\end{proof}

\begin{definition}[homomorphism of measure algebras] \label{app:def:measalg-hom}
  Given two measure algebras $(\frakA,\overline\mu)$
  and $(\frakB,\overline\nu)$,
  a \emph{measure-algebra homomorphism from $(\frakA,\overline\mu)$
  to $(\frakB,\overline\nu)$}
  is a complete-boolean-algebra homomorphism
  $f : \frakA\to\frakB$
  which preserves measures: $\overline\nu(f(a)) = \overline\mu(a)$
  for all $a$ in $\frakA$.
\end{definition}

See \fremlinft{324I}{32}{25} for more information on measure-algebra
homomorphisms; just as our definition of measure algebra
corresponds to Fremlin's Dedekind-complete measure algebra,
our definition of measure-algebra homomorphism
corresponds to Fremlin's order-continuous
measure-algebra homomorphism.

\begin{definition}[simple product of measure algebras] \label{app:def:measalg-simple-product}
  Given two measure algebras $(\frakA,\overline\mu)$
  and $(\frakB,\overline\nu)$,
  their \emph{simple product}~(\fremlin{\href{https://www1.essex.ac.uk/maths/people/fremlin/chap32.pdf\#page=9}{322L}})
  is the measure algebra $(\frakA\times\frakB,\overline\lambda)$
  where $\frakA\times\frakB$ is the product boolean algebra
  with boolean-algebra operations computed pointwise
  and $\overline\lambda$ is the measure defined by
  $\overline\lambda(a, b) = \overline\mu(a) + \overline\nu(b)$
  for all $a\in\frakA$ and $b\in\frakB$.
  Since joins and meets are pointwise,
  the boolean algebra $\frakA\times\frakB$
  is Dedekind-complete if $\frakA,\frakB$
  are.
\end{definition}

The idea is that the simple product
is an algebraic counterpart to the coproduct
of measure spaces: $\measalg(X + Y) \cong \measalg X\times\measalg Y$.
It is known that every standard probability space
is isomorphic to the coproduct of countably many atoms
and possibly an interval equipped with the Lebesgue measure~\cite[p. 25, \S 2]{rohlin1949fundamental}.
This fact motivates the following definition:

\begin{definition}[standard probability algebra] \label{app:def:stdprobalg-def}
  Say a probability algebra
  $(\frakA,\overline\mu)$
  is \emph{standard}
  if it is composed of at most countably many atoms
  and an interval; that is, there exists some $p\in(0,1]$ 
  and a countable family of weights $(q_i)_{i\in I}$
  such that $(\frakA,\overline\mu)$
  decomposes as a simple product
  \[ (\frakA,\overline\mu) \cong ([0,p],\overline\lambda) \times \prod_{i\in I} \atom{q_i}
  \qquad\text{or}\qquad
  (\frakA,\overline\mu) \cong \prod_{i\in I} \atom{q_i}
  \]
  where
  $([0,p],\overline\lambda)$ is the measure algebra of Lebesgue measure on the interval $[0,p]$
  and $\atom {q_i}$ is
  the algebra with two elements $\bot,\top$ where $\top$ has measure $q_i$.
\end{definition}

\begin{definition}[the category $\StdProbAlg$] \label{app:def:stdprobalg}
  Let $\StdProbAlg$
  be the category whose objects are standard
  probability algebras
  and whose morphisms are \nref{app:def:measalg-hom}{measure-algebra homomorphisms}.
\end{definition}

The operation $\measalg$
extends to 
a functor $\StdProb\to\StdProbAlgop$:

\begin{construction} \label{app:cons:measalg-hom}
  Let $(X,\calF,\mu)$ and $(Y,\calG,\nu)$
  be probability spaces
  and $\measalg(X,\calF,\mu) = (\frakA,\overline\mu)$
  and $\measalg(Y,\calG,\nu) = (\frakB,\overline\nu)$
  the corresponding measure algebras.
  Let $f$ be a measure-preserving map
  $(X,\calF,\mu)\to(Y,\calG,\nu)$.
  The map $\measalg(f) : (\frakB,\overline\nu)\to(\frakA,\overline\mu)$
  defined by setting $\measalg(f)[G] = [f^{-1}(G)]$
  for all $G\in\calG$
  is a measure-algebra homomorphism.
  This construction respects almost-sure equality of
  measure-preserving maps, sends the identity
  map to the identity homomorphism, and
  sends composition of measure-preserving maps
  $f\circ g$
  to flipped-composition of homomorphisms
  $\measalg(g)\circ\measalg(f)$.
\end{construction}
\begin{proof}
  If $G\triangle G'$ is $\nu$-negligible
  then $f^{-1}(G\triangle G') = f^{-1}(G)\triangle f^{-1}(G)$
  is $\mu$-negligible since $f$ measure-preserving,
  so $\measalg(f)$ is well-defined.
  Taking $f$-preimages distributes over all
  unions and intersections, so $\measalg(f)$
  is a homomorphism of complete boolean-algebras;
  this homomorphism preserves the measures $\overline\mu$
  and $\overline\nu$ because $f$ is measure-preserving.
  If $f\aseq g$, then $f^{-1}(G)\triangle g^{-1}(G)$
  is $\mu$-negligible for all $G\in\calG$,
  so $\measalg(f) = \measalg(g)$.
  Finally, $f$ respects identities and composition
  because taking preimages does.
\end{proof}

\begin{lemma} \label{app:lem:stdprob-dual}
  The functor $\measalg:\StdProb\to\StdProbAlgop$
  witnesses an equivalence
  $\StdProb\simeq\StdProbAlgop$.
\end{lemma}
\begin{proof}
  The functor $\measalg$ is essentially surjective
  on objects because any standard probability algebra
  is isomorphic to the measure algebra
  generated by the coproduct of countably many atoms and possibly an interval
  equipped with the Lebesgue measure.
  Let $(X,\calF,\mu)$ and $(Y,\calG,\nu)$
  be standard probability spaces
  with associated probability algebras
  $\measalg(X,\calF,\mu) = (\frakA,\overline\mu)$
  and $\measalg(Y,\calG,\nu) = (\frakB,\overline\nu)$.
  Fix an arbitrary measure-algebra homomorphism
  $f^* : (\frakB,\overline\nu)\to(\frakA,\overline\mu)$.
  The probability space $(X,\calF,\mu)$ is complete by definition
  and strictly localizable by \fremlinft{221L}{21}{4}.
  The probability space $(Y,\calG,\nu)$ is finite, hence semi-finite~\fremlinf{221F}{21}{4},
  and nonempty.
  Thus the preconditions of 
  \fremlinft{343B}{34}{22} are satisfied.
  As the completion of a probability measure on a
  standard Borel space, $(Y,\calG,\nu)$
  is a compact measure space~\fremlinf{433X(e)(i)}{43}{15}~\fremlinf{342G(b)}{34}{15},
  hence also locally compact~\fremlinf{342H(a)}{34}{17},
  so point (i) of \fremlinft{343B}{34}{22} holds.
  The implication (i)$\Rightarrow$(vi) of \fremlinft{343B}{34}{22}
  then gives a measure-preserving map $f : (X,\calF,\mu)\to (Y,\calG,\nu)$
  such that $\measalg(f) = f^*$.
  This establishes fullness of the functor $\measalg$.
  The space $(Y,\calG,\nu)$ is
  countably separated (because it arises from a metrizable space),
  so by \fremlinft{343G}{34}{26} the map
  $f$ is unique up to almost-sure equivalence.
  This establishess faithfulness of $\measalg$.
\end{proof}

\begin{lemma} \label{app:lem:stdprobalg-mono}
  Every map in $\StdProbAlg$ is mono.
\end{lemma}
\begin{proof}
  All measure-preserving measure algebra homomorphisms
  are injective~\fremlinf{324K(a)}{32}{25}.
\end{proof}

\begin{lemma} \label{app:lem:std-prob-epi}
  Every map in $\StdProb$ is epi.
\end{lemma}
\begin{proof}
  Combine \cref{app:lem:stdprob-dual} and \cref{app:lem:stdprobalg-mono}.
\end{proof}

The following lemma seems closely connected to the theorems of \citet{edalat1999semi}.
\begin{lemma} \label{lem:std-prob-right-ore}
  The category $\StdProb$
  has the right Ore property: for $f,g$
  with types as shown below,
  there exists a space $W$ and maps $h,k$
  such that $fh = gk$:
  \[ \commsquare{(W,\calK,\tau)}{k}{(Y,\calG,\nu)}{h}{g}{(X,\calF,\mu)}{f}{(Z,\calH,\rho)} \]
\end{lemma}
\begin{proof}
  As standard probability spaces,
  $X,Y,Z$ are all completions of Borel measures on
  Polish spaces.
  Let $\calF',\calG',\calH'$ be the Borel sets
  given by the topologies on $X,Y,Z$ respectively.
  Let $W$ be the product space $X\times Y$
  with Borel sets $\calK'$ generated by rectangles
  $F'\times G'$ for $F'\in\calF',G'\in\calG'$ as usual.
  This product is Polish because each of its factors are.
  Since $f,g$ are maps of standard probability spaces, there exist
  disintegrations $\{\mu|_z\}_{z\in Z}$ and $\{\nu|_z\}_{z\in Z}$.
  Let $\tau'$ be the function $\calK'\to[0,1]$ defined by
  \[ \tau'(K') = \int (\mu|_z\otimes \nu|_z)(K')\,\rho({\rm d}z), \]
  a probability measure because each $\mu|_z\otimes\nu|_z$ is.
  Let $(W,\calK,\tau)$ be the completion of $(W,\calK',\tau')$,
  standard because $\tau'$ is a measure on a Polish space by construction.
  Let $h$ and $k$ be the projections $\pi_1$ and $\pi_2$ respectively.
  These maps are measure-preserving: for all $F\in\calF$ and $G\in\calG$,
  \begin{align*}
    \tau(h^{-1}(F)) = \rho(F\times Y) 
    = \int (\mu|_z\otimes \nu|_z)(F\times Y) \,\rho({\rm d}z)
    = \int \mu|_z(F)\nu|_z(Y) \,\rho({\rm d}z)
    = \int \mu|_z(F) \,\rho({\rm d}z) = \mu(F) \\
    \tau(k^{-1}(G)) = \rho(X\times G)
    = \int (\mu|_z\otimes \nu|_z)(X\times G) \,\rho({\rm d}z)
    = \int \mu|_z(X)\times\nu|_z(G) \,\rho({\rm d}z)
    = \int \nu|_z(G) \,\rho({\rm d}z) = \nu(G)
  \end{align*}
  Finally, the square commutes almost-surely:
  \begin{align*}
    &\Pr_{(x,y)\sim\tau}[f(h(x,y)) = g(k(x,y))]
     = \int (\mu|_z\otimes \nu|_z)(\{(x,y) \mid fx = gy\}) \,\rho({\rm d}z)\\
     &= \int (\mu|_z\otimes \nu|_z)\left(\biguplus_{z'} (f^{-1}(z') \times g^{-1}(z'))\right)\,\rho({\rm d}z) \\
     &= \int_{z\in \lambda} (\mu|_z\otimes \nu|_z)\left((f^{-1}(z) \times g^{-1}(z)) \cap \biguplus_{z'} (f^{-1}(z') \times g^{-1}(z'))\right) \,\rho({\rm d}z)
        \text{ because $(\mu|_z\otimes\nu|_z)(f^{-1}(z)\times g^{-1}(z)) = 1$ for a.a.\,$z$} \\
     &= \int (\mu|_z\otimes \nu|_z)(f^{-1}(z) \times g^{-1}(z)) \,\rho({\rm d}z)
     = 1
  \end{align*}
\end{proof}

\subsection{Standard enhanced measurable spaces and measurable algebras} \label{app:sec:stdmblespc}

Given a standard probability space $(X,\calF,\mu)$,
one can choose to forget everything about the measure $\mu$
except for which measurable subsets are negligible,
leaving behind a tuple $(X,\calF,\calN)$
where $\calN$ is a $\sigma$-ideal of $\calF$.
Given two standard probability spaces $(X,\calF,\mu)$
and $(Y,\calG,\nu)$ where $\mu$ has negligibles $\calN$
and $\nu$ has negligibles $\calM$,
and given a morphism $[f] : (X,\calF,\mu)\to (Y,\calG,\nu)$ of standard probability spaces,
one can choose to forget everything about the fact that $[f]$
is measure-preserving except for the fact that
$\nu(G) = 0$ iff $\mu(f^{-1}(G)) = 0$,
leaving behind an equivalence class
$[f]$ with $f^{-1}(G)\in\calN$ iff $G\in\calM$
for all $G\in\calG$.
(Note this is well-defined, as only the negligible sets
are needed to determine whether two morphisms $f,f'$ are almost-surely equal.)
This idea is made precise as follows.

\begin{definition}[standard enhanced measurable space] \label{app:def:stems}
  An \emph{enhanced measurable space}
  is a tuple $(X,\calF,\calN)$
  for which there exists a measure $\mu$
  with negligibles $\calN$ such that
  $(X,\calF,\mu)$ is a probability space.
  A \emph{standard enhanced measurable space}
  is an enhanced measurable space for which there
  exists a measure $\mu$ making it a \nref{app:def:sps}{standard probability space}.
\end{definition}

\begin{definition}[negligible-preserving, negligible-reflecting] \label{app:def:neg-refl}
  Given two standard measurable spaces
  $(X,\calF,\calN)$ and $(Y,\calG,\calM)$,
  say a measurable map $f : (X,\calF)\to(Y,\calG)$
  is \emph{negligible-reflecting}
  if $f^{-1}(M)\in \calN$
  for all $M\in\calM$
  and \emph{negligible-preserving}
  if $f^{-1}(G)\in\calN$
  implies $G\in\calM$
  for all $G\in\calG$.
\end{definition}

\begin{definition}[map of standard enhanced measurable spaces] \label{app:def:stems-mor}
  Given two standard enhanced measurable spaces
  $(X,\calF,\calN)$ and $(Y,\calG,\calM)$,
  a \emph{map of standard enhanced measurable spaces
  from $(X,\calF,\calN)$ to $(Y,\calG,\calM)$}
  is a measurable map $f : (X,\calF)\to(Y,\calG)$
  that is both negligible-preserving and negligible-reflecting,
  quotiented by almost-sure equality:
  $f\aseq g$ iff $\{x\in X\mid f(x)\ne g(x)\} \in\calN$.
\end{definition}

\begin{definition}[the category $\StdMble$] \label{app:def:mblecat}
  Let $\StdMble$ be the category of
  \nref{app:def:stems}{standard enhanced measurable spaces}
  and \nref{app:def:stems-mor}{maps between them}.
\end{definition}

The following gives
  a more concrete picture of
  what kinds of (equivalence classes of) maps live in $\StdMble$:
  they can be equivalently seen as
  those maps that induce injective complete-boolean-algebra homomorphisms,
  and as measure-preserving maps that have forgotten the 
  fact that they were measure-preserving.

\begin{lemma} \label{app:lem:mblecat-map-char}
  Let $(X,\calF,\calN)\xrightarrow f(Y,\calG,\calM)$
  be a \nref{app:def:neg-refl}{negligible-reflecting} map of 
  enhanced measurable spaces.
  The following are equivalent:
  \begin{enumerate}
  \item $f$ is a \nref{app:def:stems-mor}{$\StdMble$-map} from $(X,\calF,\calN)$ to $(Y,\calG,\calM)$
  \item $f$ is \nref{app:def:neg-refl}{negligible-preserving}
  \item The complete-boolean-algebra homomorphism
    $\calF/\calN\xleftarrow{\measalg(f)}\calG/\calM$ 
    described in \cref{app:cons:measalg-hom}
    is injective
  \item There exist measures $\mu$ on $\calF$ and $\nu$ on $\calG$
    with negligibles $\calN$ and $\calM$ respectively
    such that $f$ is a $\StdProb$-map
    $(X,\calF,\mu)\to(Y,\calG,\nu)$
  \end{enumerate}
\end{lemma}
\begin{proof}~
  \begin{itemize}
  \item $(1)\Leftrightarrow(2)$. By definition.
  \item $(2)\Rightarrow(3)$.
    If $f$ is negligible-preserving
    then $\measalg(f)[E] = \bot$ implies $[E] = \bot$,
    so $\measalg(f)$ has trivial kernel.
  \item $(3)\Rightarrow(4)$.
    If $\measalg(f) : \calF/\calN\hookleftarrow\calG/\calM$ 
    is an injective complete-boolean-algebra homomorphism,
    then unforgetting a standard probability measure $\mu$
    on $(X,\calF,\calN)$ gives a standard probability space
    $(X,\calF,\mu)$ and corresponding measure algebra
    $(\calF/\calN, \overline\mu)$;
    restricting $\overline\mu$ along the injective homomorphism $\measalg(f)$
    gives a measure algebra $(\calG/\calM,\overline\nu)$
    making $\measalg(f)$
    a measure-algebra homomorphism $(\calF/\calN,\overline\mu)
    \hookleftarrow(\calG/\calM,\overline\nu)$.
    Define $\mu : \calF\to[0,1]$ by $\mu E = \overline\mu[E]$.
    This is a measure,
    since $\mu\emptyset = \overline\mu[\emptyset] = \overline\mu \bot = 0$
    and if $\{E_i\}_i$ is a countable disjoint family of events
      then $\mu(\biguplus_i E_i) =
      \overline\mu[\biguplus_i E_i]
      = \overline\mu(\biguplus_i [E_i])
      = \sum_i \overline\mu [E_i]
      = \sum_i \mu E_i
      $,
    Moreover, $\mu$ has negligibles $\calM$,
    since if $N\in\calM$ then $\mu N = \overline\mu [N]
      = \overline\mu\bot = 0$, and conversely if $\mu E = \overline\mu[E] = 0$
       then $[E] = \bot$ so $E\in\calM$.
    To show $(Y,\calG,\nu)$ standard,
    unforget a collection of Borel sets $\calG'$
    arising from a Polish topology on $Y$
    and a measure $\lambda$ making $(Y,\calG,\lambda)$
    into a standard probability space, with $\lambda$
    the completion of the Borel measure $\lambda|_{\calG'}$.
    The measure $\mu$ is correspondingly the completion of the Borel measure
    $\mu|_{\calG'}$ because $\lambda|_{\calG'}$ and $\mu|_{\calG'}$
    have the same Borel-negligibles~\fremlinf{212E(d)}{21}{11},
    so $(Y,\calG,\nu)$ standard.
  \item $(4)\Rightarrow(2)$.
    Any measure-preserving map is negligible-preserving.
  \end{itemize}
\end{proof}

\begin{definition} \label{app:def:forgetmu}
  Let $\forgetmu : \StdProb\to \StdMble$ be the forgetful functor
  that sends a standard probability space $(X,\calF,\mu)$
  to the standard enhanced measurable space $(X,\calF,\negligibles(\mu))$
  and a $\StdProb$-map $[f] : (X,\calF,\mu)\to(Y,\calG,\nu)$
  to itself (now considered as a negligible-preserving-and-reflecting map of standard enhanced measurable spaces).
\end{definition}

\begin{lemma} \label{app:lem:stdmble-img}
  The category $\StdMble$ is the image of $\forgetmu$:
  the functor $\forgetmu$ is surjective on objects,
  and any morphism of standard enhanced measurable spaces
  arises from a measure-preserving map equipping
  those spaces with standard probability measures.
\end{lemma}
\begin{proof}
  The object part follows by definition of standard enhanced measurable space.
  The morphism part is (1)$\Rightarrow$(4) of \cref{app:lem:mblecat-map-char}.
\end{proof}

Just as there is an equivalence $\StdProb\simeq\StdProbAlgop$ (\cref{app:lem:stdprob-dual}),
there is an equivalence between $\StdMble$
and a category of measure algebras that have forgotten their measures.

\begin{definition}[measurable algebra] \label{app:def:measurable-alg}
  A \emph{measurable algebra} is a complete boolean algebra $\frakA$
  for which there exists a measure $\overline\mu : \frakA\to[0,1]$
  making $(\frakA,\overline\mu)$ a \nref{app:def:measalg}{probability algebra}.
\end{definition}

As with the definition of measure algebra, we deviate slightly from
\fremlinft{391B(a)}{39}{2} in requiring measurable algebras to be
complete as boolean algebras.

\begin{definition}[standard measurable algebra] \label{app:def:stdmblealg-def}
  Call a measurable algebra \emph{standard}
  if there exists a measure $\mu$ on it
  that makes it into a \nref{app:def:stdprobalg-def}{standard probability algebra}.
\end{definition}

\begin{definition}[the category $\StdMbleAlg$] \label{app:def:stdmblealg}
  Let $\StdMbleAlg$ be the category of standard measurable algebras
  and injective complete-boolean-algebra homomorphisms.
\end{definition}

\begin{lemma} \label{app:lem:mblecat-dual}
  The operation $\measalg$
  defines a functor $\StdMble\to\StdMbleAlgop$
  witnessing an equivalence
  $\StdMble\simeq\StdMbleAlgop$.
\end{lemma}
\begin{proof}
  Forgetting about the parts of
  \cref{app:cons:measalg,app:cons:measalg-hom}
  to do with measures makes $\measalg$
  a functor of the required type.
  The proof that this functor is an equivalence
  is analogous to the proof of \cref{app:lem:stdprob-dual}.
\end{proof}

\begin{corollary} \label{app:cor:stdmble-epis}
  Every morphism in $\StdMble$ is epi.
\end{corollary}
\begin{proof}
  Combine (1)$\Rightarrow$(3) from \cref{app:lem:mblecat-map-char} 
  and \cref{app:lem:mblecat-dual}.
\end{proof}

\begin{construction} \label{app:cons:pbmeas}
  Suppose $(X,\calF,\calN)\in\StdMble$
  and $(Y,\calG,\nu)\in\StdProb$.
  Suppose $f : (X,\calF,\calN)\to \forgetmu(Y,\calG,\nu)$,
  or in other words that $f$ is a measurable map $(X,\calF)\to(Y,\calG)$
  such that $f^{-1}(G)\in\calN$ iff $\nu(G) = 0$.
  The function $\pbmeas f\nu : f^{-1}\calG\to[0,1]$
  defined on the pullback $\sigma$-algebra
  $f^{-1}\calG := \{ f^{-1}(G) \mid G\in\calG\}$
  by $\pbmeas f\nu(f^{-1}(G)) = \nu(G)$
  is a probability measure,
  and $f$ is measure-preserving
  as a map $(X,f^{-1}\calG,\pbmeas f\nu)\to(Y,\calG,\nu)$.
\end{construction}
\begin{proof}
  The function $\pbmeas f\nu$ is well-defined: if $f^{-1}(G) = f^{-1}(G')$
  for some $G,G'\in\calG$, then $f^{-1}(G\triangle G') = \emptyset \in \calN$,
  so $\nu(G\triangle G') = 0$, so $\nu(G) = \nu(G')$.
  It indeed defines a measure, because $\nu$ is a measure and taking $f$-preimages
  preserves the empty set and complements
  and distributes over countable disjoint unions.
  Finally, $f$ is measure-preserving by definition of $\pbmeas f\nu$.
\end{proof}

\begin{note}
  The probability space $(X,f^{-1}\calG,\pbmeas f\nu)$
  is not necessarily standard:
  for example, if $f$ is the map $(x\mapsto [x<1/2]) : [0,1]\to \{\vtrue,\vfalse\}$
  and $\{\vtrue,\vfalse\}$ is given the uniform measure,
  then $f^{-1}\calG$ is the atomic $\sigma$-algebra
  on $[0,1]$ generated by $[0,1/2)$
  and $[1/2,1]$ and $\pbmeas f\nu$ assigns each atom probability $1/2$.
  The triple $(X,f^{-1}\calG,\pbmeas f\nu)$ is not a standard probability space
  because its measure algebra has two atoms but there is no
  bijection from it onto the two-point space.
\end{note}

\begin{lemma} \label{app:lem:stdext}
  For $(X,\calF,\calN)\in\StdMble$
  and $(Y,\calG,\nu)\in\StdProb$ and
  $f : (X,\calF,\calN)\to \forgetmu(Y,\calG,\nu)$,
  there exists a measure $\mu$ extending $\pbmeas f \nu$
  making $(X,\calF,\mu)$ into a standard probability space
  with negligibles $\calN$
  and $f$ into a measure-preserving
  map $(X,\calF,\mu)\to (Y,\calG,\nu)$.
\end{lemma}
\begin{proof}
  Unforget a measure $\lambda$ with negligibles $\calN$
  such that $(X,\calF,\lambda)$ is a standard probability space.
  The measure $\pbmeas f\nu : f^{-1}\calG\to[0,1]$
  is absolutely continuous with respect to
  the restriction $\lambda|_{f^{-1}\calG}$
  of $\lambda$ to the pullback $\sigma$-algebra $f^{-1}\calG$:
  if $\pbmeas f\nu(f^{-1}(G)) =\nu(G) = 0$
  then $f^{-1}(G)\in\calN$, so $\mu(f^{-1}(G)) = 0$.
  Thus $\pbmeas f\nu$ has a Radon-Nikodym derivative
  $g : (X,f^{-1}\calG)\to \R$.
  Since $f^{-1}\calG\subseteq\calF$,
  the derivative $g$ is also measurable as a function $(X,\calG)\to\R$.
  Let $\mu$ be the measure $\calF\to[0,1]$
  defined by $\mu(F) = \int [x\in F]g(x)\,\lambda({\rm d}x)$
  for all $F\in\calF$.
  This extends $\pbmeas f\nu$, and so makes $f$ measure-preserving
  as a map $(X,\calF,\mu)\to(Y,\calG,\nu)$.

  All that's left is to show $(X,\calF,\mu)$ standard with negligibles $\calN$.
  Since $g$ is measurable as a function $(X,f^{-1}\calG)\to\R$,
  there must be some $G\in\calG$ with $\{x\mid g(x) = 0\} = g^{-1}(0) = f^{-1}(G)$.
  This $G$ must be $\nu$-negligible, since
  \[\nu(G) = \pbmeas f\nu(f^{-1}(G))
  = \int [f(x)\in G] g(x) \,\lambda_{f^{-1}\calG}({\rm d}x) 
  = \int [g(x) = 0] g(x) \,\lambda_{f^{-1}\calG}({\rm d}x) 
  = 0.\]
  Since $f$ is negligible-reflecting, this implies $f^{-1}(G) = \{x\mid g(x) = 0\}$
  is $\lambda$-negligible, so the derivative $g$ is $\lambda$-almost-everywhere
  strictly positive.
  This implies $\nu(F) = \int [x\in F]g(x)\,\lambda({\rm d}x) = 0$
  iff $\lambda(F) = 0$, so
  $\lambda$ and $\mu$ have the same negligible sets $\calN$.
  Finally, unforget a collection of Borel sets $\calF'$
  arising from a Polish topology on $X$ such that $\lambda$
  is the completion of the Borel measure $\lambda|_{\calF'}$.
  The measure $\mu$ is correspondingly the completion of the Borel measure
  $\mu|_{\calF'}$ because $\lambda|_{\calF'}$ and $\mu|_{\calF'}$
  have the same Borel-negligibles~\fremlinf{212E(d)}{21}{11},
  so $(X,\calF,\mu)$ standard.
\end{proof}

The following lemma can also be proved via a mild
strengthening of \citet[Definition 3.1]{wendt1998change}.

\begin{lemma} \label{app:lem:mble-right-ore}
  The category $\StdMble$
  has the right Ore property: for $f,g$
  with types as shown below,
  there exists a space $W$ and maps $h,k$
  such that $fh = gk$:
  \[ \commsquare{(W,\calK,\calS)}{k}{(Y,\calG,\calM)}{h}{g}{(X,\calF,\calN)}{f}{(Z,\calH,\calR)} \]
\end{lemma}
\begin{proof}
  By 
  (1)$\Rightarrow$(4) of \cref{app:lem:mblecat-map-char},
  there exist measures $\mu$ and $\rho$
  with negligibles $\calN$ and $\calR$
  respectively making $f$ a $\StdProb$-map
  from $(X,\calF,\mu)$ to $(Z,\calH,\rho)$.
  By \cref{app:lem:stdext} there exists
  a measure $\nu$ with negligibles $\calM$
  making $(Y,\calG,\nu)$ a standard probability space
  and $g$ a $\StdProb$-map
  from $(Y,\calG,\nu)$ to $(Z,\calH,\rho)$.
  The result follows by applying
  \cref{lem:std-prob-right-ore} and forgetting all measures.
\end{proof}

\subsubsection{Semicartesian structure of $\StdMble$}

\begin{fact}[semicartesian structure of $\StdProb$] \label{app:lem:stdprob-semicart}
  The category $\StdProb$ of standard probability spaces
  and measure-preserving maps
  is symmetric \nref{app:def:semicartesian-monoidal-category}{semicartesian monoidal}.
  The symmetric monoidal product of two standard probability spaces
  $(X,\calF,\mu)$
  and $(Y,\calG,\nu)$
  has underlying set $X\times Y$
  with $\sigma$-algebra generated by rectangles
  $F\times G$ for $F\in\calF,G\in\calG$,
  and probability measure defined by
  $(\mu\otimes\nu)(E) = \iint [(x,y)\in E]\,{\rm d}\mu(x){\rm d}\nu(y)$~\cite[Definition 5.25]{axler2020measure}.
  This preserves standardness~\cite[p. 37, \S 2.7]{rohlin1949fundamental}.
  The unit is the one-point probability space.
\end{fact}

The fact that $\StdMble$
is the image of $\forgetmu$ (\cref{app:lem:stdmble-img})
suggests the semicartesian symmetric monoidal structure on $\StdMble$
ought to be the image of the corresponding structure on $\StdProb$.
This is indeed the case, as we now show.

\begin{lemma} \label{app:lem:stdprob-prod-forgets-uniquely}
  Let $X$ be an object of $\StdMble$
  and $X_1,X_2$ two objects in $\StdProb$
  that forget to it (i.e., $\forgetmu X_1 = \forgetmu X_2 = X$).
  Similarly let $Y$ be an object of $\StdMble$ and $Y_1,Y_2$
  two objects of $\StdProb$ with $\forgetmu Y_1 = \forgetmu Y_2 = Y$.
  Then the monoidal products $X_1\otimes Y_1$
  and $X_2\otimes Y_2$ forget to the same
  standard enhanced measurable space: $\forgetmu (X_1\otimes Y_1) = \forgetmu (X_2\otimes Y_2)$.
\end{lemma}
\begin{proof}
  Both $X_1\otimes Y_1$ and $X_2\otimes Y_2$
  have the same underlying set (given by set-theoretic product)
  and $\sigma$-algebra (generated by rectangles),
  so all that's left is to show
  that the measures on $X_1\otimes Y_1$ and $X_2\otimes Y_2$
  have the same negligible sets.
  Let $\mu_{X_1}$ be the measure on $X_1$
  and $\mu_{X_2}$ the measure on $X_2$.
  Since $\forgetmu (X_1) = \forgetmu (X_2) = X$,
  the measures $\mu_{X_1}$ and $\mu_{X_2}$
  have the same negligibles,
  so $\mu_{X_2}$ is absolutely continuous
  with respect to $\mu_{X_1}$
  and has a Radon-Nikodym derivative
  $f : X\to \R$.
  The derivative $f$ can be taken to be strictly positive:
  being a derivative forces
   $f >_{\text{$\mu_{X_1}$-a.e.}} 0$, in or other words that
  $E = \{x \mid f(x) = 0\}$ is $\mu_{X_1}$-negligible,
  for otherwise we would have
  $\mu_{X_2}(E) = \int [x\in E] f(x) \,{\rm d}\mu_{X_1}(x) = 0$
  contradicting the hypothesis that
  $\mu_{X_1}$ and $\mu_{X_2}$ have the same negligible sets.
  Running the same argument on $\forgetmu (Y_1) = \forgetmu (Y_2) = Y$
  shows $\mu_{Y_2}$ has a strictly positive Radon-Nikodym derivative
  $g : X\to \R$ with respect to $\mu_{X_1}$.
  Therefore the product measure $\mu_{X_2}\otimes \mu_{Y_2}$
  has Radon-Nikodym derivative $h(x,y) = f(x)g(y)$
  with respect to product measure $\mu_{X_1}\otimes \mu_{Y_1}$,
  strictly positive because $f,g$ are, and by Fubini's theorem
  \[
    (\mu_{X_2}\otimes \mu_{Y_2})(E)
    = \iint [(x,y)\in E]\,{\rm d}\mu_{X_2}(x){\rm d}\mu_{Y_2}(y)
    = \iint [(x,y)\in E]f(x)g(y)\,{\rm d}\mu_{X_1}(x){\rm d}\mu_{Y_1}(y)
  \]
  is zero exactly when $(\mu_{X_1}\otimes\mu_{Y_1})(E)$ is, as required.
\end{proof}

\begin{definition}[tensor product of standard enhanced measurable spaces] \label{app:def:stdmble-tensor}
  The \emph{tensor product}
  of standard enhanced measurable spaces $X,Y$,
  written $X\otimes Y$, is defined to be
  $\forgetmu (X'\otimes_{\StdProb} Y')$.
  where
  $X'$ and $Y'$ are arbitrary standard probability spaces
  with $\forgetmu (X') = X$ and $\forgetmu (Y') = Y$.
  This is well-defined by 
  \cref{app:lem:stdprob-prod-forgets-uniquely}: the choice of
  $X',Y'$ does not matter.

  This operation extends to a functor $-_1\otimes -_2 : \StdMble\times\StdMble\to\StdMble$:
  given $\StdMble$-maps $f : X\to Y$ and $g : A\to B$,
  pick arbitrary $X',Y'\in\StdProb$ with $\forgetmu (X') = X$ and $\forgetmu (Y') = Y$
  making $f$ into a $\StdProb$-morphism
  (which exist by \cref{app:lem:stdmble-img})
  and similarly pick $A',B'$ for $g$,
  then set $f\otimes g$ to be the map $\forgetmu (f\otimes_{\StdProb} g) : X\otimes A\to Y\otimes B$.
  The choice of $X',Y',A',B'$ doesn't matter because $\forgetmu$ faithful.
  Functoriality follows from functoriality of $\otimes_{\StdProb}$.
\end{definition}

\begin{lemma} \label{app:lem:stdmble-semicartesian}
  The category $\StdMble$
  is \nref{app:def:semicartesian-monoidal-category}{semicartesian monoidal}
  with symmetric monoidal product $\nref{app:def:stdmble-tensor}{(\otimes)}$.
\end{lemma}
\begin{proof}
  The associator, unitor, and swapping map are
  the images under $\forgetmu$ of their counterparts
  in $\StdProb$ (\cref{app:lem:stdprob-semicart}).
  The coherence diagrams commute because
  they are the image under $\forgetmu$ of
  the corresponding diagrams for $\otimes_{\StdProb}$.
\end{proof}

\subsection{Subalgebras with enough room} \label{app:sec:subalgebras-with-enough-room}

Throughout this section,
when we say ``subalgebra'' we mean what \citet{fremlin2000measure}
calls ``closed subalgebra''.
There are two possible meanings for the word ``closed'',
one order-theoretic and one topological, but thankfully the
two coincide in our situation: \fremlinft{323H}{32}{19}
shows ``closed'' has a
canonical meaning for localizable (hence also for probability) algebras.
We will use the order-theoretic definition, so a closed subalgebra
is a boolean subalgebra closed under all small joins~\fremlinf{313D(a)}{31}{22}.

\begin{definition}[enough room] \label{app:def:enough-room}
  Say a subalgebra $\frakC$ of $[0,1]$
  has \emph{enough room}
  if
  every standard probability algebra
  is isomorphic to
  a subalgebra $\frakD$ of $\frakA$
  such that $\frakC$ and $\frakD$
  are
  stochastically independent~\fremlinf{325L}{32}{35}.
\end{definition}

\begin{intuition}
  A subalgebra $\frakC$ of $[0,1]$ having enough room
  is analogous to Lilac's requirement
  that sub-$\sigma$-algebras of the Hilbert cube
  have finite footprint~\lilac{Definition 2.6}{15}.
  The main idea is that, if $[0,1]$ is to be
  an inexhaustible source of randomness,
  then a given subalgebra
  must be small enough that 
  one can always allocate fresh (i.e., stochastically independent)
  standard probability algebras inside of it.
\end{intuition}

\ifdefined\flavortext

\begin{note}
  We have a lot of freedom in how we choose to picture
  the measure algebra $[0,1]$. As a measure algebra, it is
  isomorphic to $[0,1]^2$ (the Lebesgue measure on the unit square)
  and $[0,1/2] \times [0,1/2]$ and $[0,1/3]\times[0,2/3]$
  (where $\times$ is the simple product of measure algebras, so that these algebras
  correspond to the disjoint
  union of two intervals) and so on.
  Because of this, in pictures below $[0,1]$ may be drawn as any of these spaces.
\end{note}

\subsubsection{A subalgebra with enough room}

\begin{example}
  Let $\frakA$ be the probability algebra given by the Lebesgue measure on the unit square $[0,1]^2$.
  Let $\frakC$ be the subalgebra generated by the projection $\pi_1 : [0,1]^2\to[0,1]$,
  corresponding to the sub-$\sigma$-algebra consisting of rectangles of the form
  $F\times [0,1]$ for $F$ a Lebesgue-measurable subset of $[0,1]$.
  The subalgebra $\frakC$ has enough room, because any standard probability algebra
  can be embedded as a subalgebra of $\frakA$
  corresponding to a sub-$\sigma$-algebra consisting of rectangles of the form 
  $[0,1]\times F'$.
\end{example}

\input{subalgebra-with-not-enough-room}

\fi 

\subsubsection{Characterizing subalgebras with enough room}

\ifdefined\flavortext
We now formalize the intuitions sketched in the
above example, and characterize the subalgebras with enough room.

\begin{lemma} \label{app:lem:indep-incomparable}
  Let $\frakC$ be a subalgebra of a measure algebra $(\frakA,\mu)$
  and $a$ an element of $\frakA$ with nontrivial measure (so $\mu(a)\notin\{0,1\}$)
  that is independent of $\frakC$ in the sense that $\mu(a\cap c) = \mu(a)\mu(c)$
  for all $c$ in $\frakC$.
  Then $a$ is 
  incomparable with every element of $\frakC$;
  that is, for all $c$ nontrivial in $\frakC$ we have neither $a\subseteq c$ nor $a\supseteq c$.
\end{lemma}
\begin{proof}
  If $a\subseteq c$ for some nontrivial $c$ then $\mu(a) = \mu(a\cap c) = \mu(a)\mu(c)$
  forcing $\mu(c) = 0$ or $\mu(a) = 1$, impossible because $a,c$ nontrivial.
  Similarly, if $a\supseteq c$ for some nontrivial $c$
  then $\mu(c) = \mu(a\cap c) = \mu(a)\mu(c)$
  forcing $\mu(a) = 0$ or $\mu(c) = 1$, impossible because $a,c$ nontrivial.
\end{proof}

\fi 

\begin{lemma} \label{app:lem:enough-room-interval}
  A subalgebra $\frakC$ of $[0,1]$ has \nameref{app:def:enough-room}
  iff there exists a subalgebra $\frakD$ isomorphic to $[0,1]$
  with $\frakC,\frakD$ independent.
\end{lemma}
\begin{proof}
  If $\frakC$ has enough room then there certainly
  is an independent subalgebra isomorphic to $[0,1]$,
  since $[0,1]$ is a \nref{app:def:stdprobalg}{standard probability algebra}.
  Conversely, any standard probability algebra $\frakF$
  can be embedded as subalgebra of $[0,1]$,
  so if $\frakC$ is independent of $\frakD$ isomorphic to $[0,1]$
  then the subalgebra of $\frakD$ corresponding to $\frakF$
  is also independent of $\frakC$.
\end{proof}

\begin{lemma} \label{app:lem:enough-room-mono}
  If $\frakC$ has \nameref{app:def:enough-room}
  and $\frakB\subseteq\frakC$
  then $\frakB$  has enough room.
\end{lemma}
\begin{proof}
  There exists $\frakD$ with $\frakD\cong[0,1]$ and $\frakC,\frakD$ independent,
  so $\frakB,\frakD$ independent,
  so $\frakB$ has enough room by \cref{app:lem:enough-room-interval}.
\end{proof}

\begin{lemma} \label{app:lem:enough-room-iso}
  If $\frakC$ has \nameref{app:def:enough-room}
  and $\pi$ is an automorphism of $[0,1]$
  then $\pi \frakC$ has enough room.
\end{lemma}
\begin{proof}
  There exists $\frakD$ with $\frakD\cong[0,1]$ and $\frakC,\frakD$ independent,
  so $\pi\frakC,\pi\frakD$ independent and $\pi\frakD\cong[0,1]$,
  so $\pi\frakB$ has enough room by \cref{app:lem:enough-room-interval}.
\end{proof}

\begin{lemma}
  If a subalgebra $\frakC$ of $\frakA$ has \nameref{app:def:enough-room}
  then for any standard probability algebra there exists
  a subalgebra $\frakD$ isomorphic to it and independent of $\frakC$
  such that the subalgebra generated by $\frakC,\frakD$
  still has enough room.
\end{lemma}
\begin{proof}
  Since $\frakC$ has enough room,
  there is a copy of $[0,1]^2$ embedded in $\frakA$ independent of $\frakC$.
  Any standard probability space can be embedded
  as a subalgebra $\frakD$ of $[0,1]$,
  hence as a subalgebra $\frakD \lmafp\top$
  of $[0,1]^2 = [0,1]\lmafp[0,1]$~\fremlinf{325F}{32}{32},
  where $\lmafp$ is the localizable measure algebra free product~\fremlinf{325E}{32}{32}.
  By associativity of independent subalgebras~\fremlinf{272K}{27}{14}~\fremlinf{325X(g)}{32}{37}
  this implies the subalgebra generated by $\frakC,\frakD$
  is independent of the subalgebra $\top \lmafp[0,1] \cong [0,1]$
  of the copy of $[0,1]^2$ embedded in $\frakC$.
  Thus we have found $\frakD$ such that the subalgebra generated by $\frakC,\frakD$
  is independent of a copy of $[0,1]$,
  which implies the subalgebra generated by $\frakC,\frakD$
  has enough room by \cref{app:lem:enough-room-interval}.
\end{proof}

\begin{lemma} \label{app:lem:enough-room-rel-atomless}
  If $\frakC\subseteq[0,1]$ has enough room,
  then $[0,1]$ is relatively atomless over $\frakC$~\fremlinf{331A}{33}{1}.
\end{lemma}
\begin{proof}
  Since $\frakC$ has enough room, there is a subalgebra $\frakD$
  of $[0,1]$ with $\frakD\cong[0,1]$ and $\frakC,\frakD$ independent.
  Let $\angled{\frakC,\frakD}$ be the subalgebra generated by
  $\frakC$ and $\frakD$.
  Since $\frakD\cong[0,1]$,
  there is an isomorphism
  $\pi : \angled{\frakC,\frakD}\cong \frakC\lmafp\frakD \cong \frakC\lmafp[0,1]$~\fremlinf{325L}{32}{35}.
  Under this isomorphism, elements $c$ of $\frakC$
  look like rectangles $\pi c \lmafp \top_{[0,1]}$,
  and the image of $\pi$ is the subalgebra $\{c \lmafp \top \mid c \in [0,1]\}\subseteq \frakC\lmafp[0,1]$ of all such rectangles.
  The measure algebra $\frakC\lmafp[0,1]$ is relatively atomless
  over this subalgebra, because the principal ideal generated by any rectangle $c\lmafp \top$
  contains elements such as $c\lmafp[0,1/2]$ that are not of the form $(c\lmafp\top)\cap (d\lmafp\top)$
  for any $d$.
  Thus $\angled{\frakC,\frakD}$ is relatively atomless over $\frakC$.
  Since $\frakC\subseteq\angled{\frakC,\frakD}\subseteq[0,1]$,
  this implies $[0,1]$ relatively atomless over $\frakC$~\fremlinf{331Y(a)}{33}{9}.
\end{proof}

\begin{lemma} \label{app:lem:rel-atomless-projection}
  If $[0,1]$ is $\frakC$-relatively atomless,
  then there is an isomorphism $\pi : [0,1]\to \frakC\lmafp[0,1]$
  with $\pi c = c\lmafp \top_{[0,1]}$
  for all $c$ in $\frakC$.
\end{lemma}
\begin{proof}
  Consider the decomposition
   \begin{align*}
      \pi &: [0,1] \xrightarrow{\sim}\prod_{n\in\N} \frakC_{e_n} \times \prod_{\kappa\in K}\frakC_{e_\kappa}\lmafp \frakB_\kappa \\
      \text{for all $c$ in $\frakC$, }\quad\pi c &= ((c \cap e_n)_{n\in \N}, ((c\cap e_\kappa)\lmafp \top_{\frakB_\kappa})_{\kappa\in K})
   \end{align*}
   of $[0,1]$ with respect to $\frakC$
   given by
   \fremlin{\href{https://www1.essex.ac.uk/maths/people/fremlin/chap33.pdf\#page=29}{333K}},
  where $\frakB_\kappa$ denotes the standard
  measure algebra on $\{0,1\}^\kappa$~\fremlinf{333A(d)}{33}{22}.
  (We have changed the notation slightly
  from the statement of this decomposition in Fremlin, writing $\top_{\frakB_\kappa}$
  for the top element of the algebra $\frakB_\kappa$ instead of $1$,
  to avoid confusion with the real number $1$ in $[0,1]$.)
  As discussed in \fremlinft{333K(a)}{33}{29},
  each factor $\frakC_{e_n}$ corresponds to a relative atom of $[0,1]$ over $\frakC$.
  Since $[0,1]$ is $\frakC$-relatively atomless,
  we must have $e_n = \bot$ for all $n$,
  so that the product $\prod_n \frakC_{e_n}$ disappears and $\pi$ can be rewritten as
   \begin{align*}
      \pi &: [0,1] \xrightarrow{\sim} \prod_{\kappa\in K}\frakC_{e_\kappa}\lmafp \frakB_\kappa \\
      \text{for all $c$ in $\frakC$, }\quad\pi c &= ((c\cap e_\kappa)\lmafp 1)_{\kappa\in K}
   \end{align*}
  Each factor $\frakC_{e_\kappa}\lmafp\frakB_\kappa$
  corresponds to a principal ideal of $[0,1]$ with Maharam type $\kappa$~\fremlinf{333G(a)}{33}{25}.
   The measure algebra $[0,1]$ is Maharam-type-homogenous
   with Maharam type $\omega$~\fremlinf{331K}{33}{7}~\fremlinf{254K}{25}{53},
   so does not have any principal ideals
   with Maharam type $\kappa>\omega$~\fremlinf{331H(c)}{33}{5}.
  Thus, if $\pi$ is to be an isomorphism, we must have $K = \{\omega\}$, so
  the product $\prod_\kappa \frakC_{e_\kappa}\lmafp \frakB_\kappa$
  reduces to $\frakC_{e_\omega}\lmafp \frakB_\omega$.
  Moreover, the measure algebra $\frakB_\omega$
  is isomorphic to $[0,1]$~\fremlinf{254K}{25}{53}.
  Thus $\pi$ can be rewritten as
   \begin{align*}
      \pi &: [0,1] \xrightarrow{\sim} \frakC_{e_\omega}\lmafp [0,1] \\
      \text{for all $c$ in $\frakC$, }\quad\pi c &= (c\cap e_\omega)\lmafp \top_{[0,1]}
   \end{align*}
  Finally, the fact that $\pi$ preserves measure forces $\mu(e_\omega) = 1$;
  this in turn forces $e_\omega$ to have measure $1$, and so $e_\omega = \top_\frakC$.
  Thus
   \begin{align*}
      \pi &: [0,1] \xrightarrow{\sim} \frakC\lmafp [0,1]\\
      \text{for all $c$ in $\frakC$, }\quad\pi c &= c\lmafp \top_{[0,1]}
   \end{align*}
  as claimed.
\end{proof}

\begin{lemma} \label{app:lem:projection-enough-room}
  If there is an isomorphism $\pi : [0,1]\to \frakC\lmafp[0,1]$
  with $\pi c = c\lmafp \top_{[0,1]}$
  for all $c$ in $\frakC$,
  then 
  $\frakC$ has \nameref{app:def:enough-room}.
\end{lemma}
\begin{proof}
  The subalgebra $\top_\frakC\lmafp[0,1]$ generated by
  rectangles $\top_\frakC\lmafp a$ for $a$ in $[0,1]$
  is isomorphic to $[0,1]$ and independent
  of $\frakC\lmafp\top_{[0,1]}$.
  Transporting back along $\pi$
  gives a subalgebra $\pi^{-1}(\top_\frakC\lmafp[0,1])$
  of $[0,1]$ independent of $\frakC$,
  so $\frakC$ has enough room by \cref{app:lem:enough-room-interval}.
\end{proof}

\begin{theorem}[characterization of subalgebras with enough room] \label{app:thm:enough-room-char}
  The following are equivalent:
  \begin{enumerate}
    \item $\frakC\subseteq [0,1]$ has \nameref{app:def:enough-room}
    \item 
      There is an isomorphism $\pi : [0,1]\to \frakC\lmafp[0,1]$
      with $\pi c = c\lmafp \top_{[0,1]}$
      for all $c$ in $\frakC$
    \item $[0,1]$ is $\frakC$-relatively atomless
  \end{enumerate}
\end{theorem}
\begin{proof}
  \cref{app:lem:enough-room-rel-atomless},
  \cref{app:lem:rel-atomless-projection},
  and \cref{app:lem:projection-enough-room}
  give the cycle
  $1\Rightarrow2\Rightarrow3\Rightarrow 1$.
\end{proof}

\begin{intuition}
  \cref{app:thm:enough-room-char} says
  subalgebras with enough room
  have a very nice form:
  up to measure-algebra isomorphism,
  they are 
  projections out of a product space with an interval's worth of independence
  available for allocating fresh randomness.
  Also, the characterization based on relative-atomlessness
  shows that the idea of having enough room does not depend
  on measures in an essential way, as
  relative-atomlessness is purely a property of the underlying
  Boolean algebras.
\end{intuition}

\subsection{Automorphisms of the interval} \label{app:sec:autos}

\begin{lemma} \label{app:lem:rel-atomless-omega}
  If $[0,1]$ is relatively atomless over a subalgebra $\frakC$,
  then $\href{https://www1.essex.ac.uk/maths/people/fremlin/chap33.pdf\#page=21}{\tau_{\frakC_a}([0,1]_a)} = \omega$ for all $a\ne\bot$ in $[0,1]$.
\end{lemma}
\begin{proof}
  Relative atomlessness implies $\tau_{\frakC_a}[0,1]_a \ge \omega$
  for all $a\ne\bot$ in $[0,1]$~\fremlinf{333B(d)}{33}{22},
  and
  \[\tau_{\frakC_a}[0,1]_a 
  \stackrel{\text{\href{https://www1.essex.ac.uk/maths/people/fremlin/chap33.pdf\#page=22}{333B(a)}}}{\le} \tau_{\frakC}[0,1] 
  \stackrel{\text{\href{https://www1.essex.ac.uk/maths/people/fremlin/chap33.pdf\#page=22}{333B(e)}}}{\le} \tau[0,1] 
  \stackrel{\text{\href{https://www1.essex.ac.uk/maths/people/fremlin/chap33.pdf\#page=7}{331K}+\href{https://www1.essex.ac.uk/maths/people/fremlin/chap25.pdf\#page=53}{254K}}}= \omega.\]
\end{proof}

\begin{lemma} \label{app:lem:rel-atomless-projection-forgot}
  Let $\frakA$ be a subalgebra of
  $[0,1]$ considered as a \nref{app:def:measurable-alg}{measurable algebra}.
  If $[0,1]$ is $\frakA$-relatively atomless,
  then there exists a measurable-algebra isomorphism $\pi : [0,1]\to \frakA\lmafp [0,1]$
  such that $\pi a = a\lmafp\top_{[0,1]}$ for all $a$ in $\frakA$.
\end{lemma}
\begin{proof}
  Put the usual Lebesgue measure $\lambda$ on $[0,1]$,
  rendering $\frakA$ a closed sub-measure-algebra of the measure algebra
  $([0,1],\lambda)$, apply \cref{app:lem:rel-atomless-projection}
  to get a measure algebra isomorphism
  $\pi : ([0,1],\lambda)\to(\frakA,\lambda|_\frakA)\lmafp([0,1],\lambda)$
  with $\pi a = a\lmafp \top_{[0,1]}$ for all $a$ in $\frakA$,
  and then forget all the measures and the fact that $\pi$ preserves them.
\end{proof}

\begin{lemma}[Homogeneity] \label{app:lem:mble-homogeneity}
  Let $\frakA,\frakB$ be standard probability subalgebras of $[0,1]$ that render it relatively atomless.
  Let $f : \frakB\hookrightarrow \frakA$
  be an injective order-continuous Boolean algebra homomorphism.
  There exists an order-continuous Boolean algebra automorphism $\pi : [0,1]\to[0,1]$
  refining $f$ in the sense that $\pi b = f b$ for all $b$ in $\frakB$.
  In other words, the diagram
\[ \begin{tikzcd}
    {[0,1]}                        & {[0,1]} \arrow[l, "\pi"', dashed]                     \\
    {\frakA} \arrow[u, hook] & {\frakB} \arrow[l, "f", hook] \arrow[u, hook]
    \end{tikzcd}\]
   commutes, where the vertical arrows are the inclusions $\frakA\subseteq[0,1]$ and $\frakB\subseteq[0,1]$.
\end{lemma}
\begin{proof}
  Give $[0,1]$ in the top-left corner the usual Lebesgue measure $\lambda$.
  Since $f$ and the inclusion $\frakA\subseteq [0,1]$ are injective,
  restricting $\lambda$ along these gives measures $\mu,\nu$ such that
   \[ \begin{tikzcd}
    {([0,1],\lambda)}                        &  \\
    {(\frakA,\mu)} \arrow[u, hook] & {(\frakB,\nu)} \arrow[l, "f", hook] 
    \end{tikzcd}\]
  is a diagram in $\StdProbAlg$,
  the category of standard probability algebras and measure-algebra homomorphisms.
  By \cref{app:lem:rel-atomless-projection-forgot},
  there exists an order-continuous Boolean algebra isomorphism $\pi : [0,1]\to \frakB\lmafp[0,1]$
  such that the triangle
  \[
\begin{tikzcd}
  {\frakB\lmafp[0,1]} & {[0,1]} \arrow[l, "\pi"', two heads, hook]                 \\
                      & \frakB \arrow[u, hook] \arrow[lu, "{-\lmafp\top_{[0,1]}}"]
  \end{tikzcd}
  \]
  commutes.
  The homomorphism $-\lmafp \top_{[0,1]}$
  is measure-preserving as a map $(\frakB,\nu)\to (\frakB,\nu)\lmafp([0,1], \lambda)$,
  since
  \[ 
    (\nu\lmafp\lambda)(b\lmafp\top_{[0,1]})
    = \nu b\cdot \lambda\top_{[0,1]}
    = \nu b
    \]
  for all $b$ in $\frakB$.
  Therefore it fits in with the other $\StdProbAlg$-diagram above to give the following
  $\StdProbAlg$-diagram:
   \[ 
\begin{tikzcd}
  {([0,1],\lambda)}              & {(\frakB,\nu)\lmafp([0,1],\lambda)} &                                                                           \\
  {(\frakA,\mu)} \arrow[u, hook] &                     & {(\frakB,\nu)} \arrow[ll, "f", hook] \arrow[lu, "{-\lmafp \top_{[0,1]}}"]
  \end{tikzcd}
    \]
  This diagram can be completed into a commutative quadrilateral.
  Note that
  $[0,1]$ is $\frakA$-relatively atomless
  and $(\frakB,\nu)\lmafp([0,1],\lambda)$
  is $\frakB$-relatively atomless,
  and $(\frakB,\nu)\lmafp([0,1],\lambda)$
  has countable Maharam type~\fremlinf{333G(a)}{33}{26}
  because $\frakB$ standard and so has at-most-countable Maharam type.
  Therefore, the same argument as in \cref{app:lem:rel-atomless-omega}
  gives $\tau_{[0,1]_a}([0,1]) = \tau_{(\frakB\lmafp[0,1])_{b\lmafp \top}}(\frakB\lmafp[0,1]) =\omega$
  for all $a$ in $[0,1]$ and $b$ in $\frakB$.
  Now \fremlin{\href{https://www1.essex.ac.uk/maths/people/fremlin/chap33.pdf\#page=22}{333C(b)}}
  gives a measure algebra isomorphism
  $\sigma$ completing
  the above diagram into a commutative quadrilateral:
  \[
\begin{tikzcd}
  {([0,1],\lambda)}              & {(\frakB,\nu)\lmafp([0,1],\lambda)} \arrow[l, "\sigma"', two heads, hook, dashed] &                                                                           \\
  {(\frakA,\mu)} \arrow[u, hook] &                                                           & {(\frakB,\nu)} \arrow[ll, "f", hook] \arrow[lu, "{-\lmafp \top_{[0,1]}}"]
  \end{tikzcd}
  \]
  Forgetting all measures and
  combining this quadrilateral with the triangle above gives the following commutative rectangle
  of measurable algebras:
  \[\begin{tikzcd}
  {[0,1]}              & {\frakB\lmafp[0,1]} \arrow[l, "\sigma"', two heads, hook] & {[0,1]} \arrow[l, "\pi"', two heads, hook]                                                        \\
  {\frakA} \arrow[u, hook] &                                          & {\frakB} \arrow[ll, "f", hook] \arrow[u, hook] \arrow[lu, "{-\lmafp \top}"]
  \end{tikzcd}\]
  The composite $\sigma\pi$ is an automorphism of the form required.
\end{proof}

\begin{lemma}[Correspondence] \label{app:lem:mble-distinguishability}
  For $\frakC\subseteq[0,1]$, let $\Fix\frakC$ be the subgroup of $\Aut_{\StdMble}[0,1]$
  consisting of autos fixing every $c$ in $\frakC$:
  \[ \Fix\frakC := \{ \pi \mid \pi c =c\text{ for all }c \in \frakC \}. \]
  If $[0,1]$ is $\frakC$-relatively atomless
  and $\Fix\frakC\subseteq\Fix\frakD$,
  then $\frakD\subseteq\frakC$.
\end{lemma}
\begin{proof}
  If every $\StdMble$-auto fixing $\frakC$ fixes $\frakD$,
  then surely every $\StdProb$-auto fixing $\frakC$
  fixes $\frakD$.
  Now suppose for contradiction that there exists $d$ in $\frakD$ not in $\frakC$,
  with aim to construct a $\StdProb$-auto fixing $\frakC$ but not $d$.
  By \cref{app:thm:enough-room-char},
  there is an isomorphism $\sigma : [0,1]\to \frakC\lmafp[0,1]$
  such that $\sigma c = c\lmafp \top_{[0,1]}$ for all $c$ in $\frakC$.
  By \fremlin{\href{https://www1.essex.ac.uk/maths/people/fremlin/chap33.pdf\#page=35}{333R(ii)$\Rightarrow$(iii)}},
  $\frakC$ arises as the fixed-point subalgebra of a particular automorphism
   $\pi : [0,1]\xrightarrow\sim[0,1]$,
   so that $c\in\frakC$ iff $\pi c = c$.
  Since $d\notin\frakC$ we must have $\pi d \ne d$,
  so $\pi$ is an auto of the form required.
\end{proof}

\subsection{The target practice lemma}

\begin{lemma} \label{app:lem:target-practice}
  Let $\frakB$ be a \nref{app:def:stdmblealg-def}{standard measurable algebra},
  $\frakA\subseteq\frakB$ a standard measurable subalgebra,
  and $b$ an element of $\frakB$ not in $\frakA$.
  There exists a standard measurable algebra
  $\frakC$ and homomorphisms of standard measurable algebras
  $f,g : \frakB\hookrightarrow\frakC$
  such that $f|_\frakA = g|_\frakA$
  and $f(b)\notin\img g$.
\end{lemma}
\begin{proof}
  Let $l$ be the largest element of $\frakA$ contained in $b$
  and $u$ the smallest element of $\frakA$ containing $b$:
  \[ l := \bigvee_{\frakA\ni a\le b} a \qquad u := \bigwedge_{\frakA\ni a\ge b} a \]
  These exist and are elements of $\frakA$ because $\frakA$ is a complete boolean algebra.
  Since $b$ is not in $\frakA$, neither $l$ nor $u$ can be equal to it.
  Thus $b$ is sandwiched between $l$ and $u$ by a sequence of strict inequalities
  $l < b < u$.

  Consider an arbitrary $a$ in $\frakA$ less than $u-l$.
  Any such $a$ can be decomposed
  into a disjoint union $a_{u-b} + a_{b-l}$
  with $a_{u-b} \le u-b$ and $a_{b-l}\le b-l$,
  by setting $a_{u-b} := a\wedge (u-b)$ and $a_{b-l} := a\wedge(b-l)$.
  The following diagram depicts the situation (and explains the lemma's name):
  
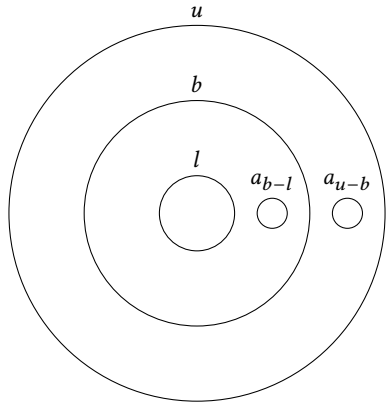
\[
    \begin{tikzpicture}
      \draw 
         (0,0) circle (2.5) (0,2.5) node [above] {$u$}
         (0,0) circle (1.5) (0,1.5) node [above] {$b$}
         (0,0) circle (0.5) (0,0.5) node [above] {$l$}
         (1,0) circle (0.2) (1,0.2) node [above] {$a_{b-l}$}
         (2,0) circle (0.2) (2,0.2) node [above] {$a_{u-b}$}
        ;
    \end{tikzpicture}
  \]
  The idea is to make use of the following fact:
  for any $a$ in $\frakA$ that is nonzero (i.e., $a\ne\bot$),
  both $a_{u-b}$ and $a_{b-l}$ as shown above must be nonzero too.
  For if $a_{u-b} = \bot$ and $a_{b-l}$ nonzero,
  then $a_{b-l} + l$ would be an element of $\frakA$ contained in $b$
  and larger than $l$, contradicting $l$ maximal.
  Symmetrically, if $a_{b-l} = \bot$ and $a_{u-b}$ nonzero,
  then $u - a_{u-b}$ would be an element of $\frakA$
  containing $b$ and smaller than $u$, contradicting $u$ minimal.
  This property --- that nonzero elements of $\frakA$ contained in the annulus $u-l$ depicted above
  decompose into disjoint unions $a_{u-b} + a_{b-l}$
  with $a_{u-b},a_{b-l}$ both nonzero --- distinguishes
  elements of $\frakA$ from $b$, since $b \wedge (u - b) = \bot$.
  The idea of the proof is to construct maps $f,g$
  making this difference visible, allowing to separate the
  element $b$ from the subalgebra $\frakA$.

  To do this, we first formulate the above remarks
  in terms of boolean algebras.
  The sequence of inequalities $l < b < u$
  gives a partition of unity $\{\neg u, u - b, b - l, l\}$,
  which gives a decomposition of $\frakB$ as a simple product
  of principal ideals
  $
    \frakB \cong \frakB_{\neg u} \times \frakB_{u-b} \times \frakB_{b-l} \times \frakB_l
  $.
  Since $l$ and $u$ are both in the subalgebra $\frakA$, it decomposes similarly as
  $
    \frakA \cong \frakA_{\neg u} \times \frakA_{u-l} \times \frakA_l
  $.
  The inclusion $\frakA\subseteq\frakB$ decomposes correspondingly into three inclusions of principal ideals:
  \[
    \frakA_{\neg u}\subseteq \frakB_{\neg u}
    \qquad
    \frakA_{u-l}\subseteq \frakB_{u-b} \times \frakB_{b-l}
    \qquad
    \frakA_l\subseteq \frakB_l
  \]
  Let $i$ be the inclusion $\frakA_{u-l}\subseteq \frakB_{u-b} \times \frakB_{b-l}$.
  As a homomorphism into a product of boolean algebras,
  $i$ can be written uniquely as
  \[
    i = (i_{u-b},i_{b-l})
    \quad\text{for some complete-boolean-algebra homomorphisms}\quad
    \begin{aligned}
      i_{u-b} &: \frakA_{u-l} \to \frakB_{u-b} \\
      i_{b-l} &: \frakA_{u-l} \to \frakB_{b-l} 
    \end{aligned}
  \]
  The earlier discussion shows that $i_{u-b}(a)$ and $i_{b-l}(a)$
  are both nonzero for all nonzero $a\in\frakA_{u-l}$.
  More is true: the homomorphisms $i_{u-b},i_{b-l}$
  are injective.
  To see that $i_{u-b}$ is injective, pick arbitrary $a,a'\in \frakA_{u-l}$
  with $i_{u-b}(a) = i_{u-b}(a')$.
  Unwinding definitions, this is equivalent to saying
  $a_{u-b} = a'_{u-b}$,
  where $a$ decomposes as $a_{u-b} + a_{b-l}$
  and $a'$ as $a'_{u-b} +a'_{b-l}$.
  This forces $a_{b-l} = a'_{b-l}$ and hence $a=a'$ showing $i_{u-b}$ injective,
  for otherwise
  $ a - a' = (a_{u-b} +a_{b-l}) - (a'_{u-b} + a'_{b-l}) = a_{b-l} - a'_{b-l}$
  would be an element of $\frakA_{u-l}$ contained entirely in $b-l$,
  contradicting $l$ maximal.
  An analogous argument shows $i_{b-l}$ injective.

  Since $i_{u-b},i_{b-l}$ are injective complete-boolean-algebra homomorphisms
  and therefere homomorphisms of standard measurable algebras, they correspond
  by \cref{app:lem:mblecat-dual}
  to maps of standard enhanced measurable spaces with common codomain.
  By \cref{app:lem:mble-right-ore}, this cospan
  completes to a commutative square.
  Passing this square back through \cref{app:lem:mblecat-dual}
  gives 
  a standard measurable algebra $\frakC_{u-l}$
  and homomorphisms $j,k$
  fitting into the following commutative square:
  \[\commsquare{\frakA}{i_{u-b}}{\frakB_{u-b}}{i_{b-l}}{j}{\frakB_{b-l}}{k}{\frakC_{u-l}}\]
  All ingredients needed to construct $f,g$ are now at hand.
  Let $\frakC$ be the standard measurable algebra $\frakB_{\neg u}\times (\frakC_{u-l}\lmafp[0,1]) \times \frakB_l$.
  Let $p,q$ be the injective homomorphisms $\frakC_{u-l}\times\frakC_{u-l}\hookrightarrow \frakC_{u-l}\lmafp[0,1]$
  defined by the following equations:
  \begin{align*}
     p(c_1,c_2) &= (c_1\lmafp [0,1/2]) + (c_2\lmafp [1/2,1]) \\
     q(c_1,c_2) &= (c_1\lmafp [0,1/3]) + (c_2\lmafp [1/3,1]) 
  \end{align*}
  Let $f,g : \frakB\hookrightarrow\frakC$ be the following composites
  (note the only difference in their definitions is whether $p$ or $q$ is used at the end):
  \begin{align*}
     f &= \left(\frakB
       \xrightarrow{\sim} \frakB_{\neg u}\times \frakB_{u-b}\times\frakB_{b-l}\times \frakB_l
       \xhookrightarrow{1 \times j\times k \times 1} \frakB_{\neg u}\times (\frakC_{u-l}\times\frakC_{u-l})\times \frakB_l
       \xhookrightarrow{1 \times p \times 1} \frakB_{\neg u}\times (\frakC\lmafp[0,1])\times \frakB_l
       =\frakC
       \right) \\
     g &= \left(\frakB
       \xrightarrow{\sim} \frakB_{\neg u}\times \frakB_{u-b}\times\frakB_{b-l}\times \frakB_l
       \xhookrightarrow{1 \times j\times k \times 1} \frakB_{\neg u}\times (\frakC_{u-l}\times\frakC_{u-l})\times \frakB_l
       \xhookrightarrow{1 \times q \times 1} \frakB_{\neg u}\times (\frakC\lmafp[0,1])\times \frakB_l
       =\frakC
       \right) 
  \end{align*}
  Since every element $x$ of $\frakB$ decomposes into a disjoint union
  $x_{\neg u} +x_{u-b} + x_{b-l} + x_l$ for some $x_{\neg u}\in\frakB_{\neg u}$
  and $x_{u-b} \in\frakB_{u-b}$ and $x_{b-l}\in\frakB_{b-l}$ and $x_l\in\frakB_l$,
  the action of $f$ and $g$ on arbitrary elements of $\frakB$ can be described 
  by the following equations:
  \begin{align*}
    f(x_{\neg u} +x_{u-b} + x_{b-l} + x_l) &= (x_{\neg u}, (j(x_{u-b})\lmafp[0,1/2]) + (k(x_{b-l})\lmafp[1/2,1]), x_l) \\
    g(x_{\neg u} +x_{u-b} + x_{b-l} + x_l) &= (x_{\neg u}, (j(x_{u-b})\lmafp[0,1/3]) + (k(x_{b-l})\lmafp[1/3,1]), x_l)
  \end{align*}
  We now verify $f,g$ have the desired property.
  First $f,g$ agree on the subalgebra $\frakA$:
  for any $a$ in $\frakA$, it holds that $j(a_{u-b}) = k(a_{b-l})$
  because the square $ji_{u-b} = ki_{b-l}$ commutes, so
  \begin{align*}
    f(a_{\neg u} + a_{u-b} + a_{b-l} + a_l)
    &= (a_{\neg u}, (j(a_{u-b})\lmafp[0,1/2]) + (k(a_{b-l})\lmafp[1/2,1]), a_l) \\
    &= (a_{\neg u}, (j(a_{u-b})\lmafp[0,1/2]) + (j(a_{u-b})\lmafp[1/2,1]), a_l) \\
    &= (a_{\neg u}, j(a_{u-b})\lmafp([0,1/2]+[1/2,1]), a_l) \\
    &= (a_{\neg u}, j(a_{u-b})\lmafp \top, a_l) \\
    &= (a_{\neg u}, j(a_{u-b})\lmafp([0,1/3]+[1/3,1]), a_l) \\
    &= (a_{\neg u}, (j(a_{u-b})\lmafp[0,1/3]) + (j(a_{u-b})\lmafp[1/3,1]), a_l) \\
    &= (a_{\neg u}, (j(a_{u-b})\lmafp[0,1/3]) + (k(a_{b-l})\lmafp[1/3,1]), a_l) \\
    &= g(a_{\neg u} +a_{u-b} + a_{b-l} + a_l).
  \end{align*}
  Second, $f(b)\notin\img g$:
  every element in the image of $g$ is of the form
  \begin{align*}
    (x_{\neg u}, (j(x_{u-b})\lmafp[0,1/3]) + (k(x_{b-l})\lmafp[1/3,1]), x_l)
  \end{align*}
  but
  \begin{align*}
    f(b)
    &= f(\bot + \bot + (b-l) + l)\\
    &= (\bot, (j(\bot)\lmafp[0,1/2]) + (k(\top_{\frakB_{b-l}})\lmafp[1/2,1]), \top_{\frakB_l}) \\
    &= (\bot, (\bot\lmafp[0,1/2]) + (\top\lmafp[1/2,1]), \top) \\
    &= (\bot, \top\lmafp[1/2,1], \top)
  \end{align*}
  and there are no $x_{u-b},x_{b-l}$
  such that $(j(x_{u-b})\lmafp[0,1/3]) + (k(x_{b-l})\lmafp[1/3,1]) = \top\lmafp[1/2,1]$.
\end{proof}

\section{General sheaf theory} \label{app:sec:sheaf-theory}

\subsection{Atomic sheaves} \label{app:sec:atomic-sheaves}

\begin{notation} \label{app:notation:psh-action}
  If $F$ is a presheaf
  on $C$ and $f:X\to Y$ a morphism in $C$
  and $y\in FY$,
  we will write $y\cdot f$ for the element $F(f)(y)\in FX$.
\end{notation}

\begin{definition}[$f$-invariance] \label{app:def:f-inv}
  Let $F$ be a presheaf on $C$,
  $y\in Fd$ an element of $F$,
  and $f : d\to c$ a morphism in $C$.
  Following \citet{simpson2017probability},
  say $y$ is \emph{$f$-invariant}
  if for all $g,h : e\to d$
  with $fg = fh$
  it holds that
  $y\cdot g = y\cdot h$.
  The following diagram illustrates the situation:
  \[ \begin{tikzcd}
  y\cdot g = y\cdot h \in Fe                                & y\in Fd \arrow[l, maps to] &   \\
  e \arrow[r, "h"', shift right] \arrow[r, "g", shift left] & d \arrow[r, "f"]           & c
  \end{tikzcd}\]
\end{definition}

\begin{fact} \label{app:fact:right-ore-atomic-topology}
  The atomic topology exists for $C$ iff $C$ satisfies the \emph{right Ore property}:
  for any two morphisms $f : c\to e$ and $g : d\to e$,
  there exists an object $b$ and morphisms $h : b\to c$
  and $k : b\to d$ such that
  \[ \commsquare{b}{k}{d}{h}{g}{c}{f}{e} \]
  commutes~\cite[Example III.2(f), p.115]{maclane2012sheaves}.
\end{fact}

\begin{definition}[atomic sheaf] \label{app:def:atomic-sheaf-condition}
  Let $C$ be a category for which the atomic topology exists.
  A presheaf $F$ on $C$ is an \emph{atomic sheaf}
  if and only if it satisfies
  the following condition~\cite[Lemma III.4.2]{maclane2012sheaves}:
  for all morphisms $f : d\to c$ in $C$,
  the function $Ff$ is an inclusion $Fc\hookrightarrow Fd$
  whose image is the subset of \nref{app:def:f-inv}{$f$-invariant} elements
  of $Fd$.
  More explicitly: for all morphisms
  $f : d\to c$ and $f$-invariant elements $y\in Fd$,
  there exists a unique $x\in Fc$
  with $y = x\cdot f$.
\end{definition}

\subsection{Continuous group-invariant sets} \label{app:sec:tgsets}

\begin{definition}[category of $G$-sets] \label{app:def:gsets}
  For $G$ a topological group,
  the \emph{category of continuous $G$-sets},
  written $\tgsets G$,
  is the category whose objects are sets $X$ equipped
  with a continuous right action $(\cdot_X) : X\times G\to X$
  (where $X$ is given the discrete topology)
  and whose morphisms from $X$ to $Y$
  are functions $f : X\to Y$
  that are equivariant: $f(x\cdot_X g) = f(x)\cdot_Y g$
  for all $x\in X$ and $g\in G$.
\end{definition}

\begin{lemma} \label{app:lem:left-right-gsets}
  For any topological group $G$, there is an
  isomorphism of categories
  $\tgsets G\cong \tgsets{G\op}$.
\end{lemma}
\begin{proof}
  Whereas objects of $\tgsets G$
  are sets $X$ equipped with a continuous right action
  $(\cdot) : X\times G\to X$,
  objects of $\tgsets{G\op}$
  are sets $X$ equipped with a continuous left action
  $(\cdot\op) : G\times X\to X$.
  Every $(X,\cdot)\in\tgsets G$
  corresponds to $(X, \cdot\op)\in\tgsets{G\op}$
  by setting $g \cdot\op x := x\cdot g^{-1}$.
  (This is indeed a left action, since
  $g\cdot\op h \cdot\op x = x\cdot h^{-1}\cdot g^{-1}
  = x\cdot (gh)^{-1} = gh\cdot\op x$,
  and it is continuous because
  $(\cdot)$ and $(\cdot\op)$
  yield the same stabilizer subgroups.)
  A morphism $f : (X,\cdot_X)\to(Y,\cdot_Y)$ in $\tgsets G$
  is a function $f : X\to Y$
  satisfying $f(x\cdot_X g) = f(x)\cdot_Y g$ for all $g\in G$,
  which is equivalent to $f(x\cdot_X g^{-1}) = f(x)\cdot_Y g^{-1}$
  for all $g\in G$,
  making $f$ also a morphism $(X,\cdot\op_X)\to(Y,\cdot\op_Y)$
  in $\tgsets{G\op}$.
\end{proof}

\subsection{Presheaves with minimal supports}

\begin{definition}[category with minimal supports] \label{app:def:supported-category}
  A category $C$ \emph{has minimal supports} if every coslice
  $c/C$ has a terminal object.
  For any object $c$ of $C$, call the terminal object
  $(c^*, p : c\to c^*)$ of $c/C$ the \emph{support} of $c$.
  An object $c$ has \emph{trivial support}
  or is \emph{trivially-supported}
  if its support is the identity map $1 : c\to c$.
\end{definition}

\begin{lemma} \label{app:lem:supported-category-stability}
  If $C$ has minimal supports
  and $c$ has support $p : c \to c^*$
  and $f : c\to d$,
  then $d$ has support $!_f : d\to c^*$
  where $!_f$ is the unique morphism $f\to p$ in $c/C$.
\end{lemma}
\begin{proof}
Consider the following diagram.
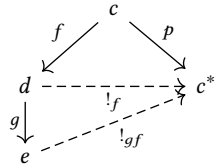
\[\begin{tikzcd}
      & c \arrow[ld, "f"'] \arrow[rd, "p"] &                 \\
  d \arrow[d, "g"'] \arrow[rr, "!_f"', dashed] &                                                                
   & c^* \\
  e \arrow[rru, "!_{gf}"', dashed]                        &                                                                  &                
  \end{tikzcd}\]
The upper triangle depicts the situation given: $d$
has support $p$, and,
because $p$ terminal in $c/C$,
there exists a unique map
$!_f$ making the upper triangle commute in $C$.
To see that the map $!_f$ is the terminal object of $d/C$,
fix an arbitrary object $g$ of $d/C$ as shown.
Because $p$ terminal in $c/C$,
there exists a unique $!_{gf}$ such that the outer quadrilateral 
commutes. The composite $!_{gf}g$ is just as
good as $!_f$ when it comes to making the upper triangle commute, so
uniqueness of $!_f$ implies the lower triangle commutes.
Uniqueness of $!_{gf}$ in making the lower triangle commute
then follows from its uniqueness in making the quadrilateral commute.
\end{proof}

\begin{lemma} \label{app:lem:supported-category-support-id}
  If $C$ has minimal supports and $c$ has support $p : c \to c^*$,
  then $c^*$ has support $1 : c^*\to c^*$.
  Thus every object of $C$ has a map to an object
  with trivial support.
\end{lemma}
\begin{proof}
  Set $f = p$ in \cref{app:lem:supported-category-stability};
  have $!_p = 1_{c^*}$ by uniqueness of $!_p$ in making
  the triangle $!_p p = p$ commute.
\end{proof}

\begin{lemma} \label{app:lem:supported-category-mono}
  If $C$ has minimal supports and $c$ has support $p : c\to c^*$
  and $f : d\to c$, then $d$ has support $pf : d\to c^*$.
  In particular, if $c$ has trivial support
  then any map $f : d\to c$ supports $d$.
\end{lemma}
\begin{proof}
  Consider the following diagram.
  
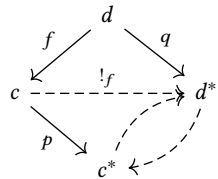
\[\begin{tikzcd}
                                               & d \arrow[ld, "f"'] \arrow[rd, "q"] &                                   \\
  c \arrow[rd, "p"'] \arrow[rr, "!_f", dashed] &                                    & d^* \arrow[ld, dashed, bend left] \\
                                               & c^* \arrow[ru, dashed, bend left]  &                                  
  \end{tikzcd}\]
  The solid arrows depict the situation given: $c$ has support $p$
  and $f:d\to c$, and, since $C$ supported, $d$ has some support $q : d\to d^*$.
  To show $pf$ supports $d$, it suffices to show $pf\cong q$ as objects
  of $d/C$.
  By \cref{app:lem:supported-category-stability}, the unique map $!_f$
  making the upper triangle commute is a support for $c$.
  Thus $!_f$ and $p$ are both supports for $c$;
  thus $!_f\cong p$ as objects of $c/C$,
  and there exists the unlabelled dashed isomorphism making
  the lower triangle commute.
  Since both the upper and lower triangles commute, the whole
  diagram commutes, and
  the dashed isomorphism gives $f\cong q$ in $d/C$ as desired.
\end{proof}

\begin{lemma} \label{app:lem:supported-category-epis-iso}
  If $C$ has minimal supports and contains only epis,
  and $c$ has trivial support,
  then every $f : c\to d$ is an isomorphism.
\end{lemma}
\begin{proof}
  Since $1 : c\to c$ terminal in $c/C$, there is a unique
  $g : d\to c$ with $gf = 1$,
  so $f$ epi and a left inverse, so $f$ iso.
\end{proof}

\begin{definition}[presheaf with minimal supports] \label{app:def:supported-sheaf}
  \labelword{Say}{def:catelts} a presheaf $F$ on $C$ \emph{has minimal supports}
  if its category of elements $\catelts(F)$ is a \nameref{app:def:supported-category}.
  Unwinding definitions, $F$ has minimal supports if for every $x : F c$
  there exists $x^* : F c^*$ and $p :  c\to c^*$
  with $x = x^*\cdot p$
  such that for all
  $x' : F c'$ and $p' :  c\to c'$
  with $x = x'\cdot p$,
  there exists a unique $q :  c'\to c^*$
  with $x' = x^*\cdot q$ 
  and $p = qp'$.
  As a terminal object, $( c^*,x^*,p)$
  is unique up to unique isomorphism,
  which in this case means unique mod
  $( c^*,x^*,p) \sim ( c^\sharp, x^*\cdot i, i^{-1}p)$
  for all isomorphisms $i :  c^\sharp\to  c^*$.
\end{definition}

\begin{intuition}
  For sheaves on a category of measurable spaces,
  a sheaf $F$ is supported if every element $x : F\Omega$
  can be expressed in terms of a ``smallest sample space'' $\Omega^*$.
  This sample space is not necessarily unique---any two-point space
  will do for modelling a boolean random variable, for example---but
  is unique up to unique isomorphism of measurable spaces.
\end{intuition}

\begin{note}
  This notion of support is related
  to the one presented in \citet[Section 4.4.1]{staton2007name}:
  a sheaf with minimal supports as defined here corresponds,
  under the terminology there, to a sheaf for which every element has a least support.
\end{note}

\begin{lemma} \label{app:lem:yo-supps}
  For any category $C$ and object $c$ of $C$,
  the representable presheaf $\yo c$ has \nref{app:def:supported-sheaf}{minimal supports}.
\end{lemma}
\begin{proof}
  Fix $f : C(x,c)$.
  The morphism $(f : C(x,c)) \xrightarrow f (1 : C(c,c))$
  is terminal in $f/\catelts(\yo c)$. To see this, fix arbitrary
  $(f : C(x,c))\xrightarrow p (g : C(y,c))$ with aim to find
  a unique dashed morphism making the following triangle commute
  in $f/\catelts(\yo c)$:
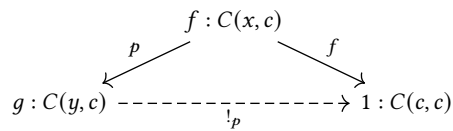
   \[
\begin{tikzcd}
  & f:C(x,c) \arrow[ld, "p"'] \arrow[rd, "f"] &                 \\
g:C(y,c) \arrow[rr, "!_p"', dashed] &                                                     & 1 : C(c,c)
\end{tikzcd}\]
Such a dashed arrow $!_p$ must satisfy $1\circ {!_p} = g$, forcing $!_p = g$.
It only remains to check $!_p$ makes the triangle commute;
this follows from the fact that $p$ is a morphism
from $f$ to $g$ in $\catelts(\yo c)$.
\end{proof}

\begin{definition} \label{app:def:rep-mod-g-general}
  For $U : D\to C$ a functor
  and $G$ an essentially small groupoid in $D$,
  define $\yo{(UG)}$ to be the presheaf $\colim_{c\in G} \yo{(Uc)}$,
  where $\yo$ is the Yoneda embedding. More explicitly, the action of $\yo{(UG)}$
  on objects is
  \begin{align*}
    \yo{(UG)}(c) &= \{\text{morphisms $f : c\to Ud$ for $d$ in $G$}\} / \sim 
  \end{align*}
  where $(f : c\to Ud) \sim (g : c\to Ue)$
  if and only if $U(\pi) f = g$ for some $\pi : d\to e$ in $G$,
  and the action on morphisms is
  \begin{align*}
    \yo(UG)(f : c\to d) &= \left(
      \begin{aligned}
        \yo(UG)(d) &\to \yo(UG)(c) \\
        [g : d\to Ue] &\mapsto [gf : c\to Ue]
      \end{aligned}
    \right).
  \end{align*}
\end{definition}

\begin{definition}[representables modulo a groupoid] \label{app:def:rep-mod-g}
  For $C$ any category and $G$ a small groupoid in $C$,
  let $\yo {G}$ be the \emph{presheaf of representables modulo $G$}.
  This is a specialization of \cref{app:def:rep-mod-g-general}
  to the case $U=1_C$.
\end{definition}

\begin{lemma} \label{app:lem:yo-groupoid-supps}
  For $U:D\to C$
  with $C$ a category of epis and $G$ a groupoid in $D$,
  the presheaf \nref{app:def:rep-mod-g-general}{$\yo(UG)$}
  has \nref{app:def:supported-sheaf}{minimal supports}.
\end{lemma}
\begin{proof}
  The proof is similar to \cref{app:lem:yo-supps}.
  Fix an equivalence class $[f : C(x,Uc)]$.
  The morphism $[f : C(x,Uc)] \xrightarrow f [1 : C(Uc,Uc)]$
  is terminal in $[f]/\catelts(\yo(UG))$. To see this, fix arbitrary
  $p : x\to y$ and
  $[f : C(x,Uc)]\xrightarrow p [g : C(y,Ud)]$ with aim to find
  a unique dashed morphism making the following triangle commute
  in $[f]/\catelts(\yo(UG))$:
   \[
\begin{tikzcd}
  & {[f:C(x,Uc)]} \arrow[ld, "p"'] \arrow[rd, "f"] &                 \\
{[g:C(y,Ud)]} \arrow[rr, "!_p"', dashed] & & {[1 : C(Uc,Uc)]}
\end{tikzcd}\]
Unwinding definitions, this diagram says
$[g]\cdot p = [gp] = [f]$, so $U(\pi) gp = f$
for some $\pi : d\to c$ in $G$.
Commutativity of the above triangle requires solving
$!_p p = f$ for $!_p$. Since $p$ epi and $U(\pi) g p = f$,
the composite $U(\pi) g$ is the only possible solution.
It only remains to check that
setting $!_p := U(\pi) g$ gives
a morphism in $\catelts(\yo(UG))$
from $[g]$ to $[1]$, and indeed
$[1]\cdot {!_p} = [1]\cdot U(\pi) g = [U(\pi) g] = [g]$.
\end{proof}

\subsection{The Day convolution of sheaves with minimal supports}

In this section we describe conditions under which Day convolution preserves atomic sheaves.

\begin{definition}[tensor product of presheaves] \label{app:def:tensor-product-of-presheaves} \label{app:def:tensor-product-of-sheaves}
  Let $(C,\otimes,\monunit)$ be a symmetric monoidal category.
  Given two presheaves $P,Q$ on $C$,
  their \emph{Day convolution}~\citep{day2006closed}
  is defined by the following coend:
  \[ (P\otimes Q) c = \int^{c_P,c_Q\in C} Pc_P \times Qc_Q \times C(c,c_P\otimes c_Q)
  \]
  This product makes the category $\psh(C)$ of presheaves on $C$
  into a symmetric monoidal category,
  with unit $\yo \monunit$ where $\yo$ is the Yoneda embedding.
\end{definition}

The following lemma gives a concrete representation for the
Day convolution of two atomic sheaves in the special case
where the sheaves \nref{app:def:supported-sheaf}{have minimal supports}
and the base category is made only of epis.

\begin{lemma} \label{app:lem:ussheaf-conv}
  Let $(C,\otimes,\monunit)$ be a symmetric monoidal category of epis
  for which the atomic topology exists.
  If $F$ and $G$ are atomic sheaves on $C$
  with \nref{app:def:supported-sheaf}{minimal supports},
  then $iF\otimes iG$
  is the presheaf
  \[ (iF\otimes iG)( c)
    = \left(\begin{aligned}
      &\text{tuples $( c_x, c_y, f :  c\to c_x\otimes  c_y, x : F c_x, y : G c_y)$, abbreviated $(f,x:F c_x,y:G c_y)$,} \\
      &\text{where $x$ and $y$ have \nref{app:def:supported-category}{trivial support},
       mod the equivalence relation} \\
      &\quad (f,x:F c_x,y:G c_y) \sim (g,a:F c_a,b:G c_b) \\
      &\text{iff there exist isos $h :  c_a\to  c_x$ and $k :  c_b\to  c_y$} \\
      &\text{such that $f = (h\otimes k)g$ and $a =x\cdot h$ and $b = y\cdot k$}
    \end{aligned}\right)
    \]
\end{lemma}
\begin{proof}
  The Day convolution $iF\otimes iG$ is a presheaf that sends $ c$ to
  \begin{align*}
    (iF\otimes iG)( c)
    &= \int^{ c_x,  c_y}
    C( c, c_x\otimes  c_y) \times
    F c_x \times G c_y \\
    &= \left(\begin{aligned}
      &\text{tuples $( c_x, c_y, f :  c\to c_x\otimes  c_y, x : F c_x, y : G c_y)$, abbreviated $(f,x:F c_x,y:G c_y)$,} \\
      &\text{mod the equivalence relation generated by} \\
      &\quad ((g\otimes h)\circ f,x:F c_x,y:G c_y)\sim(f,x\cdot g:F c_x',y\cdot h:G c_y') \\
      &\text{for all $f :  c\to c_x'\otimes  c_y'$ and $g :  c_x'\to c_x$ and $h :  c_y'\to c_y$}
    \end{aligned}\right)
  \end{align*}
  Since $F$ and $G$ have minimal supports,
  every element $[f,x:F c_x,y:G c_y]$
  of $(iF\otimes iG)( c)$
  (writing $[-]$ for equivalence class)
  is of the form $[f,x^*\cdot p_x:F c_x^*, y^* \cdot p_y : G c_y^*]$
  for some $(x^*, p_x)$ supporting $x$ and $(y^*, p_y)$ supporting $y$.
  This simplifies things:
  the elements $x^*$ and $y^*$ have support $1_{ c_x^*}$ and $1_{ c_y^*}$
  by \cref{app:lem:supported-category-support-id},
  and $[f,x^*\cdot p_x, y^*\cdot p_y] = [(p_x\otimes p_y)\circ f, x^*,y^*]$,
  so every element of $(iF\otimes iG)( c)$
  is of the form $[f, x : F c_x, y : G c_y]$
  for some $x$ and $y$ whose supports are the identity maps $1_{ c_x}$
  and $1_{ c_y}$ respectively:
  \begin{align*}
    (iF\otimes iG)( c) &= \left(\begin{aligned}
      &\text{trivially-supported tuples $(f, x : F c_x, y : G c_y)$} \\
      &\text{mod the equivalence relation generated by $(\sim)$ as above} \\
    \end{aligned}\right)
  \end{align*}
  Suppose two trivially-supported tuples $(f,x\ofty F c_x,y\ofty G c_y)$
  and $(g,a\ofty F c_a,b\ofty G c_b)$ are
  related by $(\sim)$,
  so there exist $h :  c_a\to  c_x$ and $k :  c_b\to c_y$ such that
  \[ 
    f = (h\otimes k)  g
    \qquad\text{and}\qquad
    a = x\cdot h
    \qquad\text{and}\qquad
    b = y\cdot k.
    \]
  The equation $a = x\cdot h$ corresponds to a
  morphism
  $(a : F c_a)\stackrel h\longrightarrow (x : F c_x)$
  in $\catelts(F)$.
  Now $a$ has support $1_{ c_a}$
  by assumption and $\catelts(F)$ contains only epis
  because $ C$ does by hypothesis,
  so $h$ iso in $\catelts(F)$
  by \cref{app:lem:supported-category-epis-iso}.
  This implies $h$ iso in $ C$.
  Running the same argument on the equation $b = y\cdot k$
  gives $k$ iso in $ C$.
  Thus
  \begin{align} \label{app:def:sheaf-tensor-product-eqv-rel}
    (f,x,y) \sim (g,a,b)
    \iff 
    \text{there exist $h,k$ \textbf{iso} such that $f = (h\otimes k) g$ and $a = x\cdot h$ and $b = y\cdot k$.}
  \end{align}
  This is an equivalence relation on trivially-supported tuples:
  \begin{itemize}
  \item Reflexivity: choose $h = k = 1$.
  \item Symmetry: if $f = (h\otimes k) g$ and $a = x\cdot h$ and $ b = y\cdot k$,
    then $(h^{-1}\otimes k^{-1}) f = g$ and $a\cdot  h^{-1} = x$ and $b\cdot k^{-1} = y$,
    so if $h,k$ witness $(f,x,y) \sim (g,a,b)$ then $h^{-1},k^{-1}$
    witness $(g,a,b)\sim(f,x,y)$.
  \item Transitivity:
    suppose $(f,x,y) \sim_{p,q} (g,a,b) \sim_{r,s} (h,u,v)$,
    where the subscripts on $\sim$ indicate the witnesses for the given relation.
    This gives
    \begin{align*}
      f &= (p\otimes q)  g = (p\otimes q)  (r\otimes s)  h = (pr\otimes qs)h\\
      u &= a\cdot r = x\cdot p \cdot r = x\cdot pr \\
      v &= b\cdot s = y\cdot q \cdot s = y\cdot qs
    \end{align*}
    which together says $(f,x,y) \sim_{pr,qs} (h,u,v)$.
  \end{itemize}
  Thus the quotienting done by the coend in $(iF\otimes iG)( c)$
  is precisely a quotient by $(\sim)$ on trivially supported tuples,
  as claimed:
  \begin{align*}
    (iF\otimes iG)( c) &= 
      \left(\text{trivially-supported tuples $(f, x : F c_x, y : G c_y)$}\right)
      ~/~\!\sim
  \end{align*}
\end{proof}

\begin{definition}[semicartesian monoidal category] \label{app:def:semicartesian-monoidal-category}
  A monoidal category $(C, \otimes, \monunit)$
  is \emph{semicartesian} if $\monunit$ is the terminal object of $C$.
  This implies the existence of projection maps
  $\scmfst  : a\otimes b\to a$
  and $\scmsnd  : a\otimes b\to b$,
  defined by the composites
  \begin{align*}
    \scmfst  &= \left(a\otimes b \xrightarrow{1\otimes !} a\otimes 1 \cong a \otimes \monunit \cong a\right) \\
    \scmsnd  &= \left(a\otimes b \xrightarrow{!\otimes 1} 1\otimes b \cong \monunit \otimes b \cong a\right)
  \end{align*}
  where every occurrence of $!$ denotes the unique morphism into the terminal object.
\end{definition}

\begin{definition}[category of supports] \label{app:def:category-of-supports}
  Call a symmetric \nameref{app:def:semicartesian-monoidal-category} $(C,\otimes,\monunit)$
  a \emph{category of supports} if \begin{itemize}
    \item Every map in $C$ is epi;
    \item The two projection maps $\scmfst ,\scmsnd $ 
      are jointly monic: two maps $f,g : c\to d\otimes e$
      are equal iff $\scmfst  f = \scmfst  g$ and $\scmsnd  f = \scmsnd  g$;
    \item The atomic topology exists for $C$;
    \item For every groupoid $G$ in $C$,
      the \nref{app:def:rep-mod-g}{presheaf of representables modulo $G$}
      is an atomic sheaf.
  \end{itemize}
\end{definition}

\begin{lemma} \label{app:lem:tensor-preserves-ussheaves}
  Let $C$ be a \nameref{app:def:category-of-supports}
  and $F,G$ atomic sheaves on $C$ with \nref{app:def:supported-sheaf}{minimal supports}.
  The Day convolution for presheaves $iF\otimes iG$
  is an atomic sheaf with minimal supports.
\end{lemma}
\begin{proof}
  We use the concrete representation for $iF\otimes iG$
  calculated in \cref{app:lem:ussheaf-conv}.
  \begin{itemize}
    \item $iF\otimes iG$ is a sheaf:
    Fix $p :  c'\to  c$
    and an equivalence class
     $[f, x\ofty F c_x, y \ofty G c_y] : (iF\otimes iG)( c')$
     of trivially-supported tuples
    that is $p$-invariant, so
    for all 
    $q,r :  c''\to c'$
    satisfying $pq = pr$ it holds that
    $[f,x,y]\cdot q = [f,x,y]\cdot r$.
    We are done if we can show that there exists a unique
    extension of $[f,x,y]$ to an element $[g,a,b]$ of $(iF\otimes iG)( c)$
    such that $[g,a,b]\cdot p = [f,x,y]$.
  
    For any $q$ and $r$, we have $[f,x,y] \cdot q= [fq,x,y] = [fr,x,y] = [f,x,y]\cdot r$
    iff $(fq,x,y)\sim (fr,x,y)$, iff there exist isos $h,k$
    with
    \[ \text{$fq = (h\otimes k)fr$ and $x = x\cdot h$ and $y = y\cdot k$}. \]

    Thus $p$-invariance of $[f,x,y]$ says that for all $q,r$
    with $pq=pr$ it holds that $fq = (h\otimes k)fr$
    for some isos $h :  c_x\to c_x$ and $k: c_y\to c_y$.
    This is precisely what it means for $[f]_G$
    to be $p$-invariant as an element of $ C( c',G)$,
    where $G$ is the groupoid of isos of the form $h\otimes k$.
    (Note $h$ and $k$ happen to be automorphisms
    here, but $G$ also contains non-automorphisms.
    This will be relevant later.)
    The presheaf $C(-,G)$ is an atomic sheaf
    because $C$ is a category of supports, so
    there exist
    $ c_x'$ and $ c_y'$
    and $\overline f :  c\to  c_x'\otimes  c_y'$
    with $\overline f p \sim_G f$,
    unique up to $(\sim_G)$.
    Existence of $( c_x', c_y',\overline f)$
    is equivalent to the existence of a map $\overline f :  c\to c_x\otimes c_y$
    with $\overline f p = f$:
    any tuple $(c_x',c_y',\overline f)$ with $\overline f p = (h\otimes k)f$
    for some $h,k$ in $G$ yields a map
    $(h^{-1}\otimes k^{-1})\overline f$
    with $(h^{-1}\otimes k^{-1})\overline f p = f$,
    and conversely if $\overline f p = f$
    then one has a tuple $(c_x,c_y,\overline f)$ with $\overline f p \sim_G f$.

    The equivalence class $[\overline f, x, y]$
    is an element of $(iF\otimes iG)( c)$
    satisfying
      $[\overline f, x, y]\cdot p
      = [f,x,y]$.
    It only remains to show that it is the unique such.
    Suppose $[g,a,b]\cdot p = [gp,a,b] = [f,x,y]$,
    so $(f,x,y) \sim (gp,a,b)$, so
    there exist isos $h,k$ with
    \[ f = (h\otimes k)gp
    \qquad\text{and}\qquad
       a = x\cdot h
       \qquad\text{and}\qquad 
       b = y\cdot k. \]
    Then $(h\otimes k)gp = f = \overline fp$,
    so $(h\otimes k)g = \overline f$ by $p$ epi,
    so
    \[ [g,a,b] = [g,x\cdot h,y\cdot k] = [(h\otimes k)g,x,y] = [\overline f,x,y]\]
    establishing uniqueness of $[\overline f,x,y]$.
  \item $iF\otimes iG$ is has minimal supports:
    fix $[f,x:F c_x,y:G c_y] : (iF\otimes iG)( c)$.
    We have $[f,x,y] = [1_{ c_x\otimes  c_y}, x, y]\cdot f$,
    giving an object
    $[f,x, y] \xrightarrow f [1, x, y]$
    of $[f,x,y]/\catelts(iF\otimes iG)$.
    This object is terminal: fix arbitrary
    $[f,x,y] \xrightarrow p [g,a,b]$
    with aim to find $!_p$ making
    \[
\begin{tikzcd}
  & {[f,x,y]} \arrow[ld, "p"'] \arrow[rd, "f"] &           \\
{[g,a,b]} \arrow[rr, "!_p"', dashed] &                                            & {[1,x,y]}
\end{tikzcd}
      \]
    commute.
    Since $C$ is a category of epis and $p$ a morphism in $C$,
    any dashed morphism completing this triangle must be unique,
    so it only remains to find one such.
    Unpacking the arrow $[f,x,y] \xrightarrow p[g,a,b]$ gives the equations
    \[ f = (h\otimes k)gp \qquad a = x\cdot h \qquad b = y\cdot k \]
    Commutativity of the triangle requires $!_pp = f$,
    which suggests setting $!_p = (h\otimes k)g$.
    It only remains to check that
    $[1,x,y] \,\cdot\, !_p = [g,a,b]$.
    Indeed, $[1,x,y]\,\cdot\, !_p = [(h\otimes k)g, x, y] = [g,x\cdot h,y\cdot k] = [g,a,b]$.
  \end{itemize}
\end{proof}

\begin{lemma} \label{app:lem:tensor-sub-product}
  Let $C$ be a \nameref{app:def:category-of-supports}
  and $F,G$ atomic sheaves on $C$ with \nref{app:def:supported-sheaf}{minimal supports}.
  The natural transformation
  \begin{align*}
    &i : F\shotimes G \hookrightarrow F\times G \\
    &i_ c[f, x\ofty F c_x, y\ofty G c_y] = (x\cdot \scmfst f : F c,\, y\cdot \scmsnd  f : G c)
  \end{align*}
  is a monic map of sheaves, making $F\shotimes G$ a subobject of $F\times G$.
\end{lemma}
\begin{proof}
  Because $F$ and $G$ have minimal supports,
  the sheaf tensor product $F\shotimes G$
  coincides with the presheaf tensor product (\cref{app:lem:tensor-preserves-ussheaves}).
  The map $i$ is defined above on \nref{app:lem:ussheaf-conv}{trivially-supported tuples};
  it respects the \nref{app:lem:ussheaf-conv}{equivalence relation} because
  \begin{align*}
    i_ c[f,x\cdot h,y\cdot k]
    = (x\cdot h\cdot \scmfst  f,\, y\cdot k \cdot \scmsnd  f)
    = (x\cdot \scmfst  (h\otimes k) f, \, y\cdot \scmsnd  (h\otimes k) f]
    = i_ c[(h\otimes k)f,x,y].
  \end{align*}
  It is a natural transformation:
  \begin{align*}
    i_ c[fp,x,y]
    = (x\cdot \scmfst  fp, y\cdot\scmsnd  fp)
    = (x\cdot \scmfst  f, y\cdot \scmsnd  f)\cdot p
    = i_ c[f,x,y] \cdot p
  \end{align*}
  Finally, each component of $i$ is monic.
  Fix arbitrary trivially-supported $[f,x,y]$ and $[g,a,b]$
  and
  suppose $i_ c[f,x,y] = i_ c[g,a,b]$,
  so $x\cdot \scmfst  f = a\cdot \scmfst  g$
  and $y\cdot \scmsnd  f = b\cdot \scmsnd  g$.
  This corresponds to the following diagrams in $\catelts(F)$
  and $\catelts(G)$ respectively:
  \[
\begin{tikzcd}
  & x\cdot \scmfst  f = a\cdot \scmfst  g \arrow[ldd, "\scmfst  f"'] \arrow[rdd, "\scmfst  g"] &                                                    &  &                                                         & y\cdot\scmsnd  f=b\cdot\scmsnd  g \arrow[ldd, "\scmsnd  f"'] \arrow[rdd, "\scmsnd  g"] &                                                    \\
  &                                                                                &                                                    &  &                                                         &                                                                            &                                                    \\
x \arrow[rr, "h^{-1}"', dashed, bend right, shift left] &                                                                                & a \arrow[ll, "h"', dashed, bend right, shift left] &  & y \arrow[rr, "k^{-1}"', dashed, bend right, shift left] &                                                                            & b \arrow[ll, "k"', dashed, bend right, shift left]
\end{tikzcd}
    \]
  The solid arrows depict the situation given.
  Since both $x$ and $a$ have trivial support,
  the common value $x\cdot\scmfst  f = a\cdot\scmfst  g$
  has both $\scmfst  f$ and $\scmfst  g$ as supports
  by \cref{app:lem:supported-category-mono}.
  Any two supports for the same object are isomorphic, so
  $\scmfst  f \cong \scmsnd  g$ in the slice category
  $(x\cdot \scmfst  f) / \catelts(F)$,
  giving $h,h^{-1}$ making the triangle on the left commute.
  Analogously, $\scmsnd  f$ and $\scmsnd g$ are both supports
  for the common value $y\cdot \scmsnd  f = b\cdot \scmsnd  g$,
  giving $k,k^{-1}$ making the triangle on the right commute.
  Unpacking what it means for $h$ and $k$ to be morphisms
  in $\catelts(F)$ and $\catelts(G)$ respectively gives
  $a = x\cdot h$ and $b = y \cdot k$,
  so to get $[f,x,y] = [g,a,b]$ it only remains to show
  $f = (h\otimes k)g$.
  Two maps $ c\to c_x\otimes c_y$ are equal
  iff they are equal when postcomposed with the projections $\scmfst $ and $\scmsnd $,
  and indeed
  \begin{align*}
    \scmfst (h\otimes k)g = h\scmfst  g \stackrel{(*)}= \scmfst  f
    \qquad\text{and}\qquad
    \scmsnd (h\otimes k)g = k\scmsnd  g \stackrel{(*)}= \scmsnd  f
  \end{align*} 
  where the equations marked $(*)$ follow from commutativity
  of the triangles above.
\end{proof}

\subsection{Nominal situations}

\begin{definition}[nominal situation] \label{app:def:nominal-situation}
  A \emph{nominal situation}
  is a tuple $(C, C_\infty, \cinfty : C_\infty, (i_c : c\hookrightarrow \cinfty)_{c:C}, G)$ where
  \begin{itemize}
  \item $C$ is a full subcategory of $C_\infty$
  \item $C$ and $C_\infty$ consist only of monic maps
  \item The atomic topology exists for $C\op$
  \item $G$ is a subgroup of $\Aut(\cinfty)$
  \item (Closure) The special monos $i_c$ hit every subobject
    of the form $\pi i_c$ in
    $\cinfty$ with $\pi\in G$. That is, for every $c$
    and auto $\pi \in G$
    there exists an isomorphism $f : c\xrightarrow\sim c'$ in $C$ with $\pi i_c = i_{c'} f$:
    \[\begin{tikzcd}
    \cinfty \arrow[r, "\pi"]                 & \cinfty                    \\
    c \arrow[u, "i_c", hook] \arrow[r, "f"', dashed] & c' \arrow[u, "i_{c'}"', hook']
    \end{tikzcd}\]
  \item (Homogeneity)
    For every map $f : c\hookrightarrow d$ in $C$
    there exists $\pi\in G$
    with $\pi i_c = i_d f$:
    \[\begin{tikzcd}
    \cinfty \arrow[r, "\pi", dashed]                 & \cinfty                    \\
    c \arrow[u, "i_c", hook] \arrow[r, "f"', hook] & d \arrow[u, "i_d"', hook']
    \end{tikzcd}\]
  \item (Correspondence)
    The map
    \[ \Fix i := \{ \pi \mid \pi i = i \}\subseteq G \]
    that sends every mono $i : c\hookrightarrow\cinfty$
    to the subgroup of $G$ fixing it
    gives an contravariant equivalence between subobjects of 
    $c_\infty$ and subgroups of $G$ fixing those subobjects.
    (This mapping is automatically faithful
    because its domain --- the subobjects of $\cinfty$ --- is
    a thin category.
    The nontrivial part is the requirement that $\Fix$ be full,
    which is to say that if $\Fix i\subseteq\Fix j$
    then $j$ factors through $i$, so that the triangle
    \[ \begin{tikzcd}
                                                    & \cinfty &                               \\
    \dom j \arrow[ru, "j", hook] \arrow[rr, dashed] &         & \dom i \arrow[lu, "i"', hook]
    \end{tikzcd}\]
    commutes.)
  \item (Cofinality)
      For every finite family of objects $(c_j)_{j\in J}$ in $C$
      there exists an object $c^*$ in $C$
      with $\Fix i_{c^*}\subseteq\bigcap_j \Fix i_{c_j}$.
  \end{itemize}
\end{definition}

\begin{definition} \label{app:def:nominal-refinement}
  In a \nameref{app:def:nominal-situation}
  $(C,C_\infty, \cinfty,i_\bullet,G)$,
  say an automorphism $\pi \in G$
  \emph{refines} a map $f : c\to d$
  if the following square commutes:
    \[\begin{tikzcd}
    \cinfty \arrow[r, "\pi"]                 & \cinfty                    \\
    c \arrow[u, "i_c", hook] \arrow[r, "f"', hook] & d \arrow[u, "i_d"', hook']
    \end{tikzcd}\]
  In this language, Homogeneity says every map is refined by some automorphism.
\end{definition}

\begin{lemma} \label{app:lem:refinement-respects-comp}
  \nref{app:def:nominal-refinement}{Refinement}
  respects composition:
  if $\pi_f$ refines $f : c\to d$
  and $\pi_g$ refines $g : d\to e$,
  then $\pi_g\pi_f$ refines $gf$,
  and if $\pi$ refines an iso $f$ then $\pi^{-1}$
  refines $f^{-1}$.
\end{lemma}
\begin{proof}
  First, if $\pi_f$ refines $f$ and $\pi_g$ refines $g$
  then pasting the respective commutative squares
  together gives the following commutative rectangle witnessing
  $\pi_g\pi_f$ refines $gf$:
\[  \begin{tikzcd}
  \cinfty \arrow[r, "\pi_f"]               & \cinfty \arrow[r, "\pi_g"]               & \cinfty                  \\
  c \arrow[r, "f"'] \arrow[u, "i_c", hook] & d \arrow[u, "i_d", hook] \arrow[r, "g"'] & e \arrow[u, "i_e", hook]
  \end{tikzcd}\]
  Second, if $\pi$ refines $f : c\to d$ iso then $\pi i_c = i_d f$,
  so $i_c f^{-1} = \pi^{-1} i_d$, so $\pi^{-1}$ refines $f^{-1}$.
\end{proof}

\newcommand\refineset{{\nref{app:def:refinement-topology}{\calR}}}

\begin{definition}[refinement topology] \label{app:def:refinement-topology}
  Let $(C,C_\infty, \cinfty,i_\bullet,G)$
  be a \nameref{app:def:nominal-situation}.
  For any iso $f : c\xrightarrow\sim d$
  in $C$,
  let $\refineset_{f}$ be the
  collection of all automorphisms in $G$
  that \nref{app:def:nominal-refinement}{refine} $f$:
  \[ \refineset_f := \{ \pi : \cinfty\xrightarrow\sim\cinfty \mid \pi i_c = i_d f \} \]
  The \emph{refinement topology}
  is the topology on $G$
  consisting of unions of finite intersections of sets $\refineset_f$ for isos $f$ of $C$.
\end{definition}

\begin{lemma} \label{app:lem:aut-continuous}
  In a \nameref{app:def:nominal-situation}
  $(C,C_\infty, \cinfty,i_\bullet,G)$,
  the group $G$ is continuous for the \nameref{app:def:refinement-topology}.
\end{lemma}
\begin{proof}
  We show that the inverse and multiplication maps are continuous.
  \begin{itemize}
  \item Inverse is continuous:
    for all isos $f$ of $C$,
    \[\mathrm{inverse}^{-1}(\refineset_f)
    = \{ \pi \mid \pi^{-1} \in \refineset_f \}
    = \{ \pi \mid \pi^{-1} \text{ refines } f \}
    \nref{app:lem:refinement-respects-comp}{=} \{ \pi \mid \pi \text{ refines } f^{-1} \}
    = \{ \pi \mid \pi \in \refineset_{f^{-1}} \}
    = \refineset_{f^{-1}}
    \]
    so the preimage of every basis element is open.
  \item Multiplication is continuous:
    fix an iso $f$ of $C$ and
    suppose $(\sigma,\pi) \in \textrm{mul}^{-1}(\refineset_f)$
    with aim to find an open neighborhood $\calO_\sigma \times \calO_\pi$
    such that $(\sigma,\pi)\in\calO_\sigma\times\calO_\pi\subseteq\textrm{mul}^{-1}(\refineset_f)$.
    By assumption $\sigma\pi$ refines $f$, so
    \[\begin{tikzcd}
    \cinfty \arrow[r, "\pi"]                  & \cinfty \arrow[r, "\sigma"] & \cinfty                  \\
    c \arrow[u, "i_c", hook] \arrow[rr, "f"'] &                             & d \arrow[u, "i_d", hook]
    \end{tikzcd}\]
    commutes. By Closure,
    the subobject $\pi i_c$ 
    is isomorphic to $i_{c'}$ for some $c'$.
    That is, there exists an object $c'$ and isomorphism $g : c\to c'$
    such that
\[ \begin{tikzcd}
  \cinfty \arrow[r, "\pi"]                 & \cinfty                      \\
  c \arrow[u, "i_c", hook] \arrow[r, "g"'] & c' \arrow[u, "i_{c'}", hook]
  \end{tikzcd}\]
  commutes. Combining this square with the rectangle above gives
  \[
\begin{tikzcd}
  \cinfty \arrow[r, "\pi"]                                   & \cinfty \arrow[r, "\sigma"]                                                         & \cinfty                  \\
  c \arrow[u, "i_c", hook] \arrow[rr, "f"'] \arrow[rd, "g"'] &                                                                                     & d \arrow[u, "i_d", hook] \\
                                                             & c' \arrow[uu, "i_{c'}" description, hook, bend left] \arrow[ru, "fg^{-1}"', dashed] &                         
  \end{tikzcd}
  \]
  where the back rectangle and the left quadrilateral commute.
  Since $g$ iso, the dashed arrow exists and makes the lower triangle commute.
  By diagram chase, the right quadrilateral commutes when precomposed
  with $g$; since $g$ iso, this implies the right quadrilateral commutes.
  This gives the following commutative rectangle:
  \[\begin{tikzcd}
  \cinfty \arrow[r, "\pi"]                 & \cinfty \arrow[r, "\sigma"]                        & \cinfty            \\
  c \arrow[u, "i_c", hook] \arrow[r, "g"'] & c' \arrow[u, "i_{c'}", hook] \arrow[r, "fg^{-1}"'] & d \arrow[u, "i_d"]
  \end{tikzcd}\]
  Translating this rectangle into words, we have that $\pi$ refines $g$
  and $\sigma$ refines $fg^{-1}$
  and and $\sigma\pi$ refines $f$.
  But note that for any other $\pi'$ and $\sigma'$
  we would still have $\sigma'\pi'$ refining $f$
  so long as $\pi'$ refines $g$ and $\sigma'$ refines $fg^{-1}$.
  In other words,
   \[ (\sigma,\pi) \in \refineset_{fg^{-1}}\times \refineset_{g}\subseteq 
      \mathrm{mul}^{-1}(\refineset_f)  \]
  and we have found a suitable open neighborhood as required.
  \end{itemize}
\end{proof}

\begin{definition} \label{app:def:fix-subgroup-category}
  In a nominal situation 
  $(C,C_\infty, \cinfty,i_\bullet,G)$,
  let $\subgrpcat_C$
  be the category whose objects are
  subgroups $\Fix i_c$ for all $c$ in $C$
  and whose morphisms $\Fix i_c\to \Fix i_d$
  are cosets $(\Fix i_d)\pi$
  such that $\pi g\pi^{-1}$ fixes $i_d$
  for all $g\in G$ that fix $i_c$,
  with composition
  $\Fix{i_c}
  \xrightarrow{(\Fix{i_d})\pi} \Fix{i_d}
  \xrightarrow {(\Fix {i_e})\sigma} \Fix {i_e}$
  given by $(\Fix {i_e})\sigma\pi$.
\end{definition}

\begin{lemma} \label{app:lem:fix-subgroup-equiv}
   For any \nameref{app:def:nominal-situation} $(C,C_\infty,\cinfty,i_\bullet,G)$,
   there is an equivalence of categories
   $C^{\rm op} \simeq \nref{app:def:fix-subgroup-category}{\subgrpcat_C}$.
\end{lemma}
\begin{proof}
  We will construct a functor $F : C^{\rm op}\to \subgrpcat_C$
  and show that it is full, faithful, and surjective on objects.
  \begin{itemize}
  \item Send $c$ in $C$ to $\Fix i_c$ in $\subgrpcat_C$.
  \item Send $f : d\to c$ in $C$ to $(\Fix i_d)\pi^{-1} : \Fix{i_c}\to \Fix{i_d}$ in $\subgrpcat_C$,
    where $\pi$ is an automorphism
    such that $\pi i_d = i_c f$,
    guaranteed to exist by Homogeneity.
    This automorphism is indeed a map $\Fix{i_c}\to \Fix{i_d}$,
    because if $g$ fixes $i_c$
    then
    \[ \pi^{-1} g \pi i_d = \pi^{-1} g i_c f = \pi^{-1} i_c f 
    = \pi^{-1} \pi i_d  = i_d\]
    so $\pi^{-1} g \pi$ fixes $i_d$.
    The choice of $\pi$ does not matter:
    for any other $\overline\pi$ with $\overline\pi i_d = i_c f$,
    it holds that $\pi^{-1}\overline\pi i_d = \pi^{-1} i_c f = i_d$,
    so $\pi^{-1}\overline\pi\in\Fix i_d$,
    so $(\Fix i_d)\pi^{-1} = (\Fix i_d)\overline\pi^{-1}$.
  \item This assignment is functorial.
    The identity $1_c : c\to c$ is sent to
    a coset $(\Fix {i_c})\pi$ for some $\pi i_c = i_c 1_c$.
    But this means $\pi$ fixes $i_c$, so $(\Fix{i_c})\pi = \Fix{i_c}$.
    Given $f : e\to d$ and $g : d\to c$,
    we have $F(f) = (\Fix {i_e})\pi_f^{-1}$
    and $F(g) = (\Fix {i_d}) \pi_g^{-1}$
    and $F(gf) = (\Fix{i_e}) \pi_{fg}^{-1}$
    with $\pi_f$, $\pi_g$, and $\pi_{fg}$ fitting into the following commutative rectangles:
    \[\begin{tikzcd}
    \cinfty \arrow[r, "\pi_f"]               & \cinfty \arrow[r, "\pi_g"]               & \cinfty                  \\
    e \arrow[r, "f"'] \arrow[u, "i_e", hook] & d \arrow[r, "g"'] \arrow[u, "i_d", hook] & c \arrow[u, "i_c", hook]
    \end{tikzcd}
    \qquad\qquad
\begin{tikzcd}
    \cinfty \arrow[rr, "\pi_{fg}"]           &                   & \cinfty                  \\
    e \arrow[r, "f"'] \arrow[u, "i_c", hook] & d \arrow[r, "g"'] & c \arrow[u, "i_e", hook]
    \end{tikzcd}\]
    We need $F(gf) = (\Fix{i_e})\pi_{fg}^{-1} = (\Fix{i_e})\pi_f^{-1}\pi_g^{-1} = F(f)F(g)$.
    Since in general two cosets $gH,hH$ are equal iff $gh^{-1}\in H$,
    this amounts to showing $(\pi_f^{-1}\pi_g^{-1})(\pi_{fg}^{-1})^{-1}\in \Fix{i_e}$.
    The commutativity of the above rectangles implies
    \[ 
       (\pi_f^{-1}\pi_g^{-1})(\pi_{fg}^{-1})^{-1} i_e
       = \pi_f^{-1}\pi_g^{-1}\pi_{fg} i_e
       = \pi_f^{-1}\pi_g^{-1} i_c g f
       = \pi_f^{-1}\pi_g^{-1} \pi_g\pi_f i_e
       = i_e
        \]
    so $(\pi_f^{-1}\pi_g^{-1})(\pi_{fg}^{-1})^{-1}$ fixes $i_e$ as required.
  \item Full:
    let $(\Fix{i_d})\pi^{-1}$ be a morphism
    $\Fix{i_c}\to \Fix{i_d}$ in $\subgrpcat_C$,
    the goal being to find $f : d\to c$ in $C$
    such that the following square commutes:
    \[
    \begin{tikzcd}
    \cinfty \arrow[r, "\pi"]                 & \cinfty                    \\
    d \arrow[u, "i_d", hook] \arrow[r, "f"', dashed] & c \arrow[u, "i_c"', hook']
    \end{tikzcd}\]
    Such an $f$ exists iff the mono $\pi i_d$ factors through $i_c$.
    By Correspondence we just need to show $\Fix{i_c}\subseteq\Fix{\pi i_d}$.
    This follows from the fact that $(\Fix i_d)\pi^{-1}$ is a morphism
    in $\subgrpcat_C$:
    if $\sigma$ fixes $i_c$ then
    $\pi^{-1} \sigma \pi$ fixes $i_d$,
    so $\pi^{-1}\sigma\pi i_d = i_d$,
    so $\sigma \pi i_d = \pi i_d$,
    so $\sigma$ fixes $\pi i_d$.
    Since $\sigma$ was arbitrary we have
    $\Fix{i_c}\subseteq \Fix{(\pi i_d)}$ as required.
  \item Faithful:
    suppose $f,g : d\to c$ and $F(f) = F(g) : \Fix{i_c}\to \Fix{i_d}$.
    By definition $F(f) = (\Fix{i_d})\pi_f^{-1}$
    and $F(g) = (\Fix{i_d})\pi_g^{-1}$
    for some $\pi_f,\pi_g$ fitting into the following commutative squares:
    \[\begin{tikzcd}
    \cinfty \arrow[r, "\pi_f"]               & \cinfty                  \\
    d \arrow[r, "f"'] \arrow[u, "i_d", hook] & c \arrow[u, "i_c", hook]
    \end{tikzcd}
    \qquad\qquad 
    \begin{tikzcd}
    \cinfty \arrow[r, "\pi_g"]               & \cinfty                  \\
    d \arrow[r, "g"'] \arrow[u, "i_d", hook] & c \arrow[u, "i_c", hook]
    \end{tikzcd}
    \]
    The assumption $F(f) = F(g)$ implies $\pi_f^{-1}(\pi_g^{-1})^{-1}\in \Fix{i_d}$,
    iff $\pi_f^{-1}\pi_g\in\Fix{i_d}$,
    iff $\pi_f^{-1}\pi_g i_d = i_d$, iff $\pi_g i_d = \pi_f i_d$.
    Thus the bottom-left-to-top-right routes of the two squares above
    are equal. Commutativity of these squares implies
    $i_cf = i_c g$, so $f=g$ because $i_c$ mono.
  \item Surjective on objects:
    every object
    $\Fix{i_c}$ of $\subgrpcat_C$
    is equal to $F(c)$, so in the image of $F$.
  \end{itemize}
\end{proof}

\begin{lemma} \label{app:lem:fix-is-cofinal}
   Let $(C,C_\infty,\cinfty,i_\bullet,G)$
   be a \nameref{app:def:nominal-situation}.
  The set $\mathcal U := \{\Fix{i_c} \mid c \in \mathrm{Ob}(C)\}$
  is cofinal in the open subgroups of 
  the \nref{app:lem:aut-continuous}{topological group} $G$,
  in the sense that
  any open subgroup $H$ of $G$
  contains $\Fix{i_c}$ for some $c$.
\end{lemma}
\begin{proof}
  Every open subgroup $H$ contains the identity $1_\cinfty$,
  so contains an open neighborhood around $1_\cinfty$
  of the form $\bigcap_j \refineset_{f_j}$
  for $(f_j : c_j\to d_j)_{j\in J}$ a finite set of isos in $C$.
  Unwinding the definition of $\refineset$
  in $1_\cinfty\in\refineset_{f_j}$ shows
  $i_{d_j} f_j = i_{c_j}$ for all $j\in J$.
  Thus $\pi\in\refineset_{f_j}$
  iff $\pi i_{c_j} = i_{d_j} f_j = i_{c_j}$
  iff $\pi\in\Fix i_{c_j}$
  for all $\pi\in\Aut(\cinfty)$ and $j\in J$,
  so $\refineset_{f_j} = \Fix i_{c_j}$ for all $j\in J$,
  and by Cofinality there exists $c^*$
  with
  \[ H \supseteq \bigcap_j \refineset_{f_j} = \bigcap_j \Fix i_{c_j} \supseteq \Fix i_{c^*}\ni \calU \]
  as needed.
\end{proof}

\begin{theorem} \label{app:thm:nominal-situation}
  For any \nameref{app:def:nominal-situation} $(C,C_\infty,\cinfty,i_\bullet,G)$,
  there is an equivalence of categories
  \[
    \Shjat(C^{\rm op})
    \simeq
    \tgsets{G}
   \]
   where $G$ is given the \nameref{app:def:refinement-topology}.
   Across this equivalence,
   an atomic sheaf $F$ on $C\op$
   corresponds to an $G$-set $\tilde F$
   whose carrier is
   $\colim_{c\in P\op} Fc$ where $P$ is the
   preorder of subobjects of $\cinfty$ of the form $i_c$ for $c\in C$.
   Elements of $\tilde F$
   are equivalence classes $[c,x:Fc]$
   where $[c,x] = [c',x']$
   iff there exists $c^*$ with $c\sqsubseteq c^*\sqsupseteq c'$
   where $\sqsubseteq$ is the preorder $P$
   such that $x\cdot j = x'\cdot j'$,
   where $j$ is the unique $C$-morphism witnessing $c\sqsubseteq c^*$
   via the equation $i_{c^*} j = i_c$
   and $j'$ is the unique $C$-morphism witnessing $c'\sqsubseteq c^*$
   via the equation $i_{c^*} j' = i_{c'}$.
   The action of $G$
   on these equivalence classes
   is given by $[c, x:Fc]\cdot \pi = [c', F(f)(x):Fc']$
   where $(c',f:c\to c')$ is an arbitrary iso
   $\pi$ refines, guaranteed to exist by Closure.
\end{theorem}
\begin{proof}
  $\subgrpcat_C \simeq C^{\rm op}$ by \cref{app:lem:fix-subgroup-equiv}
  and the subgroups $\Fix i_c$ are cofinal in the open subgroups of $G$ (\cref{app:lem:fix-is-cofinal}),
  so \citet[Theorem III.9.2]{maclane2012sheaves} gives
  the desired equivalence.
  Rifling through the proof of \citet[Theorem III.9.2]{maclane2012sheaves}, we find that 
  it constructs for each $F\in\Shjat(\subgrpcat_C)$
  the $G$-set
  with carrier $\colim_{\Fix i_c\in\calU} F(\Fix i_c)$
  where the colimit is over $\calU$
  ordered by subgroup inclusion,
  and the diagram the colimit is taken over sends every such
  inclusion $\Fix i_c\subseteq \Fix i_d$
  to the $\subgrpcat_C$-morphism
  $\Fix{i_c}\xrightarrow{(\Fix{i_d})1_{\cinfty}}\Fix{i_d}$ (\citet[p. 153]{maclane2012sheaves}).
  By Correspondence every such morphism corresponds to
  the canonical morphism $i_d\to i_c$
  in $C_\infty/\cinfty$ witnessing
  the ordering relation $d\subseteq c$,
  so transporting
  $\colim_{\Fix i_c\in\calU} F(\Fix i_c)$
  across the equivalence $\subgrpcat_C\simeq C\op$  gives the colimit
  in the statement.
  Further rifling through the proof of \citet[Theorem III.9.2]{maclane2012sheaves}
  reveals that the action of a $\cinfty$-auto $\pi$ 
  sends an equivalence class $[\Fix{i_c},x : F(\Fix{i_c})]$
  to $[\pi^{-1}(\Fix{i_c})\pi,\, F(\overline\pi)(x):F(\pi^{-1}(\Fix{i_c})\pi)]$
  where $\overline\pi$ is the $\subgrpcat_C$-morphism
  $\pi^{-1}(\Fix{i_c})\pi \xrightarrow{(\Fix{i_c})\pi} \Fix{i_c}$ (\citet[p. 153]{maclane2012sheaves}).
  Pick an arbitrary $c'$ and iso $f : c\to c'$ that $\pi^{-1}$ refines, guaranteed to exist by Closure.
  No matter the choice of $c',f$, we have 
  $\pi^{-1}(\Fix{i_c})\pi = \Fix{i_{c'}}$,
  since for all $\sigma$ fixing $i_c$ it holds that
  $\pi^{-1}\sigma\pi i_{c'} 
  = \pi^{-1}\sigma i_c f^{-1} = \pi^{-1} i_c f^{-1} 
  = i_{c'} ff^{-1} = i_{c'}$
  so $\pi^{-1}\sigma \pi$ fixes $i_{c'}$,
  and conversely for all $\sigma$ fixing $i_{c'}$ it holds that
  $\pi\sigma\pi^{-1}i_c 
  = \pi\sigma i_{c'}f
  = \pi i_{c'}f = i_c f^{-1} f = i_c$
  so $\pi\sigma\pi^{-1}$ fixes $i_c$.
  Thus the $\subgrpcat_C$-morphism $\overline\pi$
  has domain $\Fix{i_{c'}}$ and codomain $\Fix{i_c}$,
  and the action of a $\cinfty$-auto $\pi$
  sends $[\Fix i_c, x:F(\Fix i_c)]$
  to $[\Fix i_{c'}, F(\overline\pi)(x) : F(\Fix i_{c'})]$
  for $(c', f : c'\to c)$ refining $\pi^{-1}$.
  Transporting this across
  the equivalence $\subgrpcat_C\simeq C\op$
  gives the action in the statement:
  the subgroups $\Fix i_c,\Fix i_{c'}$
  correspond to the objects $c,c'$, and
  the morphism $\overline\pi$
  given by the coset 
  $(\Fix i_c)\pi$
  is the image of $f$ under the functor 
  $C\op\to\subgrpcat_C$ constructed in the proof of 
  \cref{app:lem:fix-subgroup-equiv}.
\end{proof}

\section{\Goodsheavesnolink} \label{app:sec:goodsheaves}

\subsection{Probabilistic concepts as sheaves}

\begin{definition} \label{app:def:goodsheaf}
  An \emph{\goodsheaf{}}
  is an object of the category $\Shjat(\StdMble)$.
\end{definition}

\begin{lemma} \label{app:lem:stdmble-pfwd}
  If $(X,\calF,\calN)$
  is a standard enhanced measurable space
  and $(Y,\calG)$
  a measurable space
  arising from a Polish space
  and $f : (X,\calF)\to(Y,\calG)$
  a measurable map,
  then
  the set $\calM := \{G\in \calG \mid f^{-1}(G)\in \calN\}$
  makes $(Y,\calG,\calM)$ a standard enhanced measurable space
  and $f$ a morphism in $\StdMble$.
\end{lemma}
\begin{proof}
  The set $\calM$ is a $\sigma$-ideal in $\calG$
  because $\calN$ is a $\sigma$-ideal in $\calF$
  and taking preimages preserves all $\sigma$-algebra operations,
  so $(Y,\calG,\calM)$ is an enhanced measurable space.
  The map $f$ preserves and reflects negligibles by construction,
  since $M\in \calM$ iff $f^{-1}(M)\in \calN$.
  All that's left is to show $(Y,\calG,\calM)$ is standard.
  This follows from $Y$ Polish, by unforgetting
  a standard probability measure on $(X,\calF,\calN)$
  and pushing it forward onto $Y$ through $f$~\cite[Section 2.7, p.24]{rohlin1949fundamental}.
\end{proof}

\begin{lemma} \label{app:lem:rv-mble-sheaf}
  For any measurable space $(A,\calG)$
  arising from a Polish space, the random variable presheaf
  \begin{align*}
    &\RV_A(\Omega,\calF,\calN) = \{ \text{measurable maps }(\Omega,\calF)\to (A,\calG) \}/\aseq
    \text{where $X\aseq Y$ iff $\{\omega \mid X\omega\ne Y\omega\}\in\calN$}
    \\
    &\RV_A(p : \Omega' \to \Omega)([X] : \RV_A(\Omega)) : \RV_A(\Omega') = [X\circ p]
  \end{align*}
  is an atomic sheaf.
\end{lemma}
\begin{proof}
  Given two random variables $X,Y$ let $\neset X Y := \{\omega \mid X\omega\ne Y\omega\}$
  be the event they disagree.
  The action of $\RV_A$ on morphisms is well-defined because $p$
  negligible-reflecting: if $\neset{X}{Y}$ negligible in $\Omega$
  then $\neset{X\circ p}{Y\circ p} = p^{-1}(\neset X{Y})$ negligible in $\Omega'$.
  Unwinding the definition of $\RV_A$ and making the quotienting
  of random variables up to almost-everywhere explicit,
  the presheaf $\RV_A$ is an atomic sheaf
  if and only if
  \[
    \forall p~[X].~\underbrace{(\forall q~r.~ p\circ q = p\circ r \implies X\circ q \aseq X\circ r)}_{\text{$X$ is $p$-invariant}}
    \implies \exists! [\overline X].~ \overline X p \aseq X
  \]
  The following diagram gives the types of all variables involved:
\[\begin{tikzcd}
  \Omega'' \arrow[r, "r"', shift right=1] \arrow[r, "q", shift left=1] & \Omega' \arrow[r, "p"] \arrow[d, "X"] & \Omega \arrow[ld, "\overline X"] \\
                                                                       & A                                     &                                 
  \end{tikzcd}\]
  Suppose $[X]$ is $p$-invariant.
  By \cref{app:lem:stdmble-pfwd},
  pushing the negligibles of $\Omega'$ forward along $X$ 
  makes $A$ into an standard enhanced measurable space and $X$ a morphism of
  standard enhanced measurable spaces.
  By \cref{app:lem:mblecat-dual},
  this gives a diagram
\[\begin{tikzcd}
    \mathrm{alg}(\Omega')                     & \mathrm{alg}(\Omega) \arrow[l, "p^{-1}", hook] \\
    \mathrm{alg}(A) \arrow[u, "X^{-1}", hook] &                                               
    \end{tikzcd}\]
  in $\StdMbleAlg$.
  There is an inclusion of standard measurable algebras $\img(X^{-1})\subseteq \img(p^{-1})$:
  otherwise, if there were some $E$ in $\img(X^{-1})$
  not in $\img(p^{-1})$,
  then by \cref{app:lem:target-practice} and \cref{app:lem:mblecat-dual}
  there would exist $\StdMble$-maps $q,r : \Omega''\to\Omega'$
  with $pq = pr$ but $q^{-1}(E) \ne r^{-1}(E)$,
  so $X\circ q \not\aseq X\circ r$,
  contradicting the assumption that $X$ is $p$-invariant.
  The inclusion $\img(X^{-1})\subseteq \img(p^{-1})$ gives
  a corresponding (injective) homomorphism $\operatorname{alg}(A)\hookrightarrow\operatorname{alg}(\Omega)$
  making the above diagram of homomorphisms into a commutative triangle.
  By \cref{app:lem:mblecat-dual} again, 
  such a homomorphism arises from a measurable map
  $\overline X : \Omega\to A$,
  and commutativity of the triangle implies
  $\overline X p\aseq X$.
  Finally, for any $\overline Y : \Omega\to A$
  with $\overline Y p \aseq X$,
  the string of equations $\overline X p\aseq X \aseq \overline Y p$
  implies
  $\neset{\overline X p}{\overline Y p}$
  negligible in $\Omega'$;
  since $\neset{\overline X p}{\overline Y p} = p^{-1}(\neset{\overline X}{\overline Y})$
  and $p$ negligible preserving,
  this implies $\neset{\overline Y}{\overline X}$ negligible in $\Omega$,
  which is to say $\overline Y\aseq \overline X$,
  establishing uniqueness of $\overline X$.
\end{proof}

\begin{lemma} \label{app:lem:mblecat-quot-sheaf}
  Let $U : C\to\StdMble$ be a functor and $G$ a groupoid in $C$.
  The presheaf \nref{app:def:rep-mod-g-general}{$\yo{(UG)}$},
  more explicitly given by
  \begin{align*}
    &\StdMble(\Omega,UG) = \left(
      \begin{aligned}
        &\{\text{maps $\Omega\to UA$ for some $A\in G$}\}/\sim \\
        &\text{where $(f:\Omega\to UA)\sim (g:\Omega\to UB)$
           iff $f = U(\pi) g$ for some $\pi : B\to A$ in $G$}
      \end{aligned}\right) \\
    &[f]\cdot p = [fp]
  \end{align*}
  is an atomic sheaf on $\StdMble$.
\end{lemma}
\begin{proof}
  The proof is analogous to the proof of \cref{app:lem:rv-mble-sheaf}.
  Suppose $p : \Omega'\to\Omega$
  and $[f] : \StdMble(\Omega',UG)$ is $p$-invariant for some $A\in G$ and $f : \Omega'\to UA$.
  The goal is to find $\overline f : \StdMble(\Omega',UG)$
  with $\overline f p \sim f$, which is to equivalent to
  finding $A'$ and $\overline f : \Omega'\to UA'$ and $\pi : A\to A'$ in $G$ with $\overline f p = U(\pi) f$,
  which is equivalent (by left-multiplying both sides by $U(\pi^{-1})$)
  to finding $\overline f : \Omega'\to UA$ with $\overline f p = f$.
  As in the proof of \cref{app:lem:rv-mble-sheaf},
  such an $\overline f$ exists
  if there is an inclusion of measurable algebras $\img(f^{-1})\subseteq \img(p^{-1})$,
  so it suffices to establish this inclusion;
  uniqueness of $\overline f$ then follows from $p$ epi.
  Rephrasing the assumption that $f$ is $p$-invariant in terms of
  measurable algebras, we have
  \[ 
    \forall q~r.~ q p^{-1} = r p^{-1} \implies \exists \pi\in G.~ q f^{-1} = r f^{-1} U(\pi)^{-1}
  \]
  with types of the variables involved depicted by the following diagram:
  \[
\begin{tikzcd}
  \measalg(\Omega'') & \measalg(\Omega') \arrow[l, "q"', hook, shift right] \arrow[l, "r", hook, shift left] & \measalg(\Omega) \arrow[l, "p^{-1}", hook] \\
                     & \measalg(A) \arrow[u, "f^{-1}", hook]                                                 &                                           
  \end{tikzcd}
    \]
  Note that if $qf^{-1} = rf^{-1}U(\pi)^{-1}$ for some $\pi$
  then $qf^{-1}$ and $rf^{-1}$ have the same image in $\measalg(\Omega'')$,
  so $p$-invariance of $f$ implies
  \begin{align} \label{app:eqn:yoneda-quot-prop-to-be-contradicted}
    \forall q~r.~ q p^{-1} = r p^{-1} \implies \img (q f^{-1}) = \img (r f^{-1}).
  \end{align}
  We are now ready to establish the inclusion $\img(f^{-1})\subseteq \img(p^{-1})$.
  Suppose for contradiction that there exists $E$ in $\measalg(A)$ with $f^{-1}(E)$ not in $\img(p^{-1})$.
  By \cref{app:lem:target-practice},
  there exists a standard measurable algebra $\Omega''$
  and homomorphisms $q,r$ as depicted above such that $qp^{-1} = rp^{-1}$
  but $qf^{-1}E \notin \img r$,
  contradicting~(\ref{app:eqn:yoneda-quot-prop-to-be-contradicted}).
\end{proof}

\begin{lemma} \label{app:lem:mblecat-subcanonical}
  For any standard probability space $\Omega$,
  the representable functor $\StdMble(-,\Omega)$
  is an atomic sheaf with \nref{app:def:supported-sheaf}{minimal supports}.
\end{lemma}
\begin{proof}
  Sheafhood follows from \cref{app:lem:mblecat-quot-sheaf}
  by setting $U = 1$ and $G$ to the trivial groupoid $\{1_\Omega\}$;
  the minimal support property follows similarly from \cref{app:lem:yo-groupoid-supps}.
\end{proof}

\begin{definition}[sheaf of probability spaces] \label{app:def:mble-pspcs}
  The \emph{sheaf of probability spaces} $\pspcs$ is
  \nref{app:def:rep-mod-g-general}{$\yo (\forgetmu G)$} with
  $\forgetmu$ the forgetful functor $\StdProb\to\StdMble$
  and $G$ the maximal subgroupoid of $\StdProb$.
  Concretely, the action of $\pspcs$ on objects is
  \begin{align*}
    \pspcs(\Omega,\calF,\calN) &= \left(\begin{aligned}
      &\text{pairs $((A,\calG,\mu),X)$ with $(A,\calG,\mu)\in\StdProb$ 
      and $X$ a $\StdMble$-map
      $(\Omega,\calF,\calN)\to (A,\calG,\negligibles(\mu))$} \\
      &\text{mod $\sim$, where $((A,\calG,\mu),X)\sim((A',\calG',\mu'),X')$ iff
        exists $\StdProb$-iso $i : (A,\calG,\mu)\to (A',\calG',\mu')$
        with $X' = U(i)\, X$}
    \end{aligned}\right)
  \end{align*}
  and the action on morphisms is by precomposition:
  \begin{align*}
    \pspcs(f : \Omega'\to\Omega) &= \left(\begin{aligned}
      \pspcs\Omega &\to \pspcs\Omega' \\
      [(A,\calG,\mu),X] &\mapsto [(A,\calG,\mu),Xf] 
    \end{aligned}\right)
  \end{align*}
  This presheaf is an atomic sheaf on $\StdMble$ by
  \cref{app:lem:mblecat-quot-sheaf},
  and has \nref{app:def:supported-sheaf}{minimal supports} by \cref{app:lem:yo-groupoid-supps}.
\end{definition}

\subsection{Separation as Day convolution} \label{app:sec:dayconv}

\begin{lemma} \label{app:lem:stdmble-supports}
  The category $\StdMble$ is a \nameref{app:def:category-of-supports}.
\end{lemma}
\begin{proof}
  The category $\StdMble$ is symmetric \nref{app:def:semicartesian-monoidal-category}{semicartesian monoidal} by
  \cref{app:lem:stdmble-semicartesian},
  contains only epis by \cref{app:cor:stdmble-epis},
  and has the atomic topology by \cref{app:lem:mble-right-ore} 
  and \cref{app:fact:right-ore-atomic-topology}.
  Finally, setting $U=1_\StdMble$ in \cref{app:lem:mblecat-quot-sheaf} shows representables mod $G$
  are atomic sheaves for all groupoids $G$.
\end{proof}

\begin{lemma} \label{app:lem:mble-pspcs-tensor-pspcs}
  The tensor product of $\nref{app:def:mble-pspcs}\pspcs$ with itself in sheaves can be computed as if in presheaves:
  \[ i(\pspcs\nref{app:def:tensor-product-of-sheaves}{\shotimes} \pspcs) \cong i\pspcs\otimes i\pspcs\]
\end{lemma}
\begin{proof}
  $\StdMble$ is a category of supports (\cref{app:lem:stdmble-supports})
  and $\pspcs$ has minimal supports,
  so \cref{app:lem:tensor-preserves-ussheaves} applies.
\end{proof}

\begin{lemma} \label{app:lem:pspcs-day-conv-char}
  The tensor product $\pspcs\nref{app:def:mble-pspcs}\shotimes\pspcs$
  is a sheaf of ``independent probability spaces'', with action on objects
  \[ (\pspcs\shotimes\pspcs)(\Omega)
     = \left(\begin{aligned}
      &\text{pairs $([A,X]:\pspcs\Omega,[B,Y]:\pspcs\Omega)$ with $X:\Omega\to \forgetmu A$ and $Y:\Omega\to \forgetmu B$} \\
      &\text{that factor through a tensor product; i.e.,
        there exist $\Omega_1,\Omega_2$ and $f : \Omega\to\Omega_1\otimes\Omega_2$
        and $X' : \Omega_1\to \forgetmu A$ and $Y':\Omega_2\to \forgetmu B$} \\
      &\text{with $X = X'\,\scmfst \,f$
      and $Y = Y'\,\scmsnd \,f$}
     \end{aligned}\right)\]
  and action on morphisms
  \[ (\pspcs\shotimes\pspcs)(f: \Omega'\to\Omega)[(A,X),(B,Y)]
     = [(A,Xf),(B,Yf)].
  \]
\end{lemma}
\begin{proof}
  The tensor product $\pspcs\shotimes\pspcs$
  is the image of the inclusion $\pspcs\shotimes\pspcs\hookrightarrow\pspcs\times\pspcs$
  defined by \cref{app:lem:tensor-sub-product}.
\end{proof}

\begin{definition} \label{app:def:pspcs-tensor-join}
  There is a map $\pcmjoin : \pspcs\shotimes\pspcs\to\pspcs$
  defined by
  \begin{align*}
    \pcmjoin_\Omega = \left(\begin{aligned}
      \pspcs\Omega \shotimes \pspcs\Omega &\to \pspcs\Omega \\
      ([A,X], [B,Y]) &\mapsto [A\otimes B, (X,Y)]\text{ where $(X,Y)(\omega) = (X\omega,Y\omega)$}
    \end{aligned}\right)
  \end{align*}
  where we have used the representation calculated
  in \cref{app:lem:pspcs-day-conv-char} to define $\pcmjoin$
  on pairs of independent probability spaces.
  The assumption that $[A,X]$ and $[B,Y]$ 
  factor through some tensor product is needed in order for the map
  $(X,Y)$
  to be well-formed: otherwise it may fail to be negligible-reflecting
  as a map of enhanced measurable spaces $\Omega\to \forgetmu(A\otimes B)$.
  (For example, if $\Omega = [0,1]$ and $X = Y = 1_{[0,1]}$,
  the image of the map $(X,Y)$ is the diagonal in the unit square,
  a negligible set with nonnegligible preimage.)
\end{definition}

\begin{note} \label{app:def:pspcs-tensor-join-note}
  Another way to arrive at the map $(X,Y)$
  in \cref{app:def:pspcs-tensor-join}
  is by unforgetting
  the tensor product that $[A,X]$ and $[B,Y]$
  factor through, giving a commutative diagram
  
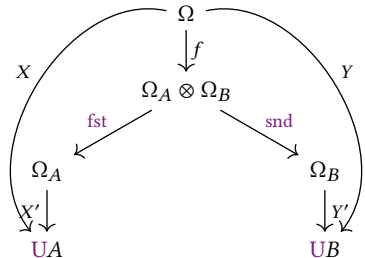
\[\begin{tikzcd}
                            & \Omega \arrow[d, "f"] \arrow[lddd, "X"', bend right=71] \arrow[rddd, "Y", bend left=71] &                          \\
                            & \Omega_A\otimes\Omega_B \arrow[ld, "\scmfst "'] \arrow[rd, "\scmsnd "]          &                          \\
  \Omega_A \arrow[d, "X'"'] &                                                                                         & \Omega_B \arrow[d, "Y'"] \\
  \forgetmu A                         &                                                                                         & \forgetmu B                       
  \end{tikzcd}\]
  and then setting $(X,Y)$
  to be the composite $(X'\otimes Y') f$.
  (Note $\forgetmu(A\otimes B) = \forgetmu A\otimes \forgetmu B$ by definition,
  so this typechecks.)
  \cref{app:lem:tensor-sub-product}, used to prove
  \cref{app:lem:pspcs-day-conv-char}, shows
  this construction does not depend on the choice of $f,\Omega_A,\Omega_B,X',Y'$.
\end{note}

\begin{definition} \label{app:def:pspcs-join-unit}
  There is a map $\pcmunit : 1\to\pspcs$
  defined by
  \begin{align*}
    \pcmunit_\Omega = \left(\begin{aligned}
      1 &\to \pspcs\Omega \\
      \_ &\mapsto [1, ! : \Omega\to 1]
    \end{aligned}\right).
  \end{align*}
\end{definition}

\begin{lemma} \label{app:def:pspcs-join-unit-cm}
  The tuple $(\pspcs, \pcmjoin,\pcmunit)$ is
  a commutative monoid internal to the symmetric monoidal
  category $(\Shjat(\StdMble), \shotimes, 1)$,
  which is to say that the following equations hold in the
  internal linearly-typed language of this category:
  \begin{itemize}
    \item (Unit) $p \ofty \pspcs \vdash \pcmjoin(\pcmunit, p) = p : \pspcs$
    \item (Commutativity) $p\ofty\pspcs,q\ofty\pspcs\vdash\pcmjoin(p, q) = \pcmjoin(q, p) : \pspcs$
    \item (Associativity) $p\ofty\pspcs,q\ofty\pspcs,r\ofty\pspcs\vdash\pcmjoin(p, \pcmjoin(q, r))
      = \pcmjoin(\pcmjoin(p, q), r) : \pspcs$
  \end{itemize}
\end{lemma}
\begin{proof}
  Each of the equations holds because
  the corresponding property holds of $\shotimes$:
  \begin{itemize}
    \item (Unit)
      The goal is to show
      $\left(\pspcs \xrightarrow{\sim} 1\shotimes\pspcs \xrightarrow{\pcmunit\otimes 1} \pspcs\shotimes\pspcs
      \xrightarrow{\pcmjoin} \pspcs\right)
      = 1_\pspcs$.
      At stage $\Omega$ and given an element $[A,X]:\pspcs\Omega$, 
      the left side is
      $\pcmjoin_\Omega([1,!_\Omega], [A,X]) = [1\otimes A, (!_\Omega,X)]$
      and the right side is $[A, X]$.
      These are the same equivalence classes via the isomorphism $1\otimes A\cong A$.
    \item (Commutativity)
      The goal is to show
      $\left(\pspcs\shotimes\pspcs \xrightarrow{\pcmjoin} \pspcs \right)
      =\left(\pspcs\shotimes\pspcs \xrightarrow{s} \pspcs\shotimes\pspcs \xrightarrow{\pcmjoin} \pspcs \right)$
      where $s$ witnesses symmetry of the monoidal product $\shotimes$.
      At stage $\Omega$ and given an element of $(\pspcs\shotimes\pspcs)(\Omega)$,
      which by  \cref{app:lem:pspcs-day-conv-char} amounts to 
      a pair of elements $([A,X],[B,Y]):\pspcs\Omega\times\pspcs\Omega$
      that factor through a tensor product,
      the left side is
      $\pcmjoin_\Omega([A,X], [B,Y]) = [A\otimes B, (X,Y)]$
      and the right side is
      $\pcmjoin_\Omega([B,Y], [A,X]) = [B\otimes A, (Y,X)]$.
      These two equivalence classes are equal
      via the isomorphism $A\otimes B\cong B\otimes A$.
    \item (Associativity)
      The goal is to show
      \[\left(\pspcs\shotimes(\pspcs\shotimes\pspcs) \xrightarrow{1\otimes \pcmjoin}
        \pspcs\shotimes\pspcs \xrightarrow\pcmjoin \pspcs \right)
      =\left(\pspcs\shotimes(\pspcs\shotimes\pspcs) \xrightarrow{\sim}
      (\pspcs\shotimes\pspcs)\shotimes\pspcs\xrightarrow{\pcmjoin\otimes 1}
      \pspcs\shotimes\pspcs \xrightarrow{\pcmjoin} \pspcs \right).\]
      At stage $\Omega$ and given an element of $(\pspcs\shotimes(\pspcs\shotimes\pspcs))(\Omega)$,
      which by two applications of \cref{app:lem:tensor-sub-product} amounts to 
      a tuple of elements $([A,X],([B,Y],[C,Z])):\pspcs\Omega\times(\pspcs\Omega\times\pspcs\Omega)$
      that factor through tensor products in the proper way,
      the left side is
      \[\pcmjoin_\Omega([A,X], \pcmjoin_\Omega([B,Y], [C,Z]))=[A\otimes (B\otimes C), (X,(Y,Z))],\]
      and the right side is
      \[\pcmjoin_\Omega(\pcmjoin_\Omega([A,X],[B,Y]),[C,Z])=[(A\otimes B)\otimes C, ((X,Y),Z)].\]
      These two equivalence classes are equal
      via the isomorphism $A\otimes (B\otimes C)\cong (A\otimes B)\otimes C$.
  \end{itemize}
\end{proof}

\begin{lemma} \label{app:def:pspcs-join-unit-pcm}
  The tuple $(\pspcs, \pcmjoin,\pcmunit)$ is
  a partial commutative monoid internal
  to the symmetric monoidal category $(\Shjat(\StdMble),\times,1)$,
  in the sense that
  \begin{itemize}
  \item Unit: for all $p : F\to \pspcs$
    the map $(\pcmunit!,p)$ factors through $i$
    and $\pcmjoin(\pcmunit!,p) = p$.
  \item Commutativity: If a map $(p,q) : F\to \pspcs\times\pspcs$
    factors through the inclusion $i:\pspcs\shotimes\pspcs\hookrightarrow\pspcs\times\pspcs$,
    there exists a unique map $f: F\to \pspcs\shotimes\pspcs$
    with $if = (p,q)$.
    Abusing notation and writing $f$ as $(p,q)$, 
    relying on types to disambiguate,
    commutativity holds if
    whenever $(p,q)$ factors through $i$
    then so does $(q,p)$,
    and
     $\pcmjoin(p,q) = \pcmjoin(q,p)$.
  \item Associativity:
    for all $p,q,r : F\to \pspcs$
    such that $(p,q)$ factors through $i$
    and $(\pcmjoin(p,q),r)$ factors through $i$,
    it holds that $(q,r)$ factors through $i$
    and $(p,\pcmjoin(q,r))$ factors through $i$
    and $\pcmjoin(\pcmjoin(p,q),r)=\pcmjoin(p,\pcmjoin(q,r))$.
  \end{itemize}
\end{lemma}
\begin{proof}
  A map $(p,q) : F\to\pspcs\times\pspcs$
  factors through $i$ iff
  for any $x : F\Omega$
  the pair
  $(p(x), q(x)) : \pspcs\Omega\times\pspcs\Omega$
  factors through a tensor product.
  \begin{itemize}
    \item (Unit) If $p(x) = [A,X]$,
      the pair $([1,!], [A,X])$
      always factors through the tensor product $1\otimes \forgetmu A$,
      so $(\pcmunit!,p)$ always factors through $i$.
      That $\pcmunit$ is a unit for $\pcmjoin$ follows from
       \cref{app:def:pspcs-join-unit-cm}.
    \item (Commutativity)
      If $(p(x),q(x))$ factors through a
      tensor product $\Omega_p\otimes\Omega_q$, then $(q(x),p(x))$
      factors through the tensor product $\Omega_q\otimes\Omega_p$.
      Thus if $(p,q)$ factors through $i$ then  so does $(q,p)$.
      Commutativity of $\pcmjoin$ follows from
       \cref{app:def:pspcs-join-unit-cm}.
    \item (Associativity)
      Fix arbitrary $x$ and write $p(x) = [A,X]$
      and $q(x) = [B,Y]$ and $r(x) = [C,Z]$.
      Suppose $([A,X],[B,Y])$
      factors through
      a tensor product $\Omega_p\otimes\Omega_q$
      and $(\pcmjoin([A,X],[B,Y]), [C,Z]) = ([A\otimes B, (X,Y)], [C,Z])$
      factors through a tensor product $\Omega_{pq}\otimes\Omega_r$.
      The situation is illustrated by the following commutative diagram:
      
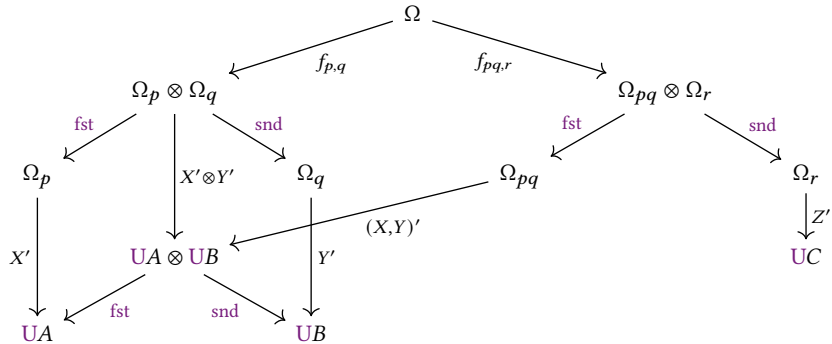
\[
\begin{tikzcd}
  &                                                                                                           &                           & \Omega \arrow[lld, "{f_{p,q}}"] \arrow[rrd, "{f_{pq,r}}"'] &                                      &                                                                                   &                          \\
  & \Omega_p\otimes\Omega_q \arrow[ld, "\scmfst "'] \arrow[rd, "\scmsnd "] \arrow[dd, "X'\otimes Y'"] &                           &                                                            &                                      & \Omega_{pq}\otimes\Omega_r \arrow[ld, "\scmfst "'] \arrow[rd, "\scmsnd "] &                          \\
\Omega_p \arrow[dd, "X'"'] &                                                                                                           & \Omega_q \arrow[dd, "Y'"] &                                                            & \Omega_{pq} \arrow[llld, "{(X,Y)'}"] &                                                                                   & \Omega_r \arrow[d, "Z'"] \\
  & \forgetmu A\otimes \forgetmu B \arrow[ld, "\scmfst "] \arrow[rd, "\scmsnd "']                                         &                           &                                                            &                                      &                                                                                   & \forgetmu C                        \\
\forgetmu A                          &                                                                                                           & \forgetmu B                         &                                                            &                                      &                                                                                   &                         
\end{tikzcd}\]
      The root-to-leaf paths from $\Omega$ to $\forgetmu A,\forgetmu B,\forgetmu C$
      are the random variables $X,Y,Z$ respectively.
      The two paths $\Omega\to \forgetmu A\otimes \forgetmu B$ represent the random variable $(X,Y)$:
      the path through $\Omega_p\otimes\Omega_q$
      is the one constructed by $\pcmjoin([A,X],[B,Y])$, as described in 
      \cref{app:def:pspcs-tensor-join-note},
      and the path through $\Omega_{pq}$
      witnesses the fact that $([A\otimes B,(X,Y)],[C,Z])$
      factors through the tensor product $\Omega_{pq}\otimes \Omega_r$.

      With this diagram, $([B,Y],[C,Z])$
      visibly factors through $\Omega_{pq}\otimes \Omega_r$.
      The trickier case is to show
      $([A,X],\pcmjoin([B,Y],[C,Z]))$ factors through a tensor product.
      Contract the pentagon in the middle of the above diagram and
      the root-to-leaf path from $\Omega$ to $\forgetmu C$ to get
      
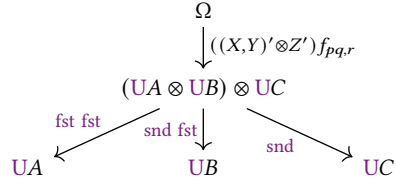
\[
\begin{tikzcd}
  & \Omega \arrow[d, "{((X,Y)'\otimes Z')f_{pq,r}}"]                                                                                         &   \\
  & (\forgetmu A\otimes \forgetmu B)\otimes \forgetmu C \arrow[ld, "{\scmfst \,\scmfst }"'] \arrow[d, "{\scmsnd \,\scmfst }"'] \arrow[rd, "\scmsnd "'] &   \\
\forgetmu A & \forgetmu B                                                                                                                                        & \forgetmu C
\end{tikzcd}
      \]
      Since all we have done is to contract and delete nodes in the previous diagram,
      the root-to-leaf paths still denote the random variables $X,Y,Z$.
      The expression $\pcmjoin([B,Y],[C,Z])=[B\otimes C,(Y,Z)]$ can be constructed following \cref{app:def:pspcs-tensor-join-note},
      giving
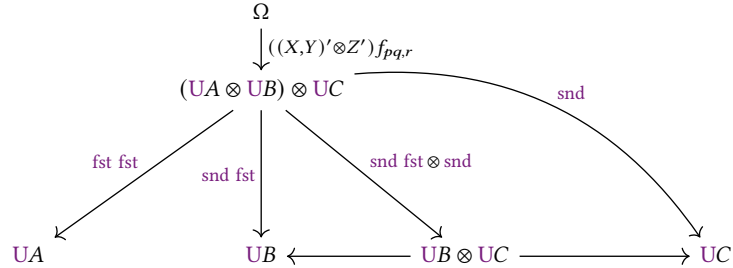
      \[
\begin{tikzcd}
  &  & \Omega \arrow[d, "{((X,Y)'\otimes Z')f_{pq,r}}"]                                                                                                                                                                             &                                  &  &  &   \\
  &  & (\forgetmu A\otimes \forgetmu B)\otimes \forgetmu C \arrow[llddd, "{\scmfst \,\scmfst }"'] \arrow[ddd, "{\scmsnd \,\scmfst }"'] \arrow[rrrrddd, "\scmsnd ", bend left] \arrow[rddd, "{\scmsnd \,\scmfst \otimes\scmsnd }"] &                                  &  &  &   \\
  &  &                                                                                                                                                                                                                              &                                  &  &  &   \\
  &  &                                                                                                                                                                                                                              &                                  &  &  &   \\
\forgetmu A &  & \forgetmu B                                                                                                                                                                                                                            & \forgetmu B\otimes \forgetmu C \arrow[l] \arrow[rrr] &  &  & \forgetmu C
\end{tikzcd}
      \]
      where the path $\Omega\to \forgetmu B\otimes \forgetmu C$
      is $(Y,Z)$. Tidying and applying the isomorphism
      $(\forgetmu A\otimes \forgetmu B)\otimes \forgetmu C\cong \forgetmu A\otimes (\forgetmu B\otimes \forgetmu C)$
      gives
      \[
\begin{tikzcd}
  & \Omega \arrow[d, "{a((X,Y)'\otimes Z')f_{pq,r}}"]                            &            \\
  & \forgetmu A\otimes (\forgetmu B\otimes \forgetmu C) \arrow[ld, "\scmfst "'] \arrow[rd, "\scmsnd "] &            \\
\forgetmu A &                                                                              & \forgetmu B\otimes \forgetmu C
\end{tikzcd}
      \]
      Since the path $\Omega\to \forgetmu A$ is $X$
      and the path $\Omega\to \forgetmu B\otimes \forgetmu C$ is $(Y,Z)$,
      this diagram shows $(X,(Y,Z))$
      factors through a tensor product as needed.
      Finally, the associativity equation follows from
       \cref{app:def:pspcs-join-unit-cm}.
  \end{itemize}
\end{proof}

\begin{fact}
  Let $\Prop$ denote the subobject classifier in $\Shjat(\StdMble)$,
  which is the constant sheaf $\Prop(\Omega) = \{\vtrue,\vfalse\}$.
\end{fact}

\begin{definition}[ordering on probability spaces] \label{app:def:mble-pspcs-ordering}
  Let $(\sqsubseteq) : \pspcs\times\pspcs\to \Prop$
  be the map defined by
  \begin{align*}
    [A,X] \sqsubseteq_\Omega [B,Y]
    \iff
    \text{there exists a morphism $f : B\to A$ such that $fY = X$.}
  \end{align*}
  This respects the equivalence classes
  $[A,X]$ and $[B,Y]$ because if $fY = X$
  and $[A,X] = [\overline A, \overline X]$
  and $[B,Y] = [\overline B,\overline Y]$,
  so $\overline X = iX$ and $\overline Y = jY$
  for isos $i:A\to\overline A$ and $j:B\to\overline B$,
  then setting $\overline f$ to $ifj^{-1}$
  gives $ifj^{-1}\overline Y = ifj^{-1}jY = ifY = iX = \overline X$.
  And it is natural in $\Omega$ because if $p : \Omega'\to\Omega$
  and $fY = X$ witnesses $[A,X]\sqsubseteq_\Omega[B,Y]$
  then $fYp = Xp$ witnesses $[A,Xp]\sqsubseteq_{\Omega'}[B,Yp]$
  (in words, the ordering relation is invariant under extensions
  of the sample space).
\end{definition}

\begin{theorem} \label{app:thm:mble-pspcs-krm}
  The tuple $(\pspcs, \pcmjoin,\pcmunit,\krmorder)$ is
  a partially defined monoid~\cite[\S 5.3]{galmiche2005semantics}
  internal to $\Shjat(\StdMble)$; in other words,
  $(\pspcs,\pcmjoin,\pcmunit)$
  forms an internal PCM and the following monotonicity
  condition holds:
  if $p \krmorder_\Omega p'$
  and $q\krmorder_\Omega q'$
  and $(p',q')$ factors through a tensor product,
  then $(p,q)$ does too
  and $\pcmjoin_\Omega(p,q)\krmorder_\Omega\pcmjoin_\Omega(p',q')$.
\end{theorem}
\begin{proof}
  The tuple $(\pspcs,\pcmjoin,\pcmunit)$
  is a PCM by \cref{app:def:pspcs-join-unit-pcm}.
  For monotonicity, suppose
  $[A,X]\krmorder[A',X']$
  and
  $[B,Y]\krmorder[B',Y']$
  and $(X',Y')$ factor through a tensor product.
  Unwinding the definition of $(\krmorder)$,
  there exist $f : A'\to A$ and $g : B'\to B$
  with $X = fX'$ and $Y = gY'$.
  Now consider the following commutative diagram:
  
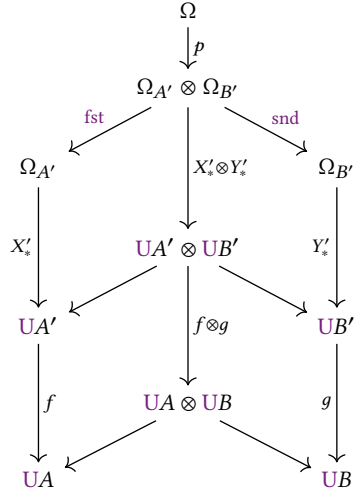
\[
\begin{tikzcd}
  & \Omega \arrow[d, "p"]                                                                                               &                                 \\
  & \Omega_{A'}\otimes\Omega_{B'} \arrow[ld, "\scmfst "'] \arrow[rd, "\scmsnd "] \arrow[dd, "X'_*\otimes Y'_*"] &                                 \\
\Omega_{A'} \arrow[dd, "X'_*"'] &                                                                                                                     & \Omega_{B'} \arrow[dd, "Y'_*"'] \\
  & \forgetmu A'\otimes \forgetmu B' \arrow[ld] \arrow[rd] \arrow[dd, "f\otimes g"]                                                         &                                 \\
\forgetmu{A'} \arrow[dd, "f"]   &                                                                                                                     & \forgetmu{B'} \arrow[dd, "g"']  \\
  & \forgetmu A\otimes\forgetmu B \arrow[ld] \arrow[rd]                                                                 &                                 \\
\forgetmu A                     &                                                                                                                     & \forgetmu B                    
\end{tikzcd}
\]
  Since $(X',Y')$
  factors through a tensor product,
  there exist $\Omega_{A'},\Omega_{B'},p,X'_*,Y'_*$ with
  $X' = X'_*\,\scmfst \,p$ and $Y' = Y'_*\,\scmsnd \,p$ as shown.
  All paths from $\Omega$ to $\forgetmu{A'}$ are equal to $X'$.
  Analogously, all paths from $\Omega$ to $\forgetmu{B'}$ are equal to $Y'$.
  Since $X = fX'$ and $Y = gY'$, the root-to-leaf paths from $\Omega$
  to $\forgetmu A$ and $\Omega$ to $\forgetmu B$
  are $X$ and $Y$ respectively.
  The diagram shows $X$ and $Y$ factor through $\Omega_{A'}\otimes\Omega_{B'}$,
  so $\pcmjoin_\Omega([A,X],[B,Y])$ is defined.
  It only remains to show
  $\pcmjoin_\Omega([A,X],[B,Y])\sqsubseteq_\Omega
  \pcmjoin_\Omega([A',X'],[B',Y'])$.
  Since $\pcmjoin_\Omega([A,X],[B,Y]) = [A\otimes B, (X,Y)]$
  and $\pcmjoin_\Omega([A',X'],[B',Y']) = [A'\otimes B', (X', Y')]$,
  it's enough to show $(X,Y)$ factors through $(X',Y')$.
  This is visible in the diagram:
  the map $(X,Y)$ is the path from $\Omega$
  to $\forgetmu A\otimes\forgetmu B$,
  the map $(X',Y')$ is the path from $\Omega$
  to $\forgetmu{A'}\otimes\forgetmu{B'}$,
  and the diagram shows
  the factorization $(X,Y) = (f\otimes g)(X',Y')$.
\end{proof}

\section{\Agoodnamenolink{s}} \label{app:sec:agoodnames}

\subsection{A nominal situation for standard enhanced measurable spaces}

\newcommand\micat{{\nref{app:def:micat}{\hcube\StdMble}}}
\newcommand\rami{{\nref{app:def:rami}{\hcube\mathbf{EMS}_{\rm std}^{\rm{ff}}}}}

\begin{definition} \label{app:def:hcube}
  The Hilbert cube $\hcube$ is the standard enhanced measurable space $([0,1]^\omega,\calF,\calN)$
  of infinite sequences in the interval $[0,1]$.
  The $\sigma$-algebra $\calF$
  and negligibles $\calN$ are those obtained in constructing the usual Lebesgue measure $\lambda$
  on $\hcube$, given by extending the function
  \[ \lambda([a_1,b_1]\times\dots\times[a_n,b_n] \times [0,1]^\omega) = 
     (b_1 - a_1)\times\dots\times(b_n - a_n),
    \]
  defined on finite-dimensional boxes in $\hcube$,
  to all Borel sets of $\hcube$ and then taking the completion.
  As with the interval $[0,1]$, we will write $\hcube$
  for both the standard enhanced measurable space and its associated
  standard measurable algebra, relying on context to disambiguate.
\end{definition}

Informally speaking, our equivalence result
validates the idea that a Hilbert cube's worth of
randomness is enough so long as one
embeds every measurable space
needed into $\hcube$ in a way that leaves
enough room for new randomness.
To set up this result,
we first introduce some auxiliary categories
to help track the particular way in which
a standard enhanced measurable space
can be embedded into $\hcube$;
these will be helpful later, when
it comes to finding simple descriptions
of the \agoodname{s} corresponding to
the \goodsheaves{} introduced in \cref{app:sec:goodsheaves}.

\renewcommand\hcubeproj[1]{{\nref{app:def:projn}{\mathrm{proj}_{1..#1}}}}
\renewcommand\embedmap[1]{{#1\nref{app:def:embedmap}{\otimes1_\hcube}}}

\begin{definition} \label{app:def:projn}
  For $n\in \N$ let $\hcubeproj n$
  be the projection $\hcube\to[0,1]^n$
  defined by $\hcubeproj n (x_1,\dots,x_n,\dots) = (x_1,\dots,x_n)$.
\end{definition}

\begin{definition} \label{app:def:mblespcramap}
  A $\StdMble$-map 
  $f : \hcube\to Y$
  has \emph{finite footprint}
  if $f$ factors through $\hcubeproj n$ for some $n$.
\end{definition}

\begin{definition} \label{app:def:micat}
  Let $\micat$ be the category 
  whose objects are
  $\StdMble$-maps
  $\hcube \xrightarrow p X$
  for some standard enhanced measurable space $X$,
  and whose morphisms
  from $\hcube \xrightarrow p X$
  to $\hcube\xrightarrow qY$
  are $\StdMble$-maps
  $f : X\to Y$.
\end{definition}

\begin{definition} \label{app:def:rami}
  Let $\rami$ be the full subcategory
  of $\micat$ spanned by
  objects $\hcube\xrightarrow p X$
  with \nref{app:def:mblespcramap}{finite footprint}.
\end{definition}

Note no commutativity conditions are imposed on
$\micat$-morphisms or $\rami$-morphisms $f$ in relation to their domains $p$
and codomains $q$. This may make $p,q$ seem
superfluous --- \cref{app:lem:raii-eq} below
establishes $\micat\simeq\StdMble$ --- but having
them around allows associating to each space $X$
a particular way in which it sits inside of $\hcube$;
this makes calculations easier later on.

\begin{lemma} \label{app:lem:raii-eq}
  There are equivalences $\StdMble\simeq\micat\simeq\rami$.
\end{lemma}
\begin{proof}
  Fix an arbitrary standard enhanced measurable space $X$.
  Since $X$ arises from a standard probability space,
  it is isomorphic to a coproduct $Y$ of countably many atoms
  and an interval $[0,p]$. For every such coproduct $Y$
  there is at least one map $f : [0,1]\to Y$
  given by assigning a half-open subinterval
  to each atom and the remainder to $[0,p]$,
  so the composite $\left(\hcube\xrightarrow{\hcubeproj1}[0,1]\xrightarrow fY\xrightarrow\sim X\right)$
  is a $\StdMble$-map $\hcube\to X$ with finite footprint.
  Thus for every standard enhanced measurable space $X$,
  there is at least one finite-footprint $\StdMble$-map
  of the form $\hcube\to X$.
  This establishes surjectivity-on-objects
  of the forgetful functor $\rami\to\StdMble$
  sending $\hcube\xrightarrow pX$
  to $X$; it follows that this functor witnesses $\StdMble\simeq\rami$.
  The equivalence 
  $\StdMble\simeq\rami$ holds similarly.
\end{proof}

\newcommand\caninc{\mathrm{inc}}

\begin{lemma} \label{app:lem:hcube-measalg-iso}
  For all $n$,
  there is an isomorphism of standard enhanced measurable spaces
  $[0,1]^n\cong[0,1]$,
  and hence also of their associated standard measurable algebras.
\end{lemma}
\begin{proof}
  Both $[0,1]^n$ and $[0,1]$ arise from
  forgetting the measure on an
  atomless standard probability space.
\end{proof}

\begin{lemma} \label{app:lem:proj-relatomless}
  For all $n < m$ it holds that
  $[0,1]^m$ is relatively atomless over $\img(\measalg(p))$,
  where $p : [0,1]^m\to[0,1]^n$
  is the projection $p(x_1,\dots,x_n,\dots) = (x_1,\dots,x_n)$.
\end{lemma}
\begin{proof}
  Every subalgebra $\img(\measalg(p))$
  has enough room: one can always allocate
  independent standard probability algebras in dimensions
  $n+1$ and above,
  so the claim follows from \cref{app:thm:enough-room-char}.
\end{proof}

\begin{definition} \label{app:def:embedmap}
  For $n\in \N$
  and a $\StdMble$-automorphism $\pi : [0,1]^n\xrightarrow\sim [0,1]^n$,
  let $\embedmap \pi : \hcube\to\hcube$
  be the automorphism
  defined by $(\embedmap \pi)(x_0,\dots,x_{n-1},x_n,\dots)
  =(y_0,\dots,y_{n-1},x_n,\dots)$
  where $\pi(x_0,\dots,x_{n-1}) = (y_0,\dots,y_{n-1})$.
\end{definition}

\begin{definition} \label{app:def:autoff}
  A $\StdMble$-automorphism
  $\pi : \hcube\to \hcube$
  has \emph{width $n$}
  if $\pi = \embedmap{\pi'}$
  for some automorphism $\pi' : [0,1]^n\to[0,1]^n$.
\end{definition}

\begin{definition} \label{app:def:auti}
  Let $\auti$ be the subgroup of $\Aut_\StdMble\hcube$
  consisting of the \nref{app:def:autoff}{finite-width}
  automorphisms of $\hcube$,
  topologized via the \nameref{app:def:refinement-topology}
  with respect to $\rami$-isos:
  for all $\StdMble$-maps $\hcube\xrightarrow pX$
  and $\hcube\xrightarrow qY$ with finite footprint
  and $\StdMble$-isos $f : p\stackrel\sim\to q$,
  the set $\refineset_f := \{\pi : \auti \mid fp = q\pi\}$
  is a basic open.
\end{definition}

\begin{lemma} \label{app:lem:mble-nominal-situation}
  Let $\caninc$ be the $\rami$-indexed family
  $\caninc(\hcube\xrightarrow p X) := p$,
  sending each object $\hcube\xrightarrow pX$ of $\rami$
  to the $\StdMble$-map $p$, now considered as a
  morphism $1_{\hcube}\to p$ in $\micat$.
  The tuple $(\rami\op,\micat\op,1_{\hcube},\caninc,\auti)$
  is a \nameref{app:def:nominal-situation}.
\end{lemma}
\begin{proof}
  ~\begin{itemize}
    \item ($\rami\op$ is a full subcategory of $\micat\op$)
      By \cref{app:cor:stdmble-epis}.
    \item ($\rami\op$ and $\micat\op$ consist only of monic maps)
      Both $\rami$ and $\micat$ contain only epis by transporting \cref{app:cor:stdmble-epis}
      across \cref{app:lem:raii-eq}.
    \item (The atomic topology exists for $\rami^{\rm op op}$)
      By \cref{app:lem:mble-right-ore} and \cref{app:lem:raii-eq}.
    \item (Closure)
      Fix a map $\caninc(\hcube\xrightarrow p X)$,
      which amounts to a finite-footprint map
      $\hcube\xrightarrow p X$ factoring through $\hcubeproj n$
      for some $n$,
      and an auto $\pi$ in $\auti$,
      which amounts to a $\StdMble$-auto
      with width $m$.
      Then $p$ also factors through $\hcubeproj{\max(m,n)}$
      and $\pi$ also has width $\max(m,n)$,
      so $p = p'\hcubeproj{\max(m,n)}$
      and $\pi = \embedmap{\pi'}$
      for some $p' : [0,1]^{\max(m,n)}\to X$
      and $\pi' : [0,1]^{\max(m,n)}\xrightarrow\sim[0,1]^{\max(m,n)}$.
      Thus $p \pi = p'\hcubeproj{\max(m,n)}(\embedmap{\pi'})
      = p'\pi' \hcubeproj{\max(m,n)}$
      showing $\hcube\xrightarrow{p\pi} X$
      is an object of $\rami$.
      The map $X\xrightarrow 1 X$
      is a morphism $p\to p\pi$ in $\rami$,
      and the relevant square needed for Closure commutes in $\micat$:
      $\caninc(p)\pi = p\pi = 1_X\caninc(p\pi)$.
    \item (Homogeneity)
      Since $\rami\op\stackrel{\text{\cref{app:lem:raii-eq}}}\simeq\StdMbleop
      \stackrel{\text{\cref{app:lem:mblecat-dual}}}\simeq\StdMbleAlg$,
      it suffices to show that for $\StdMble$-maps
      $\hcube\xrightarrow pX$
      and $\hcube\xrightarrow qY$
      with finite footprint
      and $\StdMbleAlg$-morphisms
      $f : \measalg(X)\to\measalg(Y)$,
      there exists a finite-width automorphism $\pi :\hcube\to\hcube$
      such that the following square commutes in $\StdMbleAlg$:
      \[
        \begin{tikzcd}
        \hcube \arrow[r, "\measalg(\pi)", dashed]                 & \hcube                    \\
        \measalg(X) \arrow[u, "\measalg(p)", hook] \arrow[r, "f"', hook] & \measalg(Y) \arrow[u, "\measalg(q)"', hook']
        \end{tikzcd}\]
      Suppose $p$ factors through $\hcubeproj n$
      and $q$ factors through $\hcubeproj m$.
      Then $p$ and $q$ also factor through $\hcubeproj{\max(m,n)+1}$,
      so $p = p'\hcubeproj{\max(m,n)+1}$
      and $q = q'\hcubeproj{\max(m,n)+1}$
      for some $p' : [0,1]^{\max(m,n)+1}\to X$
      and $q' : [0,1]^{\max(m,n)+1}\to Y$.
      It suffices to find an automorphism $\pi' : [0,1]^{\max(m,n)+1}\to[0,1]^{\max(m,n)+1}$
      making the following square commute, as then $\embedmap{\pi'} : \hcube\to\hcube$
      will be a finite-width automorphism of the form required:
      \[
        \begin{tikzcd}
        {[0,1]^{\max(m,n)+1}}\arrow[r, "\measalg(\pi')", dashed]                 & {[0,1]^{\max(m,n)+1}}\\
        \measalg(X) \arrow[u, "\measalg(p')", hook] \arrow[r, "f"', hook] & \measalg(Y) \arrow[u, "\measalg(q')"', hook']
        \end{tikzcd}\]
      By \cref{app:lem:hcube-measalg-iso},
      there is an isomorphism $i:[0,1]^{\max(m,n)+1}\cong[0,1]$,
      so it suffices to find an automorphism $[0,1]\to[0,1]$.
      Since $p$ and $q$ factor through $\hcubeproj n$ and $\hcubeproj m$,
      the maps $p'$ and $q'$ factor through the canonical projections 
      \[[0,1]^{\max(m,n)+1}\to[0,1]^n
      \qquad\text{and}\qquad
      [0,1]^{\max(m,n)+1}\to[0,1]^m,\]
      so $[0,1]^{\max(m,n)+1}$
      is relatively atomless over $\img(\measalg(p'))$ and $\img(\measalg(q'))$
      by \cref{app:lem:proj-relatomless}.
      Hence $[0,1]$ is relatively atomless over
         $\img(i\measalg(p'))$ and $\img(i\measalg(q'))$.
      \cref{app:lem:mble-homogeneity} then gives an automorphism of the form required.
    \item (Correspondence)
      Fix
      $\StdMble$-maps $\hcube\xrightarrow p X$
      and $\hcube\xrightarrow qY$
      with finite footprint
      such that $\Fix p\subseteq\Fix q$,
      so $p\pi = p$ implies $q\pi = q$ for all finite-width automorphisms
      $\pi$ of $\hcube$.
      The goal is to show $q$ factors through $p$.
      Since $p$ and $q$ have finite footprint,
      $p$ factors through $\hcubeproj n$ and $q$ through $\hcubeproj m$
      for some $n,m$.
      Therefore, $p$ and $q$ also both factor through
      $\hcubeproj{\max(m,n)+1}$,
      so $p = p'\hcubeproj{\max(m,n)+1}$
      and $q = q'\hcubeproj{\max(m,n)+1}$
      for some $p' : [0,1]^{\max(m,n)+1}\to X$
      and $q' : [0,1]^{\max(m,n)+1}\to Y$.
      The inclusion $\Fix p \subseteq\Fix q$
      implies every width-$({\max(m,n)+1})$ auto
      fixing $p$ fixes $q$,
      so $p(\embedmap{\pi}) = p$
      implies $q(\embedmap{\pi}) = q$
      for all autos $\pi$ of $[0,1]^{\max(m,n)+1}$.
      Since $p(\embedmap{\pi}) = p'\hcubeproj{\max(m,n)+1}(\embedmap{\pi})
      =p'\pi\hcubeproj{\max(m,n)+1}$
      and similarly $q(\embedmap{\pi}) = q'\hcubeproj{\max(m,n)+1}(\embedmap{\pi})
      =q'\pi\hcubeproj{\max(m,n)+1}$,
      this implies $p'\pi = p'$ and $q'\pi = q'$
      for all automorphisms $\pi$ of $[0,1]^{\max(m,n)+1}$,
      or in other words that $\Fix_{\Aut{[0,1]^{\max(m,n)+1}}}p'
      \subseteq\Fix_{\Aut{[0,1]^{\max(m,n)+1}}}q'$.
      The maps $p'$ and $q'$
      induce corresponding 
      subalgebras $\img(\measalg(p'))$ and $\img(\measalg(q'))$
      of ${[0,1]^{\max(m,n)+1}}$,
      and ${[0,1]^{\max(m,n)+1}}$ is relatively atomless
      over these subalgebras by \cref{app:lem:proj-relatomless},
      so \cref{app:lem:mble-distinguishability} applies,
      giving an inclusion of subalgebras $\img(\measalg(q'))\subseteq\img(\measalg(p'))$.
      This in turn gives a $\StdMbleAlg$-morphism $\measalg(Y)\to\measalg(X)$
      factoring $\measalg(q')$ through $\measalg(p')$
      (since $\measalg(X)\cong\img(\measalg(p'))$
      and $\measalg(Y)\cong\img(\measalg(q'))$),
      which by \cref{app:lem:mblecat-dual}
      gives a $\StdMble$-map $X\to Y$ factoring $q'$ through $p'$
      and hence also $q$ through $p$ as required.
    \item (Cofinality)
      Given two objects $\hcube\xrightarrow p X$
      and $\hcube\xrightarrow q Y$ of $\rami$
      with finite footprints witnessed by the factorizations
      $p = p' \hcubeproj n$
      and $q = q'\hcubeproj m$ for some $n,m\in \N$,
      the map $\hcubeproj{\max(m,n)} : \hcube\to[0,1]^{\max(m,n)}$
      is relatively atomless by \cref{app:lem:proj-relatomless}
      and there is an inclusion of subgroups
      $\Fix {\hcubeproj {\max(m,n)}}\subseteq \Fix p\cap \Fix q$.
      This extends to finite families of objects by induction.
  \end{itemize}
\end{proof}

\begin{theorem} \label{app:thm:nominal}
  $\Shjat(\StdMble)
  \stackrel{(1)}\simeq\Shjat(\rami)
  \stackrel{(2)}\simeq \tgsets{\big(\auti\big)\op}
  \stackrel{(3)}\simeq \tgsets{\auti}$.
\end{theorem}
\begin{proof}
  ~\begin{itemize}
    \item (1): by \cref{app:lem:raii-eq}.
    \item (2): by \cref{app:thm:nominal-situation}
      and \cref{app:lem:mble-nominal-situation},
      and $\Aut_{\micat\op}\hcube
      = \Aut_{\StdMbleop}\hcube = (\Aut_{\StdMble}\hcube)\op$.
    \item (3): by \cref{app:lem:left-right-gsets}.
  \end{itemize}
\end{proof}

\begin{definition} \label{app:def:pnom}
  An \emph{\agoodname{}} is an object of the category
  $\tgsets{\auti}$ of continuous $\auti$-sets.
\end{definition}

\subsection{Probabilistic concepts as \agoodnamenolink{s}} \label{app:sec:prob-concepts-as-goodnames}

We now use the equivalence
established in \cref{app:thm:nominal}
to calculate the $\auti$-set
counterparts to the sheaves
defined in \cref{app:sec:goodsheaves}.

\subsubsection{Random variables}

\newcommand\ramap{\mathbin{\nref{app:def:ramap}{\xrightarrow{\rm ff}}}}

\begin{definition} \label{app:def:ramap}
  For $A$ a Polish space,
  a random variable $X : \hcube\to A$ (equivalently,
  an element of $\RV_A\hcube$ where $\RV_A$ is the \nref{app:lem:rv-mble-sheaf}{sheaf of random variables})
  has \emph{finite footprint}
  if it factors through $\hcubeproj n$ for some $n$.
  Write $\hcube\ramap A$ for the collection of random variables with finite footprint.
\end{definition}

\begin{definition} \label{app:def:rv-pnom}
  For $A$ a Polish space,
  the \emph{\agoodname{} of $A$-valued random variables}
  is the set $\gRV A := (\hcube\ramap A)$
  of $A$-valued random variables with finite footprint,
  with action $X\cdot\pi = X\pi$.
\end{definition}

We now show $\gRV_A$ indeed defines
\aagoodname{}, and moreover corresponds to the
\goodsheaf{} $\nref{app:lem:rv-mble-sheaf}{\RV_A}$.

\begin{lemma} \label{app:lem:rv-rami-sheaf}
  Across the equivalence
  $\Shjat(\StdMble) \simeq \Shjat(\rami)$
  given by \cref{app:thm:nominal},
  the \nref{app:lem:rv-mble-sheaf}{sheaf}
  $\RV_A\in\Shjat(\StdMble)$
  of random variables
  corresponds to a sheaf $\widehat{\RV_A}\in\Shjat(\rami)$
  of random variables that factor through maps with finite footprint:
  \begin{align*}
    &\widehat{\RV_A}(p : \hcube\to \Omega)
    = \{ X\in \RV_A{\hcube} \mid \text{there exists a unique $X' : \RV_A(\Omega)$ with $X \aseq X'p$}\} \\
    &\widehat{\RV_A}(q : p'\to p)(X) = X'qp'
    \text{ where $X'$ is the unique such that $X \aseq X'p$}
  \end{align*}
\end{lemma}
\begin{proof}
  Let $\widetilde\cdot$ be the equivalence $\Shjat(\StdMble)\to\Shjat(\rami)$;
  inspecting the proof of \cref{app:lem:raii-eq}
  shows that if $F$ is a sheaf on $\StdMble$
  then $\widetilde F$ is a sheaf on $\rami$
  defined by $\widetilde F(p : \hcube\to \Omega) = F(\Omega)$
  on objects and $\widetilde F(q : p\to p') = F(q)$ on morphisms.
  In the case of $\RV_A$,
  this gives the following:
  \begin{align*}
    &\widetilde{\RV_A}(p : \hcube\to\Omega) = \RV_A(\Omega) \\
    &\widetilde{\RV_A}(q : p'\to p)(X) = Xq
  \end{align*}
  We now show $\widetilde{\RV_A}\cong\widehat{\RV_A}$.
  For $p : \hcube\to\Omega$ an object of $\rami$,
  let $\alpha_p$
  be the map $\widetilde{\RV_A}(p)\to\widehat{\RV_A}(p)$
  that sends $X\in\widetilde{\RV_A}(\Omega)$
  to $Xp\in\widehat{\RV_A}(\Omega)$.
  This is well-defined: it produces elements
  of type $\widehat{\RV_A}(p)$ because $X$ is unique,
  for
  if there were some other $X'\in\RV_A(\Omega)$
  with $Xp\aseq X'p$,
  then $p^{-1}(\{\omega \mid X\omega\ne X'\omega\})$
  negligible, so $\{\omega \mid X\omega\ne X'\omega\}$
  negligible because $p$ negligible-preserving,
  so $X\aseq X'$.
  This automatically makes $\alpha_p$ bijective,
  since its inverse is
  the map that sends $X\in\widehat{\RV_A}(p)$
  to the unique $X'$ with $X \aseq X' p$.
  Finally, $\alpha$ is natural in $p$:
  if $X : \Omega\to A$
  and $q : p'\to p$ a morphism in $\rami$,
  then
  $\alpha_{p'}(\widetilde{\RV_A}(q)(X))
  =\alpha_{p'}(Xq)
  =Xqp'
  = \widehat{\RV_A}(q)(Xp)
  = \widehat{\RV_A}(q)(\alpha_p(X))$.
\end{proof}

\begin{lemma} \label{app:lem:rv-pnom}
  Across the equivalence
  $\Shjat(\rami) \simeq \tgsets{\auti}$
  given by \cref{app:thm:nominal},
  the \nref{app:lem:rv-rami-sheaf}{sheaf}
  $\widehat{\RV_A}\in\Shjat(\rami)$
  defined in \cref{app:lem:rv-rami-sheaf}
  corresponds to the $\Aut_{\StdMble}\hcube$-set
  $\big(\hcube\ramap A\big)$ of $A$-valued
  random variables with finite footprint,
  with action $X\cdot \pi = X\pi$.
\end{lemma}
\begin{proof}
  Let $\widetilde\cdot$ be the equivalence $\Shjat(\rami)\to \tgsets{\Aut_{\StdMble}\hcube}$.
  The proofs involved in its construction
  are: \cref{app:thm:nominal-situation} to pass from sheaves to $\big(\auti\big)\op$-sets,
  and then \cref{app:lem:left-right-gsets} to pass from left to right actions.
  Inspecting these reveals that if $F$ is a sheaf on $\rami$
  then $\widetilde F$ is an $\Aut_\StdMble\hcube$-set
  with carrier $\colim_{(1_{\hcube}\xrightarrow pp)\in P\op} F(p)$ where $P\op$ is the
  set of $\micat$-morphisms out of $1_{\hcube}$
  of the form $\caninc(\hcube\xrightarrow p X)$,
  preordered by $(1_{\hcube}\xrightarrow p p)
  \preceq(1_{\hcube}\xrightarrow q q)$
  if there exists a morphism $r$ with $rp = q$ (with $r$
  necessarily unique because every $\StdMble$-map is epi).
  The following diagram illustrates:
  \[
\begin{tikzcd}
  & {\hcube} \arrow[d, "{1_{\hcube}}"]         &                           \\
{\hcube} \arrow[d, "p"'] & {\hcube} \arrow[ld, "p"'] \arrow[rd, "q"] & {\hcube} \arrow[d, "q"]    \\
X \arrow[rr, "r", dashed]                       &                                          & Y
\end{tikzcd}
  \]
  The vertical arrows depict the objects $1_{\hcube},p,q$
  of $\micat$. The diagonal arrows are two objects of the preorder $P\op$:
  $\micat$-morphisms $1_{\hcube}\xrightarrow p p$ and $1_{\hcube}\xrightarrow qq$,
  equal to $\caninc(p)$ and $\caninc(q)$ by definition.
  The dashed arrow
  witnesses the inequality $p\preceq q$ in $P\op$:
  it is a $\micat$-morphism $p\xrightarrow r q$ 
  such that $\big(1_{\hcube}\xrightarrow pp\xrightarrow r q\big)
  =\big(1_{\hcube}\xrightarrow qq\big)$,
  or equivalently a $\StdMble$-morphism $r$
  with $rp = q$.

  Specializing to our case $F = \widehat{\RV_A}$
  where $\widehat{\RV_A}$ is the sheaf defined in 
  \cref{app:lem:rv-rami-sheaf}, the carrier of $
  \widetilde F = \widetilde{\widehat {\RV_A}}$
  is the colimit over a diagram
  whose $\caninc(p)$th component (for some $p :\hcube\to \Omega$)
  is
  \[
    \widehat{\RV_A}(\hcube\xrightarrow p \Omega)
    = \{ X\in \RV_A\hcube \mid \text{$X$ factors (uniquely) through $p$}\}.
  \]
  Thus each component of the colimiting diagram is a subset of $\RV_A\hcube$.
  Since $\widehat{\RV_A}$ is an atomic sheaf, every morphism in the colimiting diagram
  is an injective $\Set$-function~(\cref{app:def:atomic-sheaf-condition});
  unwinding definitions reveals that inequalities $(\hcube\xrightarrow p\Omega')\preceq (\hcube\xrightarrow q\Omega)$ in $P\op$,
  which is to say maps $r : \Omega'\to\Omega$ with $rp = q$,
  are sent by $\widehat{\RV_A}$
  to inclusions
  \begin{align*}
    \{X\in \RV_A\hcube\mid \text{$X$ factors through $q$}\}
    &\hookrightarrow
    \{X\in \RV_A\hcube\mid \text{$X$ factors through $p$}\}
    \\
    (X'q\text{ for some }X'\in\RV_A\Omega)
    &\mapsto
    X'rp
  \end{align*}
  Since $rp = q$ by assumption, these inclusion maps
  are the canonical ones among subsets of $\RV_A\hcube$.
  Thus the colimiting diagram defining the carrier of $\widetilde{\widehat{\RV_A}}$
  is a diagram of canonical inclusions of subsets of $\RV_A\hcube$,
  and its colimit is the union of all such subsets:
  \begin{align*}
    \widetilde{\widehat{\RV_A}}
    &\stackrel{\phantom{(*)}}= \bigcup_{\big(\hcube\xrightarrow p\Omega\big)\in\rami}\{X\in \RV_A\hcube\mid \text{$X$ factors through $p$}\}\\
    &\stackrel{\phantom{(*)}}= \{X\in \RV_A\hcube\mid \text{$X$ factors through $p$ for some $\big(\hcube\xrightarrow p\Omega\big)\in\rami$}\}\\
    &\stackrel{(*)}= \{X\in \RV_A\hcube\mid \text{$X$ has finite footprint}\}\\
    &\stackrel{\phantom{(*)}}= \hcube\ramap A
  \end{align*}
  The equation $(*)$ holds: if $X$ factors through a map $p$ with finite footprint
  then $X$ has finite footprint, and conversely if $X$ has finite footprint
  then it factors through $\hcubeproj n$ for some $n$.
  This shows the $\auti$-set corresponding to $\widehat{\RV_A}$
  has carrier $\hcube\ramap A$ as claimed.

  Further inspecting \cref{app:thm:nominal},
  which transports sheaves on $\rami$ to left $\auti$-sets
  as described in \cref{app:thm:nominal-situation},
  shows $\widehat{\RV_A}$ corresponds to the left $\auti$-action
  on equivalence classes $[p\in\rami, X\ofty\widehat{\RV_A}(p)]$ defined by
  \[\pi\cdot \big[(\hcube\xrightarrow p\Omega), X : \widehat{\RV_A}(p)\big]
  = \big[(\hcube\xrightarrow q\Omega'), \widehat{\RV_A}(r)(X)\big]
  \text{ where $(q, r : p\to q)$
  is an arbitrary $\rami\op$-iso $\pi$ refines.}\]
  The following diagram illustrates:
  
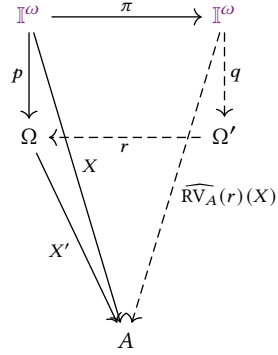
\[
\begin{tikzcd}
  {\hcube} \arrow[rr, "\pi"] \arrow[dd, "p"'] \arrow[rdddddd, "X"] &   & {\hcube} \arrow[dd, "q", dashed] \arrow[ldddddd, "\widehat{\RV_A}(r)(X)", dashed] \\
                                                                  &   &                                                                                  \\
  \Omega \arrow[rdddd, "X'"']                                     &   & \Omega' \arrow[ll, "r", dashed]                                                  \\
                                                                  &   &                                                                                  \\
                                                                  &   &                                                                                  \\
                                                                  &   &                                                                                  \\
                                                                  & A &                                                                                 
  \end{tikzcd}
    \]
  The triangle on the left depicts
  the equivalence class $[p,X]$: by the calculation of the 
  carrier of $\widetilde{\widehat{\RV_A}}$
  above, this equivalence class corresponds to a random variable $X$
  that factors through $p$ via $X'$ as shown.
  The dashed arrows depict the action of the automorphism $\pi$ 
  on $X$. There exists by Closure arbitrary $\Omega',q,r$
  with $r$ a $\rami\op$-iso from $p$ to $q$ that $\pi$ refines;
  this amounts to $\Omega',q,r$ with $r : \rami(q,p)$ iso making the square commute
  as shown. The result of the action is the composite $\widehat{\RV_A}(r)(X) = X'rq$,
  illustrated by the dashed arrow $\hcube\to A$.
  Commutativity of the diagram and $\pi,r$ iso
  implies $\widehat{\RV_A}(r)(X) = X'rq = X'p\pi^{-1} = X\pi^{-1}$.
  Thus the left $\auti$-set corresponding to
  the sheaf $\widehat{\RV_A}$
  across equivalence (2) of \cref{app:thm:nominal}
  has action $(\pi,X)\mapsto X\pi^{-1}$.
  This corresponds under equivalence (3) of \cref{app:thm:nominal}
  to the right-$\auti$-action $(X,\pi)\mapsto X\pi$, as claimed.
\end{proof}

\begin{theorem} \label{app:lem:rv-corresp}
  The \goodsheaf{}
  $\nref{app:lem:rv-mble-sheaf}{\RV_A}$
  corresponds to the \agoodname{}
  $\gRV_A$
  across
  the equivalence
  in \cref{app:thm:nominal}.
\end{theorem}
\begin{proof}
  Combine \cref{app:lem:rv-rami-sheaf}
  and \cref{app:lem:rv-pnom}.
\end{proof}

\subsubsection{Probability spaces} \label{app:sec:pspcs-nom-calc}

\renewcommand\epbmeas[2]{{\nref{app:def:epb}{#1^*}}#2}
\begin{definition} \label{app:def:epb}
  Let $(X,\calF,\calN)$
  be an standard enhanced measurable space,
  $(Y,\calG,\mu)$ a standard probability space,
  and $f : (X,\calF,\calN)\to\forgetmu(Y,\calG,\mu)$
  a $\StdMble$-map.
  The \emph{pullback of $(Y,\calG,\mu)$ along $f$},
  written $\epbmeas f(\calG,\mu)$,
  is the pair $(\epbmeas f \calG,\epbmeas f \mu)$ defined by
  \begin{align*}
    &\epbmeas f\calG = \{f^{-1}(G)\triangle N \mid G\in\calG,N'\in\calN\} \\
    &\epbmeas f\mu(f^{-1}(G)\triangle N) = \mu(G)\text{ for all $G\in\calG,N\in\calN$}
  \end{align*}
  This operation makes $(X,\epbmeas f \calG,\epbmeas f \mu)$
  a probability space with negligibles $\calN$
  and $f$ a measure-preserving map $(X,\epbmeas f\calG,\epbmeas f\mu)\to(Y,\calG,\mu)$.
\end{definition}
\begin{proof}
  The set $\epbmeas f\calG$ is a $\sigma$-algebra:
  it contains the empty set because $\emptyset = f^{-1}(\emptyset)\triangle\emptyset\in\epbmeas f\calG$,
  it's closed under complements
  because $(f^{-1}(G)\triangle N)^{\rm c}
  = f^{-1}(G)^{\rm c}\triangle N
  = f^{-1}(G^{\rm c})\triangle N\in \epbmeas f\calG$
  for all $G\in\calG$.
  For closure under countable unions, fix a countable family $(f^{-1}(G_i)\triangle N_i)_{i\in \N}$.
  First
  \begin{align*}
    &x\in \bigcup_i (f^{-1}(G_i)\triangle N_i)\setminus \bigcup_i f^{-1}(G_i)\\
    &\text{iff } (\exists i. f(x)\in G_i\Leftrightarrow x\notin N_i)\land\forall i. f(x)\notin G_i\\
    &\text{iff } (\exists i. x\in N_i)\land\forall i. f(x)\notin G_i\\
    &\text{iff } x\in \bigcup_i N_i \cap \bigcap_i f^{-1}(G_i^{\rm c})
  \end{align*}
  and $\bigcup_i N_i \cap \bigcap_i f^{-1}(G_i^{\rm c})$
  is in $\calN$ because $\bigcup_i N_i$ is in $\calN$
  and $(X,\calF,\calN)$ arises from a complete probability space.
  Second
  \begin{align*}
    &x\in \bigcup_i f^{-1}(G_i)\setminus\bigcup_i (f^{-1}(G_i)\triangle N_i) \\
    &\text{iff } (\exists i. f(x)\in G_i)\land\forall i. f(x)\in G_i\Leftrightarrow x\in N_i \\
    &\text{iff } (\exists i. x\in N_i)\land\forall i. f(x)\in G_i\Leftrightarrow x\in N_i \\
    &\text{iff } x\in \bigcup_i N_i \cap \bigcap_i f^{-1}(G_i^{\rm c})\triangle N_i
  \end{align*}
  and $\bigcup_i N_i \cap \bigcap_i f^{-1}(G_i^{\rm c})\triangle N_i)$
  negligible similarly.
  Thus
  $
    \underbrace{\left(\bigcup_i (f^{-1}(G_i)\triangle N_i)\right)}_A
    \triangle \underbrace{\left(\bigcup_i f^{-1}(G_i)\right)}_B
    = (\underbrace{A\setminus B}_{\in\calN}) \uplus (\underbrace{B\setminus A}_{\in \calN})
    \in\calN
  $
  as required.

  The measure $\mu'$ is
  well-defined:
  if $f^{-1}(G)\triangle N' = f^{-1}(\overline G)\triangle \overline N'$
  for $G,\overline G\in\calG$ and $N',\overline N'\in\calN'$
  then rearranging using the algebraic properties of $\triangle$ gives
  $f^{-1}(G\triangle \overline G) = N'\triangle \overline N'\in \calN'$,
  so $G\triangle\overline G\in\calN$ because $f$ negligible-preserving,
  so $\mu(G\triangle\overline G) = 0$ because $\mu$ has negligibles $\calN$,
  so $\mu(G) = \mu(\overline G)$.

  The equation defining $\epbmeas f\mu$
  makes $f$ measure-preserving,
  and $\epbmeas f\mu$ has negligibles $\calN$ because
  if $G\in\calG$ and $N\in\calN$
  then $\epbmeas f\mu(f^{-1}(G)\triangle N) = 0$
  iff $\mu(G) = 0$, iff $G\in\negligibles(\mu)$,
  iff $f^{-1}(G)\in\calN$ (by $f$ negligible-reflecting),
  iff $f^{-1}(G)\triangle N\in\calN$.
\end{proof}

\begin{note}
  Informally, what this operation is doing
  is completing the usual pullback measure on $f^{-1}(\calG)$
  with respect to the the negligibles $\calN$.
  For example: if $f$ is the map $[0,1]\to \{\vtrue,\vfalse\}$
  given by the indicator function $\lambda x. [x<1/2]$
  where $[0,1]$ is given Lebesgue-negligibles
  and $\{\vtrue,\vfalse\}$ is given the uniform probability measure,
  then the usual pullback $\sigma$-algebra is
  the $4$-element $\sigma$-algebra $\calF$
  generated by the atoms $\{[0,1/2), [1/2, 1]\}$,
  with each atom getting probability $1/2$.
  But the operation described in \cref{app:def:epb}
  produces not the $\sigma$-algebra $\calF$
  but rather $\calF$ plus all Lebesgue-negligible subsets of $[0,1]$,
  so that it includes not just $[0,1/2)$
  but also the closed interval $[0,1/2]$,
  sets of the form $[0,1/2)\setminus \{x\}$ for all $x\in[0,1]$,
  the set $[0,1/2] \setminus \Q$, and so on.
\end{note}

\begin{definition}[probability space on an enhanced measurable space] \label{app:def:pspc-on-ems}
  Let $(\Omega,\calF,\calN)$ be an enhanced measurable space.
  A \emph{probability space on} $(\Omega,\calF,\calN)$
  is a pair $(\calG,\mu)$
  such that $\calN\subseteq \calG\subseteq\calF$
  and $\mu$ is a probability measure with negligibles $\calN$.
  Call such a pair \emph{standardizable} if $(\Omega,\calG,\mu)$
  arises via \nref{app:def:epb}{pullback}
  along a map $X : (\Omega,\calF,\calN)\to \forgetmu(Y,\calG,\mu)$
  with $(Y,\calG,\mu)$ a standard probability space.
  In case $(\Omega,\calF,\calN) = \hcube$,
  a standardizable probability space $(\calG,\mu)$
  has \emph{finite footprint}
  if it arises by pullback along a \nref{app:def:mblespcramap}{map with finite footprint}.
\end{definition}

\begin{definition} \label{app:def:pspcs-pnom}
  The \emph{\agoodname{} of probability spaces}
  is the set $\gpspcs$
  of standardizable probability spaces on $\hcube$ with finite footprint,
  and action $(\calF,\mu)\cdot\pi = \epbmeas{\pi}{(\calF,\mu)}$.
\end{definition}

We now show $\gpspcs$ indeed defines \aagoodname{}, and moreover corresponds to the
\goodsheaf{} $\pspcs$.

\begin{lemma} \label{app:lem:epb-resp-iso}
  Let $(\Omega,\calF,\calN)$
  be an enhanced measurable space,
  $X : (\Omega,\calF,\calN)\to\forgetmu(A,\calG,\mu)$
  a $\StdMble$-map,
  and $\pi : (A,\calG,\mu)\to(A',\calG',\mu')$
  a $\StdProb$-iso.
  Then $\epbmeas X{(\calG,\mu)}
  =\epbmeas{(\pi X)}{({\calG'}, {\mu'})}$.
\end{lemma}
\begin{proof}
  Let $\tau : (A',\calG',\mu')\to(A,\calG,\mu)$ be the inverse of $\pi$
  (to avoid confusion with the operation $\pi^{-1}$ of taking $\pi$-preimages).
  First, the $\sigma$-algebras $\epbmeas X\calG$ and $\epbmeas{(\pi X)}{\calG'}$
  are equal.
  For any event $X^{-1}(G)\triangle N\in \epbmeas X\calG$
  with $G\in\calG,N\in\calN$,
  \begin{align*}
    X^{-1}(G) \triangle (\pi X)^{-1}(\tau^{-1} G)
    = X^{-1}(G)\triangle X^{-1}(\pi^{-1}(\tau^{-1} G)) 
    = X^{-1}(G)\triangle X^{-1}(M)
  \end{align*}
  for some $M\in\negligibles(\mu)$,
  so $X^{-1}(G)\triangle N 
  = (\pi X)^{-1}\underbrace{(\tau^{-1}G)}_{\in\calG'}\triangle 
  \underbrace{(N\triangle X^{-1}(M))}_{\in\calN} \in \epbmeas X{\calG'}$.
  This shows $\epbmeas X\calG\subseteq\epbmeas{(\pi X)}{\calG'}$.
  Running the same argument with the roles of $\pi$ and $\tau$ swapped
  shows the converse inclusion.
  Next, the measures are equal:
  for $X^{-1}(G)\triangle N\in \epbmeas X\calG$
  with $G\in\calG,N\in\calN$,
  \begin{align*}
    \epbmeas {(\pi X)}{\mu'}(X^{-1}(G)\triangle N)
    &= \epbmeas {(\pi X)}{\mu'}((\pi X)^{-1}\underbrace{(\tau^{-1}G)}_{\in\calG'}\triangle \underbrace{X^{-1}(M)}_{\in \calN})
      \text{ for some $M\in\negligibles(\mu)$, as in the argument above} \\
    &= \mu'(\tau^{-1}(G)) \text{ by definition of $\epbmeas{(\pi X)}{\mu'}$}\\
    &= \mu(G) \text{ because $\tau$ measure-preserving}\\
    &=\epbmeas X\mu(X^{-1}(G)\triangle N).
  \end{align*}
\end{proof}

\begin{lemma} \label{app:lem:epb-functorial}
  Pullback of probability spaces respects
  respects composition of $\StdMble$-maps:
  for all $\StdMble$-maps
  $f : (X,\calF,\calN)\to (Y,\calG,\calM)$
  and $g : (Y,\calG,\calM)\to\forgetmu(Z,\calH,\mu)$
  it holds that $\epbmeas{(gf)}{(\calH,\mu)}
  = \epbmeas f{\epbmeas g{(\calH,\mu)}}$.
\end{lemma}
\begin{proof}
  First, the $\sigma$-algebras are equal.
  Every $(gf)^{-1}(H)\triangle N\in \epbmeas{(gf)}\calH$
  is equal to $f^{-1}(g^{-1}(H)\triangle\emptyset)\triangle N\in\epbmeas{f}{\epbmeas g\calH}$
  and conversely every $f^{-1}(g^{-1}(H)\triangle M)\triangle N\in\epbmeas{f}{\epbmeas g\calH}$
  is equal to $(gf)^{-1}(H)\triangle \underbrace{f^{-1}(M)\triangle N}_{\in\calN}\in\epbmeas{(gf)}{\calH}$.
  Second, the measures are equal:
  \begin{align*}
    \epbmeas f{\epbmeas g\mu}(f^{-1}(g^{-1}(H)\triangle M)\triangle N)
    =\epbmeas g\mu(g^{-1}(H)\triangle M)
    =\mu(H)
    =\epbmeas {(gf)}\mu((gf)^{-1}(H)\triangle (f^{-1}(M)\triangle N))
    =\epbmeas {(gf)}\mu(f^{-1}(g^{-1}(H)\triangle M)\triangle N)
  \end{align*}
  for all $(f^{-1}(g^{-1}(H)\triangle M)\triangle N)\in\epbmeas f{\epbmeas g\calH}$.
\end{proof}

\begin{lemma} \label{app:lem:epb-bij-hom}
  For $f : (X,\calF,\calN)\to\forgetmu(Y,\calG,\mu)$
  a $\StdMble$-map,
  the measure algebra homomorphism
  $\measalg(X,\epbmeas f\calG,\epbmeas f \mu)\xleftarrow{\measalg(f)}\measalg(Y,\calG,\mu)$
  is bijective.
\end{lemma}
\begin{proof}
  The homomorphism $\measalg(f)$ is automatically injective
  because $f$ measure-preserving as a map $(X,\epbmeas f\calG,\epbmeas f\mu)\to(Y,\calG,\mu)$
  and measure-preserving homomorphisms are injective~\fremlinf{324K(a)}{32}{25}.
  It's surjective because every element of $\measalg(X,\epbmeas f \calG,\epbmeas f\mu)$
  is an equivalence class
  of the form $[f^{-1}(G)\triangle N]$
  in the quotient boolean algebra $\epbmeas f \calG/\negligibles(\epbmeas f \mu)$
  for some $G\in\calG,N\in\calN$,
  and \begin{align*}
    [f^{-1}(G)\triangle N]
    \stackrel{(*)}= [f^{-1}(G)]
    = \measalg(f)[G]
  \end{align*}
  where $(*)$ holds because $\epbmeas f\mu$ has negligibles $\calN$.
\end{proof}

\begin{lemma} \label{app:lem:mble-pspcs-repr-on}
  The sheaf $\pspcs$ is equivalently
  a sheaf of standardizable probability spaces, whose action on objects is
  \begin{align*}
    \pspcs\Omega
    \cong \{\text{standardizable probability spaces $(\calG,\mu)$ \nref{app:def:pspc-on-ems}{on} $\Omega$}\}
  \end{align*}
  and whose action on morphisms takes \nref{app:def:epb}{pullbacks of probability spaces}:
  \begin{align*}
    \pspcs(f : \Omega'\to\Omega)(\calG,\mu)
    = \epbmeas f{(\calG,\mu)}
  \end{align*}
\end{lemma}
\begin{proof}
  Any $[(A,\calF,\mu),X : \Omega\to A]\in \pspcs\Omega$
  gives rise to a standardizable probability space $\epbmeas X(\calF,\mu)$.
  This operation respects the equivalence class $[A,X]$
  by \cref{app:lem:epb-resp-iso}, and
  is surjective by definition of standardizable probability space.
  It is natural in $\Omega$ by \cref{app:lem:epb-functorial}.
  All that's left is to show injectivity.
  Fix $[(A,\calF,\mu),X],[(B,\calG,\nu),Y]$
  with $(\epbmeas X\calF,\epbmeas X\mu) = (\epbmeas Y\calG,\epbmeas Y\nu)$.
  Applying \cref{app:lem:epb-bij-hom}
  to $X$ and $Y$
  gives bijective measure-algebra homomorphisms
  \[ \measalg(X) : \measalg(A,\calF,\mu)\to \measalg(\Omega,\epbmeas X\calF,\epbmeas X\mu)
  \qquad\text{and}\qquad
    \measalg(Y) : \measalg(B,\calG,\nu)\to \measalg(\Omega,\epbmeas Y\calG,\epbmeas Y\nu).
  \]
  The composition $i^* := \measalg(X)^{-1}\measalg(Y) : \measalg(B,\calG,\nu)\to\measalg(A,\calF,\mu)$,
  well-typed because $(\epbmeas X\calF,\epbmeas X\mu) = (\epbmeas Y\calG,\epbmeas Y\nu)$,
  satisfies $\measalg(X) i^* = \measalg(Y)$,
  so corresponds by \cref{app:lem:mblecat-dual}
  to a $\StdProb$-iso $i : (A,\calF,\mu)\to(B,\calG,\nu)$
  such that $\forgetmu(i)X = Y$,
  witnessing $[(A,\calF,\mu),X] = [(B,\calG,\nu),Y]$.
\end{proof}

\begin{lemma} \label{app:lem:epbmeas-inj}
  For $f:(X,\calF,\calN)\to(Y,\calG,\calM)$ a $\StdMble$-map
  and $(\calH,\mu),(\calH',\mu')$ probability spaces on $(Y,\calG,\calM)$,
  if $\epbmeas f (\calH,\mu)= \epbmeas f({\calH'},{\mu'})$
then $(\calH,\mu) = (\calH',\mu')$.
\end{lemma}
\begin{proof}
  If $H\in\calH$ then
  $f^{-1}(H) = f^{-1}(H)\triangle\emptyset\in\epbmeas f \calH = \epbmeas f{\calH'}$,
  so $f^{-1}(H) = f^{-1}(H')\triangle N$ for some $H'\in\calH'$
  and $N\in\calN$, so $f^{-1}(H)\triangle f^{-1}(H') = f^{-1}(H\triangle H') = N\in\calN$,
  so $H\triangle H'\in\calM$ since $f$ negligible-preserving,
  so $H\triangle H'\in\calH'$ since $\calH'\supseteq\calM$,
  so $H'\triangle (H\triangle H') = H\in\calH'$.
  This shows $\calH\subseteq\calH'$;
  the converse inclusion follows from an analogous argument
  with the roles of $\calH,\calH'$ swapped.
  Thus $\calH = \calH'$.
  Finally, if $H\in\calH$
  then $f^{-1}(H) = f^{-1}(H')\triangle N$ for some $H'\in\calH',N\in\calN$,
  so $\mu(H) = \epbmeas f\mu(f^{-1}(H))
  =\epbmeas f\mu'(f^{-1}(H')\triangle N)
  = \mu'(H') = \mu(H)$
  where the last equality follows from $f^{-1}(H)\triangle f^{-1}(H') = N\in\calN$
  and $f$ negligible-preserving and $\negligibles(\mu) = \negligibles(\mu') = \calM$.
\end{proof}

\begin{lemma} \label{app:lem:pspcs-rami}
  Across the equivalence
  $\Shjat(\StdMble) \simeq \Shjat(\rami)$
  given by \cref{app:thm:nominal},
  the \nref{app:def:mble-pspcs}{sheaf}
  $\pspcs\in\Shjat(\StdMble)$
  of standardizable probability spaces
  described in \cref{app:lem:mble-pspcs-repr-on}
  corresponds to a sheaf $\widehat{\pspcs}\in\Shjat(\rami)$
  of standardizable probability spaces that factor through maps with finite footprint:
  \begin{align*}
    &\widehat{\pspcs}(p : \hcube\to \Omega)
    = \{ (\calG,\nu)\text{ a standardizable probability space \nref{app:def:pspc-on-ems}{on} $\hcube$}
    \mid \text{there exists a unique $(\calF,\mu)$ \nref{app:def:pspc-on-ems}{on} $\Omega$ with $(\calG,\nu) = \epbmeas p(\calF,\mu)$}
    \} \\
    &\widehat{\pspcs}(q : p'\to p)(\epbmeas p(\calF,\mu)) = {\epbmeas{(qp')}{(\calF,\mu)}}
  \end{align*}
\end{lemma}
\begin{proof}
  The proof is similar to \cref{app:lem:rv-rami-sheaf}.
  Let $\widetilde\cdot$ be the equivalence $\Shjat(\StdMble)\to\Shjat(\rami)$;
  across this equivalence and \cref{app:lem:mble-pspcs-repr-on},
  the sheaf $\widetilde\pspcs \in\Shjat(\rami)$ is
  \begin{align*}
    &\widetilde{\pspcs}(p : \hcube\to\Omega) = \{\text{standardizable probability spaces $(\calG,\mu)$ on $\Omega$}\} \\
    &\widetilde{\pspcs}(q : p'\to p)(\calF,\mu) = \epbmeas q{(\calF,\mu)}
  \end{align*}
  We now show $\widetilde{\pspcs}\cong\widehat{\pspcs}$.
  For $p : \hcube\to\Omega$ an object of $\rami$,
  let $\alpha_p$
  be the map $\widetilde{\pspcs}(p)\to\widehat{\pspcs}(p)$
  that sends $(\calF,\mu)\in\widetilde{\pspcs}(\Omega)$
  to $\epbmeas p{(\calF,\mu)}\in\widehat{\pspcs}(\Omega)$.
  This is well-defined: it produces elements
  of type $\widehat{\pspcs}(p)$ because $(\calF,\mu)$ is unique
  by \cref{app:lem:epbmeas-inj}.
  This automatically makes $\alpha_p$ bijective,
  with inverse 
  the map sending $(\calG,\nu)\in\widehat{\pspcs}(p)$
  to the unique $(\calF,\mu)\in\widetilde{\pspcs}(p)$ with $(\calG,\nu) = \epbmeas p{(\calF,\mu)}$.
  Finally, $\alpha$ is natural in $p$:
  if $(\calF,\mu)\in\widetilde{\pspcs}(p)$
  and $q : p'\to p$ a $\rami$-morphism,
  then
  $\alpha_{p'}(\widetilde{\pspcs}(q)(\calF,\mu))
  =\alpha_{p'}(\epbmeas q{(\calF,\mu)})
  =\epbmeas{(qp')}{(\calF,\mu)}
  = \widehat{\pspcs}(q)(\epbmeas p{(\calF,\mu)})
  = \widehat{\pspcs}(q)(\alpha_p(\calF,\mu))$
  by \cref{app:lem:epb-functorial}.
\end{proof}

\begin{lemma} \label{app:lem:pspcs-pnom}
  Across the equivalence
  $\Shjat(\rami) \simeq \tgsets{\auti}$
  given by \cref{app:thm:nominal},
  the sheaf
  $\widehat{\pspcs}\in\Shjat(\rami)$
  defined in \cref{app:lem:pspcs-rami}
  corresponds to the $\Aut_{\StdMble}\hcube$-set
  of standardizable probability spaces on $\hcube$ with finite footprint,
  with action $(\calF,\mu)\cdot \pi = \epbmeas\pi{(\calF,\mu)}$.
\end{lemma}
\begin{proof}
  The proof is similar to \cref{app:lem:rv-pnom}.
  Let $\widetilde\cdot$ be the equivalence $\Shjat(\rami)\to \tgsets{\Aut_{\StdMble}\hcube}$.
  The proofs involved in its construction
  are: \cref{app:thm:nominal-situation} to pass from sheaves to $\big(\auti\big)\op$-sets,
  and then \cref{app:lem:left-right-gsets} to pass from left to right actions.
  As in the proof of \cref{app:lem:rv-pnom},
  the sheaf $\widehat{\pspcs}$
  is sent to an $\auti$-set $\widetilde{\widehat{\pspcs}}$
  whose carrier is a colimit over $\widehat\pspcs(p)$
  as $p$ ranges over the preorder on $\rami$-objects
  of the form
  $\hcube\xrightarrow p p$
  with ordering relation
  $(\hcube\xrightarrow p p)\preceq (\hcube\xrightarrow q q)$
  iff $rp = q$ for some (necessarily unique) $r$.
  For $\hcube\xrightarrow p p$,
  \[
    \widehat{\pspcs}(\hcube\xrightarrow p \Omega)
    = \{ (\calG,\nu)\text{ on }\hcube
    \mid \text{there exists a unique $(\calF,\mu)$ with $(\calG,\nu) = \epbmeas p{(\calF,\mu)}$} \},
  \]
  so all objects in the colimiting diagram are subsets of the set of standardizable probability spaces on $\hcube$.
  Ordering relations $(\hcube\xrightarrow p\Omega')\preceq (\hcube\xrightarrow q\Omega)$,
  which is to say maps $r : \Omega'\to\Omega$ with $rp = q$,
  are sent by $\widehat\pspcs$ to inclusions
  \begin{align*}
    \widehat\pspcs(q)
    &\hookrightarrow
    \widehat\pspcs(p)
    \\
    \epbmeas q{(\calF,\mu)}
    &\mapsto
    \epbmeas {(rp)}{(\calF,\mu)} = \epbmeas q{(\calF,\mu)}\text{ since $rp = q$}
  \end{align*}
  so morphisms of the colimiting diagram are
  the canonical ones among subsets of the set of standardizable probability spaces on $\hcube$.
  Thus the colimit defining the carrier of $\widetilde{\widehat{\RV_A}}$
  is the union of all such subsets:
  \begin{align*}
    \widetilde{\widehat{\pspcs}}
    &= \bigcup_{\big(\hcube\xrightarrow p\Omega\big)\in\rami}\{(\calG,\nu)\text{ standardizable on }\hcube\mid \text{$(\calG,\nu)$ factors through $p$}\}\\
    &= \{(\calG,\nu)\text{ standardizable on }\hcube\mid \text{$(\calG,\nu) = \epbmeas p{(\calF,\mu)}$
      for some $\big(\hcube\xrightarrow p\Omega\big)\in\rami$
      and $(\calF,\mu)$ standardizable on $\Omega$}\}\\
    &= \{(\calG,\nu)\text{ standardizable on }\hcube\mid \text{$(\calG,\nu)$ has finite footprint}\}
  \end{align*}
  Further inspecting \cref{app:thm:nominal},
  which transports sheaves on $\rami$ to left $\auti$-sets
  as described in \cref{app:thm:nominal-situation},
  shows $\widehat{\pspcs}$ corresponds to the left $\auti$-action
  on equivalence classes $[p\in\rami, (\calG,\nu))\ofty\widehat{\pspcs}(p)]$ defined by
  \[\pi\cdot \big[(\hcube\xrightarrow p\Omega), (\calG,\nu) \ofty \widehat{\pspcs}(p)\big]
  = \big[(\hcube\xrightarrow q\Omega'), \widehat{\pspcs}(r)(\calG,\nu)\big]
  \text{ where $(q, r : p\to q)$
  is an arbitrary $\rami\op$-iso $\pi$ refines.}\]
  In this case, this action amounts
  to sending a standardizable probability space $\epbmeas p{(\calF,\mu)}$
  (where $p$ is some map $\hcube\to\Omega$ with finite footprint)
  to $\epbmeas{(rq)}(\calF,\mu)$ where $q,r$ arbitrary such that
  $r$ iso and
  \[
\begin{tikzcd}
  \hcube \arrow[rr, "\pi"] \arrow[dd, "p"'] &  & \hcube \arrow[dd, "q"]  \\
                                            &  &                                 \\
  \Omega                                    &  & \Omega' \arrow[ll, "r"]
  \end{tikzcd}
    \]
  commutes. This implies $rq = p\pi^{-1}$,
  so the action sends $\epbmeas p{(\calF,\mu)}$
  to $\epbmeas{(p\pi^{-1})}{(\calF,\mu)}=\epbmeas{(\pi^{-1})}{\epbmeas{p}{(\calF,\mu)}}$~(\cref{app:lem:epb-functorial}).
  Since every element of $\widetilde{\widehat\pspcs}$
  is of the form $\epbmeas p{(\calF,\mu)}$ for some $p,\calF,\mu$,
  this shows the left $\auti$-set corresponding to
  the sheaf $\widehat{\pspcs}$
  across equivalence (2) of \cref{app:thm:nominal}
  has action $(\pi,(\calG,\nu))\mapsto \epbmeas{(\pi^{-1})}{(\calG,\nu)}$.
  This corresponds under equivalence (3) of \cref{app:thm:nominal}
  to the right $\auti$-action $((\calG,\nu),\pi)\mapsto \epbmeas\pi{(\calG,\nu)}$, as claimed.
\end{proof}

\begin{theorem} \label{app:lem:pspcs-corresp}
  The \goodsheaf{}
  $\pspcs$
  corresponds to the \agoodname{}
  $\gpspcs$
  across the equivalence
  in \cref{app:thm:nominal}.
\end{theorem}
\begin{proof}
  Combine \cref{app:lem:pspcs-rami}
  and \cref{app:lem:pspcs-pnom}.
\end{proof}

\subsection{Separation as independent combination}

\cref{app:sec:dayconv} established the following: \begin{itemize}
  \item The Day convolution $\pspcs\shotimes\pspcs$
    is a subobject of the sheaf $\pspcs\times\pspcs$
    of pairs of probability spaces, consisting of those
    pairs which factor through a tensor product.
  \item There is a map $\pcmjoin : \pspcs\shotimes\pspcs\to\pspcs$
    combining such pairs of probability spaces.
    Viewed as a partial map $\pspcs\times\pspcs\to\pspcs$
    and in combination with an ordering relation
    $(\krmorder) : \pspcs\times\pspcs\to\Prop$,
    it forms a partially defined monoid (PDM) internal to the category
    of \goodsheaves{}~(\cref{app:thm:mble-pspcs-krm}).
\end{itemize}
By the equivalence of \goodsheaves{}
and \agoodname{s} in \cref{app:thm:nominal},
there is a corresponding PDM
internal to the category of \agoodname{s}.
We show this PDM is --- modulo small differences regarding
negligible sets --- the PDM of \citet{li2023lilac},
and in particular that the
monoidal operation is independent combination.

\begin{definition}[independent sub-$\sigma$-algebras~\fremlinf{272A(b)}{27}{10}]
  Let $(\Omega,\calF,\mu)$
  be a probability space.
  Two sub-$\sigma$-algebras $\calG,\calH\subseteq\calF$
  are \emph{independent}
  if $\mu(G\cap H) = \mu(G)\mu(H)$ for all $G\in\calG,H\in\calH$.
\end{definition}

\begin{definition}[ordering on probability spaces] \label{app:def:pspc-ord}
  Let $(\Omega,\calF,\calN)$ be a standard enhanced measurable space.
  Let $\gkrmorder$ be the ordering on standardizable probability spaces on $\Omega$
  given by $(\calG,\mu)\gkrmorder (\calG',\mu')$ iff $\calG\subseteq\calG'$ and $\mu = \mu'|_\calG$.
\end{definition}

\begin{definition}[independent combination] \label{app:def:indepcomb}
  Let $(\Omega,\calF,\calN)$
  be a standard enhanced measurable space.
  Two standardizable probability spaces $(\calG,\mu),(\calH,\nu)$
  on $(\Omega,\calF,\calN)$ (\cref{app:def:pspc-on-ems})
  are \emph{independently combinable}
  if there exists a standardizable probability space $(\calK,\rho)$
  on $(\Omega,\calF,\calN)$
  such that $(\calG,\mu)\gkrmorder(\calK,\rho)\gkrmorderop(\calH,\nu)$
    and $\calG$ and $\calH$ are independent sub-$\sigma$-algebras
    in the probability space $(\Omega,\calK,\rho)$.
  If $(\calK,\rho)$ is the smallest such standardizable probability space
  with respect to the ordering $\gkrmorder$,
  then it is called the \emph{independent combination}
  of $(\calG,\mu)$ and $(\calH,\nu)$.
\end{definition}

Compare \cref{app:def:indepcomb} with \lilact{Definition 2.2}{10}:
they are essentially the same, 
up to the $\sigma$-ideal
$\calN$ of negligible sets and the stipulation that probability spaces
be standardizable.

\begin{definition}[empty probability space] \label{app:def:emp-gpspcs}
  Let $\calN$ be the negligibles of $\hcube$.
  The \emph{empty probability space on $\hcube$},
  written $\gpcmunit$, is the probability space on $\hcube$
  with $\sigma$-algebra generated by $\calN$
  and measure defined by $\mu(N) = 0$ for all $N\in\calN$,
  standardizable because it arises via pullback 
  along the unique $\StdMble$-map $! : (\Omega,\calF,\calN)\to (1,\{\emptyset,1\},\emptyset)$,
  finite-footprint because $!$ factors through $\hcubeproj 0$.
\end{definition}

\newcommand\indepable{{\nref{app:def:indepable}{{\mathbb P}^2_{\perp}}}}

In \cref{app:sec:pspcs-nom-calc}
we showed that, across the equivalence in \cref{app:thm:nominal},
the sheaf $\pspcs$ corresponds to the \agoodname{}
$\gpspcs$ of standardizable probability spaces with finite footprint.
In this section we extend this correspondence with the following: \begin{itemize}
  \item The map $\pcmunit : 1\to \pspcs$
     corresponds to the empty probability space $\gpcmunit : \gpspcs$
     (equivalently a map $1\to\gpspcs$ in the category of \agoodname{s})~(\cref{app:lem:pcmunit-corresp}).
  \item The ordering relation $\krmorder : \pspcs\times\pspcs\to\Prop$ 
    corresponds to the ordering $\gkrmorder : \overline\pspcs\times\overline\pspcs\to\Prop$~(\cref{app:lem:krmorder-corresp}).
  \item The Day convolution $\pspcs\shotimes\pspcs$ corresponds to
    \aagoodname{} $\indepable$ of pairs of independently-combinable probability spaces (\cref{app:lem:indepable-corresp}).
  \item The partial map $\pcmjoin : \pspcs\times\pspcs\to\Prop$ corresponds to
    a partial map $\gpcmjoin : \gpspcs\times\gpspcs\to\gProp$
    that sends
    independently-combinable pairs to their independent combination (\cref{app:thm:pcmjoin-corresp}).
\end{itemize}
Putting this together shows the PDM $(\pspcs,\pcmjoin,\pcmunit,\krmorder)$
internal to \goodsheaves{}
corresponds to the PDM $(\gpspcs,\gpcmjoin,\gpcmunit,\gkrmorder)$
internal to \agoodname{s}~(\cref{app:thm:krm-corresp}).

The recipe 
for showing a morphism of \goodsheaves{}
corresponds to a morphism of \agoodname{s}
across the equivalence in \cref{app:thm:nominal}
is as follows. For the sheaves
 $F$ introduced in \cref{app:sec:goodsheaves}, the corresponding \agoodname{} $\overline F$
is a union of the images of the
embeddings $F(p) : F\Omega\hookrightarrow F\hcube$
for $\StdMble$-maps $p : \hcube\to\Omega$ with finite footprint.
A natural transformation of \goodsheaves{} $\alpha : F\to G$
then corresponds to a map of \agoodname{s} $f : \overline F\to \overline G$
across the equivalence in \cref{app:thm:nominal} if
for all $\Omega$ and $p : \hcube\to\Omega$ with finite footprint,
it holds that $\left(F\Omega\xhookrightarrow{F(p)}
\overline F\xrightarrow f\overline G\right) 
= \left(F\Omega\xrightarrow{\alpha_\Omega}G\Omega
\xhookrightarrow{G(p)}\overline G\right)$;
i.e., $f$ behaves like $\alpha$ on elements
$x\in F\Omega$ when embedded into $\overline F$ via $F(p)$.

\begin{lemma} \label{app:lem:pcmunit-corresp}
  The constant function at the empty probability space
  $\gpcmunit$ defines a map $1\to\gpspcs$
  of \agoodname{s},
  and this map corresponds to
  the map $\pcmunit : 1\to \pspcs$
  across the equivalence in \cref{app:thm:nominal}.
\end{lemma}
\begin{proof}
  At stage $\Omega\in\StdMble$,
  the natural transformation $\pcmunit$
  picks out the equivalence class $[1,!:\Omega\to \forgetmu 1]\in\pspcs\Omega$
  where $1\in\StdProb$ is the
  one-point probability space.
  For any $\StdMble$-map $p : \hcube\to\Omega$ with finite footprint,
  the element $\pspcs(p)[1,!]$
  of $\gpspcs$ corresponding to $[1,!]$
  is the pullback $\epbmeas{(!p)}{1}$
  on $\hcube$; this is precisely $\gpcmunit$.
\end{proof}

\begin{lemma} \label{app:lem:subalg-implies-ordering}
  Let $(X,\calF,\calN)$ be an enhanced measurable space.
  Let $f : (X,\calF,\calN)\to\forgetmu(A,\calG,\mu)$
  and $g : (X,\calF,\calN)\to\forgetmu(B,\calH,\nu)$
  be two $\StdMble$-maps.
  If $\epbmeas f(\calG,\mu) \gkrmorder \epbmeas f(\calH,\nu)$
  then there exists a $\StdProb$-map $p : (B,\calH,\nu)\to(A,\calG,\mu)$
  such that $f = \forgetmu(p) g $.
\end{lemma}
\begin{proof}
  The inequality $\epbmeas f(\calG,\mu) \gkrmorder \epbmeas f(\calH,\nu)$
  implies a corresponding inclusion of measure algebras
  $i : \measalg(\epbmeas f(\calG,\mu)) \hookrightarrow \measalg(\epbmeas f(\calH,\nu))$.
  By \cref{app:lem:epb-bij-hom}, the maps $f$ and $g$
  give rise to $\StdProb$-isos
  \[
    \measalg(f) : \measalg(A,\calG,\mu)\xrightarrow\sim
      \measalg(\epbmeas f(\calG,\mu))
    \qquad\text{and}\qquad
    \measalg(g) : \measalg(B,\calH,\nu)\xrightarrow\sim
      \measalg(\epbmeas g(\calH,\nu)).
  \]
  The composite $p^* := \measalg(g)^{-1}i\measalg(f)$
  is a $\StdProbAlg$-morphism $\measalg(A,\calG,\mu)\to\measalg(B,\calH,\nu)$
  with $\measalg(g)p^* = i\measalg(f)$, so
  by \cref{app:lem:mblecat-dual} corresponds 
  to a $\StdProb$-map $p : (B,\calH,\nu)\to(A,\calG,\mu)$
  such that $f = \forgetmu(p) g$ as needed.
\end{proof}

\begin{lemma} \label{app:lem:krmorder-corresp-helper}
  For any $\Omega\in\StdMble$
  and $(A,\calF,\mu),(B,\calG,\nu)\in\StdProb$
  and $\StdMble$-maps $X:\Omega\to\forgetmu(A,\calF,\mu)$
  and $Y:\Omega\to\forgetmu(B,\calG,\nu)$,
  the following are equivalent: \begin{enumerate}
    \item $X = qY$ for some $\StdMble$-map $q : (B,\calG,\nu)\to(A,\calF,\mu)$
    \item $\epbmeas{X}(\calF,\mu) \gkrmorder \epbmeas{Y}(\calG,\nu)$
  \end{enumerate}
\end{lemma}
\begin{proof}
  If $X = qY$ then
  $\epbmeas{X}{(\calF,\mu)}
  =\epbmeas{qY}{(\calF,\mu)}
  =\epbmeas{Y}{\epbmeas q{(\calF,\mu)}}
  \gkrmorder\epbmeas{Y}{(\calG,\nu)}
  $ where the final inequality follows from
  the fact that $\epbmeas q{(\calF,\mu)}\gkrmorder (\calG,\nu)$
  as probability spaces on $(A,\calG,\nu)$.
  Conversely, if (2) holds then \cref{app:lem:subalg-implies-ordering}
  gives $q : (B,\calG,\nu)\to(A,\calF,\mu)$
  so $qY = X$.
\end{proof}

\begin{lemma} \label{app:epb-resp-ord}
  If $f : (X,\calF,\calN)\to(Y,\calG,\calM)$
  a $\StdMble$-map
  and $(\calH,\mu),(\calH',\mu')$
  are two probability spaces on $(Y,\calG,\calM)$,
  then $(\calH,\mu)\gkrmorder(\calH',\mu')$
  iff $\epbmeas f{(\calH,\mu)}\gkrmorder\epbmeas f{(\calH',\mu')}$.
\end{lemma}
\begin{proof}
  The left-to-right implication is straightforward.
  For the converse, suppose
 $\epbmeas f{(\calH,\mu)}\gkrmorder\epbmeas f{(\calH',\mu')}$.
  Then $\calH\subseteq\calH'$,
  because if $H\in\calH$
  then $f^{-1}(H)\in\epbmeas f\calH$,
  so $f^{-1}(H) = f^{-1}(H')\triangle N$
  for some $H'\in\calH',N\in\calN$
  by $\epbmeas f\calH=\epbmeas f{\calH'}$,
  so $f^{-1}(H)\triangle f^{-1}(H')
  =f^{-1}(H\triangle H')=N\in\calN$,
  so $H\triangle H'\in\calN$ by $f$ negligible-preserving,
  so $H' = H'\triangle(H\triangle H')\in\calH'$
  by $\calH'\supseteq\calN$.
  And $\mu'|_\calH = \mu$,
  because if $H\in\calH$
  then \begin{align*}
    \mu(H)
     &= \epbmeas f\mu(f^{-1}(H)) \\
     &= \epbmeas f{\mu'}(f^{-1}(H)) \text{ by $\epbmeas f\mu=\epbmeas f{\mu'}|_{\epbmeas f{\calH}}$}\\
     &= \epbmeas f{\mu}(f^{-1}(H')\triangle N) \text{ for $H'\in\calH',N\in\calN$ as above} \\
     &= \mu'(H') \\
     &= \mu'(H'\triangle(H\triangle H')) \text{ where $H\triangle H'\in\calN$ as above}\\
     &= \mu'(H') \text{ because $\calN=\negligibles(\mu')$}.
  \end{align*}
\end{proof}

\begin{lemma} \label{app:lem:krmorder-corresp}
  The ordering $\gkrmorder$
  is a map $\gpspcs\times\gpspcs\to\Prop$
  of \agoodname{s}
  corresponding to 
  the map $\krmorder : \pspcs\times\pspcs\to\Prop$
  of \goodsheaves{}
  across the equivalence in \cref{app:thm:nominal}.
\end{lemma}
\begin{proof}
  For any $\Omega\in\StdMble$
  and $\StdMble$-map $p : \hcube\to\Omega$ with finite footprint,
  two elements $[(A,\calF,\mu),X]$
  and $[(B, \calG,\nu), Y]$ of $\pspcs\Omega$
  are related by $\krmorder_\Omega$
  iff $X = qY$ for some $\StdMble$-map $q : (B,\calG,\nu)\to(A,\calF,\mu)$,
  iff $\epbmeas X{(\calF,\mu)} \gkrmorder \epbmeas Y{(\calG,\nu)}$ by \cref{app:lem:krmorder-corresp-helper},
  iff $\epbmeas{(Xp)}{(\calF,\mu)} \gkrmorder \epbmeas{(Yp)}{(\calG,\nu)}$ by \cref{app:epb-resp-ord},
  and $\epbmeas{(Xp)}{(\calF,\mu)}$ and $\epbmeas{(Yp)}{(\calG,\nu)}$
  are the standardizable probability spaces on $\hcube$
  corresponding to $[A,X]$ and $[B,Y]$ across the equivalence
  in \cref{app:thm:nominal}.
\end{proof}

\begin{definition} \label{app:def:indepable}
  Let $\indepable$ be the $\preauti$-set
  of pairs $((\calG,\mu),(\calH,\nu))\in\gpspcs\times\gpspcs$
  for which $(\calG,\mu)$ and $(\calH,\nu)$
  are independently combinable,
  with action inherited from $\gpspcs\times\gpspcs$.
\end{definition}

\begin{lemma} \label{app:lem:indepcomb-pbstable}
  If $(\calH,\mu)$ and $(\calK,\nu)$
  are independently combinable on
  a standard enhanced measurable space $(Y,\calG,\calM)$
  and 
  $f$ is a $\StdMble$-map $(X,\calF,\calN)\to(Y,\calG,\calM)$,
  then $\epbmeas f(\calH,\mu)$
  and $\epbmeas f(\calK,\nu)$
  are independently combinable on $(X,\calF,\calN)$.
\end{lemma}
\begin{proof}
  Suppose $(\calL,\rho)$ witnesses
  the independent combinability of 
  $(\calH,\mu)$ and $(\calK,\nu)$,
  so $\calH\subseteq\calL\supseteq\calK$
  and $\mu=\rho|_\calH$ and $\nu=\rho|_\calK$
  and $\rho(H\cap K)=\rho(H)\rho(K)$
  for all $H\in\calH,K\in\calK$.
  It is straightforward to show
  $\epbmeas f\calH\subseteq\epbmeas f\calL\supseteq\epbmeas f\calK$
  and $\epbmeas f\calH = \epbmeas f\calL|_{\epbmeas f\calH}$
  and $\epbmeas f\calK = \epbmeas f\calL|_{\epbmeas f\calK}$.
  Fix arbitrary $f^{-1}(H)\triangle N\in\epbmeas f\calH$
  and $f^{-1}(K)\triangle N'\in\epbmeas f{\calK}$
  for $H\in\calH$ and $K\in\calK$ and $N,N'\in\calN$.
  Since $\cap$ distributes over $\triangle$,
  \begin{align*} 
    \epbmeas f\rho((f^{-1}(H)\triangle N)\cap(f^{-1}(K)\triangle N'))
    &= 
    \epbmeas f\rho(
      (f^{-1}(H)\cap f^{-1}(K))\triangle
      \underbrace{(f^{-1}(H)\cap N')\triangle
      (N\cap f^{-1}(K))\triangle
      (N\triangle N')}_{\in\calN}
    )\\
    &= \epbmeas f\rho(f^{-1}(H\cap K))
      \text{ because $\calN = \negligibles(\epbmeas f\rho)$}
      \\
    &= \rho(H\cap K)\\
    &= \rho(H)\rho(K)\\
    &= \epbmeas f\rho(f^{-1}(H)\triangle N)
       \epbmeas f\rho(f^{-1}(K)\triangle N')
  \end{align*}
  so $\epbmeas f\calH$ and $\epbmeas f\calK$
  are independent sub-$\sigma$-algebras in $(X,\epbmeas f\calL,\rho)$.
\end{proof}

\begin{lemma} \label{app:lem:factor-through-tensor-iff-indep-combable}
  For any $\Omega\in\StdMble$ 
  and $(A,\calF,\mu),(B,\calG,\nu)\in\StdProb$
  and $\StdMble$-maps $X:\Omega\to\forgetmu(A,\calF,\mu)$
  and $Y:\Omega\to\forgetmu(B,\calG,\nu)$,
  the following are equivalent: \begin{enumerate}
    \item The pairs $((A,\calF,\mu),X)$ and $((B,\calG,\nu),Y)$
      factor through a tensor product;
      i.e., there exist $\Omega_1,\Omega_2$ and $f : \Omega\to\Omega_1\otimes\Omega_2$
      and $X':\Omega_1\to \forgetmu{(A,\calF,\mu)}$
      and $Y':\Omega_2\to\forgetmu{(B,\calG,\nu)}$
      with $X = X'\,\scmfst\,f$
      and $Y = Y'\,\scmsnd\,f$.
    \item The pullbacks $\epbmeas{X}(\calF,\mu)$
     and $\epbmeas{Y}(\calG,\nu)$ are independently combinable.
  \end{enumerate}
  Moreover, in the case where both hold,
  the pullback $\epbmeas{((X'\otimes Y')f)}{(\calF\otimes\calG,\mu\otimes\nu)}$
  of the tensor product $(A,\calF,\mu)\otimes(B,\calG,\nu)$
  along the composite $\left(\Omega\xrightarrow f\Omega_1\otimes\Omega_2 \xrightarrow{X'\otimes Y'}A\otimes B\right)$
  is the independent combination of
    $\epbmeas{X}(\calF,\mu)$
     and $\epbmeas{Y}(\calG,\nu)$.
\end{lemma}
\begin{proof}
  Suppose (1).
  Form the tensor product $(A\times B,\calF\otimes\calG,\mu\otimes\nu)$
  of $(A,\calF,\mu)$ and $(B,\calG,\nu)$,
  so that $X = \scmfst(X'\otimes Y')f$
  and $Y=\scmsnd(X'\otimes Y')f$.
  By construction the sub-$\sigma$-algebras $\epbmeas\scmfst\calF$
  and $\epbmeas\scmsnd\calG$ are independent
  in the tensor product and the measure $\mu\otimes\nu$
  restricts to $\epbmeas\scmfst\mu$ on $\epbmeas\scmfst\calF$
  and $\epbmeas\scmsnd\nu$ on $\epbmeas\scmsnd\calG$,
  so $\epbmeas\scmfst(\calF,\mu)$ and $\epbmeas\scmsnd(\calG,\nu)$
  are independently combinable in $(A\times B,\calF\otimes\calG,\negligibles(\mu\otimes\nu))$.
  Hence $\epbmeas{X}{(\calF,\mu)}=\epbmeas{(\scmfst(X'\otimes Y')f)}{(\calF,\mu)}$
     and $\epbmeas{Y}{(\calG,\nu)}=\epbmeas{(\scmsnd(X'\otimes Y')f)}{(\calG,\nu)}$
     are independently combinable in $\Omega$
     by \cref{app:lem:indepcomb-pbstable,app:lem:epb-functorial}.

   Conversely, suppose (2),
   and let $(\calH,\rho)$ be the standardizable probability space
   on $\Omega$ witnessing independent combinability
   of $\epbmeas X{(\calF,\mu)}$
   and $\epbmeas Y{(\calG,\nu)}$.
   There is a $\StdProb$-map
   $g : (\Omega,\calH,\rho)\to(A,\calF,\mu)\otimes(B,\calG,\nu)$
   defined by $f(\omega)=(X(\omega),Y(\omega))$,
   measure-preserving because $\rho(g^{-1}(F\times G))
   = \rho(X^{-1}(F)\cap Y^{-1}(G))
   = \rho(X^{-1}(F))\rho(Y^{-1}(G))
   = \mu(F)\nu(G) = (\mu\otimes\nu)(F\times G)$
   for all rectangles $F\times G\in\calF\otimes\calG$.
   Moreover, $g$ satisfies $\scmfst\,g = X$ and $\scmsnd\,g = Y$,
   so $\forgetmu g$ is a $\StdMble$-map
   showing that $[A,X]$ and $[B,Y]$ factor through the tensor product
   $(A,\calF,\negligibles(\mu))\otimes(B,\calG,\negligibles(\nu))$.

   It only remains to show, in case (1) and (2) both hold,
   that $\epbmeas{((X'\otimes Y')f)}(\calF\otimes\calG,\mu\otimes\nu)$
   is the independent combination; i.e.,
   that it is the $\gkrmorder$-smallest standardizable probability space on $\Omega$
   witnessing the independent combinability
   of $\epbmeas X{(\calF,\mu)}$
   and $\epbmeas Y{(\calG,\nu)}$.
   Suppose $(\calH,\rho)$ is another such witness.
   The map $g : (\Omega,\calH,\rho)\to(A,\calF,\mu)\otimes(B,\calG,\nu)$
   constructed above
   is equal to $(X'\otimes Y')f$
   as a set-function,
   since $\scmfst\,g = X = \scmfst(X'\otimes Y')f$
   and $\scmsnd\,g = Y = \scmsnd(X'\otimes Y')f$.
   This shows $(X'\otimes Y')f$ is measure-preserving
   as a map with domain $(\Omega,\calH,\rho)$,
   so
   $\epbmeas{((X'\otimes Y')f)}{(\calF\otimes\calG,\mu\otimes\nu)}
   \gkrmorder (\calH,\rho)$ as required.
\end{proof}

\begin{lemma} \label{app:lem:indepable-corresp}
  The $\preauti$-set $\indepable$ is \aagoodname{}
  corresponding to
  the \goodsheaf{} $\pspcs\shotimes\pspcs$
  across the equivalence in \cref{app:thm:nominal}.
\end{lemma}
\begin{proof}
  By \cref{app:lem:pspcs-day-conv-char},
  each $(\pspcs\shotimes\pspcs)(\Omega)$
  is a subset of $\pspcs\Omega\times\pspcs\Omega$
  consisting of the pairs of probability spaces on $\Omega$
  that factor through a tensor product.
  Thus the corresponding \agoodname{}
  is a subset of the \agoodname{}
  $\gpspcs\times\gpspcs$
  of pairs of standardizable probability spaces
  with finite footprint.
  To determine this subset, it suffices to compute
  the image of the inclusion $\pspcs\shotimes\pspcs\hookrightarrow\pspcs\times\pspcs$
  across the equivalence in \cref{app:thm:nominal}.
  For any $\Omega\in\StdMble$
  and $\StdMble$-map $p : \hcube\to\Omega$ with finite footprint,
  this inclusion sends a pair $((\calG,\mu),(\calH,\nu))$
  of probability spaces on $\Omega$ that factor through a tensor product
  to $(\epbmeas p{(\calG,\mu)},\epbmeas p{(\calH,\nu)})$,
  a pair of independently-combinable probability spaces on $\hcube$ with
  finite footprint.
  This hits every pair of independently-combinable probability spaces
  with finite footprint by
  \cref{app:lem:factor-through-tensor-iff-indep-combable},
  so the image of the inclusion $\pspcs\shotimes\pspcs\hookrightarrow\pspcs\times\pspcs$
  corresponds to $\indepable$ as claimed.
\end{proof}

\begin{definition} \label{app:def:gpspcs-join}
  Let $(\calG,\mu)$ and $(\calH,\nu)$
  be two standardizable probability spaces on $\hcube$
  with finite footprint.
  If $(\calG,\mu)$ and $(\calH,\nu)$
  are independently combinable,
  let $\gpcmjoin((\calG,\mu),(\calH,\nu))$
  be their independent combination.
\end{definition}

\begin{lemma} \label{app:lem:icomb-stablepb}
  Let $f : (X,\calF,\calN)\to(Y,\calG,\calM)$
  be a $\StdMble$-map.
  If $(\calL,\rho)$ is the independent combination
  of $(\calH,\mu)$ and $(\calK,\nu)$
  then 
  $\epbmeas f{(\calL,\rho)}$ is the independent combination
  of $\epbmeas f{(\calH,\mu)}$ and $\epbmeas f{(\calK,\nu)}$.
\end{lemma}
\begin{proof}
  Let $h : Y\to \forgetmu A$ and $k : Y\to\forgetmu B$
  be the $\StdMble$-maps witnessing standardizability
  of $(\calH,\mu)$ and $(\calK,\nu)$,
  so $(\calH,\mu) = \epbmeas hA$ and $(\calK,\nu)=\epbmeas kB$.
  By \cref{app:lem:factor-through-tensor-iff-indep-combable},
  the independent combination $(\calL,\rho)$
  is the pullback $\epbmeas p{(A\otimes B)}$,
  where $p : Y\to\forgetmu (A\otimes B)$
  is the map defined by $p(y) = (h(y),k(y))$.
  Now $\epbmeas f{(\calH,\mu)}$ and $\epbmeas f{(\calK,\nu)}$
  arise via pullback from $hf$ and $kf$ respectively,
  so by \cref{app:lem:factor-through-tensor-iff-indep-combable}
  again their independent combination
  is $\epbmeas q{(A\otimes B)}$
  where $q : X\to\forgetmu (A\otimes B)$
  is the map defined by $q(x) = (h(f(x)),k(f(x)))$.
  By definition $q = pf$, so
  $\epbmeas q{(A\otimes B)} 
  = \epbmeas{(pf)}{(A\otimes B)}
  = \epbmeas{f}{\epbmeas p{(A\otimes B)}}
  = \epbmeas f{(\calL,\rho)}$
  as claimed.
\end{proof}

\begin{theorem} \label{app:thm:pcmjoin-corresp}
  The operation $\gpcmjoin$ is a map $\indepable\to\gpspcs$
  of \agoodname{s} corresponding to
  the map $\pcmjoin : \pspcs\shotimes\pspcs\to\pspcs$
  of \goodsheaves{}
  across the equivalence in \cref{app:thm:nominal}.
\end{theorem}
\begin{proof}
  Fix $\Omega\in\StdMble$ and $p : \hcube\to\Omega$
  a $\StdMble$-map with finite footprint.
  Let $((\calF',\mu'),(\calG',\nu'))\in\pspcs\shotimes\pspcs$
  be a pair of standardizable probability spaces
  on $\Omega$ that factor through a tensor product,
  so $(\calF',\mu') = \epbmeas {(X\,\scmfst\,f)}{(\calF,\mu)}$
  and $(\calG',\nu') = \epbmeas {(Y\,\scmsnd\,f)}{(\calG,\nu)}$
  for standard probability spaces $(A,\calF,\mu),(B,\calG,\nu)$
  and standard enhanced measurable spaces $\Omega_1,\Omega_2$
  and $\StdMble$-maps 
  $f : \Omega\to\Omega_1\otimes\Omega_2$
  and $X:\Omega_1\to A$ and $Y:\Omega_1\to B$.
  The map $\pcmjoin$ sends $((\calF',\mu'),(\calG',\nu'))$
  to $\epbmeas {((X\otimes Y)f)}{(\calF\otimes\calG,\mu\otimes\nu)}$
  by \cref{app:def:pspcs-tensor-join-note},
  their independent combination
  by \cref{app:lem:factor-through-tensor-iff-indep-combable},
  so $\epbmeas p{\epbmeas {((X\otimes Y)f)}{(\calF\otimes\calG,\mu\otimes\nu)}}$
  is the independent combination of $\epbmeas p{(\calF',\mu')}$
  and $\epbmeas p{(\calG',\nu')}$ by \cref{app:lem:icomb-stablepb}.
  Putting this together gives
  $\epbmeas p{\pcmjoin((\calF',\mu'),(\calG',\nu'))}
  = \gpcmjoin(\epbmeas p{(\calF',\mu')},\epbmeas p{(\calG',\nu')})$,
  so $\pcmjoin$ corresponds to $\gpcmjoin$ across the equivalence
  in \cref{app:thm:nominal} as claimed.
\end{proof}

\begin{theorem} \label{app:thm:krm-corresp}
The PDM $(\pspcs,\pcmjoin,\pcmunit,\krmorder)$
internal to \goodsheaves{}
corresponds to the PDM $(\gpspcs,\gpcmjoin,\gpcmunit,\gkrmorder)$
internal to \agoodname{s}
across the equivalence in \cref{app:thm:nominal}.
\end{theorem}
\begin{proof}
  \cref{app:lem:pspcs-corresp} shows $\pspcs$ corresponds to $\gpspcs$,
  \cref{app:thm:pcmjoin-corresp} shows $\pcmjoin$ corresponds to $\gpcmjoin$,
  \cref{app:lem:pcmunit-corresp} shows $\pcmunit$ corresponds to $\gpcmunit$,
  and \cref{app:lem:krmorder-corresp} shows $\krmorder$ corresponds to $\gkrmorder$.
\end{proof} \fi

\end{document}